\documentclass[a4paper, 11pt]{article}
\usepackage[usenames]{xcolor}
\usepackage{xifthen}

\def\fontsettingup{2} 

\usepackage{amsmath, amsthm, amssymb}
\usepackage[T1]{fontenc}
\ifthenelse{\fontsettingup = 1}{ \usepackage{eulerpx, eucal, tgpagella}   }{}
\ifthenelse{\fontsettingup = 2}{ \usepackage{mathpazo, eufrak, tgpagella} }{}

\usepackage{hyperref}
\hypersetup{
  colorlinks=true,
  linkcolor=blue,
  filecolor=magenta,
  urlcolor=cyan,
}

\usepackage{algorithm2e}
\usepackage{cleveref}

\usepackage{tikz}
\usepackage{pgfplots}
\usepgfplotslibrary{fillbetween}
\pgfplotsset{compat=newest}

\usepackage[a4paper]{geometry}
\newgeometry{
textheight=9in,
textwidth=6.5in,
top=1in,
headheight=14pt,
headsep=25pt,
footskip=30pt
}

\newtheorem{theorem}{Theorem}
\newtheorem{observation}[theorem]{Observation}
\newtheorem{claim}[theorem]{Claim}
\newtheorem*{claim*}{Claim}
\newtheorem{condition}[theorem]{Condition}

\newtheorem{fact}[theorem]{Fact}
\newtheorem{lemma}[theorem]{Lemma}

\newtheorem{corollary}[theorem]{Corollary}
\theoremstyle{definition}

\newtheorem{definition}[theorem]{Definition}
\newtheorem{remark}[theorem]{Remark}
\newtheorem*{remark*}{Remark}

\ifthenelse{\fontsettingup = 1}{
  \def\*#1{\mathbf{#1}} 
  \def\+#1{\mathcal{#1}} 
  \def\-#1{\mathrm{#1}} 
  \def\^#1{\mathbb{#1}} 
  \def\!#1{\mathfrak{#1}} 
}{}

\ifthenelse{\fontsettingup = 2}{
  \def\*#1{\boldsymbol{#1}} 
  \def\+#1{\mathcal{#1}} 
  \def\-#1{\mathrm{#1}} 
  \def\^#1{\mathbb{#1}} 
  \def\!#1{\mathfrak{#1}} 
}{}

\DeclareMathOperator*{\argmin}{arg\,min}

\def\oPr{\mathbf{Pr}}
\renewcommand{\Pr}[2][]{ \ifthenelse{\isempty{#1}}
  {\oPr\left[#2\right]}
  {\oPr_{#1}\left[#2\right]} } 

\def\oE{\mathbb{E}}
\newcommand{\E}[2][]{ \ifthenelse{\isempty{#1}}
  {\oE\left[#2\right]}
  {\oE_{#1}\left[#2\right]} }

\def\oVar{\mathbf{Var}}
\newcommand{\Var}[2][]{ \ifthenelse{\isempty{#1}}
  {\oVar\left[#2\right]}
  {\oVar_{#1}\left[#2\right]} }

\def\oEnt{\mathbf{Ent}}
\newcommand{\Ent}[2][]{ \ifthenelse{\isempty{#1}}
  {\oEnt\left[#2\right]}
  {\oEnt_{#1}\left[#2\right]} }

\newcommand{\DKL}[2]{\-D_{\-{KL}}\left(#1 \parallel #2\right)}
\newcommand{\DTV}[2]{\-D_{\mathrm{TV}}\left({#1},{#2}\right)}

\newcommand{\e}{\mathrm{e}}
\renewcommand{\epsilon}{\varepsilon}
\renewcommand{\emptyset}{\varnothing}
\newcommand{\norm}[1]{\left\Vert#1\right\Vert}

\newcommand{\tuple}[1]{\left(#1\right)} \newcommand{\eps}{\varepsilon}
\newcommand{\inner}[2]{\left\langle #1,#2\right\rangle}
\newcommand{\tp}{\tuple}

\newcommand{\abs}[1]{\left\vert#1\right\vert}
\newcommand{\ctp}[1]{\left\lceil#1\right\rceil}
\newcommand{\ftp}[1]{\left\lfloor#1\right\rfloor}
\newcommand{\Ind}[1]{\-{Ind}_{#1}}

\newcommand{\FD}[2][\theta]{P^{\-{FD}}_{#1, #2}}
\newcommand{\FDdown}[2][\theta]{P^{\downarrow}_{#1, #2}}
\newcommand{\FDup}[2][\theta]{P^{\uparrow}_{#1, #2}}
\newcommand{\Tmix}[1]{t_{\-{mix}}\tp{#1}}
\newcommand{\GD}[1][\nu]{P^{\-{GD}}_{#1}}
\newcommand{\sgap}[2][]{\gamma_{#1}\tp{#2}}

\newcommand{\PB}{P_{\-B}}
\newcommand{\PBdown}{\PB^\downarrow}
\newcommand{\PBup}{\PB^\uparrow}

\newcommand{\BHC}{\ensuremath{\#\mathrm{BipHardCore}}}

\begin{document}
\title{Uniqueness and Rapid Mixing\\
 in the Bipartite Hardcore Model}
\date{}
\author{Xiaoyu Chen\thanks{State Key Laboratory for Novel Software Technology, Nanjing University, 163 Xianlin Avenue, Nanjing, Jiangsu Province, China. \textnormal{E-mails: \url{chenxiaoyu233@smail.nju.edu.cn}, \url{liu@nju.edu.cn}, \url{yinyt@nju.edu.cn}}}
\and 
Jingcheng Liu\footnotemark[1]
\and 
Yitong Yin\footnotemark[1]
}

\pagenumbering{roman}
\maketitle
\begin{abstract}
	We characterize the uniqueness condition in the hardcore model for bipartite graphs with degree bounds only on one side, and provide a nearly linear time sampling algorithm that works up to the uniqueness threshold.
	We show that the uniqueness threshold for bipartite graph has almost the same form of the tree uniqueness threshold for general graphs, except with degree bounds only on one side of the bipartition.
	The hardcore model is originated in statistical physics for modeling equilibrium of lattice gas. Combinatorially, it can also be seen as a weighted enumeration of independent sets. Counting the number of independent sets in a bipartite graph (\#BIS) is a central open problem in approximate counting. Compared to the same problem in a general graph, surprising tractable regime have been identified that are believed to be hard in general.
	This is made possible by two lines of algorithmic approach: the high-temperature algorithms starting from Liu and Lu (STOC 2015), and the low-temperature algorithms starting from Helmuth, Perkins, and Regts (STOC 2019).

	In this work, we study the limit of these algorithms in the high-temperature case. 
	Our characterization of the uniqueness condition is obtained by proving decay of correlations for arguably the best possible regime, which involves locating fixpoints of multivariate iterative rational maps and showing their contraction. 
	Interestingly, we are able to show that a regime that was considered ``low-temperature'' is actually well within the uniqueness (high-temperature) regime.
	We also give a nearly linear time sampling algorithm based on simulating field dynamics only on one side of the bipartite graph that works up to the uniqueness threshold.  
	Our algorithm is very different from the original high-temperature algorithm of Liu and Lu (STOC 2015), and it makes use of a connection between correlation decay and spectral independence of Markov chains. Along the way, we also build an explicit connection between the very recent developments of negative-fields stochastic localization schemes and field dynamics.
	Last but not the least, we are able to show that the standard Glauber dynamics on both side of the bipartite graph mixes in polynomial time up to the uniqueness. Remarkably, this is a model where both the total influence and the spectral radius of the adjacency matrix can be unbounded, yet we are able to prove mixing time bounds through the framework of spectral independence.
\end{abstract}

\newpage
\setcounter{tocdepth}{2}
\tableofcontents

\clearpage

\pagenumbering{arabic}
\section{Introduction}

Counting the number of independent sets in a bipartite graph (\#BIS) is arguably one of the most important open problem in the field of approximate counting.
Many natural counting problems are known to have the same complexity as \#BIS (\#BIS-equivalent) or at least as hard (\#BIS-hard) under approximation-preserving reductions (AP-reductions), while \#BIS itself is a complete problem for a logically defined class known as \#RH$\Pi_1$~\cite{dyer2004relative}.
These problems arise from the study of counting CSPs~\cite{dyer2010approximation, dyer2012complexity, bulatov2013expressibility, galanis2021approximating}, spin systems in statistic physics~\cite{goldberg2007complexity, goldberg2012approximating, liu2014complexity, goldberg2015complexity, galanis2016ferromagnetic, cai2016bis}, and combinatorial settings~\cite{dyer2004relative, chebolu2012complexity}.
To name a few, these include counting the number of down-sets in a partial order system, counting stable matchings, counting the number of $q$-colorings in a bipartite graph, computing the partition function of the ferromagnetic Potts model~\cite{goldberg2012approximating} and ferromagnetic Ising model with mixed external fields~\cite{goldberg2007complexity}.
\#BIS also plays a major role in complexity classification for Boolean counting CSP as an intermediate class.
It is conjectured that neither does \#BIS admit fully polynomial-time randomized approximation scheme (FPRAS), nor is it as hard as \#SAT~\cite{dyer2004relative}.

For a general graph, approximately counting the number of independent sets is a well studied problem.
Its weighted version, the hardcore model, was originally used in statistical physics to model equilibrium of lattice gas.
The partition function of the hardcore model coincides with the weighted counting of independent sets, which we define next.
Given a graph $G$ and fugacity $\lambda > 0$, the Gibbs distribution of the hardcore model on $G$, denoted by  $\mu$, is given by
\begin{align*}
	\forall S \in \Ind{G}, \quad \mu(S) := \frac{\lambda^{\abs{S}}}{Z(G)},
\end{align*}
where $\Ind{G}$ is the family of independent sets of $G$, and the normalizing factor $Z(G) := \sum_{S \in \Ind{G}} \lambda^{\abs{S}}$ is the \emph{partition function} of the hardcore model.
Since then, the hardcore partition function has also found applications in the Lov\'asz local lemma~\cite{Shearer85} and its algorithmic counterparts~\cite{KS11}.
The approximability of $Z(G)$ is well understood. If the maximum degree of $G$ is $\Delta$, then there is a critical threshold $\lambda_c(\Delta):= (\Delta-1)^{\Delta-1} / (\Delta-2)^\Delta$: 
for $\lambda < \lambda_c(\Delta) $, 
efficient algorithms for approximating $Z(G)$ have been known through deterministic approximate counting~\cite{weitz2006counting, patel2017deterministic} which run in polynomial-time when $\Delta$ is bounded by a constant, 
and more recently through rapidly mixing Markov chains~\cite{anari2020spectral,chen2020rapid,chen2021optimal,chen2021rapid,anari2022entropic,chen2022localization,chen2022optimal} which can run a lot faster, especially when $\Delta$ is large; 
while for $\lambda > \lambda_c(\Delta)$,  an FPRAS for $Z(G)$ does not exist unless NP=RP~\cite{sly2010computational,sly2014counting,galanis2016inapproximability}. The threshold $\lambda_c(\Delta)$ is known as the \emph{uniqueness threshold}, as it corresponds to the uniqueness of Gibbs measure in infinite $(\Delta-1)$-ary trees.


When $G$ is restricted to bipartite graphs $G=(L \cup R, E)$, the problem of approximating a bipartite hardcore partition function gets more interesting.
It can also be seen as a weighted version of \#BIS, and it will be the main subject of this paper. We will refer to this problem by \BHC{}.
For the hardness side, if one allows for complex fugacity $\lambda\in \mathbb{C}$ outside a cardioid-shaped region, 
Bez\'{a}kov\'{a}, Galanis, Goldberg, and \v{S}tefankovi\v{c} show that there are \#P-hardness~\cite{bezakova2020inapproximability} even for approximating the complex norm, or the complex argument of a bipartite hardcore partition function.
However, we note that it is highly unlikely that the \#P-hardness results can be extended to the case of positive real fugacity $\lambda$, which is the main focus of this paper, as they are NP-easy via Valiant and Varzirani~\cite{valiant1985np}.
The main challenge of producing an NP-hardness proof for \BHC{} with real fugacity (and \#BIS in particular) is that many optimization problems become easy in a bipartite graph. In particular, the problem of approximately finding the largest independent set in a general graph is hard, but its bipartite counterpart become easy.
For $\lambda>0$, there are known \#BIS-hardness established by Cai, Galanis, Goldberg, Guo, Jerrum, {\v S}tefankovi\v{c}, and Vigoda~\cite{cai2016bis}, which show that when $\lambda > \lambda_c(\Delta)$, $\BHC(\lambda)$ become \#BIS-hard on bipartite graphs of maximum degree $\Delta \ge 3$. 
They also formulated sufficient conditions for their \#BIS-hardness gadget construction to succeed, the unary symmetry breaking and the (balanced) nearly-independent phase-correlated spins property, which roughly corresponds to supporting a balanced mixture of two phases.
To some extent, this can be viewed as characteristic for the ``hardest'' \#BIS instances.
Such characteristic for ``hardness'' is well-known in the literature, and they have appeared formally as concrete algorithmic barriers before: it was shown that any \emph{local} Markov chain Monte-Carlo (MCMC) based algorithm that moves via updating $o(n)$ vertices mixes slowly on $\Delta$-regular random bipartite graph when $\lambda > \lambda_c(\Delta)$ with high probability~\cite{dyer2002counting, mossel2009hardness}.
There are later works proposing interesting Markov chains working on subsets of edges rather than vertices, but they were also shown to mix slowly on a different family of bipartite graphs with a similar characteristic~\cite{ge2012graph, goldberg2012counterexample}.

The algorithmic fronts on \BHC{} and \#BIS have been more active. 
There are mainly two lines of algorithmic approaches: the \emph{high-temperature} ones via correlation decay, and the \emph{low-temperature} ones via cluster expansion and the polymer models.
Based on the method of correlation decay, Liu and Lu~\cite{liu2015fptas} gave an algorithm for \#BIS when $\lambda = 1$ and $\Delta_L \leq 5$, where, unlike the general case, here $\Delta_L$ is the maximum degree only on one side of the bipartition $L$.
The other line of algorithmic approach emerging recently, starting from Helmuth, Perkins, and Regts~\cite{helmuth2020algorithmic}, is based on the cluster expansion and the  polymer model (or the more sophisticated contour model) in Pirogov-Sinai theory from statistical physics.
Since then, several improvements, extensions, and generalizations give algorithms for estimating the partition function of the hardcore model on bipartite expander graph with large fugacity~\cite{jenssen2020algorithms, chen2021fast, jessen2022approximately}; on random regular bipartite graph for sufficiently large maximum degree $\Delta$ and fugacity $\lambda = \widetilde{\Omega}(1/\Delta)$~\cite{liao2019counting, jenssen2020algorithms, chen2022sampling}; on unbalanced bipartite graph with large fugacity~\cite{cannon2020counting, friedrich2023polymer, blanca2022fast}; on $d$-regular bipartite graph that runs in subexponential time provided that $d=\omega(1)$~\cite{jessen2022approximately}.

Both lines of algorithmic approaches have failed to give efficient algorithm for \BHC{} or \#BIS. As discussed earlier, one characteristic of the ``hardest'' \#BIS instances is the ability to support a balanced mixture of two phases. Essentially, for a non-random bipartite graph, both lines of algorithmic approaches give new sufficient conditions when the system does not support a balanced mixture of two phases, leading to new tractable instances. Specifically, the high-temperature ones provide new sufficient condition for uniqueness with degree bounds on one side; The low-temperature ones mainly show that the imbalance in the parameters, or the expansion of the graph can introduce a dominating phase, and then a cluster expansion around the dominant phase can provide good approximations. 
This motivates the following question: what is the limit of these algorithmic approaches? 

In this work, we resolve the question for the high-temperature case, by giving a complete characterization of the uniqueness condition in \BHC{} with degree bounds only on one side.
We complement the new characterization with a nearly linear time MCMC based sampler that works up to the uniqueness threshold.

\paragraph{Uniqueness condition in the bipartite hardcore model:}
We show that the correlation decay property holds if and only if $\lambda < \lambda_c(\Delta_L)$ where $\Delta_L$ is the degree bound only on one side.
The threshold here is exactly the tree-uniqueness critical threshold for hardcore model on general graph with maximum degree $\Delta$.
This shows that the tractable regime in Liu and Lu~\cite{liu2015fptas}, with $\lambda = 1$ and $\Delta_L \leq 5$ being the degree bound on one side, is not a coincidence.
And it confirms a heuristic that might have been suggested by the result in~\cite{liu2015fptas} in the non-weighted case, 
that the correlation decay property is guaranteed by the uniqueness condition on one side.
More surprisingly,
it seems that one should never expect to exploit such one-sided criterion of uniqueness any further:
if further allowing fugacity to be different on two sides of the bipartition, that is, $\lambda_L \neq \lambda_R$,
then the class of instances satisfying the one-sided uniqueness condition $\lambda_L < \lambda_c(\Delta_L)$ may still exhibit non-uniqueness and can be \#BIS-hard. 
Therefore,
although the correlation decay property is critically captured by the uniqueness condition on one side,
such criticality  does not hold
oblivious to the other side.
To the opposite, it crucially relies on that the two sides are in the same hardcore model with the same fugacity. 
We give a characterization of the uniqueness
in a more refined setting with arbitrary lopsided fugacity $\lambda_L$, $\lambda_R$ (formally stated in \Cref{thm:delta-unique}), which subsumes the setting with $\lambda_L=\lambda_R$ as special case.
Interestingly, we are able to identify a regime that was considered ``low-temperature'', and show that it actually lies well within the uniqueness (high-temperature) regime (see \Cref{rem:low-temp}).
%

The \BHC{} also seems to pose unique challenges that requires an approach that differs significantly from previous works for establishing correlation decay. In particular, our starting point is not the global contraction of multivariate rational maps under certain set of parameters. Instead, our uniqueness condition is defined first by requiring that all the fixpoints to be contractive. Then, we use implicit function theorems to locate the fixpoints for \emph{every} set of parameters (rather than just the fixpoints at the critical threshold), and show that the worst case contraction rate of the multivariate rational maps can be bounded by the contraction rate at the fixpoints determined by the given parameters. 
In particular, we have to do this for the parameters away from the uniqueness threshold.


\paragraph{A nearly linear time sampler up to uniqueness:}
We also give a nearly linear time sampling algorithm based on simulating the \emph{field dynamics} on one side of the \BHC{}. 
Our algorithm works for all $\lambda < \lambda_c(\Delta_L)$, where $\Delta_L$ is the degree bound only on one side.
The field dynamics is a Markov chain introduced by Chen, Feng, Yin and Zhang as a proxy for the analysis of the Glauber dynamics~\cite{chen2021rapid}.
Since then, it has found further applications in establishing optimal mixing time of Glauber dynamics (and its variants)~\cite{chen2021rapid, anari2022entropic, chen2022optimal, chen2022localization}, and the field dynamics itself has also been used in designing nearly linear time samplers~\cite{anari2022entropic, chen2022near}.
At a high level, we will first establish the decay of correlation property up to the uniqueness threshold, which has been an extremely important algorithmic tool by itself.
Then we adopt the notion of  {spectral independence} introduced by Anari, Liu, and Oveis-Gharan~\cite{anari2020spectral}.
This is a very important notion, through which numerous proofs of correlation decay has been successfully translated to proofs of rapid mixing of Markov chains.
In particular,  we follow a similar route developed by Chen, Liu, and Vigoda~\cite{chen2020rapid} to prove sharp spectral independence via {contraction}.
For \BHC{} however, we do not have bounded total influence, and we are only able to show spectral independence on one side of the bipartite graph.
To get a nearly linear time sampling algorithm, we have to also establish a stronger notion called {entropic independence} for the field dynamics, which is a novel form of entropy decay introduced by Anari, Jain, Koehler, Pham, and Vuong~\cite{anari2022entropic}.
Along the way, we also build an explicit connection between the field dynamics and the negative-fields stochastic localization scheme of Chen and Eldan~\cite{chen2022localization}.  
%
Remarkably, the case of $\Delta_L=2$ for \BHC{} is highly non-trivial. 
While we established the correlation decay property on one side of the bipartite hardcore model, when projected to the other side it is equivalent to a ferromagnetic Ising model in the so-called ``spin world'' representation. 
Indeed, if one allow fugacity to be different on two sides of the bipartite graph (that is, $\lambda_L \neq \lambda_R$), then any ferromagnetic Ising model can be equivalently represented by the bipartite hardcore model with $\Delta_L=2$.
Even though the ferromagnetic Ising model does not have decay of long range correlations in general, in the case of $\lambda_L=\lambda_R$ of the bipartite hardcore model, we are able to show that it does. 

\paragraph{Mixing of the standard Glauber dynamics up to uniqueness:}
With the new characterization of the uniqueness conditions and the decay of correlation that it implies, 
and by augmenting from a series of new tools for mixing times developed in aforementioned works and in~\cite{chen2021optimal},
we are also able to show polynomial mixing time bound for the Glauber dynamics up to the uniqueness threshold.
Remarkably, we are able to establish the mixing time bound through the spectral independence framework, despite the fact that model does not have bounded total influence, nor bounded spectral radius of the adjacency matrix (due to unbounded degree on the other side).
We also remark that these are very different algorithms from the one used in~\cite{liu2015fptas}, 
as it is an MCMC that has been successfully applied to \#BIS beyond what can be inherited from algorithms for general graphs. 
Furthermore,  such a rapid mixing result does not seem to follow from directly applying censoring inequality to the field dynamics, despite its monotonicity on one side of the bipartition. 
Our starting point is that the field dynamics mixes rapidly on only one side of the bipartite graph.
Then, we adapt the framework of approximate tensorization of variance~\cite{chen2021optimal,caputo2015approximate,cesi2001quasi} to the bipartite hardcore model and perform a comparison argument on variance decay,  and we are able to show that the field dynamics on one side can still be ``approximately tensorized'' into a single site Glauber dynamics on both side.
Another challenge in analyzing the standard Glauber dynamics on the bipartite hardcore model  arises due to the fact that we only have degree bounds on one side of the bipartition, while the other side can have vertices with unbounded degree.
Yet, we are able to show a mixing time bound that only depends on one side of the maximum degree.


\bigskip
Throughout the rest of the paper, we will only be interested in the degree bounds on one side, so we will simply write $\Delta$ for $\Delta_L$ unless otherwise stated.

\section{Main results and technical overview}
Our first result gives a tight characterization for the uniqueness of the bipartite hardcore model, in terms of a ``tree-uniqueness'' condition with degree bounds only on one side (that is to say, all the odd level of a tree). 
The classical notion of uniqueness condition for a general graph has been characterized on infinite $\Delta$-regular tree since Kelly~\cite{kelly1985stochastic}.
For technical convenience, we introduce the following notion of $\delta$-uniqueness to ensure that there is an explicit gap.  

\begin{definition} \label{def:lambda-d-unique}
	Let $\delta \in [0, 1)$ be any real number, and $d,w \in \^R_{>0}$.
  We will be interested in the fixpoints of the recurrence $F(x) := \lambda (1 + \lambda(1 + x)^{-w})^{-d}$.
  We say that $(\lambda, d) \in \^R^2_{>0}$ is \emph{$\delta$-unique} if for any $w \in \^R_{>0}$, all fixpoints $\hat{x} = F(\hat{x})$ of $F$ satisfy $F'(\hat{x}) \leq 1 - \delta$.
\end{definition}

As will be discussed in \Cref{sec:uniqueness}, our notion of $\delta$-uniqueness is defined with respect to degree bounds only on one side of the bipartite graph. 
The parameter $d+1$ correspond to the degree on one side, and $w+1$ correspond to the degree on the other side, where $w$ can be chosen arbitrarily. 
In contrast, the tree uniqueness threshold is defined over infinite $\Delta$-regular trees, with degree bounds on both side. 
We show that our one-sided uniqueness threshold on bipartite graph coincides with the uniqueness threshold on general graph, and further that the slackness $\delta$ is roughly equivalent. 
\begin{theorem} \label{thm:delta-unique-eq}
	Fix any $\Delta=d+1 \geq 3$ and any $\delta \in [0, 1)$,
    the pair $(\lambda, d)$ is $\frac{\delta}{10}$-unique if
    \[\textstyle \lambda \leq (1 - \delta) \lambda_c(\Delta) = (1 - \delta) \frac{(\Delta - 1)^{\Delta - 1}}{(\Delta - 2)^{\Delta}}. \]
\end{theorem}
The proof of this theorem is given in \Cref{sec:delta-unique} and \Cref{sec:SI}.
In \Cref{sec:delta-unique}, we will first handle the $\delta$-uniqueness regime for sufficiently small $\delta$ in \Cref{thm:solve-sys-eq}.
This gives a weaker form of the theorem that only holds for sufficiently small $\delta$.
Then, the proof will be completed in \Cref{sec:delta-unique-eq-7}, where we leverage the analysis of correlation decay to extend \Cref{thm:solve-sys-eq} for all $\delta \in (0, 1)$.
Along the way,  we also characterized the uniqueness condition for $\lambda\neq \alpha$ (see~\Cref{thm:delta-unique-lambda-d-alpha-w}), where $\lambda$ is the fugacity on one side of the bipartite graph, and $\alpha$ is the fugacity on the other side.

Our next result states that, by simulating the field dynamics, we can sample approximately from the bipartite hardcore distribution in nearly linear time, provided that the model satisfy $\delta$-uniqueness.
We measure distance between two distributions $\mu$ and $\nu$ over a finite space $\Omega$ in \emph{total variation distance (TV distance)}:
$  \DTV{\mu}{\nu} := \frac{1}{2} \sum_{X \in \Omega} \abs{\mu(X) - \nu(X)}$.
\begin{theorem} \label{thm:main-unique}
Fix a degree $\Delta = d+1 \ge 2$, $\delta\in (0,1)$, fugacity $\lambda > 0$, such that the pair $(\lambda, d)$ is $\delta$-unique, and $\eps>0$. 
	Then there is an algorithm that approximately samples, within TV distance $\epsilon$, from the hardcore distribution of all $n$-vertex  bipartite graphs of maximum degree $\Delta$ on one side, in time
  \begin{align*}
	 n \cdot \tp{\frac{\Delta \log n}{\lambda}}^{O(C/\delta)} \cdot \log^2(1/\epsilon),
  \end{align*}
  where 
  $C < \exp(15\e^2)$ is an absolute constant for $\Delta \ge 3$, and $C= (1 + \lambda)^{10}$ for $\Delta=2$.
\end{theorem}
\Cref{thm:main-unique} is proved in~\Cref{sec:main-unique} by combining results presented in~\Cref{sec:uniqueness,sec:SI-and-uniqueness,sec:ent-decay-field-general}. 

\begin{remark}
Our algorithm in \Cref{thm:main-unique} works in arguably the best possible regime. This is because when $\lambda > \lambda_c(\Delta)$, 
$\BHC{(\lambda)}$ becomes \#BIS-hard~\cite{cai2016bis}.
  We also note that when $\Delta \leq 1$, the bipartite graph $G$ is a forest, and the partition function of the hardcore model encoded by $G$ and $\lambda$ can be computed exactly in polynomial time.
\end{remark}

Last but not least, through comparison arguments between the field dynamics on one side and the  single-site Glauber dynamics on both side, we also derive a nearly cubic mixing time bound for the standard Glauber dynamics.
Let $P$ be a Markov chain with stationary distribution $\mu$ with support $\Omega(\mu)$. 
The \emph{mixing time} is defined by
$  \Tmix{\epsilon} := \min\left\{t \mid \max_{X \in \Omega(\mu)} \DTV{P^t(X, \cdot)}{\mu} \leq \epsilon \right\}$.
\begin{theorem} \label{thm:GD-mu-mix-intro}
Fix a degree $\Delta =d+1\ge 3$, $\delta\in (0,1)$, fugacity $\lambda \in( 0, (1-\delta)\lambda_c(\Delta))$, and $\eps>0$. 
	Then the mixing time for the  standard Glauber dynamics $(X_t)_{t\in\mathbb{N}}$ 
	for the hardcore distribution of of all $n$-vertex bipartite graphs of maximum degree $\Delta$ on one side,
	is bounded as
  \begin{align*}
   \Tmix{\epsilon}\le \tp{\frac{\Delta \log n}{\lambda}}^{O(C/\delta)} \cdot n^2 \cdot \tp{n \log \frac{1 + \lambda}{\min\{1, \lambda\}} + \log \frac{1}{\epsilon}},
  \end{align*}
  where 
  $C <  \exp(15\e^2)$ is an absolute constant. 

  In addition, when $\Delta=2$ and $(\lambda, d)$ is $\delta$-unique, the above mixing bound also holds with $C=  (1 + \lambda)^{10}$.
\end{theorem}
We note that $\Delta$ is only the degree bound on one side,
and the hardcore distribution does not have bounded total influence, nor does the graph have bounded spectral radius, due to unbounded degree on the other side.
\Cref{thm:GD-mu-mix-intro} proves the rapid mixing up to criticality without establishing a total influence bound or a spectral independence bound for the hardcore distribution on the bipartite graph.
	Instead, we study a ``one-sided'' distribution which has a much better total influence bound.
We believe this will serve as an important example in understanding the relationship between the total influence and the spectral independence framework.
%

In the case of $\Delta=2$, we have $\lambda_c(\Delta)=\infty$. Indeed, we show that the hardcore distribution with $\Delta=2$ can be reduced to a ferromagnetic Ising model, which is always unique (see \Cref{sec:ising-case}).
\Cref{thm:GD-mu-mix-intro} will be proved in \Cref{sec:main-GD}.

In the following, we give a technical overview of our proofs.
\subsection{Tree-uniqueness in the bipartite hardcore model}
\label{sec:uniqueness}
To characterize the uniqueness of the bipartite hardcore model, we consider a more general setting where different fugacities are allowed on two sides of the bipartite graph.
Given a bipartite graph $G = ((L, R), E)$ with maximum degree $\Delta$ on $L$ and maximum degree $W$ on $R$,
let $\lambda > 0$ be the fugacity on $L$ and $\alpha > 0$ be the fugacity on $R$.
Then, the hardcore distribution $\mu$ on $G$ is given by:
\begin{align*}
  \forall S \subseteq L \cup R, \quad \mu_G(S) \propto \lambda^{\abs{S\cap L}} \alpha^{\abs{S\cap R}} \cdot \*1[S \in \Ind{G}].
\end{align*}

As is typically the case, we study the ``tree-uniqueness'' by considering $\mu_G$ when $G$ is a tree $T$ rooted at $r$, and $T$ can have different branching numbers on the even level and on the odd level.
Suppose $r$ has $d$ children and for $i \in [d]$, the $i$-th child of $r$ has $w_i$ children.
For $i\in [d]$, we denote $u_i$ as the $i$-th child of $r$, and $T_i$ as the subtree rooted at $u_i$.
Then, for $j\in [w_i]$, we denote $u_{ij}$ as the $j$-th child of $u_i$.
Without loss of generality, we assume that $r \in L$, and thus it has fugacity $\lambda$.
Due to the independence of subtrees, we have an easy recurrence for calculating the marginal occupation ratio:
\begin{align*}
	R_r  := \frac{\Pr[S \sim \mu]{ r \in S} }{\Pr[S \sim \mu]{ r \not\in S}} 
	 = \lambda \prod_{i=1}^d \Pr[S_i \sim \mu_{T_i}]{ u_i \not\in S}  
  = \lambda \prod_{i=1}^d \tp{1 + R_{u_i}}^{-1} 
  = \lambda \prod_{i=1}^d \tp{1 + \alpha \prod_{j=1}^{w_i} \tp{1 + R_{u_{ij}}}^{-1}}^{-1},
\end{align*}
where for any vertex $v$, we use $R_v$ to denote the marginal ratio of $v$ in the subtree rooted at $v$.
\newcommand{\dwTree}{\^T_{\Delta, W}}
We consider the uniqueness of the hardcore Gibbs measure on the infinite $(\Delta, W)$-regular tree $\dwTree$, in which the recurrence is simplified to
\begin{align*}
  F(x) := \lambda (1 + \alpha (1 + x)^{-w})^{-d},
\end{align*}
where $d := \Delta - 1$ and $w := W - 1$ are the branching numbers.
The critical condition that governs the uniqueness of hardcore Gibbs measure on the infinite bi-regular tree is that all fixpoints should be attractive fixpoints.

We start by considering the tree-uniqueness threshold for $(d,w)$ infinite bi-regular tree in the case of $0$-uniqueness, which is the classical non-gapped notion of uniqueness.
\begin{definition} \label{def:lambda-d-alpha-w-unique0}
  Let $d,w \in \^R_{>0}$.
We say a tuple $(\lambda, d, \alpha, w) \in \^R^4_{>0}$ is $0$-\emph{unique},  if for all $\hat{x} \geq 0$ satisfying $F(\hat{x}) = \hat{x}$, it holds that $F'(\hat{x}) \leq 1 $.
We say a  tuple $(\lambda, d, \alpha) \in \^R^3_{>0}$ is $0$-\emph{unique}, if $(\lambda,d,\alpha,w)$ is 0-unique for all $w\in\mathbb{R}_{>0}$.
We will also refer to being $0$-unique as simply being \emph{unique}.
\end{definition}
Indeed, we implicitly obtained uniqueness criteria for the tuple $(\lambda, d, \alpha, w)$ in~\Cref{sec:delta-unique}, but they are only defined as implicit functions and consist of disconnected intervals in general. 
In this work, our main focus is when the degree $w$ can be chosen arbitrarily, including fractional degrees.
Then, the uniqueness regime become connected
and can be explicitly stated.
Specifically, the $0$-uniqueness condition for the tuple $(\lambda,d,\alpha)$ can be stated in terms of the function  ${\boldsymbol{\hat\lambda}}(d,w)$:
\begin{align} \label{eq:A0}
  {\boldsymbol{\hat\lambda}}(d, w) := \frac{ d^w (w + 1)^{w+1}}{(d w - 1)^{w+1}}.
\end{align}

\begin{theorem} \label{thm:delta-unique-lambda-d-alpha-w}
	Fix any $d,w \in \^R_{>0}$ such that $d \geq 1$ and $d w > 1$. 
	Let $\lambda_c(d,w):=  {\boldsymbol{\hat\lambda}}(w,d)$, and $\alpha_c(d,w):=  {\boldsymbol{\hat\lambda}}(d,w)$.
	Then for any $\lambda \ge \lambda_c(d,w), \alpha \le \alpha_c(d,w)$, the tuple $(\lambda,d,\alpha)$ is unique.

	Furthermore, for any fixed $d\ge 1$, such a threshold-pair $(\lambda_c(d,w), \alpha_c(d,w))$ is the best possible: 
	if we fix $\alpha = \alpha_c(d,w)$, then $\lambda_c(d,w)$ is the smallest possible for the $\lambda$ such that the tuple $(\lambda,d,\alpha)$ is unique;
	and likewise, if we fix $\lambda = \lambda_c(d,w)$, then $\alpha_c(d,w)$ is the largest possible $\alpha$ for $(\lambda,d,\alpha)$ being unique.
\end{theorem}
The proof of this theorem will be deferred to \Cref{sec:delta-unique-eq} (specifically, in \Cref{sec:0-unique-eq}). 
The following corollary immediately follows from the theorem.
\begin{corollary}
	Fix any $(\lambda,d,\alpha) \in \^R_{>0}^3$. 
	If $d \geq 1$ and there exists $w$ such that $d w > 1$, $\lambda \ge \lambda_c(d,w)$, and $\alpha \le \alpha_c(d,w)$, then the tuple $(\lambda,d,\alpha)$ is unique.
%
	In particular, for $\lambda=\alpha = \lambda_c(d+1)= \frac{d^d}{(d-1)^{(d+1)}}$, one can choose $w=d$ so that $\lambda_c(d,d) = \alpha_c(d,d) = \lambda_c(d+1)$.
	Hence, the tuple $(\lambda,d,\alpha)$ is always unique.
\end{corollary}

\begin{remark}
\label{rem:low-temp}
We highlight a surprising comparison between our characterization of uniqueness ({high-temperature}) regime and the state of the art results for ``{low-temperature}'' regime. 
In the following, we compare our uniqueness threshold for infinite $(d,w)$-ary tree, and the state of the art low-temperature algorithmic results on a bi-regular bipartite graph~\cite{cannon2020counting,blanca2022fast,friedrich2023polymer}. 

	Fix any pair of degrees $d,w$ such that $w \ge d$, and let $\alpha= \alpha_c(d,w)$ on a $(d+1,w+1)$ bi-regular bipartite graph. 
	If both $d,w$ are sufficiently large, one can show that $\alpha_c(d,w) = \Theta (1/d)$ and $\lambda_c(d,w) = \Theta(1/w)$. 
	For fixed $\alpha=\alpha_c(d,w)$, the condition required by current low-temperature algorithms~\cite{cannon2020counting,blanca2022fast,friedrich2023polymer}, denoted by $\lambda_{\mathrm{low}}$, is asymptotically $\lambda \ge \lambda_{\mathrm{low}} = \Theta(w)$.
	On the other hand, we have tree-uniqueness as soon as $\lambda\ge \lambda_c(d,w) = \Theta(1/w)$.
	This suggests that there could be a significant portion of regime that was considered ``low-temperature'', actually lies well within the ``high-temperature'' regime.
	Interestingly, Cannon and Perkins~\cite{cannon2020counting} also showed that within their regime, pairwise correlations decay exponentially fast.

Below, we also include a numerical plot for the two thresholds as a function of $w$ in the case of $d=2$ and $d=3$ in~\Cref{fig:low-temp}. 
	As $\lambda_{\mathrm{low}}:= 3(d+1)(w+1)\alpha_c(d,w) -1$ is much bigger than our $\lambda_c(d,w)$ even for small $d$, we have chosen to plot in log-scale for the two regimes $[\lambda_{\mathrm{low}},\infty)\subset [\lambda_c(d,w),\infty)$.
\begin{figure}[htbp]
\centering
\scalebox{0.7}{  
  \begin{tikzpicture}
    \begin{axis}[xmin = 0.9, xmax=8.1, ymin=-1.1, ymax=8.1, width=7.2cm, height=8.2cm]
      \addplot [name path=A, blue, mark=none, thick] table {logA_d=2.txt};
      \addplot [name path=L, orange, mark=none, thick] table {logLow_d=2.txt};
      \path [name path = T] (axis cs:1,8) -- (axis cs:8,8);
      \addplot[orange!50, opacity=0.8] fill between [of=L and T];
      \addplot[blue!50, opacity=0.5] fill between [of=A and T];
    \end{axis}
    \node at (2.6, -0.8) {$d = 2$};
  \end{tikzpicture}
}
\scalebox{0.7}{  
  \begin{tikzpicture}
    \begin{axis}[xmin = 0.9, xmax=8.1, ymin=-1.1, ymax=8.1, width=7.2cm, height=8.2cm, legend style={draw=none}, legend pos=outer north east]
      \addplot [name path=A, blue, mark=none, thick] table {logA_d=3.txt};
      \addplot [name path=L, orange, mark=none, thick] table {logLow_d=3.txt};
      \path [name path = T] (axis cs:1,8) -- (axis cs:8,8);
      \addplot[orange!50, opacity=0.8] fill between [of=L and T];
      \addplot[blue!50, opacity=0.5] fill between [of=A and T];
      \legend{$\log\tp{\lambda_c}$, $\log\tp{\lambda_{\mathrm{low}}}$}
    \end{axis}
    \node at (2.6, -0.8) {$d = 3$};
  \end{tikzpicture}
}
\caption{\small The $\lambda$ regimes (in log scale) identified  by the thresholds $\lambda_c$ and $\lambda_{\mathrm{low}}$  for $d \in\{2,3\}$}
\label{fig:low-temp}
\end{figure}
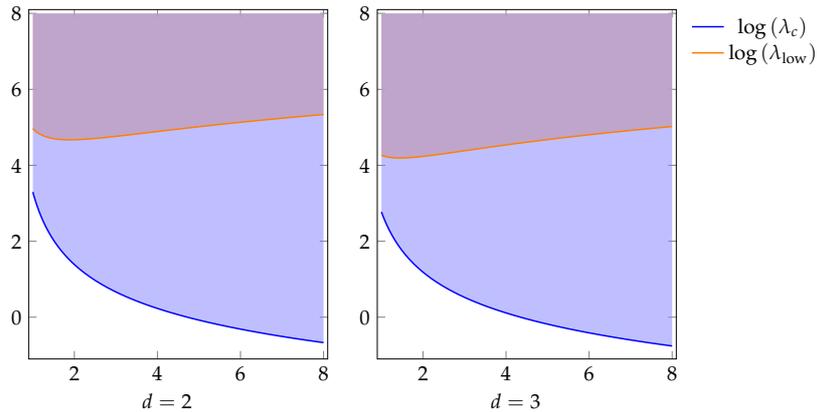

We note however, that we do not claim to have recovered the entire ``low-temperature'' regime. 
	First of all, due to technical difficulties, we have only been able to establish $0$-uniqueness for the threshold pair $(\lambda_c(d,w), \alpha_c(d,w))$.
	To get an algorithm, we will need to show a stronger gapped version, $\delta$-uniqueness, also holds whenever $(\lambda, \alpha)$ is strictly inside the uniqueness regime.
	Secondly, we have fixed $\alpha=\alpha_c(d,w)$ in the above comparison.
	To some extent, in order for existing ``low-temperature'' expansions to really shine, it seems to require that \BHC{} being more balanced, rather than more unbalanced, contrary to what previous works on ``low-temperature'' algorithm might suggest.
While it is extremely unlikely that one can recover \emph{every} application of the ``low-temperature'' paradigm (especially for the case of random regular graphs), it would also be interesting to see how much of the ``low-temperature'' regime can be matched with a better analysis of ``high-temperature'' algorithms.

\end{remark}

\subsection{$\delta$-uniqueness and spectral independence}
\label{sec:SI-and-uniqueness}

As discuss earlier, for algorithms to work, we need a gapped version that we call $\delta$-uniqueness.  In this subsection, we formally define $\delta$-uniqueness, then we outline our plans for showing that $\delta$-uniqueness implies $O(\frac{1}{\delta})$-spectral independence. 

\begin{definition} \label{def:lambda-d-alpha-w-unique}
  Let $\delta \in [0, 1)$ be a real number.
  We say a tuple $(\lambda, d, \alpha, w) \in \^R^4_{>0}$ is $\delta$-unique, if for all $\hat{x} \geq 0$ satisfying $F(\hat{x}) = \hat{x}$, it holds that $F'(\hat{x}) \leq 1 - \delta$.
\end{definition}

Our main focus is when the degree $w$ can be chosen arbitrarily, including fractional degrees.
This ensures the connectedness of the uniqueness regime. 

\begin{definition} \label{def:lambda-d-alpha-unique}
  Let $\delta \in [0, 1)$ be a real number.
  We say a tuple $(\lambda, d, \alpha) \in \^R^3_{>0}$ is $\delta$-unique, if the tuple $(\lambda, d, \alpha, w)$ is $\delta$-unique for all $w \in \^R_{>0}$.
\end{definition}
We remark that when $\lambda = \alpha$, \Cref{def:lambda-d-alpha-unique} is equivalent to \Cref{def:lambda-d-unique}.
These definitions of $\delta$-uniqueness can be rephrased in terms of $\alpha$ and $\lambda$.

\begin{theorem} \label{thm:delta-unique}
  Given $\delta \in [0, 1)$, $\alpha > 0$ and $d \geq 1 - \delta$,
  \begin{enumerate}
  \item if $\alpha \leq \frac{1-\delta}{d} \e^{1 + \frac{1-\delta}{d}}$, then $(\lambda, d, \alpha)$ is $\delta$-unique for all $\lambda > 0$; 
  \item if $\alpha > \frac{1-\delta}{d} \e^{1 + \frac{1-\delta}{d}}$, there is $\lambda_c > 0$ such that $(\lambda, d, \alpha)$ is $\delta$-unique iff $\lambda \geq \lambda_c$, where $\lambda_c = x_c (1 + \alpha (1 + x_c)^{-w_c})^d$, 
	  and $(x_c, w_c)$ is the unique positive solution of
    \begin{align*}
      \begin{cases}
        (1 - \delta)(x + 1)(\alpha + (1 + x)^w) - \alpha d w x = 0, \\
        w \log(1 + x) (\alpha d - (1 + x)^{w+1}) + \delta(x + 1)(\alpha + (x + 1)^w) = 0.
      \end{cases}
    \end{align*}
  \end{enumerate}
\end{theorem}


As a special case, when $\alpha = \lambda$, \Cref{thm:delta-unique} implies a weaker form of \Cref{thm:delta-unique-eq}, which we prove in~\Cref{thm:solve-sys-eq}. 


The proofs of \Cref{thm:delta-unique} is given in \Cref{sec:delta-unique}. At a high level, we need to locate the fixpoints of the iterative recurrence $F(x)$, and study their critical behavior.
While the recurrence $F(x)$ has been studied in the special case of $\lambda=\alpha=1$ by Liu and Lu~\cite{liu2015fptas}, 
the main difference in our work is that, in order to get a tight characterization, we have to actually locate the fixpoints and prove theorems about them.
This is the first time that the fixpoints and the critical behavior of the recurrence $F(x)$ has been identified, including for the case of $\lambda=\alpha=1$.
It is for this reason that it may seem surprising that we are able to extend the uniqueness condition on a general graph, to a bipartite graph with degree bounds only on one side.

\begin{remark}
We note that $\lambda < \lambda_c(\Delta)$ is a tight characterization of uniqueness for bipartite graphs with maximum degree $\Delta$ on one side. For any $\lambda > \lambda_c(\Delta)$, the infinite $\Delta$-regular tree is also a bipartite graph with maximum degree $\Delta$ on one side, which is known to be non-unique~\cite{kelly1985stochastic}.
\end{remark}


Next, we are ready to outline our plans of establishing spectral independence from $\delta$-uniqueness.
Spectral independence is a notion introduced by Anari, Liu and Oveis-Gharan~\cite{anari2020spectral}, which can be seen as a probabilistic formulation of local expansion without explicit reference to high-dimensional expanders.
A few definitions will be in order.

\begin{definition}[influence matrix]
  Let $\nu$ be a distribution over $\{-1, +1\}^n$, the influence matrix $\Psi_{\nu} \in \^R^{n \times n}$ is defined as
  \begin{align*}
    \forall i, j \in [n], \quad \Psi_{\nu}(i, j) :=
    \begin{cases}
      \Pr[\nu]{j \mid i} - \Pr[\nu]{j \mid \overline{i}} & \text{if $\Pr[\nu]{i} \in (0, 1)$,}  \\
      0 & \text{otherwise,}
    \end{cases}
  \end{align*}
  where we use $\Pr[\nu]{i} = \Pr[\sigma \sim \nu]{\sigma_i = +1}$ and $\Pr[\nu]{\overline{i}} = \Pr[\sigma \sim \nu]{\sigma_i = -1}$.
\end{definition}
Note that the influence matrix here has $1$ on its diagonal.

A \emph{pinning} $\tau$ is a partial configuration defined on $\Lambda \subseteq [n]$. 
Given a pinning $\tau$ on $\Lambda$, and another set $S \subseteq \Lambda$, we write $\tau_S$ for the pinning (partial configuration) restricted to $S$.

Let $\nu$ be a probability distribution over $\{-1, +1\}^n$.
We write $\Omega(\mu)$ for the support of $\mu$.
We will also write $\nu_S$ for its restriction to $S \subseteq [n]$.
Given a pinning $\tau$ on $\Lambda \subseteq [n]$, 
we write $\nu^\tau$ for the conditional distribution induced by $\nu$, on configurations that are consistent with $\tau$ on $\Lambda$. 
In particular, we write $\Psi_{\nu^\tau}$ for the influence matrix of the conditional distribution $\nu^\tau$.

\begin{definition}[spectral independence]
	Fix any $\eta > 0$. 
	A probability distribution $\nu$ over $\{-1, +1\}^n$ is called \emph{$\eta$-spectrally independent} if for any $\Lambda \subseteq [n]$ with $|\Lambda| \le n-2$, and any pinning $\tau \in \Omega(\nu_\Lambda)$, it holds that $\lambda_{\max}(\Psi_{\nu^\tau}) \leq \eta$. 
\end{definition}

We will derive spectral independence of the bipartite hardcore measure on one side, assuming the $\delta$-uniqueness on one side.
\begin{theorem} \label{thm:SI}
  Let $\delta \in (0, 1)$ be a real number.
  Let $G = ((L, R), E)$ be a bipartite graph with $\Delta = d+1 \geq 2$ be the degree bound on $L$.
  Let $\lambda > 0$ and $\alpha > 0$ be the fugacity on $L$ and $R$, respectively.
  If the tuple $(\lambda, d, \alpha)$ is $\delta$-unique, then the hardcore distribution on $G$ restricted to $L$, denoted by $\mu_L$, is $\eta$-spectrally independent for $\eta = \frac{\Delta}{d}\frac{(1 + \alpha)^\Delta}{\delta}$.
\end{theorem}
The proof of \Cref{thm:SI} is deferred to \Cref{sec:SI}.
We first prove a contraction with respect to the potential function discovered by Liu and Lu in~\cite{liu2015fptas}.
Unlike their analysis, which only verified the contraction for the cases of $d\in\{1,2,3,4\}$ and $\lambda=\alpha=1$, 
we show that their potential function can in fact support correlation decay up to criticality in the bipartite hardcore model.
However, this model seems to pose unique challenges to correlation decay analysis.
%
Traditionally, the contraction was often established by comparing to the contraction at the critical (0-unique) fixpoints. Examples of such analysis include~\cite{liu2015fptas} (or correlation decay analyses of other models~\cite{li2012approximate,sinclair2012approximation,li2013correlation,guo2018uniqueness, efthymiou2019convergence, anari2020spectral, chen2020rapid}).
In our analysis, we have to deviate from the traditional way, and establish a contraction by directly relating to the contraction at subcritical ($\delta$-unique) fixpoints.
This crucially relies on locating all the fixpoints and studying their critical behavior.
Once having established such contraction with respect to a suitable potential function, 
our arguments follow the same high-level plan of Chen, Liu and Vigoda~\cite{chen2020rapid} to bound the spectral gap of the influence matrix by the \emph{total influence} on the self-avoiding walk (SAW) tree, but only on the even depth of the SAW tree.  
The key difference here is that the total influence on the entire SAW tree can be unbounded.

\subsection{Entropy decay of field dynamics on general distributions}
\label{sec:ent-decay-field-general}

Our nearly linear time sampling algorithm is based on simulating the field dynamics.
Currently, the state-of-art techniques~\cite{anari2022entropic, chen2022optimal, chen2022localization} for proving rapid mixing of field dynamics goes through the framework of \emph{entropy decay}.
The ``entropy'' here refers to the \emph{relative entropy} (a.k.a.~\emph{KL-divergence}).
Let $\nu$ and $\mu$ be two probability distribution such that $\nu$ is absolutely continuous with respect to $\mu$, we define the relative entropy of $\nu$ with respect to $\mu$ as
$  \DKL{\nu}{\mu} := \sum_{X \in \Omega(\mu)} \nu(X) \log \frac{\nu(X)}{\mu(X)}$,
where we use the convention that $0 \log 0 = 0$.

These entropy decay analysis of field dynamics is done by using a stronger notion of spectral independence called \emph{entropic independence} introduced by Anari, Jain, Koehler, Pham, and Vuong in~\cite{anari2022entropic}.
At a high level, we follow a similar route by showing entropic independence, but there are new technical challenges that are unique to the model.
In the following, we explain the technical ingredient that we need, and the new challenges that arise in proving them.

The success of field dynamics crucially depends on the existence of a \emph{subcritical regime}, in which a mixing time bound can be established.  And the power of field dynamics is manifested by the following operation on a distribution.
\begin{definition} \label{def:external-fields} 
  Let $\nu$ be a distribution on $\{-1, +1\}^n$ and $\*\lambda \in \^R^n_{>0}$ be a positive vector of \emph{external fields}.
  The measure \emph{$\nu$ tilted by the external field $\*\lambda$} is denoted as $\*\lambda * \nu$ and is defined as
  \begin{align*}
    \forall \sigma \in \{-1, +1\}^n, \quad (\*\lambda * \nu)(\sigma) \propto \nu(\sigma) \cdot \prod_{i: \sigma_i = +1} \lambda_i.
  \end{align*}
  In particular, if $\lambda_i = \theta, \forall i \in [n]$, we simply write $\*\lambda * \nu$ as $\theta * \nu$.
\end{definition}

The main insight of \Cref{def:external-fields} is that for many ``hard'' distribution $\nu$, once we impose on it an external field $\lambda$ such that $\lambda$ is close to $0$ or $+\infty$, the biased distribution $\lambda * \nu$ becomes an ``easy'' distribution. 
We refer to such an ``easy'' biased distribution as \emph{subcritical}.
The field dynamics is a tool to take advantage of this phenomenon.
For \BHC{} with degree bounds only on one side however, the existence of a  subcritical regime is not readily available, and in particular it does not follow from the general graph case. 
To the best of our knowledge, we are the first to identify an optimal mixing subcritical regime for this model.

We explain a few more notations before describing the field dynamics.
For a configuration $\tau$, we will use $\tau_S \in \{-1, +1\}^S$ to denote the configuration $\tau$ being restricted to $S \subseteq [n]$.
For a distribution $\mu$, a pinning $\tau \in \Omega(\mu_\Lambda)$, and $S \subseteq [n]\setminus \Lambda$, let $\mu^\tau_S$ be the marginal distribution of $\mu^\tau$ restricted to $S$, which means $\mu^\tau_S = (\mu^\tau)_S$.
We also use $\pm\*1_\Lambda$ to denote the all $(\pm1)$-configuration on $\Lambda$.


\begin{definition}[field dynamics] \label{def:field-dynamics-ori}
  Let $\theta \in (0, 1)$ be a real number.
  For any distribution $\nu$ over $\{-1, +1\}^n$, the field dynamics $\FD{\nu}$ on $\nu$ with parameter $\theta$ is a Markov chain on state space $\Omega(\nu)$, with the following rule for updating a configuration $\sigma \in \Omega(\nu)$:
  \begin{enumerate}
  \item generate a random set $S$ by adding each $i \in [n]$ into $S$ with probability
    \begin{align*}
      p_i := \begin{cases}
               1 & \text{if } \sigma_i = -1; \\
               \theta & \text{if } \sigma_i = +1;
             \end{cases}
    \end{align*}
  \item replace $\sigma_S$ by a random partial configuration $\tau \sim (\theta * \nu)^{\sigma_{[n]\setminus S}}_{S} = (\theta * \nu)^{\*1_{[n]\setminus S}}_{S}$.
    That is, we deterministically set $X_i(v) = 1$ for $v \in [n]\setminus S$, and sample the remaining entries from $\theta * \nu$ conditional on the partial configuration $\sigma_{[n]\setminus S} = \*1_{[n]\setminus S}$.
  \end{enumerate}
\end{definition}

As shown in~\cite{chen2021rapid}, $\FD{\nu}$ is irreducible and aperiodic, and is reversible with respect to $\nu$.
This means that $\FD{\nu}$ converges to $\nu$ rather than $\theta * \nu$ during its evolution.
If we can show that $\FD{\nu}$ converges rapidly to $\nu$, then we have reduced the task of sampling from $\nu$ to a potentially (much) easier task of sampling from $\theta * \nu$.
To show the rapid mixing of $\FD{\nu}$, we use the following notion of entropic independence introduced in~\cite{anari2022entropic}.

\begin{definition}[entropic independence] \label{def:EI-intro}
  Let $\alpha > 0$ be a real number.
  A distribution $\nu$ over $\{-1, +1\}^n$ is said to be \emph{$\alpha$-entropically independent} if for every distribution $\mu$ which is absolutely continuous with respect to $\nu$, it holds that
  \begin{align*}
    \sum_{i \in [n]} \DKL{\mu_i}{\nu_i} \leq \alpha \cdot \DKL{\mu}{\nu}.
  \end{align*}
\end{definition}
As shown in \cite{anari2022entropic,chen2022optimal, chen2022localization}, and implicitly in \cite{chen2021optimal}, the entropic independence of a distribution $\nu$ can be established through the spectral independence of $\nu$ together with some additional requirements on $\nu$'s marginals.
For this purpose, we use the notion of \emph{marginal stability}~\cite{chen2022optimal, chen2022localization}.

  Let $\nu$ be a distribution over $\{-1, +1\}^n$.
  We write $\overline{\nu}$ for the ``flipped'' version of $\nu$ as another distribution over $\{-1, +1\}^n$, defined by $\overline{\nu}(\sigma) := \nu(\*{-1} \odot \sigma), \forall \sigma \in \{-1, +1\}^n$, where $\odot$ denotes the entry-wise product of two vectors.
  \begin{definition} [\cite{chen2022optimal, chen2022localization}] \label{def:K-stable}
    Let $K \geq 1$, a distribution $\nu$ is said to be \emph{$K$-marginally stable} if there is $\rho \in \{\nu, \overline{\nu}\}$ such that for every $i \in [n]$, $S \subseteq \Lambda \subseteq [n] \setminus \{i\}$, $\tau \in \Omega(\rho_\Lambda)$, it holds that
    \begin{align} \label{eq:K-stable}
      R^{\tau}_i \leq K \cdot R_i^{\tau_S} \quad \text{and} \quad \rho^{\tau}_i(-1) \geq K^{-1},
    \end{align}
    where $R^{\tau}_i := \rho^\tau_i(+1)/\rho^\tau_i(-1)$ is the marginal ratio on $i$ and $R^{\tau_S}_i$ is defined accordingly.
  \end{definition}

 
Through the \emph{stochastic  localization schemes} developed by Chen and Eldan in  \cite{chen2022localization}, 
optimal mixing of Glauber dynamics can be proved via entropic independence, 
by lifting from a known \emph{modified log-Sobolev inequality} in a suitable subcritical regime.
%
For \BHC{} however, such a modified log-Sobolev inequality in a subcritical regime is not readily available. 
%
Instead, we apply the result of \cite{chen2022localization}  to obtain the entropy decay for the field dynamics.

\begin{theorem}\label{thm:FD-ent-decay-short}
  Let $\theta \in (0, 1)$,  and $K, \eta \geq 1$ be real numbers.
  Let $\nu$ be a distribution over $\{-1, +1\}^n$. If 
  \begin{enumerate}
  \item $\lambda * \nu$ is  $K$-marginally stable for all $\lambda \in [\theta, 1]$,
  \item $\lambda * \nu$ is $\eta$-spectrally independent for all $\lambda \in [\theta, 1]$,
  \end{enumerate}
  then for  $\kappa := \theta^{2 \cdot 10^3 \eta K^4}$,
  for any distribution $\pi$ that is absolutely continuous respect to $\nu$, we have
  \begin{align*}
    \DKL{\pi \FD{\nu}}{\nu \FD{\nu}} \leq (1 - \kappa) \DKL{\pi}{\nu}.
  \end{align*}
\end{theorem}


\Cref{thm:FD-ent-decay-short} is proved in \Cref{sec:FD-ent-decay}, 
where an explicit connection  is provided between the field dynamics, and 
a  stochastic process built in \cite{chen2022localization}, called  the \emph{negative field localization process}.

\subsection{A fast sampler for the bipartite hardcore model and Proof of~\Cref{thm:main-unique}} \label{sec:main-unique}
Now we are ready to prove \Cref{thm:main-unique}, the main theorem for the sampling algorithm.

Let $G = ((L, R), E)$ be a bipartite graph with $\Delta = d + 1$ be the degree bound on $L$.
Let $n = \abs{L}$, 
and $\delta \in (0, 1)$ be a real number.
Let $\mu$ be the hardcore distribution on $G$ with fugacity $\lambda< \lambda_c(\Delta)$.
To sample from $\mu$, we simulate the field dynamics on $\mu_L$, where each update is simulated by a mixing Glauber dynamics. 
We note that, given a sample of $\mu$ on $L$, it is easy to generate a sample on $R$ so that their joint distribution is $\mu$.

Let $\nu = \mu_L$ and let $\overline{\nu}$ be the distribution on $\{-1, +1\}^L$ defined as
\begin{align}
  \forall \sigma \in \{-1, +1\}^L, \quad \overline{\nu}(\sigma) := \nu(\*{-1} \odot \sigma),\label{eq:sign-flipping-nu}
\end{align}
where $\odot$ denotes the entry-wise product of two vectors.

Let $\theta \in (0, 1)$. We denote the process of field dynamics $\FD{\overline{\nu}}$ by $(\overline{X_t})_{t \in \^N}$.
Define another process $(X_t)_{t \in \^N}$ as: for $t \in \^N$, let $X_t = \*{-1} \odot \overline{X_t}$.
By the definition of the field dynamics, the process $(X_t)_{t\in \^N}$ starts with state $X_0 = \*{-1} \odot \overline{X_0}$, and in the $t$-th transition, it does:
\begin{enumerate}
\item let $S = \emptyset$; \textbf{for} each $i$ with $X_{t-1}(i) = -1$, add $i$ to $S$ with prob. $1 - \theta$;
\item sample $X_t(L\setminus S) \sim (\theta^{-1} * \nu)^{\*{-1}_{S}}_{L\setminus S}$ and let $X_t(S) = X_{t-1}(S) = \*{-1}_S$, \label{item:resample}
\end{enumerate}
where \Cref{item:resample} comes from the fact that $(\theta * \overline{\nu})^{\*1_S}(\sigma) = (\theta^{-1} * \nu)^{\*{-1}_{S}}(-1 \odot \sigma)$, for $\sigma \in \{-1, +1\}^L$.
For convenience, we denote the transition matrix of the Markov chain $(X_t)_{t \in \^N}$ as $\FD[1/\theta]{\nu}$.
As suggested by this abuse of notation, the chain $\FD[1/\theta]{\nu}$, by flipping the signs of spins, can transform $\nu$ to $\theta^{-1}*\nu$ (as opposed to the field dynamics $\FD[\theta]{\nu}$ which transforms $\nu$ to $\theta*\nu$) for a $\theta\in(0,1)$.
We note that the chain  $\FD[1/\theta]{\nu}$ also has stationary distribution $\nu$.

\newcommand{\PGDm}[2]{\ensuremath{P^{(#1)}_{#2}}}
This is not an efficient algorithm yet, as it is usually hard to generate perfect samples as required by \Cref{item:resample} of $\FD[1/\theta]{\nu}$.
Instead we approximate it by running a Glauber dynamics on  $L$ with the stationary distribution $(\theta^{-1} * \nu)^{\*{-1}_S}$,  started from $\*1$ for $m$ steps. 
We denote by $\PGDm{m}{(\theta^{-1} * \nu)^{\*{-1}_S}}$ the $m$-step transition matrix of this  Glauber dynamics.
Now we describe how our algorithm works.

\begin{definition}[Our algorithm] \label{def:algo}
  For carefully chosen parameters $\theta, T, m$, starting from  $X_0 = \*{-1}$,
  \begin{enumerate}
  \item for each $t=1,2,\ldots, T$:
  \begin{enumerate}
  \item let $S = \emptyset$; \textbf{for} each $i$ with $X_{t-1}(i) = -1$, add $i$ to $S$ with probability $1 - \theta$;
  \item sample $X_{t}\in \{-1, +1\}^L$ by running \PGDm{m}{(\theta^{-1} * \nu)^{\*{-1}_S}}, started from $\*1$;
  \end{enumerate}
  \item sample $W \in \{-1, +1\}^R$ from $\mu^{X_T}_R$ and return $(X_T, W)$. 
  \end{enumerate}
\end{definition}
For convenience, in the rest of this section, we use $C := (1 + \lambda)^\Delta$.
Recall~\Cref{thm:FD-ent-decay-short}, to show that the field dynamics $\FD[1/\theta]{\nu}$ itself mixes rapidly,
it remains to check that $\theta * \overline{\nu}$ in the uniqueness regime is marginally stable.

\begin{lemma} \label{lem:verify-K-tame}
  For all $\theta \in (0, 1)$, the distribution $\theta * \overline{\nu}$ is $2C$-marginally stable.
\end{lemma}
\begin{proof}
  Note that the ``flipped'' version of $\theta * \overline{\nu}$ is $\theta^{-1} * \nu$, that is, $\theta * \overline{\nu} = \overline{\theta^{-1} * \nu}$.
  Let $\pi$ be the hardcore distribution on graph $G = ((L, R), E)$ with fugacity $\lambda_L = \lambda / \theta$ on $L$, and fugacity $\lambda_R = \lambda$ on $R$.
  It is straightforward to verify that $(\theta^{-1} * \nu) = \pi_L$ is $2C$-marginally stable.

  If $\lambda_L \leq C$, we let $\rho = \theta^{-1} * \nu$. For every $i \in L, S \subseteq \Lambda \subseteq L \setminus \{i\}$, and $\tau \in \Omega(\rho_\Lambda)$,
  it holds that
  \begin{align*}
    \frac{\rho^\tau_i(+1)}{\rho^\tau_i(-1)} \left/ \frac{\rho^{\tau_S}_i(+1)}{\rho^{\tau_S}_i(-1)} \right. \leq \frac{\lambda_L}{1} \left/ \frac{\lambda_L}{(1 + \lambda_R)^\Delta} \right. \leq C \quad \text{and} \quad \rho^\tau_i(-1) \geq \frac{1}{1 + \lambda_L} \geq \frac{1}{2C}.
  \end{align*}

  If $\lambda_L > C$, we let $\rho = \theta * \overline{\nu}$. For every $i \in L, S \subseteq \Lambda \subseteq L \setminus \{i\}$, and $\tau \in \Omega(\rho_\Lambda)$,
  it holds that
  \begin{align*}
    \frac{\rho^\tau_i(+1)}{\rho^\tau_i(-1)} \left/ \frac{\rho^{\tau_S}_i(+1)}{\rho^{\tau_S}_i(-1)} \right. \leq \frac{(1 + \lambda_R)^\Delta}{\lambda_L} \left/ \frac{1}{\lambda_L} \right. \leq C \quad &\text{and} \quad \rho^\tau_i(-1) \geq \frac{\lambda_L}{(1 + \lambda_R)^\Delta + \lambda_L} \geq \frac{1}{2}. \qedhere
  \end{align*}

\end{proof}

Now we are ready to state the entropy decay of field dynamics under $\delta$-uniqueness.
\begin{lemma} \label{lem:ent-decay-field}
  If the pair $(\lambda, d)$ is $\delta$-unique and $d \geq 1$, then for any distribution $\pi$ that is absolutely continuous with respect to $\nu$, it holds that
  \begin{align*}
    \DKL{\pi \FD[1/\theta]{\nu}}{\nu \FD[1/\theta]{\nu}} \leq  (1 - \theta^{10^5 C^5/\delta}) \DKL{\pi}{\nu}.
  \end{align*}
\end{lemma}
\begin{proof}
  Note that by definition, it holds that
  \begin{align*}
    \forall X, Y \in \{-1, +1\}^L, \quad \FD{\overline{\nu}}(X, Y) = \FD[1/\theta]{\nu}(\*{-1} \odot X, \*{-1} \odot Y),
  \end{align*}
  which implies that
  \begin{align*}
    \forall X \in \{-1, +1\}^L, \quad \overline{\pi} \FD{\overline{\nu}} (X) = \pi \FD[1/\theta]{\nu} (\*{-1} \odot X) \quad \text{and} \quad \overline{\nu} \FD{\overline{\nu}} (X) = \nu \FD[1/\theta]{\nu} (\*{-1} \odot X).
  \end{align*}
Note that $\DKL{\pi}{\nu}= \DKL{\overline{\pi}}{\overline{\nu}}$, so it is sufficient to show that
  \begin{align*}
    \DKL{\overline{\pi} \FD{\overline{\nu}}}{\overline{\nu} \FD{\overline{\nu}}} &\leq  (1 - \theta^{10^5 C^5/\delta}) \DKL{\overline{\pi}}{\overline{\nu}}.
  \end{align*}
  Recalling \Cref{thm:FD-ent-decay-short}, since we have verified  $2C$-marginal stability in \Cref{lem:verify-K-tame}, and spectral independence in \Cref{thm:SI}, we conclude the proof.
\end{proof}



As alluded to earlier, the success of field dynamics requires the existence of a subcritical regime, which is what allows us to efficiently simulate the field dynamics with Glauber dynamics. 
Recall that we let $\nu = \mu_L$ be the projection of the bipartite hardcore measure on $L$.
We show that the standard Glauber dynamics $(Z_t)_{t\in \^N}$ on $(\theta^{-1} * \nu)^{\*{-1}_S}_{L\setminus S}$ is rapidly mixing for every $S \subseteq L$.
\begin{lemma} \label{lem:GD-good-0}
	Let $k \geq \e^9$ be a real number,
	and  $\theta$ satisfy $\theta^{-1} \geq C \cdot k \Delta \frac{\log n}{ \lambda}$.
	Fix any $S \subseteq L$ and $\tau \in \Omega(\mu_S)$.  
	Then for $\ell: = \abs{L \setminus S}$ and any $T \geq 21 \cdot \ell\log \ell$, the Glauber dynamics $(Z_t)_{t\in \^N}$ on $(\theta^{-1} * \nu)^{\tau}_{L\setminus S}$ satisfies 
  \begin{align*}
    \DTV{Z_T}{(\theta^{-1} * \nu)^\tau_{L\setminus S}} &\leq 2 \ell^{- \ftp{\frac{T}{21 \cdot \ell\log \ell}}}.
  \end{align*}
\end{lemma}

\Cref{lem:GD-good-0} is proved in \Cref{sec:GD-good}.

Our last ingredient is a standard argument, which basically allows us to reduce the task of sampling from $\mu$ to the task of sampling from $\nu$.
  \begin{lemma} \label{lem:nu-to-mu}
  Let $U$ be any ground set, $\mu$ be any distribution over $\{-1, +1\}^U$ and $\nu: = \mu_S$ for some $S \subsetneq U$.
  Let $X$ be any random vector on $\{-1, +1\}^S$, and $Y\sim\mu^X$ be another random vector $\{-1, +1\}^U$.
  Then it holds that $\DTV{Y}{\mu} \leq \DTV{X}{\nu}$.
\end{lemma}
\begin{proof}
  By the coupling lemma, there is a vector $X_0 \sim \nu$ such that $\Pr{X \neq X_0} = \DTV{X}{\nu}$.
  We construct vector $Y_0$ by sampling from $\mu^{X_0}$.
  By the coupling lemma,
  \begin{align*}
    \DTV{Y}{\mu} \leq \Pr{Y \neq Y_0}
    &\leq \Pr{Y\neq Y_0 \mid X = X_0 } + \Pr{X \neq X_0} \\
    &\overset{(\star)}{=} \Pr{X \neq X_0} = \DTV{X}{\nu},
  \end{align*}
  where in $(\star)$, by definition, $X = X_0$ implies that $Y = Y_0$.
\end{proof}

Now, we are ready to prove \Cref{thm:main-unique}.

\begin{proof}[Proof of \Cref{thm:main-unique}]

  Let $\theta = (C \cdot \e^9 \Delta \log n / \lambda)^{-1}$.
  Moreover, let
  \begin{align*}
    T &=  \tp{C \cdot \e^9 \Delta \log n / \lambda}^{10^5 C^5/\delta} \cdot \log \frac{n \log C}{\epsilon^2/2} \\
    \text{and} \quad m &= \ctp{\frac{\log (4T/\epsilon)}{\log n}} \cdot 21 n\log n.
  \end{align*}

  \paragraph{The total running time:}
  We note that by our choice of parameters, it holds that
  \begin{align*}
    T =  \tp{\Delta \log n / \lambda}^{O(C^5/\delta)} \cdot \log(1/\epsilon).
  \end{align*}
  We claim that each update of the Glauber dynamics on $\nu$ can be performed in $O(\Delta)$ time,
  and the last step of our algorithm takes at most $O(|R|) = O(n\Delta)$ time.
  Hence the running time after $T$ iterations in our algorithm as described in~\Cref{def:algo} is bounded by
  \begin{align*}
	  m T  \cdot O(\Delta) + O(n\Delta) &= n \cdot  \tp{\Delta \log n / \lambda}^{O(C^5/\delta)} \cdot \log^2(1/\epsilon).
  \end{align*}

  Recall that $C = (1 + \lambda)^\Delta$.
  When $\Delta \geq 3$, it holds that $\lambda_c(\Delta) \leq 3\e^2/\Delta$.
  Hence we have $(1 + \lambda)^\Delta \leq \exp(3\e^2)$ bounded by an absolute constant for any $\lambda \le \lambda_c(\Delta)$.
  When $\Delta = 2$, such a bound does not apply. Instead, $C= (1 + \lambda)^\Delta= (1 + \lambda)^2$.

  It remains to verify that each update of the Glauber dynamics can be performed in $O(\Delta)$ time, and the last step of our algorithm takes at most $O(|R|)$ time.
Let $(Z_t)_{t\in \^N}$ be the Glauber dynamics on $\nu$. 
We maintain a $\-{Cnt}_t(v) := \abs{\{u \in \Gamma_v \mid Z_t(u) = +1\}}$ for each $v \in R$.
When $u \in L$ is picked at time $t$, the marginal probability can be calculated by
  \begin{align*}
     \Pr{Z_t(u) = +1} = \frac{\lambda}{\lambda + (1 + \lambda)^{\sum_{v \in \Gamma_u} \*1[\-{Cnt}_{t-1}(v) = 0]}}.
  \end{align*}
  Since the summation only involves neighbors of $u$, and $\-{Cnt}_t$ can be obtained from $\-{Cnt}_{t-1}$ by updating entries for neighbors of $u$,
  it is straightforward to see that in each update of the Glauber dynamics, all these operations can be implemented in $O(\Delta)$.

  Finally, for the last step of algorithm, for each $v \in R$, we sample $W_v\in \{+1, -1\}$ according to 
  \[
     \Pr{W_v = +1} = \frac{\lambda}{\lambda + 1} \cdot \*1[\-{Cnt}_{\star}(v) = 0],
  \]
  where $\-{Cnt}_{\star}$ is the vector maintained by the Glauber dynamics for generating $X_T$ as in \Cref{def:algo}.
  
\paragraph{The error bounds $\DTV{X_T}{\nu}$ and $\DTV{(X_T,W)}{\mu}$:}
Let $(X_t)_{0 \leq t \leq T}$ be the process on $L$ generated by our algorithm (except for the last step), and $ (Y_t)_{0 \leq t \leq T}$ be the process generated by the transition rule $\FD[1/\theta]{\nu}$.
We start both processes from the same initial state $X_0 = Y_0 = \*{-1}$.
Note that by our choice of $\theta$ and $T$, it holds that $\theta^{-1} \geq C\cdot \e^9\Delta \frac{\log n}{\lambda}$, which meets the requirement for \Cref{lem:GD-good-0} to apply.
Recall that in every iteration of our algorithm, the transition rule $\FD[1/\theta]{\nu}$ is replaced by:
\begin{enumerate}
	\item the same subsampling step of generating $S$;
	\item instead of sampling from $(\theta^{-1} * \nu)^{\*{-1}_S}$ directly, we run \PGDm{m}{(\theta^{-1} * \nu)^{\*{-1}_S}} to generate a sample.
\end{enumerate}
Applying~\Cref{lem:GD-good-0}, the error in TV-distance introduced by the $m$-step Glauber dynamics is at most $2 n^{-\ftp{\frac{m}{21n\log n}}}$.
This means that for any given $X_t = Y_t$, there is a coupling between $X_{t+1}$ and $Y_{t+1}$, such that
  \begin{align*}
    \Pr{X_{t+1} \neq Y_{t+1} \mid X_t, Y_t} \leq 2 n^{-\ftp{\frac{m}{21n\log n}}}
  \end{align*}
  Hence, we can construct a coupling between the two process $(X_t)_{0 \leq t \leq T}$ and $(Y_t)_{0 \leq t \leq T}$ as follows: for $1 \leq t \leq T$,
  \begin{enumerate}
  \item if $X_{t-1} = Y_{t-1}$, generate $X_t, Y_t$ from the coupling offered by \Cref{lem:GD-good-0};
  \item otherwise, generate $X_t, Y_t$ independently.
  \end{enumerate}
  Hence, by union bound, for our choice of $m$, it holds that
  \begin{align} 
   \DTV{X_T}{Y_T}\le \Pr{X_T \neq Y_T}
    \label{eq:eps-tar-1} &\leq 2 T n^{-\ftp{\frac{m}{21 n\log n}}}\le\frac{\epsilon}{2}.
  \end{align}
  On the other hand, by \Cref{lem:ent-decay-field}, it holds that 
  \begin{align}
    \DTV{Y_T}{\nu}
    \nonumber &\leq \sqrt{\frac{1}{2} \DKL{Y_T}{\nu}} = \sqrt{\frac{1}{2} \DKL{\*1_{[X = \*{-1}]} (\FD[1/\theta]{\nu})^T}{\nu (\FD[1/\theta]{\nu})^T}} \\
    \nonumber &\leq \sqrt{\frac{1}{2}  (1 - \theta^{10^5 C^5/\delta})^T \DKL{\*1_{[X = \*{-1}]}}{\nu}} \\
    \label{eq:eps-tar-2} &\leq \sqrt{\frac{1}{2}  (1 - \theta^{10^5 C^5/\delta})^T \cdot n \log (1 + \lambda)},
  \end{align}
  where the first inequality is the Pinsker's inequality, the second inequality holds by \Cref{lem:ent-decay-field}, and the last inequality holds because $1 / \nu(\*{-1}) \leq (1 + \lambda)^n$.
  Meanwhile, notice that
  \begin{align*}
    T
    &\geq  \theta^{-10^5 C^5/\delta} \cdot \log \frac{n \log C}{\epsilon^2/2} 
=  \tp{C \cdot \e^9 \Delta \frac{\log n}{ \lambda}}^{10^5 C^5/\delta}  \cdot \log \frac{n \log C}{\epsilon^2/2}.
  \end{align*}
  Plugging this into $\eqref{eq:eps-tar-2}$, we have $\DTV{Y_T}{\nu}\le {\epsilon}/{2}$. Recall $\DTV{X_T}{Y_T}\le  {\epsilon}/{2}$ in \eqref{eq:eps-tar-1}.
Therefore, 
\[
\DTV{X_T}{\nu} \leq \DTV{X_T}{Y_T} + \DTV{Y_T}{\nu}\le \epsilon.
\] 

%

Finally, we note that in the last step of our algorithm, $W \sim \mu^{X_T}$ is sampled faithfully.
By \Cref{lem:nu-to-mu}, we have $\DTV{(X_T,W) }{\mu} \le \DTV{X_T}{\nu} \le \epsilon$,

\end{proof}

\subsection{Rapid mixing of the Glauber dynamics via Markov chain comparison}\label{sec:main-GD}
%
Through a novel comparison argument between the field dynamics defined  on only one side of the bipartite graph and the single-site Glauber dynamics defined on the entire graph,
we prove the rapid mixing of the standard Glauber dynamics as stated in \Cref{thm:GD-mu-mix-intro}.

The Glauber dynamics $(X_t)_{t\in\^N}$ (a.k.a. Gibbs sampler) is a canonical single-site Markov chain for sampling. 
Consider an abstract distribution $\pi$ over $\{-1, +1\}^n$. 
The Glauber dynamics on $\pi$, denoted by $\GD[\pi]$, is a Markov chain on space $\Omega(\pi)$, with the $t$-th transition defined naturally as:
\begin{enumerate}
\item pick a coordinate $i \in [n]$ uniformly at random;
\item sample $X_t \sim \pi(\cdot \mid X_{t-1}(V\setminus \{i\}))$.
\end{enumerate}
It is well known that the chain is reversible with respect to the stationary distribution $\pi$, and moreover, $\GD[\pi]$ has non-negative spectrum  (see \Cref{lem:GD-psd} for details).

We assume the same bipartite hardcore model as in \Cref{sec:main-unique}, 
on bipartite graph $G = ((L, R), E)$ with $n = \abs{L}$ vertices and degree bound $\Delta = d + 1\ge 2$ on one side $L$.
Let $\mu$ be the hardcore distribution on $G$ with fugacity $\lambda< \lambda_c(\Delta)$, and let $\nu=\mu_L$.
We are interested in the standard Glauber dynamics $\GD[\mu]$ for the hardcore distribution $\mu$ and its one-sided version $\GD[\nu]$.

We will show the rapid mixing of these Glauber dynamics through variance decay.
%
Let $\lambda_2(P)$ be the second largest eigenvalue of a irreducible, aperiodic, and reversible chain $P$.
And denote by $\sgap{P}$ its \emph{spectral gap}:
\[
\sgap{P} := 1 - \lambda_2(P).
\]

First, we prove the following result for the spectral gap of the Glauber dynamics on one side.
\begin{theorem} \label{lem:sgap-nu}
If $(\lambda, d)$ is $\delta$-unique, 
then it holds for the Glauber dynamics $\GD[\nu]$ on $\nu=\mu_L$ that
\[
\sgap{\GD[\nu]} \geq (22 n)^{-1} \cdot  \tp{\frac{C\cdot \e^9 \Delta \log n}{\lambda}}^{-10^5 C^5/\delta} ,
\]
where $n=|L|$ and $C = (1 + \lambda)^\Delta$.
\end{theorem}

We remark that \Cref{lem:sgap-nu} already implies the rapid mixing of the Glauber dynamics on $\nu$. Furthermore, its mixing time only has a quadratic dependency on $n$.
The proof of the theorem follows from the so-called ``field dynamics comparison lemma''.

\begin{lemma}[\text{\cite[Lemma 2.4]{chen2021rapid}}] \label{lem:var-boost}
  Let $\pi$ be a distribution over $\{-1, +1\}^n$.
  For all $\theta \in (0, 1)$, we have
  \begin{align*}
    \sgap{\GD[\pi]} &\geq \sgap{\FD{\pi}} \cdot \min_{\Lambda \subseteq [n], \tau \in \Omega(\pi_\Lambda)} \sgap{\GD[(\theta * \pi)^\tau]}.
  \end{align*}
\end{lemma}

This lemma relates the spectral gap of the Glauber dynamics $\GD[\pi]$ on a distribution $\pi$ to the spectral gap of the Glauber dynamics in a subcritical regime $\theta * \pi$, through the spectral gap of the field dynamics $\FD{\pi}$.
In order to prove~\Cref{lem:sgap-nu}, we apply \Cref{lem:var-boost} to $\pi=\bar{\nu}$, where $\bar{\nu}$ is the distribution obtained from $\nu$ by flipping the signs as in \eqref{eq:sign-flipping-nu}.
Specifically:
\begin{itemize}
\item
The lower bound on the spectral gap $\sgap{\FD{\pi}}$ for the field dynamics ${\FD{\pi}}={\FD{\bar{\nu}}}={\FD[1/\theta]{\nu}}$
is implied by the entropy decay stated in \Cref{lem:ent-decay-field}.  
The implication from entropy decay to variance decay is through a standard trick called linearization~\cite{rothaus1981diffusion, jerrum2003counting, caputo2015approximate}.

\item The lower bound on the spectral gap $\sgap{\GD[(\theta * \pi)^\tau]}$ for the Glauber dynamics in subcritical regime $\theta * \pi=\theta * \bar{\nu}=\theta^{-1} * \nu$ is implied by its rapid mixing, stated in~\Cref{lem:GD-good-0}.
This implication from mixing time to spectral gap is standard (see e.g.~\cite[Corollary 12.7]{levin2017markov}).
\end{itemize}
\Cref{lem:sgap-nu} follows naturally by combining these together. 
A detailed proof is given in \Cref{sec:sgap-nu}.

Next, we prove the following comparison result for the spectral gaps of the Glauber dynamics $\GD[\mu]$ on two sides and the Glauber dynamics $\GD[\nu]$ on one side.
\begin{lemma} \label{lem:comp-nu-mu}
If $\lambda \leq (1 - \delta)\lambda_c(\Delta)$, 
then it holds for the Glauber dynamics $\GD[\mu]$ and $\GD[\nu]$ 
that 
\[
\sgap{\GD[\mu]} \geq \sgap{\GD[\nu]} \cdot \zeta \cdot ((\Delta+1) n)^{-1},
\]
where  $n=|L|$ , $\zeta = 50^{-400/\delta}$ for $\Delta \geq 3$, and $\zeta = (9 \cdot 4^7(1+\lambda)^8)^{-1}$ for $\Delta = 2$.
\end{lemma}
A main challenge for comparing $\GD[\nu]$ are $\GD[\mu]$ is that they are  defined on different distributions.
To overcome this issue, we introduce a third chain, denoted by $\PB$, which is a block dynamics on $\mu$ that resembles the behavior of $\GD[\nu]$.
Specifically, let $\PB$ be a Markov chain $(X_t)_{t\in\^N}$ on the space $\Omega(\mu)$. 
We note that $X_t$ is supported on $L \cup R$. In the $t$-th transition, it does:
\begin{enumerate}
\item pick a vertex $v\in L$ uniformly at random;
\item sample $X_t \sim \mu(\cdot\mid {X_{L\setminus \{v\}}})$.
\end{enumerate}

On the one hand, by a coupling argument, 
it is not hard to see that the rate at which $\PB$ converges to $\mu$ is bounded by that of $\GD[\nu]$ to $\nu$, in TV-distance. 
Then, by the standard connection between mixing time and spectral gap (\cite[Corollary 12.7]{levin2017markov}), we can prove  that
$\sgap{\PB} \geq \sgap{\GD[\nu]}$.

On the other hand, 
we conduct a comparison between the block dynamics $\PB$ and the Glauber dynamics $\GD[\mu]$, both on the same distribution $\mu$, and demonstrate that
$\sgap{\GD[\mu]} \geq \sgap{\PB} \cdot \zeta \cdot ((\Delta+1) n)^{-1}$.
%
To establish this result, we rely on a specialized form of the Poincar{\'e} inequality, which we prove using the framework of \emph{approximate block factorization} of variance, developed in~\cite{chen2021optimal,caputo2021block, caputo2015approximate}.
Altogether, this proves \Cref{lem:comp-nu-mu}. 
A detailed proof is presented in \Cref{sec:comp-nu-mu}.

\begin{proof}[Proof of \Cref{thm:GD-mu-mix-intro}]
%
%
%
	We start by verifying that both \Cref{lem:sgap-nu} and \Cref{lem:comp-nu-mu} are applicable.
	For $\Delta=2$, we already have $\delta$-uniqueness by assumptions.
	For $\Delta\ge 3$, recalling~\Cref{thm:delta-unique-eq}, we also have $(\lambda,d)$ being $\delta/10$-unique when $\lambda \leq (1 - \delta)\lambda_c(\Delta)$.

	Next, notice that
    $\mu_{\min} \geq \tp{\frac{\min\{1, \lambda\}}{1 + \lambda}}^{\abs{L\cup R}}\ge \tp{\frac{\min\{1, \lambda\}}{1 + \lambda}}^{\abs{(\Delta+1)n}}$.
  It is a standard result (see e.g.~\cite[Theorem 12.4]{levin2017markov}) that the mixing time of $\GD[\mu]$ can be bounded through its spectral gap, 
  which in turn is bounded through \Cref{lem:sgap-nu} and \Cref{lem:comp-nu-mu} as 
  \begin{align*}
    \Tmix{\epsilon} &\leq \sgap{\GD[\mu]}^{-1} \cdot \log\tp{\frac{1}{\epsilon\mu_{\min}}}\\
    &\leq  \tp{\frac{\Delta \log n}{\lambda}}^{O(C^5/\delta)}  \cdot n^2 \cdot \tp{n \log \frac{1 + \lambda}{\min\{1, \lambda\}} + \log \frac{1}{\epsilon}},
  \end{align*}
where $n$ is overridden as  $n=|L\cup R|\le (\Delta+1)|L|$ and $C=(1+\lambda)^{\Delta}$ is bounded as in the proof of \Cref{thm:main-unique}. 
\end{proof}

\section{Related works and discussions}
\paragraph{The hardcore model \& two-spin systems}
The hardcore model has been a very important model for equilibrium statistical physics. 
Sampling from the hardcore distribution has been widely studied not only in statistical physics, but also in combinatorics and distributed computing, as 
the hardcore model can also be seen as an enumeration of weighted independent sets. 

The hardcore model belongs to a more general family known as the two-spin systems, which can be classified as either \emph{anti-ferromagnetic} or \emph{ferromagnetic} depending on the nature of their edge interactions. 
In anti-ferromagnetic two-spin systems, where the edge interactions are \emph{repulsive} and neighboring vertices tend to take on different assignments, 
a sharp computational phase transition has been established at the uniqueness threshold: initially for the hardcore model~\cite{weitz2006counting, sly2010computational, galanis2014improved, sly2012computational}, and later extended to all anti-ferromagnetic two-spin systems in~\cite{ li2012approximate, sinclair2012approximation,li2013correlation, sly2012computational,galanis2016inapproximability}. 
There has been significant research for this model, aimed at developing faster algorithms and improving analysis of existing methods. 
New techniques and analysis continue to be developed and evolve~\cite{luby1997approximately, dyer2000markov, vigoda2001note, dyer2002counting, goldberg2003computational, hayes2006coupling, weitz2006counting, sly2010computational, galanis2014improved, sly2012computational, li2012approximate, sinclair2012approximation, galanis2016inapproximability, li2013correlation, barvinok2016combinatorics, patel2017deterministic, peters2019conjecture, liu2019fisher, bencs2018note, efthymiou2019convergence, anari2020spectral, chen2020rapid, chen2021optimal, chen2021rapid, anari2022entropic, chen2022optimal, chen2022localization}; we do not attempt to provide a comprehensive list here.

In ferromagnetic two-spin systems, where the edge interactions are \emph{attractive} and neighboring vertices tend to take on the same assignments, the computational phase transition disappears in an important special case known as the Ising model. 
However, the uniqueness condition for the general ferromagnetic two-spin system is more complicated \cite{guo2018uniqueness}, and our understanding of this model is limited. 
It has also been shown that there is a close connection between the ferromagnetic Ising model with local fields and \#BIS \cite{goldberg2007complexity}. 
Research on ferromagnetic two-spin systems is becoming more active in recent years~\cite{jerrum1993polynomial, goldberg2003computational, goldberg2007complexity, mossel2013exact, guo2018uniqueness, liu2019fisher, guo2020zeros, shao2020contraction, chen2020rapid, chen2022localization, chen2022near}.
We remark that \BHC{} when projected to one side is inherently a ferromagnetic system.
Specifically, we consider a bipartite graph $G=\left( (L \cup R), E \right)$ with $\Delta_L=2$, the hardcore measure projected to $R$ is indeed a ferromagnetic two-spin system on $R$. A proof of this fact via standard holographic transformation can be found in~\Cref{sec:ising-case}. For the same reasons, when $\Delta_L \ge 3$, \BHC{} corresponds to a \emph{hypergraph} ferromagnetic two-spin system on $R$, where multi-body interactions are also allowed in the system.
Unlike the $\Delta_L=2$ case, the system does not have symmetric edge interactions, which precludes applying known results on hypergraph ferromagnetic two-spin systems easily.

\paragraph{The bipartite hardcore model}
The unweighted version of the bipartite hardcore model (\#BIS), is originally introduced as a problem of intermediate complexity, to facilitate the complexity classifications of approximate counting problems, under the so-called approximation-preserving reductions (AP-reductions)~\cite{dyer2004relative}.
Due to its bipartite nature, two sides of the bipartition can often encode different objectives. This makes it especially flexible in building gadgets on \#BIS and reducing other important problems to \#BIS.
There are many natural problems that have been proven to be \#BIS-equivalent or \#BIS-hard~\cite{dyer2004relative, goldberg2007complexity, dyer2010approximation, dyer2012complexity, chebolu2012complexity, goldberg2012approximating, bulatov2013expressibility, liu2014complexity, goldberg2015complexity, galanis2016ferromagnetic, cai2016bis}.
\#BIS or its weighted version \BHC{} is a computational problem of its own interests.
Many sampling algorithms have been proposed for the hardcore model, and the standard ``heat-bath'' Glauber dynamics is arguably one of the most prominent ones.
However, it was soon discovered that there is a concrete algorithmic barrier for any \emph{local} Markov chains above the uniqueness threshold (as soon as $\lambda> \lambda_c(\Delta)$)~\cite{dyer2002counting,mossel2009hardness}, even when restricted to regular bipartite graphs.

Remarkably, \#BIS above the uniqueness threshold is also used as a gadget in a randomized AP-reduction in~\cite{dyer2002counting}, through which they showed that approximately counting independent sets in graphs of constant bounded degree is already NP-hard.
The reduction crucially relies on the non-uniqueness of \#BIS to succeed, and this gadget in the non-uniqueness regime is also the starting point of~\cite{sly2010computational,sly2012computational,galanis2014improved} that established computational hardness phase transitions for the hardcore model at the uniqueness threshold.
Prior to using a non-unique \#BIS as gadget, the only known reduction of showing hardness of approximately counting independent sets essentially relies on the hardness of finding the largest independent sets, and the hardness result requires a much larger maximum degree (see, e.g., the proof of Theorem 1.17 in~\cite{sinclair1993} and the proof of Theorem 4 in~\cite{luby1997approximately}).
Due to the algorithmic barriers presented by non-unique \#BIS instances, combined with its central role in the complexity classification of approximate counting problems,
this leads to a conjecture that neither does \#BIS  admit an FPRAS, nor is it as hard as \#SAT~\cite{dyer2004relative}.

Since then, algorithmic efforts have been mainly focused on identifying tractable instances or regimes of parameters.
There are mainly two line of works,  focusing on algorithms that run in either the high-temperature, or the low-temperature regime of the \BHC{} model.
In the high-temperature regime (a.k.a.~tree-uniqueness regime), \cite{liu2015fptas} gives a fast algorithm for \BHC{} with $\lambda = 1$ and $\Delta_L \leq 5$ based on the method of correlation decay~\cite{weitz2006counting}.
However, until the present work, little progress had been made on the high-temperature regime of \BHC{} or \#BIS beyond what can be inherited from faster algorithms on general graphs.

In the low-temperature regime, current progress relies heavily on the polymer representation of the \BHC{} model.
Initiated by \cite{helmuth2020algorithmic}, there has been a series of works that have developed fast algorithms for the \BHC{} model or its variants in the low-temperature regime, using the polymer representation~\cite{helmuth2020algorithmic, liao2019counting, jenssen2020algorithms, cannon2020counting, chen2021fast, jessen2022approximately, blanca2022fast, chen2022sampling, friedrich2023polymer}.
There are also interesting sampling algorithms based on Markov chains designed for \BHC{}~\cite{chen2021fast, blanca2022fast,friedrich2023polymer}.
However, these are Markov chains running on the polymer representation, and they are very different from the standard Glauber dynamics.

We note that there are also more detailed study on special classes of bipartite graphs, such as on lattices (e.g.,~\cite{martinelli1994approach2, lubetzky2014cutoff}), and on trees (e.g.,~\cite{restrepo2014phase}).

\paragraph{The bipartite hardcore model and the Lov\'asz local lemma}
The hardcore partition function also gives a characterization of the worst-case extremal measure in applications of the Lov\'asz local lemma~\cite{Shearer85} and its algorithmic counterparts~\cite{KS11}. 
For specific applications (such as $k$-SAT), however, the dependency graph that naturally arises often comes with extra structure.
Taking advantage of these extra structure lead to improved LLL framework such as the cluster expansion LLL~\cite{bissacot2011improvement,harvey2015algorithmic}.
Given our new uniqueness characterization of \BHC{}, which seems to outperform a convergent cluster expansion in ``unbalanced'' settings,
it would also be interesting to see if this can lead to improved LLL conditions when the dependency graph is bipartite.
Many constraint satisfaction problem (CSP) can be naturally modeled with a dependency graph that is bipartite. 
Specifically, one side of the bipartite graph can encode the ``equality'' constraint, the other side can encode the actual constraint of the CSP instance, and then the edges in between can be viewed as duplicated variables.
Our current uniqueness characterization does not immediately give any new criterion for such ``bipartite LLL'', 
for many of these applications, one needs to consider \emph{negative} fugacities.
However, we believe that our framework of locating fixpoints and studying their critical behavior will be very useful in such analysis, which we leave as future work.

\paragraph{High-dimensional expanders and field dynamics related}
In a seminal work~\cite{anari2020spectral}, Anari, Liu, and Oveis-Gharan introduced the concept of spectral independence based on previously developed techniques for analyzing down-up random walks in high-dimensional expanders~\cite{nima2019log, cryan2019modified, alev2020improved}.
This notion is first applied to the hardcore model, which gives the first rapid mixing result of the Glauber dynamics up to the uniqueness threshold without any other assumptions.
Then the notion of spectral independence is also generalized to multi-spin systems such as $q$-colorings~\cite{chen2021coloring,feng2022rapid}.
\cite{chen2020rapid} generalizes the result in~\cite{anari2020spectral} from the hardcore model to general anti-ferromagnetic two-spin systems, while also giving a sharp bound for spectral independence. Then \cite{chen2021optimal} refines this result to give an optimal $O(n \log n)$ mixing time of the Glauber dynamics for bounded degree instances.

Further developments may be roughly categorized into two lines of works.
One of them features removing the degree or marginal bounds assumption in~\cite{chen2021optimal}.
This line of work has recently reached this goal by introducing a new Markov chain called field dynamics and the entropy version of spectral independence called entropic independence~\cite{chen2021rapid, anari2022entropic, chen2022optimal, chen2022localization}.
%
Another line of work concerns establishing spectral independence.
Up to this date, a connection has been established between spectral independence and many of the techniques that was used to analyze and design fast samplers: correlation decay~\cite{anari2020spectral, chen2020rapid}; real-stability and zero-freeness of the partition functions~\cite{alimohammadi2021fractionally, chen2021spectral}; matrix trickle-down method~\cite{abdolazimi2021matrix, abdolazimi2022improved}; topological method~\cite{efthymiou2022spectral}; contractive coupling of local Markov chains~\cite{liu2021coupling, blanca2022mixing}.


\section{Organization of the paper}
The rest of the paper is organized as follows.
We start by introducing notations and conventions in~\Cref{sec:prelim}.
Then, in~\Cref{sec:delta-unique} we prove our characterization of uniqueness for the bipartite hardcore model with degree bound on one side, by locating fixpoints and analyzing their critical behavior.
Next, we establish spectral independence by proving correlation decay and crucially a contraction property in~\Cref{sec:SI}, by reducing the contraction rate of the system to the contraction rate near fixpoints.
As explained earlier, our algorithm is based on simulating the field dynamics on one side by a Glauber dynamics on one side.
For this to work, we need to identify a subcritical regime of bipartite hardcore model in which a coupling argument can succeed, so that the Glauber dynamics can efficiently simulate the field dynamics. 
This is shown in~\Cref{sec:GD-good}.
Then, we need to show that the field dynamics itself on one side is rapidly mixing. 
We show entropy decay of the field dynamics by building on results in negative fields localization schemes and field dynamics in~\Cref{sec:FD-ent-decay}.
Last but not the least, we show that the field dynamics on one side can be "approximately tensorized" into a single site Glauber dynamics on both sides in~\Cref{sec:GD-mu-mix}.
This allows us to use a comparison of variance decay to conclude a polynomial mixing time bound for the standard Glauber dynamics.
\section{Preliminaries} \label{sec:prelim}
\subsection{Notations and conventions}
\paragraph{Graph related}
Given a graph $G = (V, E)$ and a vertex $v \in V$, we use $\Gamma_v(G)$ to denote the set of neighbors of $v$ in graph $G$.
When the context is clear, we may omit $G$ and simply write $\Gamma_v$.
Given a vertex $v \in V$, we may use $\deg_G(v)$ to denote the degree of $v$ in $G$.

\paragraph{Vectors}
Let $U$ be some ground set and let $\sigma \in \^R^U$ be a vector.
For $i \in U$, we will use $\sigma_i$ or $\sigma(i)$ to denote the value of the $i$-th coordinate of $\sigma$.
Let $S \subseteq U$, let $\sigma_S \in \^R^S$ (or in some place, written as $\sigma(S)$) be a projection of $\sigma$ to $S$ such that for $i \in S$, $(\sigma_S)_i = \sigma_i$, and undefined for $i\not\in S$.
For $c \in \^R$, let $\sigma^{-1}(c) = \{i \in U \mid \sigma_i = c\}$ be the pre-image of $\sigma$.
For $\sigma \in \{-1, +1\}^U$, let $\abs{\sigma}_\pm := \abs{\sigma^{-1}(\pm 1)}$ to denote the number of $+1$s in $\sigma$.
For two vector $\sigma, \tau \in \^R^U$, we will use $\sigma \oplus \tau := \{i \in U \mid \sigma_i \neq \tau_i \}$ to denote the set of unequal coordinates of $\sigma$ and $\tau$; and we will use $\sigma \odot \tau$ to be the entry-wise product of $\sigma$ and $\tau$.
For convenience, for $\tau \in \^R^U$, we may also use $\overline{\tau}$ to denote the vector $\tau \odot \*{-1}$.
For $\Lambda \subseteq V$, we may use $\*1_{\Lambda}$ (and $\*{-1}_\Lambda$) to denote the all-$(+1)$ (and all-$(-1)$) state on $\Lambda$.
We may omit the subscript $\Lambda$ when the context is clear.

\paragraph{Distributions}
Let $\mu$ be a distribution over $\{-1, +1\}^U$.
Let $\Omega(\mu) \subseteq \{-1, +1\}^U$ be the support of $\mu$, that is, $X \in \Omega(\mu)$ iff $\mu(X) > 0$.
For $S\subseteq U$, let $\mu_S$ be the projection of distribution $\mu$ on the set $S$, that is, $\mu_S(\tau) := \sum_{\sigma: \sigma_S = \tau} \mu(\sigma), \forall \tau \in \{-1, +1\}^S$.
For convenience, if $S = \{i\}$, then we will write $\mu_i$ instead of $\mu_{\{i\}}$.
For $\Lambda \subseteq U$ and $\sigma \in \Omega(\mu_\Lambda)$, we let $\mu^\sigma$ be $\mu$ condition on the coordinates in $\Lambda$ be fixed to $\sigma$, that is, $\mu^\sigma(\tau) = \mu(\tau \mid \tau_S = \sigma)$.
Moreover, if $S \subseteq V\setminus \Lambda$, then we use the notation $\mu^\sigma_S := (\mu^\sigma)_S$.
For two distribution $\mu$ and $\nu$, we say $\mu$ is \emph{absolutely continuous} with respect to $\nu$ if $\Omega(\mu) \subseteq \Omega(\nu)$.

Given a distribution $\nu$ over $\{-1, +1\}^U$, we define its ``flipped'' version $\overline{\nu}$ as another distribution over $\{-1, +1\}^U$ defined by $\overline{\nu}(\sigma) := \nu(\*{-1} \odot \sigma), \forall \sigma \in \{-1, +1\}^U$.

In this work, most distributions are on $\{-1, +1\}^U$ for some ground set $U$.
However, for some minor cases, we will also use distributions over $2^U$.
We note that for a distribution $\mu$ over $2^U$, we could redefine $\mu$ on $\{-1, +1\}^U$ by considering $\mu(S)$ as $\mu(\*{-1}_{U\setminus S}, \*1_S)$.
Hence we do not distinguish these two cases.

\paragraph{Others}
All the logarithm used in this work is based on $\e$.
Let $n \geq 1$ be an integer, we use the notation $[n] := \{0, 1, \cdots, n-1\}$.

\subsection{Markov chain and related topics} \label{sec:prelim-MC}
\paragraph{Basic definitions}
Let $(X_t)_{t\in \^N}$ be Markov chain over a finite state space $\Omega$ with transition matrix $P \in \^R^{\Omega \times \Omega}_{\geq 0}$.
If it is clear from the context, we will also use $P$ to refer to the Markov chain directly.
We say the Markov chain is
\begin{itemize}
\item \emph{irreducible}, if for any $X, Y \in \Omega$, there is $t > 0$ such that $P^t(X, Y) > 0$;
\item \emph{aperiodic}, if for any $X \in \Omega$, it holds that $\-{gcd}\{t > 0 \mid P^t(X, X) > 0\} = 1$.
\end{itemize}
A distribution $\mu$ over $\Omega$ is call a stationary distribution of $P$ if $\mu = \mu P$, where we use $\mu$ as a row vector.
If a Markov chain is both irreducible and aperiodic, then it has a unique stationary distribution.
The Markov chain $P$ is \emph{reversible} with respect to $\mu$ if the following \emph{detailed balanced equation} holds
\begin{align*}
  \forall X, Y, \in \Omega, \quad \mu(X)P(X, Y) = \mu(X)P(Y, X).
\end{align*}
This also implies that $\mu$ is a stationary distribution of $P$.

\paragraph{Variance}
Let $P$ be the transition matrix of a reversible, irreducible, and aperiodic Markov chain with stationary distribution $\mu$.
Let $\lambda_2(P)$ be the second largest eigenvalue of $P$.
Let $\lambda_\star(P) := \max \{\abs{\lambda} \mid \text{$\lambda$ is an eigenvalue of $P$, $\lambda \neq 1$}\}$ be the second largest eigenvalue of $P$ in terms of absolute value (see \cite[Lemma 12.1]{levin2017markov}).
For convenience, let $\sgap{P} := 1 - \lambda_2(P)$ and $\sgap[\star]{P} = 1 - \lambda_\star(P)$, be the spectral gap and absolute spectral gap of $P$, respectively.
Given any function $f: \Omega(\mu) \to \^R$, the \emph{variance} of $f$ with respect to $\mu$ is defined as
\begin{align*}
  \Var[\mu]{f} := \E[\mu]{ \tp{f - \E[\mu]{f}}^2 } = \frac{1}{2} \sum_{X, Y \sim \Omega(\mu)} \mu(X) \mu(Y) (f(X) - f(Y))^2.
\end{align*}
Given any function $f, g: \Omega(\mu) \to \^R$, the inner product of $f, g$ with respect to $\mu$ is
\begin{align*}
  \inner{f}{g}_\mu := \sum_{X\in \Omega(\mu)} \mu(X) f(X) g(X).
\end{align*}
The \emph{Dirichlet form} associated to $(P, \mu)$ for function $f$, $g$ is 
\begin{align*}
  \+E_P(f, g)
  &:= \inner{(I - P) f}{g}_\mu 
  = \frac{1}{2} \sum_{X, Y \in \Omega(\mu)} \mu(\sigma) P(X, Y) (f(X) - f(Y)) (g(X) - g(Y)).
\end{align*}
We also use $\+E_P(f)$ to denote $\+E_P(f, f)$.
The spectral gap could be characterized by
\begin{align*}
  \sgap{P} &= \min_{\substack{f: \Omega(\mu) \to \^R \\ \Var[\mu]{f} \neq 0}} \frac{\+E_P(f)}{\Var[\mu]{f}}.
\end{align*}
See~\cite[Remark 13.8]{levin2017markov} for details.
The Poincar{\'e} inequality follows from above characterization:
\begin{align}
  \forall f:\Omega(\mu) \to \^R, \quad \sgap{P} \cdot \Var[\mu]{f} \leq \+E_P(f).\label{eq:Poincare-inequality}
\end{align}
The mixing time and the eigenvalues of a Markov chain are closely related.
Below are some classic results under this topic.

\begin{lemma}[\text{\cite[Corollary 12.7]{levin2017markov}}] \label{lem:mix=var}
  Let $P$ be a reversible, irreducible, and aperiodic Markov chain with stationary $\mu$.
  It holds that
  \begin{align*}
    \lim_{t\to\infty} \max_{X, \in \Omega(\mu)} \DTV{P^t(X, \cdot)}{\mu}^{1/t} = \lambda_\star(P).
  \end{align*}
\end{lemma}

\begin{lemma}[\text{\cite[Theorem 12.4]{levin2017markov}}] \label{lem:var-to-mix}
  Let $P$ be a reversible, irreducible, and aperiodic Markov chain with stationary $\mu$.
  It holds that
  \begin{align*}
    \Tmix{\epsilon} \leq \frac{1}{\sgap[\star]{P}} \log \tp{\frac{1}{\epsilon \mu_{\min}}},
  \end{align*}
  where $\mu_{\min} := \min_{X \in \Omega(\mu)} \mu(X)$.
\end{lemma}

\begin{remark}
  Note that \Cref{lem:mix=var} and \Cref{lem:var-to-mix} only work for $\sgap[\star]{P}$.
  However, as we will see later, all the Markov chains considered in this work are irreducible, aperiodic, reversible and only have non-negative eigenvalues.
  Hence \Cref{lem:mix=var} and \Cref{lem:var-to-mix} also work for $\sgap{P}$ in our setting.
\end{remark}

\paragraph{Entropy}
The \emph{KL-divergence} is also used to measure how close two distributions are.
When $\nu$ is absolutely continuous with respect to $\mu$, it is defined as
\begin{align*}
  \DKL{\nu}{\mu} := \sum_{X \in \Omega(\mu)} \nu(X) \log \frac{\nu(X)}{\mu(X)},
\end{align*}
where we use the convention that $0 \log 0 = 0$.
We note that $\DKL{\mu}{\nu}$ may not equal to $\DKL{\nu}{\mu}$ so that $\DKL{\cdot}{\cdot}$ is not a metric.
The total variation distance and the KL-divergence is connected by the well-known Pinsker's inequality.
Suppose $\nu$ is absolutely continuous with respect to $\mu$, 
\begin{align*}
  \DTV{\nu}{\mu} \leq \sqrt{\frac{1}{2} \DKL{\nu}{\mu}}.
\end{align*}
We also use a notion of entropy that is closely connected to the KL-divergence.
For any function $f:\Omega \to \^R_{\geq 0}$, the entropy of $f$ with respect to $\mu$ is defined as
\begin{align*}
  \Ent[\mu]{f} := \E[\mu]{f\log \frac{f}{\E[\mu]{f}}},
\end{align*}
where we also use the convention that $0\log 0 = 0$.
Note that if we let $f = \frac{\nu}{\mu}$, then $\Ent[\mu]{f} = \DKL{\nu}{\mu}$.

Most result for entropy also works for variance (one-way direction).
The key insight here is a standard trick called linearization~\cite{rothaus1981diffusion, jerrum2003counting, caputo2015approximate}.
\begin{lemma} \label{lem:ent-to-var}
  Let $\mu$ be a distribution.
  For every $f: \Omega(\mu) \to \^R$ , and sufficiently small $\epsilon > 0$, it holds that
  \begin{align*}
    \Ent[\mu]{1 + \epsilon f} = \frac{\epsilon^2}{2} \Var[\mu]{f} + o(\epsilon^2).
  \end{align*}
\end{lemma}

\begin{proof}
  Note that by the Taylor's series, it holds that
  \begin{align*}
    (x + 1) \log (x + 1) = x + \frac{x^2}{2} + o(x^2).
  \end{align*}
  Then by the definition, we have
  \begin{align*}
    \Ent[\mu]{1 + \epsilon f}
    &= \E[\mu]{(1 + \epsilon f)\log (1 + \epsilon f)} - \E[\mu]{1 + \epsilon f} \log \E[\mu]{1 + \epsilon f} \\
    &= \E[\mu]{\epsilon f + \frac{\epsilon^2 f^2}{2} + o(\epsilon^2)} - \epsilon \E[\mu]{f} - \frac{\epsilon^2 \E[\mu]{f}^2}{2} - o(\epsilon^2) \\
    &= \frac{\epsilon^2}{2} \tp{\E[\mu]{f^2} - \E[\mu]{f}^2} + o(\epsilon^2). \qedhere
  \end{align*}
\end{proof}

\paragraph{Glauber dynamics}
The Glauber dynamics $(X_t)_{t\in\^N}$ (a.k.a. Gibbs sampler) is the canonical single-site Markov chain for sampling from distribution $\mu$ over $\{-1, +1\}^n$.
Let $X_0$ is picked arbitrary from $\Omega(\mu)$.
Then, in the $t$-th step, the Glauber dynamics does as follows:
\begin{enumerate}
\item pick a coordinate $i \in [n]$ uniformly at random;
\item sample $X_t \sim \mu(\cdot \mid X_{t-1}(V\setminus \{i\}))$.
\end{enumerate}
It can be verify that $\mu$ is reversible with respect to Glauber dynamics.
We will use $\GD[\mu]$ to denote the transition matrix of the Glauber dynamics on $\mu$.
There is a classic result showing that the transition matrix of Glauber dynamics only has non-negative eigenvalues (see \cite{dyer2014structure, levin2017markov, alev2020improved}).

\begin{lemma} \label{lem:GD-psd}
  For distribution $\mu$ over $\{-1, +1\}^n$.
  The Glauber dynamics for $\mu$ is reversible with respect to $\mu$, and its transition matrix only has non-negative eigenvalues.
\end{lemma}

\section{Uniqueness condition for the bipartite hardcore model}
\label{sec:delta-unique}
In this section, 
we will characterize the $\delta$-uniqueness  in the bipartite hardcore model and prove \Cref{thm:delta-unique} and a weaker version of \Cref{thm:delta-unique-eq},
by exactly resolving the $\delta$-uniqueness conditions.
Along the way,  in \Cref{sec:0-unique-eq}, we will prove the explicit uniqueness criterion stated in \Cref{thm:delta-unique-lambda-d-alpha-w}.
And \Cref{thm:delta-unique-eq} will be fully proved in \Cref{sec:SI} by using potential analysis of correlation decay.

Throughout the section, we assume the following setting for parameters: 
\begin{align}\label{eq:uniqueness-parameters}
\delta\in[0,1),\quad \lambda,\alpha>0,\quad d \geq 1 - \delta,\quad \text{ and }\quad w > 0.
\end{align}
As a reminder, $\delta$ is for $\delta$-uniqueness, $\lambda$ is the fugacity on $L$ (one side of the bipartite graph), $d+1$ is the maximum degree on $L$, and $\alpha$ is the fugacity on $R$ (the other side of the bipartite graph), and $w+1$ is a parameter that can be roughly regarded as a ``fractional degree'' on $R$.
This section is organized as follows:
\begin{enumerate}
	\item \label{item:delta-unique-1} First, we show that for any given $d, \alpha, \delta$, there is an implicitly defined threshold $\lambda_{2, c}^\delta$ as a function of $d$, $\alpha$ and $\delta$, such that $(\lambda, d, \alpha)$ is $\delta$-unique iff $\lambda \geq \lambda_{2,c}^\delta$. This is achieved by \Cref{thm:delta-unique-implicit}.
\item \label{item:delta-unique-2} Then, we will solve for $\lambda_{2, c}^\delta$ by studying a system of equations that characterizes $\lambda_{2, c}^\delta$. This is achieved by \Cref{thm:resolve-lambda-2-c}.
\item \label{item:delta-unique-3} Finally, we impose further that $\lambda = \alpha$ in the system  as in \eqref{eq:delta-sys-eq}, which we resolve in \Cref{thm:solve-sys-eq}. In the special case of $\delta = 0$, we will also solve it explicitly.
\end{enumerate}
\Cref{thm:delta-unique} will be proved by \Cref{item:delta-unique-1} and \Cref{item:delta-unique-2}.
Then combining \Cref{thm:delta-unique} and \Cref{item:delta-unique-3} will lead us to \Cref{thm:delta-unique-lambda-d-alpha-w} and a part of \Cref{thm:delta-unique-eq} (sufficiently small $\delta$).
We start by observing some nice properties about the fixpoints.




Recall the two-level tree-recursion for marginal ratios deduced in \Cref{sec:uniqueness}:
\begin{align}\label{eq:univariate-tree-recursion-F(x)}
F(x) := \lambda (1 + \alpha(1 + x)^{-w})^{-d}.
\end{align}

By \Cref{def:lambda-d-alpha-w-unique}, the tuple $(\lambda, d, \alpha, w)$ is $\delta$-unique, if and only if $F$ only has attractive fixpoints:
\begin{align} \label{eq:lambda-d-alpha-w-unique}
  \forall \hat{x} \geq  0 \text{ such that } F(\hat{x}) = \hat{x}, \text{ we have } F'(\hat{x}) \leq 1 - \delta.
\end{align}
The constraint in \eqref{eq:lambda-d-alpha-w-unique} is imposed on the fixpoints $\hat{x}$ of $F$, which seems like nowhere to launch an analysis at first glance.
Here, we have the following crucial observation.

\begin{observation}\label{obs:hat-x-lambda}
  Given $d, \alpha, w$ and $\hat{x} > 0$, there is a unique $\lambda > 0$ such that $\hat{x}$ is a fixpoint of $F$.
\end{observation}
\begin{proof}
  If $F(\hat{x}) = \hat{x}$, then by \eqref{eq:univariate-tree-recursion-F(x)} we have $\lambda = \hat{x}(1 + \alpha(1 + \hat{x})^{-w})^d$.
\end{proof}
\Cref{obs:hat-x-lambda} suggests that we define the following function:
\begin{align} \label{eq:lambda-x}
  \lambda(x) := x(1 + \alpha(1 + x)^{-w})^d.
\end{align}
Then, instead of resolving \eqref{eq:lambda-d-alpha-w-unique} directly,  one can first resolve \eqref{eq:lambda-d-alpha-w-unique} under the coordinate system $(\hat{x}, d, \alpha, w)$, then use the function $\lambda(x)$ to translate the result from $(\hat{x}, d, \alpha, w)$ to $(\lambda, d, \alpha, w)$.

To resolve \eqref{eq:lambda-d-alpha-w-unique} under the coordinate system $(\hat{x}, d, \alpha, w)$, note that $F'(x)$ is given by
\begin{align}\label{eq:F(x)-derivative}
  F'(x) &= \alpha d w (x + 1)^{-w-1} (1 + \alpha (x + 1)^{-w})^{-1} F(x).
\end{align}
And hence $F'(x) \leq 1 - \delta$ if and only if
\begin{align*}
  (1 - \delta) (x + 1) (\alpha + (1 + x)^w) \geq \alpha d w F(x).
\end{align*}
Recall that at a fixpoint we have $\hat{x}=F(\hat{x})$.
This motivates us to define the following function:
\begin{align}\label{eq:function-T-delta}
  T_\delta(x) := (1 - \delta)(x + 1)(\alpha + (1 + x)^w) - \alpha d w x.
\end{align}
For a tuple $(\hat{x}, d, \alpha, w)$, the $\delta$-uniqueness condition $F'(\hat{x})\le 1-\delta$ at the fixpoint $\hat{x}=F(\hat{x})$ is equivalent to $T_\delta(\hat{x})\ge 0$.
%
In the coordinate system $(\hat{x}, d, \alpha, w)$, the $\delta$-uniqueness changes sharply around the root of the equation $T_\delta(\hat{x}) = 0$.
Moreover, the sign of $\lambda'(x)$ is also governed by the sign of $T_0(x)$ (see \Cref{sec:delta-unique-implicit} for more details). 
To some extent, the roots of $T_\delta(x)=0$ is what characterizes the uniqueness regime.
Hence, we need to study the roots of the equation $T_\delta(x) = 0$, and we summarize the result in the following results.

\begin{fact}\label{fact:T-delta}
The followings hold for the function $T_\delta(x)$ defined in \eqref{eq:function-T-delta}:
\begin{align*}
  T_\delta'(x) &= (1 - \delta)(1 + w)(1 + x)^w - \alpha(d w  - (1 - \delta));\\
  T_\delta''(x) &= (1 - \delta)w(1 + w)(1 + x)^{w - 1} > 0;\\ 
\lim_{x \to 0} T_\delta(x) &= (1 - \delta)(1 + \alpha) > 0.
\end{align*}
\end{fact}
\begin{lemma} \label{lem:T-root-lifecycle}
Fix any $\delta \in [0,1)$, $\alpha>0$, and $d \geq 1 - \delta$.
  The followings hold for the roots of $T_\delta(x) = 0$ as an equation in $x$:
  \begin{enumerate}
  \item if $\alpha \leq \frac{1-\delta}{d} \cdot \e^{1 + \frac{1-\delta}{d}}$, then the equation $T_\delta(x) = 0$ has no positive solution;
  \item if $\alpha > \frac{1-\delta}{d} \cdot \e^{1 + \frac{1-\delta}{d}}$, 
  then the following equation in $w$:
  \begin{align}\label{eq:equation-w-delta}
      \alpha &= \frac{(1 - \delta) d^w (w + 1)^{w+1}}{(d w - (1 - \delta))^{w+1}}
   \end{align}
  has a unique positive solution $w_\delta$, such that
    \begin{itemize}
    \item if $w < w_\delta$, the equation $T_\delta(x) = 0$ has no positive solution;
    \item if $w = w_\delta$, the equation $T_\delta(x) = 0$ has a unique positive solution;
    \item if $w > w_\delta$, the equation $T_\delta(x) = 0$ has two positive solutions.
    \end{itemize}
  \end{enumerate}
\end{lemma}
The proof of \Cref{lem:T-root-lifecycle}  is postponed  to \Cref{sec:T-root-lifecycle}.
We remark that the roots of $T_\delta(x) = 0$ is what governs the location of the (unique) positive fixpoint of $F$.

Given $\alpha, d$, when the equation $T_\delta(x) = 0$ has at least one solution, we use the following notation to refer to its roots. The corresponding $\lambda$ is obtained by translating the roots back to $\lambda$ using \eqref{eq:lambda-x}.

\begin{definition}\label{def:x-lambda-delta}
Let $\alpha > \frac{1-\delta}{d} \cdot \e^{1 + \frac{1-\delta}{d}}$. 
Define $w_\delta=w_\delta(d,\alpha)$ to be the unique positive solution of \eqref{eq:equation-w-delta}.
Let $w \geq w_\delta$.
Define the positive roots of $T_\delta(x) = 0$ to be 
\begin{align*}
x^\delta_1
:=x^\delta_1(w)
\quad \le \quad
x^\delta_2:=x^\delta_2(w).
\end{align*}
For $i \in \{1, 2\}$, let $\lambda^\delta_i=\lambda^\delta_i(w)$ be defined by 
\[
\lambda^\delta_i(w) := \lambda(x^\delta_i(w))=x^\delta_i(w)(1 + \alpha(1 + x^\delta_i(w))^{-w})^d.
\]
And define
\begin{align*}
  \lambda^\delta_{1, c} := \inf_{w > w_\delta} \lambda^\delta_1(w) \quad \text{and} \quad \lambda^\delta_{2,c} := \sup_{w > w_\delta} \lambda^\delta_2(w).
\end{align*}
Furthermore, when $\alpha \le \frac{1-\delta}{d} \cdot \e^{1 + \frac{1-\delta}{d}}$, we define $\lambda^\delta_{1, c}=\lambda^\delta_{2, c}=0$ by convention.
\end{definition}
The functionality of functions and variables defined in \Cref{def:x-lambda-delta} is illustrated in \Cref{fig:illustrate-unique}.
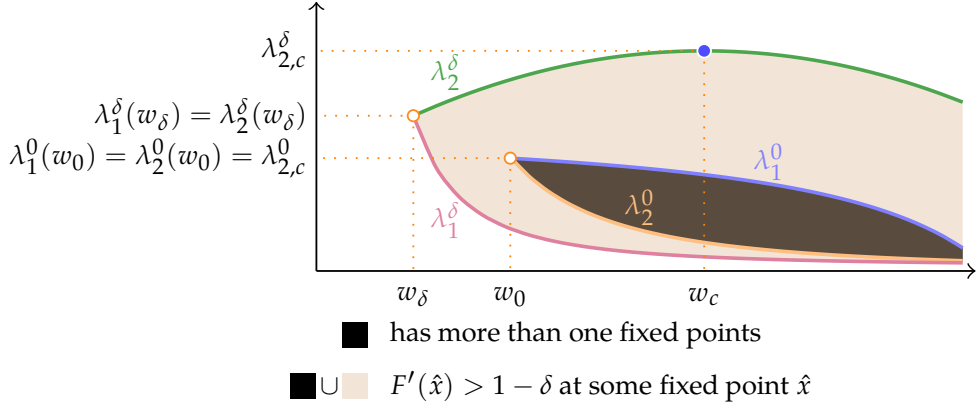
\begin{figure}[h]
  \centering
  {\centering
    
    \scalebox{1.7}{

      \begin{tikzpicture}
        \fill[domain=1.5:5, smooth, variable=\x, black, thick]
        (5, 0) -- plot ({6.5-\x}, {1.18-1.5/ \x}) -- (1.5, 0) -- cycle;
        \fill[domain=1.5:5, smooth, variable=\x, white, thick]
        (1.45, 0) -- plot ({\x}, {2/\x / \x}) -- (5.05, 0) -- cycle;

        \fill[domain=0.75:5, smooth, variable=\x, brown!50, opacity=0.4]
        (0.75, 1.211)
        -- plot ({\x}, {1.717 - 0.1 * (\x - 3) * (\x - 3)})
        -- (5, 0.168)
        -- plot ({5.75 - \x}, {0.66/ (5.75 - \x) / (5.75 - \x) + 0.037});

        \draw[domain=1.5:5, smooth, variable=\x, blue!50, thick]
        plot ({6.5-\x}, {1.18-1.5/ \x});
        \draw[domain=1.5:5, smooth, variable=\x, orange!50, thick] 
        plot ({\x}, {2/\x / \x});

        \draw[domain=0.75:5, smooth, variable=\x, purple!50, thick] 
        plot ({\x}, {0.66/\x / \x + 0.037});
        \draw[domain=0.75:5, smooth, variable=\x, green!50!black!70, thick] 
        plot ({\x}, {1.717 - 0.1 * (\x - 3) * (\x - 3)});
        
        \draw[->] (0, 0) to (5.1, 0);  
        \draw[->] (0, 0) to (0, 2.1);  

        \draw[-, dotted, draw = orange] (1.5, 0.88) to (1.5, 0);
        \node at (1.5, -0.2) {\tiny $w_0$};
        \draw[-, dotted, draw = orange] (0.75, 1.211) to (0.75, 0);
        \node at (0.75, -0.2) {\tiny $w_\delta$};
        \draw[-, dotted, draw = orange] (1.5, 0.88) to (0, 0.88);
        \node[label={left:\tiny $\lambda^0_1(w_0) = \lambda^0_2(w_0) = \lambda^0_{2, c}$}] at (0.2, 0.88) {};
        \draw[-, dotted, draw = orange] (0.75, 1.211) to (0, 1.211);
        \node[label={left:\tiny $\lambda^\delta_1(w_\delta) = \lambda^\delta_2(w_\delta)$}] at (0.2, 1.211) {};
        \draw[-, dotted, draw = orange] (3, 1.717) to (0, 1.717);
        \node[label={left:\tiny $\lambda^\delta_{2,c}$}] at (0.2, 1.717) {};
        \draw[-, dotted, draw = orange] (3, 1.717) to (3, 0);
        \node at (3, -0.2) {\tiny $w_c$};

        \node at (1, 0.4) {\color{purple!50} \tiny $\lambda^\delta_1$};
        \node at (1, 1.55) {\color{green!50!black!70} \tiny $\lambda^\delta_2$};
        \node at (2.5, 0.5) {\color{orange!50} \tiny $\lambda^0_2$};
        \node at (3.5, 0.85) {\color{blue!50} \tiny $\lambda^0_1$};

        \node[draw = orange!90, fill=white, circle, inner sep=0pt, minimum size = 2.5pt] at (1.5, 0.88) {};
        \node[draw = orange!90, fill=white, circle, inner sep=0pt, minimum size = 2.5pt] at (0.75, 1.2111) {};

        \node[draw = white, fill=blue!70, circle, inner sep=0pt, minimum size = 3pt] at (3, 1.717) {};

        \fill [black] (0.2, -0.4) -- (0.4, -0.4) -- (0.4, -0.6) -- (0.2, -0.6) -- cycle;
        \node[label={right: \tiny has more than one fixed points}] at (0.3, -0.5) {};
        \fill [brown!50, opacity=0.4] (0.2, -0.8) -- (0.4, -0.8) -- (0.4, -1.0) -- (0.2, -1.0) -- cycle;
        \node[label={right: \tiny $F'(\hat{x}) > 1 - \delta$ at some fixed point $\hat{x}$ }] at (0.3, -0.9) {};
        \fill [black] (-0.2, -0.8) -- (-0.0, -0.8) -- (-0.0, -1.0) -- (-0.2, -1.0) -- cycle;
        \node at (0.1, -0.9) {\tiny $\cup$};

      \end{tikzpicture}

    } 

    \par}
  \caption{\small Fix $d \geq 1-\delta$ and $\frac{1}{d} \cdot \e^{1 + \frac{1}{d}} < \alpha$, we note that $\lambda^\delta_i, \lambda^0_i$ ($i \in \{1, 2\}$) capture the non-unique regime and the non-contract regime for the pair $(\lambda, d, \alpha, w)$. Note that $\lambda^0_{2,c}$ is achieved at $w_0$ (\Cref{lem:lambda-c-monotone-0}); $\lambda^\delta_{2, c} > \lambda^0_{2, c}$ (\Cref{lem:lambda-c-monotone-delta}); and the maximum of $\lambda^\delta_2$ is achieved at some unique point $w_c > w_\delta$ (\Cref{thm:resolve-lambda-2-c}). Also note that as long as $w \to \infty$, we have $\lambda^\delta_i, \lambda^0_i$ ($i \in \{1, 2\}$) approach to $0$ (\Cref{lem:lb-is-0}).}
  \label{fig:illustrate-unique}
\end{figure}

The following limits of $\lambda^\delta_i(w)$ and $\lambda^\delta_{i,c}$ justify the convention assumed in \Cref{def:x-lambda-delta}. 
\begin{lemma} \label{lem:lb-is-0}
If $\alpha>\frac{1-\delta}{d} \cdot \e^{1 + \frac{1-\delta}{d}}$, then for $i \in \{1, 2\}$, it holds that
  \begin{align*}
    \lim_{w \to +\infty} \lambda^\delta_i(w) = \lim_{w \to +\infty} x^\delta_i(w) = 0.
  \end{align*}
  Consequently, $\lambda_{i,c}^{\delta}\downarrow0$ as $\alpha\downarrow \frac{1-\delta}{d} \cdot \e^{1 + \frac{1-\delta}{d}}$.
\end{lemma}
\begin{proof}
 It is easy to verify that $\lambda^\delta_i(w) = x^\delta_i(w)$ as $w\to+\infty$.
  It is then sufficient to show that $x^\delta_i(w) = 0$ as $w\to+\infty$.
  Fix an arbitrary $x > 0$. We have that $T_\delta(x)$ is finite when $w>0$ is finite and 
  \begin{align*}
    \lim_{w \to +\infty} T_\delta(x) = \lim_{w\to +\infty} (1 - \delta)(x + 1)(\alpha + (1 + x)^w) - \alpha d w x = +\infty.
  \end{align*}
  Hence, there is a finite $w_a(x)$ such that $T_\delta(x) > 0$ for all $w \geq w_a(x)$.
  Moreover, we also have that $T'_\delta(x)$ is finite when $w>0$ is finite and 
  \begin{align*}
    \lim_{w \to +\infty} T_\delta'(x) = \lim_{w\to +\infty} (1 - \delta)(1 + w)(1 + x)^w - \alpha(d w - (1 - \delta)) = +\infty.
  \end{align*}
  Together with the fact that $T''_{\delta}(x) > 0, \forall x > 0$ in \Cref{fact:T-delta}, there is a finite $w_b(x)$ such that $T_\delta'(y) > 0$ for all $w \geq w_b(x)$ and $y \geq x$.
  This means that when $w \geq \max\{w_a(x), w_b(x)\}$, it holds that $x^\delta_i(w) < x$ for $i \in \{1, 2\}$.
  This shows that  for any $\epsilon > 0$, $x^\delta_i(w) < \epsilon$ for all sufficiently large $w>0$, i.e.~$\lim_{w \to +\infty} x^\delta_i(w) = 0$.
  
  Next, as $\alpha\downarrow \frac{1-\delta}{d} \cdot \e^{1 + \frac{1-\delta}{d}}$, it can be verified that the unique positive solution $w_\delta$ of  \eqref{eq:equation-w-delta} approaches $+\infty$.
  Fix $x > 0$, when $\alpha\downarrow \frac{1-\delta}{d} \cdot \e^{1 + \frac{1-\delta}{d}}$ and $w \to \infty$, we also have $T_\delta(x)$ and $T'_\delta(x)$ approach to $+\infty$.
  Hence, by a similar argument, both $\lambda^\delta_{1,c} = \inf_{w > w_\delta}\lambda^\delta_1(w)$ and $\lambda^\delta_{2,c} = \sup_{w > w_\delta}\lambda^\delta_2(w)$ approach from the left to the limit $\lim_{w \to +\infty} \lambda^\delta_i(w)=0$. 
\end{proof}

For fixed $\alpha, d$ satisfying \eqref{eq:uniqueness-parameters}, using the notation in \Cref{def:x-lambda-delta}, we first resolve the $\delta$-uniqueness requirement under the coordinate system $(\hat{x}, d, \alpha, w)$, then use the function $\lambda(x)$ in \eqref{eq:lambda-x} to translate the result to the coordinate system $(\lambda, d, \alpha, w)$.
Notice that different $\hat{x}$ might be mapped to the same $\lambda$, the $\delta$-uniqueness regime in terms of $(\lambda, d, \alpha, w)$ defined in \eqref{eq:lambda-d-alpha-w-unique} is quite complicated as illustrated in \Cref{fig:illustrate-unique}.

However, if we strengthen the $\delta$-uniqueness requirement to hold for all $w > 0$ as in \Cref{def:lambda-d-alpha-unique}, then all complicated cases will collapse into a single and elegant requirement on the parameter $\lambda$.
The intuitive ideas are illustrated in \Cref{fig:illustrate-unique}.
By \Cref{def:lambda-d-alpha-unique}, the tuple $(\lambda, d, \alpha)$ is $\delta$-unique if and only if
\begin{align} \label{eq:lambda-d-alpha-unique}
  \forall w > 0, \quad \text{the tuple $(\lambda, d, \alpha, w)$ is $\delta$-unique}.
\end{align}
Under the strengthened $\delta$-uniqueness requirement in \eqref{eq:lambda-d-alpha-unique}, we have the following result.
\begin{theorem} \label{thm:delta-unique-implicit}
  Assuming \eqref{eq:uniqueness-parameters},
  the tuple $(\lambda, d, \alpha)$ is $\delta$-unique iff $\lambda \geq  \lambda^{\delta}_{2,c}$.   
\end{theorem}
The proof of \Cref{thm:delta-unique-implicit} is given in \Cref{sec:delta-unique-implicit}.

Notice that when $\alpha \leq \frac{1-\delta}{d}\e^{1 + \frac{1-\delta}{d}}$, by definition, $\lambda^\delta_{2,c} = 0$.
To prove \Cref{thm:delta-unique}, it is sufficient for us to resolve $\lambda^\delta_{2,c}$ when $\alpha > \frac{1-\delta}{d}\e^{1 + \frac{1-\delta}{d}}$.
In this case, by \Cref{def:x-lambda-delta}, we know that $\lambda^\delta_2$ is an implicit function of $w$.
Since we want to resolve the supremum $\lambda^\delta_{2,c}$ of the function $\lambda^\delta_2(w)$ for $w > w_\delta$, it is natural for us to consider the first order condition of the function $\lambda^\delta_2(w)$.
Towards understanding the sign of $(\lambda^\delta_2)'(w)$, we meet another function which is crucial to the analysis:
\begin{align} \label{eq:M}
  M_\delta(x) := w \log(1 + x) (\alpha d - (1 + x)^{w+1}) + \delta(x + 1)(\alpha + (x + 1)^w).
\end{align}
Recall $T_\delta$ in \eqref{eq:function-T-delta}, it turns out that $\lambda^\delta_{2,c}$ can be explicitly resolved by levering both $T_\delta$ and $M_\delta$.

\begin{theorem} \label{thm:resolve-lambda-2-c}
  Assuming \eqref{eq:uniqueness-parameters}, if $\alpha > \frac{1 - \delta}{d} \cdot \e^{1 + \frac{1 - \delta}{d}}$, then $\lambda^\delta_{2,c} = \lambda(x):=x(1 + \alpha(1 + x)^{-w})^d$, where the pair $(x, w) = (x_c, w_c)$ is the unique positive solution of the following system
  \begin{align*}
      \begin{cases}
        T_\delta := (1 - \delta)(x + 1)(\alpha + (1 + x)^w) - \alpha d w x = 0, \\
        M_\delta := w \log(1 + x) (\alpha d - (1 + x)^{w+1}) + \delta(x + 1)(\alpha + (x + 1)^w) = 0.
      \end{cases}
  \end{align*}
\end{theorem}
The proof of \Cref{thm:resolve-lambda-2-c} is given in \Cref{sec:delta-unique-general}.

\begin{proof}[Proof of \Cref{thm:delta-unique}]
  \Cref{thm:delta-unique-implicit} and \Cref{thm:resolve-lambda-2-c} together imply \Cref{thm:delta-unique} directly.
\end{proof}

Using \Cref{thm:delta-unique}, we go further to understand the $\delta$-uniqueness when $\lambda = \alpha$ is required.
Given $\alpha, d$ that satisfies \eqref{eq:uniqueness-parameters}, if $\alpha \leq \frac{1-\delta}{d} \cdot \e^{1 + \frac{1-\delta}{d}}$, then by the definition of $\lambda^\delta_{2, c}$ in \Cref{def:x-lambda-delta}, it holds that $\lambda = \alpha \geq 0 = \lambda^\delta_{2,c}$.
Together with \Cref{thm:delta-unique-implicit}, it holds that the tuple $(\lambda, d, \alpha)$ is $\delta$-unique as defined in \eqref{eq:lambda-d-alpha-unique}.
If $\alpha > \frac{1 - \delta}{d} \cdot \e^{1 + \frac{1-\delta}{d}}$, it is natural to add an extra equation which ensures $\lambda =\lambda(x)=\alpha$ into the system in \Cref{thm:resolve-lambda-2-c}.
This leads us to the following system.
\begin{align} \label{eq:delta-sys-eq}
  \begin{cases}
    T_\delta := (1 - \delta)(x + 1)(\alpha + (1 + x)^w) - \alpha d w x = 0, \\
    M_\delta := w \log(1 + x) (\alpha d - (1 + x)^{w+1}) + \delta(x + 1)(\alpha + (x + 1)^w) = 0, \\
    G       := x(1 + \alpha(1 + x)^{-w})^d - \alpha = 0.
  \end{cases}
\end{align}

Surprisingly, when $\delta = 0$, the system in \eqref{eq:delta-sys-eq} forms a beautiful symmetric structure, and can be resolved explicitly (as in \Cref{thm:delta-unique-lambda-d-alpha-w}, and will be formally restated and proved in \Cref{sec:0-unique-eq}).

Furthermore, viewing $\delta(\alpha)$ as an implicit function encoded by \eqref{eq:delta-sys-eq}, by calculating $\delta'(\alpha)$, we can approximately solve \eqref{eq:delta-sys-eq} in a sufficiently small neighborhood of $\alpha = \lambda_c(\Delta)$.

\begin{theorem} \label{thm:solve-sys-eq}
  For $\Delta = d + 1 \geq 3$, let $\delta \in [0, 1)$ be a sufficiently small real number.
  For $\lambda = (1 - \delta)\lambda_c(\Delta)$, with $\alpha = \lambda$, the system defined in \eqref{eq:delta-sys-eq} has a solution $(\delta_c, w_c, x_c)$, such that $\delta_c \geq \frac{\delta}{4}$.
  In particular, $\delta_c = 0$ when $\lambda = \lambda_c(\Delta)=\frac{(\Delta-1)^{\Delta-1}}{(\Delta-2)^{\Delta}}$.
\end{theorem}
The proof of \Cref{thm:solve-sys-eq} will be given in \Cref{sec:delta-unique-eq}.

Note that \Cref{thm:solve-sys-eq} only proves \Cref{thm:delta-unique-eq} for sufficiently small $\delta$.
And it is very unlikely that one could resolve \eqref{eq:delta-sys-eq} for every $\delta \in (0, 1)$.
In \Cref{sec:delta-unique-eq-7}, we will bypass this barrier and prove \Cref{thm:delta-unique-eq} using the analysis of correlation decay.


\subsection{The $\delta$-uniqueness regime for $(\lambda,d)$}
\label{sec:delta-unique-eq}
In this section, we prove \Cref{thm:solve-sys-eq}.
Let $\Delta = d + 1 \geq 3$, and $\delta \in [0, 1)$ be a sufficiently small real number.
%
%
%
%
%
First, we resolve the $\delta = 0$ case without assuming $\lambda = \alpha$, that is, we will assume $\delta = 0$ and remove the equation $G = 0$ from the system in \eqref{eq:delta-sys-eq}.
This simplifies \eqref{eq:delta-sys-eq} a lot as follows.
\begin{align} \label{eq:0-unique-sys}
  \begin{cases}
    T_0 := (x + 1)(\alpha + (1 + x)^w) - \alpha d w x = 0, \\
    M_0 := \alpha d - (1 + x)^{w+1} = 0.
  \end{cases}
\end{align}
\eqref{eq:0-unique-sys} can be solved exactly.
We define the following function to state the result conveniently.
\begin{align} \label{eq:A}
  A(d, w, \delta) := \frac{(1 - \delta) d^w (w + 1)^{w+1}}{(d w - (1 - \delta))^{w+1}}.
\end{align}
Recall that the ${\boldsymbol{\hat\lambda}}(d, w)$ defined in \eqref{eq:A0} is just ${\boldsymbol{\hat\lambda}}(d, w)=A(d, w, 0)$.
\begin{lemma} \label{lem:0-unique}
  Let $d \geq 1$ and $\alpha > \frac{1}{d} \e^{1 + \frac{1}{d}}$ and let $(x, w) = (x_c, w_c)$ be the solution of \eqref{eq:0-unique-sys}, it holds that
    \begin{align*} 
      \alpha = \frac{d^w (w + 1)^{w+1}}{(d w - 1)^{w+1}} = A(d, w, 0), \quad \text{and} \quad 
      \lambda(x) = \frac{w^d (d+1)^{d+1}}{(d w - 1)^{d+1}} = A(w, d, 0),
    \end{align*}
    where $\lambda(x) = x (1 + \alpha(1 + x)^{-w})^d$ as defined in \eqref{eq:lambda-x}.
\end{lemma}

\begin{proof} 
  By the equations given in \eqref{eq:0-unique-sys}, let $x = x_c$ and $w = w_c$, it holds that
  \begin{align*}
    (1 + x)^w = \frac{\alpha d}{1 + x} \quad \text{and} \quad \alpha + (1 + x)^w = \frac{\alpha d w x}{x + 1}.
  \end{align*}
  This means $\frac{\alpha d + \alpha (x + 1)}{x + 1} = \frac{\alpha d w x}{x + 1}$, which gives us
  \begin{align} \label{eq:x0-form-1}
    x = \frac{d + 1}{d w - 1}.
  \end{align}
  By $(x + 1)^{w+1} = \alpha d$ in \eqref{eq:0-unique-sys}, it holds that
  \begin{align*}
    \alpha = \frac{1}{d} (x + 1)^{w+1} \overset{\eqref{eq:x0-form-1}}{=} \frac{1}{d} \tp{1 + \frac{d+1}{dw-1}}^{w+1} = \frac{d^w(w+1)^{w+1}}{(dw - 1)^{w+1}} = A(d, w, 0).
  \end{align*}
  By $(x + 1)^{w+1} = \alpha d$ in \eqref{eq:0-unique-sys} and $x = \frac{d+1}{d w - 1}$ in \eqref{eq:x0-form-1}, we also have
  \begin{align} \label{eq:x0-form-2}
    (x + 1)^w = \frac{\alpha d}{1 + \frac{d+1}{d w - 1}} = \frac{\alpha (d w - 1)}{w + 1}.
  \end{align}
  Finally, we resolve $\lambda(x)$ as follows.
  Note that $\lambda(x) = x(1 + \alpha (1 + x)^{-w})^{d}$, we have
  \begin{align*}
    \lambda(x)
    &\overset{\eqref{eq:x0-form-2}}{=} x \tp{1 + \frac{w + 1}{d w - 1}}^d
    \overset{\eqref{eq:x0-form-1}}{=} \frac{w^d (d+1)^{d+1}}{(d w - 1)^{d+1}} = A(w, d, 0). \qedhere
  \end{align*}
\end{proof}

Now we solve \eqref{eq:delta-sys-eq} only assuming $\delta = 0$.
That is, the system in \eqref{eq:delta-sys-eq} becomes
\begin{align} \label{eq:0-unique-eq}
  \begin{cases}
    T := (x + 1)(\alpha + (1 + x)^w) - \alpha d w x = 0, \\
    M := \alpha d - (1 + x)^{w+1}  = 0, \\
    G := x(1 + \alpha(1 + x)^{-w})^d - \alpha = \lambda(x) - \alpha = 0.
  \end{cases}
\end{align}

Based on \Cref{lem:0-unique}, we note that $\alpha$ and $\lambda(x)$ forms a beautiful symmetry structure.
Leveraging this symmetry structure, we have the following corollary.

\begin{corollary} \label{lem:exact-0-unique-eq}
  Suppose $\Delta = d + 1 \geq 3$, let $w = d$ in \eqref{eq:0-unique-sys}, let $(\alpha, x)$ be the solution of \eqref{eq:0-unique-sys}.
  Then, it holds that $\alpha = \lambda(x) = \lambda_c(\Delta)$.
\end{corollary}
\begin{proof}
  According to \Cref{lem:0-unique}, if we choose $w = d$, then it holds that $\alpha = \lambda(x)$ in \Cref{lem:0-unique}.
  After that, note that a simple calculation gives us
  \begin{align*}
    \alpha = \lambda &= A(d, d, 0) = \frac{d^d(d + 1)^{d+1}}{(d^2 - 1)^{d+1}} = \frac{d^d}{(d-1)^{d+1}} = \frac{(\Delta - 1)^{\Delta-1}}{(\Delta - 2)^\Delta} = \lambda_c(\Delta). \qedhere
  \end{align*}
\end{proof}
Using \Cref{lem:0-unique}, we are also able to prove \Cref{thm:delta-unique-lambda-d-alpha-w}, the proof will be given in \Cref{sec:0-unique-eq}.

Combining \Cref{lem:exact-0-unique-eq}, \Cref{thm:delta-unique} with a monotonicity argument shows following.
\begin{lemma} \label{lem:0-unique-eq}
  Let $\Delta = d + 1 \geq 3$, and $\lambda > 0$, then the pair $(\lambda, d)$ is $0$-unique if $\lambda \leq \lambda_c(\Delta)$.
\end{lemma}
We note that a more general monotonicity result is given in \Cref{lem:delta-unique-monotone-eq} later and works for all $\delta \in (0, 1)$.
Here, we also give a self-contained proof for \Cref{lem:0-unique-eq} in \Cref{sec:0-unique-eq}.
Unlike the proof of \Cref{lem:delta-unique-monotone-eq}, this proof is quite straightforward.
However, we are not aware of a proof along this line for $\delta > 0$.


To make a summary, recall the system in \eqref{eq:delta-sys-eq} is defined as
\begin{align*} 
  \begin{cases}
    T := (1 - \delta)(x + 1)(\alpha + (1 + x)^w) - \alpha d w x = 0 \\
    M := w\log(1 + x)(\alpha d - (1 + x)^{w+1}) + \delta (x + 1) (\alpha + (x + 1)^w) = 0 \\
    G := x(1 + \alpha(1 + x)^{-w})^d - \alpha = 0.
  \end{cases}
\end{align*}
In \Cref{lem:0-unique-eq}, for fixed $d \geq 2$, we have showed that the system in \eqref{eq:delta-sys-eq} has the following solution.
\begin{align} \label{eq:0-unique-sol}
  \*p_0 &= \tp{\alpha, w, x, \delta} = \tp{\lambda_c(\Delta), d, \frac{1}{d-1}, 0}.
\end{align}
Now, towards resolving \eqref{eq:delta-sys-eq}, we are going to apply the implicit function theorem around the point $\*p_0$ to get an linear approximation of its neighborhood.
This allows us to resolve \eqref{eq:delta-sys-eq} approximately in this neighborhood.
Fix $d \geq 2$, the system in \eqref{eq:delta-sys-eq} encodes implicit function $(\delta(\alpha), w(\alpha), x(\alpha))$ by viewing $\alpha$ as the variable.
We have the following result towards this implicit function.

%
%

%

\begin{lemma} \label{lem:d-delta}
  There is an open set $U \subseteq \^R$ containing $\alpha = \lambda_c(\Delta)$ that there is a unique continuously differentiable function $(\delta(\alpha), w(\alpha), x(\alpha)) : U \to \^R^3$ such that 
  \begin{align*}
    \forall \delta \in U, \quad \text{the point $(\alpha, \delta(\alpha), w(\alpha), x(\alpha))$ is a solution of the system in \eqref{eq:delta-sys-eq}}.
  \end{align*}
  Moreover, it holds that $\left. \frac{\-d \delta(\e^y)}{\-d y} \right\vert_{y = \log \lambda_c(\Delta)} = - \frac{d-1}{d}$.
\end{lemma}
The proof of \Cref{lem:d-delta} is technical and will be given in \Cref{sec:d-delta}.
\begin{proof}[Proof of \Cref{thm:solve-sys-eq}]
  Recall $U \subseteq \^R$ is an open set defined in \Cref{lem:d-delta}.
  By Taylor series, denote $\lambda_c = \lambda_c(\Delta)$, for every $\lambda \in U$ such that $\alpha = \lambda = (1 - \delta_\star) \lambda_c \in U$ for some $\delta_\star \in (0, 1)$, it holds that
  \begin{align*}
    (\delta \circ \exp) (\log \lambda)
    &= (\delta \circ \exp)(\log\lambda_c) + (\delta \circ \exp)'(\log \lambda_c)(\log \lambda - \log \lambda_c) + o(\log \lambda - \log \lambda_c) \\
    \delta(\alpha) 
    &= -\frac{d-1}{d} \log (1 - \delta_\star) + o(\log (1 - \delta_\star)) \\
    &= \frac{d-1}{d} \log \frac{1}{1-\delta_\star} + o\tp{\log \frac{1}{1-\delta_\star}} \\
    &= \frac{d-1}{d} \frac{\delta_\star}{1 - \delta_\star} + o(\frac{\delta_\star}{1 - \delta_\star}) \\
    (d \geq 2) \quad &\geq \frac{\delta_\star}{2} + o(\delta_\star) \\
    &\geq \frac{\delta_\star}{4}. \quad (\text{for sufficiently small $\delta_\star$}) 
  \end{align*}
  This proves \Cref{thm:solve-sys-eq}.
\end{proof}

%

\subsubsection{$0$-uniqueness} \label{sec:0-unique-eq}
In this section, we prove \Cref{lem:0-unique-eq}.
It will be proved by using the monotonicity of the function $A(d, w, 0)$ with respect to $d, w$.

Recall the function $A$ defined in \eqref{eq:A}.
We have the following fact towards the function $A$.
\begin{lemma} \label{lem:alpha-c-w}
  Given $\delta \in [0, 1)$, $d, w > 0$, if $d w - (1 - \delta) > 0$, then 
  \begin{align*}
    \frac{\partial A(d, w, \delta)}{\partial d} < 0 \quad \text{and} \quad \frac{\partial A(d, w, \delta)}{\partial w} < 0.
  \end{align*}
\end{lemma}
\begin{proof}
  For $\frac{\partial A(d, w, \delta)}{\partial w}$, we have
  \begin{align*}
    (d w &- (1 - \delta)) \cdot \frac{\partial \log A(d, w, \delta)}{\partial w} \\
    &= - (d + (1 - \delta)) + (d w - (1 - \delta)) \log \tp{\frac{d (1 + w)}{d w - (1 - \delta)}} \\
    &= - (d + (1 - \delta)) + (d w - (1 - \delta)) \log \tp{1 + \frac{d + (1 - \delta)}{d w - (1 - \delta)}} < 0,
  \end{align*}
  where in the last inequality we use the fact that $\log(1 + x) < x$ for $x > 0$.

  For $\frac{\partial A(d, w, \delta)}{\partial w}$, we have 
  \begin{align*}
    \frac{\partial A(d, w, \delta)}{\partial w} &= - (1 - \delta) w (w+1)^{w+1} (d-\delta +1) d^{w-1} (d w - (1 - \delta))^{-w-2} < 0,
  \end{align*}
  where we use the fact that $d w - (1 - \delta) > 0$.
\end{proof}


%

\begin{lemma} \label{lem:alpha-vs-lambda}
	Fix any $d,w>0$ such that $d w > 1$.
  \begin{enumerate}
  \item If $d < w$, it holds that $A(d, w, 0) > A(w, d, 0)$;
  \item if $d > w$, it holds that $A(d, w, 0) < A(w, d, 0)$.
  \end{enumerate}
\end{lemma}
\begin{proof}
  When $d = w$, then it is easy to check that $A(d, w, 0) = A(w, d, 0)$.
  When $d \neq w$, then without loss of generality, we assume $d < w$.
  Define
  \begin{align*}
    \Delta_w(d) := \log\tp{A(d, w, 0)} - \log\tp{A(w, d, 0)}
  \end{align*}
  Note that we have $\Delta_w(w) = 0$ and
  \begin{align*}
    \Delta_w'(d) &= \frac{d - w}{d (d w - 1)} + \log \frac{d w - 1}{d w + w}
  \end{align*}
  For $d < w$ and $d w > 1$, it holds that $\Delta_w'(d) < 0$.
  Hence by the mean value theorem, it holds that $\Delta_w(d) > 0$ for $d < w$.
\end{proof}

\begin{proof}[Proof of \Cref{thm:delta-unique-lambda-d-alpha-w}]
  Fix $d, w$ such that $d \geq 1$ and $d w > 1$, by \Cref{lem:0-unique}, it holds that when $\alpha = \alpha_c(d, w) = A(d, w, 0) = \*{\hat{\lambda}}(d, w)$, then the system in \eqref{eq:0-unique-sys} has a solution $x$ such that $\lambda(x) = \lambda_c(d, w) = A(w, d, 0) = \*{\hat{\lambda}}(w, d)$.
  
  Note that if $\alpha \leq \frac{1}{d} \e^{1 + \frac{1}{d}}$, then by \Cref{thm:delta-unique-implicit} and \Cref{def:x-lambda-delta}, for any $\lambda \geq \lambda_c(d, w) \geq \lambda^0_{2,c} = 0$, it holds that $(\lambda, d, \alpha)$ is $0$-unique.
  For $\frac{1}{d} \e^{1 + \frac{1}{d}} < \alpha \leq \alpha_c(d, w)$, by \Cref{lem:alpha-c-w}, there is a $w' \geq w$ such that $\alpha = A(d, w', 0)$.
  By \Cref{lem:alpha-c-w}, $A(w', d, 0) \leq A(w, d, 0)$.
  By \Cref{lem:0-unique}, \Cref{thm:resolve-lambda-2-c}, and \Cref{thm:delta-unique-implicit}, it holds that $(\lambda, d, \alpha)$ is $0$-unique for all $\lambda \geq \lambda_c(w,d) = A(w, d, 0) \geq A(w', d, 0)$.
  
  This shows that for every $\alpha \leq \alpha_c$ and $\lambda \geq \lambda_c$, the tuple $(\lambda, d, \alpha)$ is $0$-unique.

  Now, we prove the furthermore part of \Cref{thm:delta-unique-lambda-d-alpha-w}.
  Fix $\alpha = \alpha_c(d, w)$, by \Cref{lem:0-unique}, \Cref{thm:resolve-lambda-2-c}, and \Cref{thm:delta-unique-implicit}, $(\lambda, d, \alpha)$ is $0$-unique iff $\lambda \geq \lambda_c(d, w)$.
  Fix $\lambda = \lambda_c(d, w) = A(w, d, 0)$, for $\alpha > \alpha_c$, by \Cref{lem:alpha-c-w}, there is a $1/d < w' < w$ such that $\alpha = A(d, w', 0)$.
  Hence, by \Cref{lem:0-unique}, \Cref{thm:resolve-lambda-2-c}, and \Cref{thm:delta-unique-implicit}, $(\lambda, d, \alpha)$ is $0$-unique iff $\lambda \geq \lambda_c(d, w') = A(w', d, 0)$.
  However, note that by \Cref{lem:alpha-c-w}, it holds that $\lambda_c(d, w') > \lambda_c(d, w) = \lambda$, which means the tuple $(\lambda, d, \alpha)$ is no longer $0$-unique.
\end{proof}

\begin{proof}[Proof of \Cref{lem:0-unique-eq}]
  To finish the proof, fix $d \geq 2$, we will consider three cases: (1) $\alpha = \lambda_c(\Delta)$; (2) $\frac{1}{d} \e^{1 + \frac{1}{d}} < \alpha < \lambda_c(\Delta)$; (3) $\alpha \leq \frac{1}{d}\e^{1+\frac{1}{d}}$.

  Consider the first case.
  Suppose $d \geq 2$ and $\alpha = \lambda_c(\Delta)$, by \Cref{lem:exact-0-unique-eq} it holds that for $\lambda = \alpha = \lambda_c(\Delta) = A(d, d, 0)$, the system $(\lambda, d)$ is $0$-unique.

  Now, for the second case, suppose $d \geq 2$ and $\frac{1}{d} \e^{1 + \frac{1}{d}} < \alpha < A(d, d, 0)$.
  By \Cref{lem:alpha-c-w}, it holds that $\alpha = A(d, w, 0)$ for some $w > d$.
  Then by $\alpha = A(d, w, 0)$, \Cref{lem:0-unique}, it holds that $\lambda^0_{2,c} = A(w, d, 0)$.
  Since $w > d$, by \Cref{lem:alpha-vs-lambda}, it holds that $\alpha = A(d, w, 0) > A(w, d, 0) = \lambda^0_{2,c}$.

Combined with \Cref{thm:delta-unique-implicit}, it implies that, if we pick $\lambda = \alpha > \lambda^0_{2,c}$, then the pair $(\lambda, d, \alpha)$ is at least $0$-unique.

  Finally, by \Cref{thm:delta-unique-implicit}, if $\lambda = \alpha \leq \frac{1}{d} \e^{1 + \frac{1}{d}}$, then $(\lambda, d)$ is also $0$-unique.

  Combining all these three cases, if $\lambda = \alpha \leq \lambda_c(\Delta)$, then $(\lambda, d)$ is $0$-unique.
\end{proof}

\subsubsection{$\delta$-uniqueness for $(\lambda, d)$ (Proof of \Cref{lem:d-delta})} \label{sec:d-delta}
\newcommand{\diff}[2]{\partial_{#2} #1\;}


In this section, we prove \Cref{lem:d-delta}.
Note that the function $\alpha(\delta), w(\delta), x(\delta)$ are implicit functions determined by the following system as in \eqref{eq:delta-sys-eq}:
\begin{align*}
  \begin{cases}
    T := (1 - \delta)(x + 1)(\alpha + (1 + x)^w) - \alpha d w x = 0 \\
    M := w\log(1 + x)(\alpha d - (1 + x)^{w+1}) + \delta (x + 1) (\alpha + (x + 1)^w) = 0 \\
    G := x(1 + \alpha(1 + x)^{-w})^d - \alpha = 0.
  \end{cases}
\end{align*}
Given a function $F \in \{T, M, G\}$, and a variable $y \in \{\alpha, x , w, \delta\}$, we use $\diff{F}{y}$ to denote the function $\partial F / \partial y$.
Therefore, by the implicit function theorem, 
\begin{align*}
  \left[
  \begin{array}{c}
    \delta'(\alpha) \\ w'(\alpha) \\ x'(\alpha)
  \end{array}
  \right]
  = -
  \left[
  \begin{array}{ccc}
    \diff{T}{\delta} & \diff{T}{w} & \diff{T}{x} \\
    \diff{M}{\delta} & \diff{M}{w} & \diff{M}{x} \\
    \diff{G}{\delta} & \diff{G}{w} & \diff{G}{x}
  \end{array}
  \right]^{-1}
  \left[
  \begin{array}{c}
    \diff{T}{\alpha} \\ \diff{M}{\alpha} \\ \diff{G}{\alpha}
  \end{array}
  \right].
\end{align*}
Using Cramer's rule, it holds that
\begin{align*}
  \delta'(\alpha)
  = \frac{\diff{G}{\alpha } \left(\diff{M}{x} \diff{T}{w}-\diff{M}{w} \diff{T}{x}\right)+\diff{G}{x} \left(\diff{M}{w} \diff{T}{\alpha }-\diff{M}{\alpha } \diff{T}{w}\right)+\diff{G}{w} \left(\diff{M}{\alpha } \diff{T}{x}-\diff{M}{x} \diff{T}{\alpha }\right)}{\diff{G}{\delta} \left(\diff{M}{w} \diff{T}{x}-\diff{M}{x} \diff{T}{w}\right)+\diff{G}{w} \left(\diff{M}{x} \diff{T}{\delta}-\diff{M}{\delta } \diff{T}{x}\right)+ \diff{G}{x} \left(\diff{M}{\delta } \diff{T}{w}-\diff{M}{w} \diff{T}{\delta }\right)}.
\end{align*}

Now, we will evaluate $\alpha'(\delta)$ around the point
\begin{align} \label{eq:p0}
  \*p_0 &= \tp{\alpha, w, x, \delta} = \tp{\lambda_c(\Delta), d, \frac{1}{d-1}, 0},
\end{align}
as defined in \eqref{eq:0-unique-sol}, where it holds that $\delta = 0$.
At this point $\*p_0$ the system in \eqref{eq:delta-sys-eq} becomes
\begin{align} \label{eq:0-unique-sys-eq}
  \begin{cases}
    T(\alpha, w, x, 0) = (x + 1)(\alpha + (1 + x)^w) - \alpha d w x = 0 \\
    M(\alpha, w, x, 0) = \alpha d - (1 + x)^{w+1} = 0 \\
    G(\alpha, w, x, 0) = x(1 + \alpha(1 + x)^{-w})^d - \alpha = 0.
  \end{cases}
\end{align}
First, note that $\diff{G}{\delta} = 0$.
Recall that we also have
\begin{align*}
  \diff{G}{x}
  &= \frac{(\hat{x} + 1)(\alpha + (1 + \hat{x})^w) - \alpha d w \hat{x}}{(\hat{x}+1) \left(\alpha +(\hat{x}+1)^w\right)} \cdot \left(\alpha  (\hat{x}+1)^{-w}+1\right)^d \\
  &= \frac{T(\alpha, w, x, 0)}{(\hat{x}+1) \left(\alpha +(\hat{x}+1)^w\right)} \cdot \left(\alpha  (\hat{x}+1)^{-w}+1\right)^d = 0,
\end{align*}
where the last equation holds by $T(\alpha, w, x, 0) = 0$ in \eqref{eq:0-unique-sys-eq}.
Hence, at the point $(\alpha, w, x, 0)$, we have
\begin{align*}
  \delta'(\alpha)
  = \frac{\diff{G}{\alpha } \left(\diff{M}{x} \diff{T}{w}-\diff{M}{w} \diff{T}{x}\right) + \diff{G}{w} \left(\diff{M}{\alpha } \diff{T}{x}-\diff{M}{x} \diff{T}{\alpha }\right)}{\diff{G}{w} \left(\diff{M}{x} \diff{T}{\delta }-\diff{M}{\delta } \diff{T}{x}\right)}.
\end{align*}
Also note that we have $\diff{T}{x} = 0$ at $(\alpha, w, x, 0)$. 
This is because $M(\alpha, w, x, 0) = \alpha d - (x + 1)^{w+1} = 0$ and $T(\alpha, w, x, 0) = 0$.
Note that $M(\alpha, w, x, 0) = 0$ gives us
\begin{align*}
  (x + 1)^{w + 1} &= \alpha d \\
  (x + 1)w(\alpha d - (x + 1)^w) &= \alpha d w x,
\end{align*}
which together with $T(\alpha, w, x, 0) = (1 + x)(\alpha + (1 + x)^w) - \alpha d w x = 0$ gives us
\begin{align} \label{eq:T'(x) = 0}
  \alpha + (x + 1)^w = w (\alpha d - (x + 1)^w).
\end{align}
By \Cref{fact:T-delta}, \cref{eq:T'(x) = 0} implies that $\diff{T}{x} = 0$.

So, we have
\begin{align*}
  \delta'(\alpha)
  &= \frac{\diff{G}{\alpha} \diff{M}{x} \diff{T}{w} - \diff{G}{w} \diff{M}{x} \diff{T}{\alpha}}{\diff{G}{w} \diff{M}{x} \diff{T}{\delta}}
    = \frac{\diff{G}{\alpha} \diff{T}{w} - \diff{G}{w} \diff{T}{\alpha}}{\diff{G}{w} \diff{T}{\delta}}
    = \frac{(\diff{G}{\alpha} / \diff{G}{w}) \cdot \diff{T}{w} - \diff{T}{\alpha}}{\diff{T}{\delta}}.
\end{align*}
Note that, we have
\begin{align}
  \diff{G}{\alpha} / \diff{G}{w}
  \nonumber &= \frac{\left(\alpha  (x+1)^{-w}+1\right)^{-d} \left(\alpha -d x \left(\alpha  (x+1)^{-w}+1\right)^d+(x+1)^w\right)}{\alpha  d x \log (x+1)} \\
  \label{eq:G-alpha-G-w} (\text{by $G(\alpha, w, x, 0) = 0$}) \quad &= \frac{\alpha - \alpha d + (x + 1)^w}{\alpha^2 d \log(x + 1)}.
\end{align}
Recall that by the definition of $\*p_0$ as in \eqref{eq:p0}, we have $\alpha = \lambda_c(\Delta)$ and $w = d$.
From $T(\alpha, w, x, 0) = M(\alpha, w, x, 0) = 0$, we also have $(x + 1)^w = \frac{\alpha (d w - 1)}{w + 1}$ (this is done by using the same calculation as in \eqref{eq:x0-form-2}).
Combining $w = d$ and $(x + 1)^w = \frac{\alpha (d w - 1)}{w + 1}$, it holds that $(x + 1)^w = \alpha (d - 1)$.
This with \eqref{eq:G-alpha-G-w} implies that
\begin{align*}
  \frac{ \diff{G}{\alpha} }{ \diff{G}{w} } =  \frac{\alpha - \alpha d + (x + 1)^w}{\alpha^2 d \log(x + 1)} = 0.
\end{align*}
Hence, we have
\begin{align*}
  \delta'(\alpha) &= - \frac{\diff{T}{\alpha}}{\diff{T}{\delta}} = - \frac{d w x-(x+1)}{(x+1) \left(\alpha +(x+1)^w\right)}.
\end{align*}
Note that that $M(\alpha, d, w, 0) = 0$ in \eqref{eq:0-unique-sys-eq}, and $x = \frac{1}{d-1}$ according to $\*p_0$ in \eqref{eq:p0}, then
\begin{align*}
  \delta'(\alpha)
  &= - \frac{d w x-(x+1)}{(x+1) \left(\alpha +(x+1)^w\right)} \\
  (\text{by $T(\alpha, w, x, 0) = 0$}) \quad &= - \frac{d w x - (x + 1)}{\alpha d w x} = - \frac{d^2 x - (x + 1)}{\alpha d^2 x} \\
  (\text{by $x = \frac{1}{d-1}$}) \quad &= - \frac{\frac{d^2}{d-1} - \frac{d}{d-1}}{\frac{\alpha d^2}{d - 1}} = - \frac{d-1}{\alpha d}.
\end{align*}
Hence, it holds that at the point $\*p_0$, 
\begin{align*}
  \left. \frac{ \-d \delta(\e^y)}{\-d y} \right\vert_{y = \log \alpha} &= \delta'(\alpha) \alpha = - \frac{d-1}{d}.
\end{align*}
This finishes the proof of \Cref{lem:d-delta}.

\subsection{Roots of $T_\delta(x)$ (Proof of  \Cref{lem:T-root-lifecycle})} \label{sec:T-root-lifecycle}
We now prove \Cref{lem:T-root-lifecycle}.
Recall that $T_\delta$ is defined as in \eqref{eq:function-T-delta},
\begin{align*}
  T_\delta(x) := (1 - \delta)(x + 1)(\alpha + (1 + x)^w) - \alpha d w x.
\end{align*}

According to \Cref{fact:T-delta}, 
$T_\delta$ achieves its minimum at a point $y$ where $T'_\delta(y) = 0$,
and  the number of solutions for the equation $T_\delta(x) = 0$ can be determined by looking at the sign of $T_\delta(y)$.

By \Cref{fact:T-delta}, we know that
\begin{align*}
  T_\delta'(x) &= (1 - \delta)(1 + w)(1 + x)^w - \alpha(d w  - (1 - \delta)).
\end{align*}
It can be verified that $T'_\delta(x) = 0$ has a positive solution if $\frac{\alpha (d w - (1 - \delta))}{(1 - \delta)(1 + w)} \geq 1$.
First, we rule out its complement case: when $\frac{\alpha (d w - (1 - \delta))}{(1 - \delta)(1 + w)} < 1$, we have
%
\begin{align*}
  \lim_{x\to 0} T'_\delta(x) = (1 + w)(1 - \delta) - \alpha(d w - (1 - \delta)) > 0.
\end{align*}
By \Cref{fact:T-delta}, we know $T''_\delta(x) > 0$ for all $x \geq 0$, which, together with the above inequality, implies that $T'_\delta(x) > 0$ for all $x \geq 0$.
And also by \Cref{fact:T-delta}, we have $\lim_{x\to 0}T_\delta(x) > 0$, which together with that $T'_\delta(x) > 0$ for all $x \geq 0$, implies that 
 $T_\delta(x) > 0$ for all $x \geq 0$.
This indicates that the equation $T_\delta(x) = 0$ has no positive solution.

In the remaining part of this section, we prove \Cref{lem:T-root-lifecycle} under the assumption that $\frac{\alpha (d w - (1 - \delta))}{(1 - \delta)(1 + w)} \geq 1$.
Note that when $\frac{\alpha (d w - (1 - \delta))}{(1 - \delta)(1 + w)} \geq 1$, the equation $T_\delta'(x) = 0$ has a unique solution
\begin{align} \label{eq:T-y}
  y &:= \tp{\frac{\alpha (d w - (1 - \delta))}{(1 - \delta)(1 + w)}}^{\frac{1}{w}} - 1.
\end{align}

Now, by \Cref{fact:T-delta}, $T_\delta$ is a strictly convex function.
And it is easy to see the following fact.
\begin{fact}\label{fact:T-delta-y}
  If $\frac{\alpha (d w - (1 - \delta))}{(1 - \delta)(1 + w)} \geq 1$, for the $y$ defined in \eqref{eq:T-y}, it holds that
  \begin{itemize}
  \item if $T_\delta(y) > 0$, then $T_\delta(x) = 0$ has no positive solution;
  \item if $T_\delta(y) = 0$, then $T_\delta(x) = 0$ has a unique positive solution $y$;
  \item if $T_\delta(y) < 0$, then $T_\delta(x) = 0$ has two positive solutions $x_1 < y < x_2$.
  \end{itemize}
\end{fact}

\begin{lemma} \label{lem:sign-T-delta-y}
  If $\frac{\alpha (d w - (1 - \delta))}{(1 - \delta)(1 + w)} \geq 1$, for the $y$ defined in \eqref{eq:T-y}, it holds that
  \begin{align*}
    \-{sign}\tp{\frac{(1 - \delta) d^w (w+1)^{w+1}}{(d w - (1 - \delta))^{w+1}} - \alpha} = \-{sign}\tp{T_\delta(y)}.
  \end{align*}
\end{lemma}
\begin{proof}
  The assumption $\frac{\alpha (d w - (1 - \delta))}{(1 - \delta)(1 + w)} \geq 1$ ensures the existence of $y$.
  Note that the inequality $T_\delta(y) > 0$ can be rewritten as
  \begin{align*}
    &&(1 - \delta) (y + 1) \tp{\alpha + \frac{\alpha (d w - (1 - \delta))}{(1 - \delta) (1 + w)}} - \alpha d w y &> 0\\
&\iff&    (1 - \delta) (y + 1) \frac{\alpha (d w + w(1 - \delta))}{(1 - \delta) (1 + w)} - \alpha d w y &> 0\\
&\iff&     (y + 1) \tp{\frac{\alpha (d w + w(1 - \delta))}{(1 + w)} - \alpha d w} + \alpha d w &> 0 \\
&\iff&     (y + 1) \frac{- \alpha w (dw - (1 - \delta))}{1 + w}  + \alpha d w &> 0.
  \end{align*}
  Plugging in the definition of $y$ in \eqref{eq:T-y}, this implies that
  \begin{align*} 
&&    y + 1 &< \frac{d (1 + w)}{d w - (1 - \delta)}\\
&\iff&     \tp{\frac{\alpha (d w - (1 - \delta))}{(1 - \delta)(1 + w)}}^{1/w} &< \frac{d (1 + w)}{d w - (1 - \delta)} \\
&\iff&     \alpha &< \frac{(1 - \delta) d^w (w + 1)^{w+1}}{(d w - (1 - \delta))^{w+1}}.
  \end{align*}
  The cases with $T_\delta(y) = 0$ and $T_\delta(y) < 0$  can be verified in the same way.
\end{proof}
For convenience, we denote $\alpha_c(w)=A(d, w, \delta)$, where recall $A(d, w, \delta)$ defined in \eqref{eq:A}, i.e.
\begin{align} \label{eq:alpha-c}
  \alpha_c(w) = \frac{(1 - \delta) d^w (w + 1)^{w+1}}{(d w - (1 - \delta))^{w+1}}. 
\end{align}

By \Cref{lem:alpha-c-w}, $\alpha_c(w) = A(d, w, \delta)$ is monotonically decreasing in  $w$.

\begin{proof}[Proof of \Cref{lem:T-root-lifecycle}]
  Recall that we assume $d \geq 1 - \delta$ and $\frac{\alpha (d w - (1 - \delta))}{(1 - \delta)(1 + w)} \geq 1$.
  Under these assumptions, the requirement $dw - (1 - \delta) > 0$ of \Cref{lem:alpha-c-w} is always satisfied.
  By \Cref{lem:alpha-c-w}, it holds that
  \begin{align*}
    \sup_{w: d w > 1 - \delta} \alpha_c(w) &= \lim_{w \to (1-\delta)/d} \alpha_c(w) = +\infty \\
    \inf_{w: d w > 1 - \delta} \alpha_c(w) &= \lim_{w \to +\infty} \alpha_c(w) =  \frac{1-\delta}{d} \e^{1 + \frac{1-\delta}{d}}.
  \end{align*}
  
  If $\alpha \leq \frac{1-\delta}{d} \e^{1 + \frac{1-\delta}{d}}$, then it holds that $\alpha < \alpha_c(w)$.
  Hence, when $\frac{\alpha (d w - (1 - \delta))}{(1 - \delta)(1 + w)} \geq 1$, then $y$ exists and by \Cref{fact:T-delta-y} and \Cref{lem:sign-T-delta-y}, $T_\delta$ has no positive solution.
  When $\frac{\alpha (d w - (1 - \delta))}{(1 - \delta)(1 + w)} < 1$, as we have dealt with this case before, $T_\delta$ also has no positive solution.
  
  In the rest part of the proof, we assume $\alpha > \frac{1-\delta}{d} \e^{1 + \frac{1-\delta}{d}}$, equation $\alpha = \alpha_c(w)$ has a unique positive solution $w_\delta$.
  We then make the following claim.
  \begin{claim} \label{claim:y-exists-check}
    If $d \geq 1 - \delta$ and $w \geq w_\delta$, it holds that
    \begin{align*}
      \frac{\alpha (d w - (1 - \delta))}{(1 - \delta)(1 + w)} \geq 1.
    \end{align*}
  \end{claim}
  \begin{proof}
  First, we prove the claim for $w = w_\delta$.
  Recall that $w_\delta$ is the unique solution of the equation $\alpha = \alpha_c(w_\delta)$, where $\alpha_c$ is defined in \eqref{eq:alpha-c}.
  It is equivalent to show $\frac{(1 - \delta)(1 + w_\delta)}{d w_\delta - (1 - \delta)} \leq \alpha = \alpha_c(w_\delta)$, which is equivalent to
  \begin{align*}
 \frac{(1 - \delta)(1 + w_\delta)}{d w_\delta - (1 - \delta)} \leq \frac{(1 - \delta) d^{w_\delta} (w_\delta + 1)^{w_\delta+1}}{(d w_\delta - (1 - \delta))^{w_\delta+1}} 
\quad\iff\quad     1 \leq \tp{\frac{d(w_\delta + 1)}{d w_\delta - (1 - \delta)}}^{w_\delta}.
  \end{align*}
  It is easy to see that the last inequality is true.
  For $w > w_\delta$, $\frac{\alpha (d w - (1 - \delta))}{(1 - \delta)(1 + w)} \geq 1$ holds by noticing that the function $w \mapsto \frac{d w - (1 - \delta)}{(1 - \delta)(1 + w)}$ is increasing when $d \geq 1 - \delta$.
\end{proof}

  Using \Cref{claim:y-exists-check}, we are able to finish the proof of \Cref{lem:T-root-lifecycle}.
  As in \Cref{lem:T-root-lifecycle}, we will consider three cases: (1) $w > w_\delta$; (2) $w = w_\delta$; (3) $w < w_\delta$.

  \paragraph{Case (1):} if $w > w_\delta$,   
  by \Cref{claim:y-exists-check}, it holds that
    \begin{align*}
      \frac{\alpha (d w - (1 - \delta))}{(1 - \delta)(1 + w)} \geq 1.
    \end{align*}
    So, $y$, as the root of the equation $T'_\delta(y) = 0$, exists.
    Moreover, the above condition also implies that $d w - (1 - \delta) > 0$.
    By \Cref{lem:alpha-c-w} and $d w - (1 - \delta) > 0$, it holds that $\alpha > \alpha_c(w)$.
    Hence, by \Cref{lem:sign-T-delta-y}, it holds that $T_\delta(y) < 0$.
    Combining $T_\delta(y) < 0$ and \Cref{fact:T-delta-y}, it holds that when $w > w_\delta$, the equation $T_\delta(x) = 0$ has two positive solutions.

    \paragraph{Case (2):} if $w = w_\delta$, the proof can be done by exactly the same argument as \textbf{Case (1)}.
    It holds that, when $w = w_\delta$, the equation $T_\delta(x) = 0$ has a unique positive solution.

    \paragraph{Case (3):} if $w < w_\delta$, note that \Cref{claim:y-exists-check} does not apply to this case.
    However, if $\frac{\alpha (d w - (1 - \delta))}{(1 - \delta)(1  +w)} < 1$, then we already know that $T_\delta(x) = 0$ has no solution since we have already dealt with this case.
    Now, assume $\frac{\alpha (d w - (1 - \delta))}{(1 - \delta)(1  +w)} \geq 1$, then the same argument as \textbf{Case (1)} can be applied to deduce that, when $w < w_\delta$, the equation $T_\delta(x) = 0$ has no solution.
\end{proof}

%


\subsection{Implicit $\delta$-uniqueness regime for $(\lambda, d, \alpha)$}
\label{sec:delta-unique-implicit}
In this section, we prove \Cref{thm:delta-unique-implicit}.
In order to translate the uniqueness regime in the coordinate system $(\hat{x}, d, \alpha, w)$ to the uniqueness regime in the coordinate system $(\lambda, d, \alpha, w)$, we want the system to have a unique fixpoint.
Put simply, we want the equation $\lambda(x) = x$ to have a unique positive solution.
To achieve this, we need to understand the monotonicity of the function
\begin{align*}
  \lambda(x) := x(1 + \alpha(1 + x)^{-w})^d,
\end{align*}
as defined in \eqref{eq:lambda-x}.
Note that the derivative of $\lambda$ is given by
\begin{align*}
  \lambda'(x)
  &= \frac{(x + 1)(\alpha + (1 + x)^w) - \alpha d w x}{(x+1) \left(\alpha +(x+1)^w\right)} \cdot \left(\alpha  (x+1)^{-w}+1\right)^d,
\end{align*}
whose sign is determined by that of
\begin{align*}
  T_0(x) = (x + 1)(\alpha + (1 + x)^w) - \alpha d w x,
\end{align*}
as defined in \eqref{eq:function-T-delta}.
Intuitively, $\delta$-uniqueness will guarantee that the function $F(x)$ has at most $1$ fixpoint.
As remarked before \Cref{thm:delta-unique-implicit}, we will use \Cref{lem:lb-is-0} to take advantage of the strengthened $\delta$-uniqueness for all $w > 0$ as defined in \eqref{eq:lambda-d-alpha-unique}.

As suggested by \Cref{lem:T-root-lifecycle}, we analyze the $\delta$-uniqueness regime for $\lambda$ in the following three separated cases: (1)~$\alpha \leq \frac{1-\delta}{d} \cdot \e^{1 + \frac{1-\delta}{d}}$; (2)~$\frac{1-\delta}{d} \cdot \e^{1 + \frac{1-\delta}{d}} < \alpha \leq \frac{1}{d} \cdot \e^{1 + \frac{1}{d}}$; (3)~$\alpha > \frac{1}{d} \cdot \e^{1 + \frac{1}{d}}$.

\paragraph{When $\alpha \leq \frac{1-\delta}{d} \cdot \e^{1 + \frac{1-\delta}{d}}$:} 
Due to \Cref{lem:T-root-lifecycle}, we have $T_\delta(x) > 0$ and $T_0(x) > 0$ for all $x \geq 0$.
Hence in this case, $(\lambda,d,\alpha)$ is always $\delta$-unique for all $\lambda>0$. And due to \Cref{def:x-lambda-delta}, $\lambda_2^\delta=0$.
Recall that we always assume $\lambda>0$. Therefore, in this case, $(\lambda,d,\alpha)$ is $\delta$-unique iff $\lambda \geq  \lambda^{\delta}_{2,c}$.

\paragraph{When $\frac{1-\delta}{d} \cdot \e^{1 + \frac{1-\delta}{d}} < \alpha \leq \frac{1}{d} \cdot \e^{1 + \frac{1}{d}}$:} 
It holds that $T_0(x) > 0$ for all $x \geq 0$. 
\begin{itemize}
\item
If $w \leq w_\delta$, then $(\lambda, d, \alpha, w)$ is $\delta$-unique for all $\lambda > 0$.
\item
If $w > w_\delta$, by \Cref{lem:T-root-lifecycle}, the $x_1^\delta, x_2^\delta$ in \Cref{def:x-lambda-delta} exist, also recall the definitions of $\lambda^\delta_i(w)$ and $\lambda^\delta_{i,c}$ for $i\in\{1,2\}$  in \Cref{def:x-lambda-delta}.
Since $T_0(x) > 0$ for all $x \geq 0$, we have $\lambda^\delta_1 \leq \lambda^\delta_2$.
Therefore, $(\lambda, d, \alpha, w)$ is $\delta$-unique iff 
\[
\lambda \in (0, \lambda^\delta_1(w)] \cup [\lambda^\delta_2(w), +\infty).
\]
By the continuity of $\lambda^\delta_1(w), \lambda^\delta_2(w)$ in $w$, the tuple $(\lambda, d, \alpha)$ is $\delta$-unique iff 
\[
\lambda \in (0, \lambda^\delta_{1, c}] \cup [\lambda^\delta_{2,c}, +\infty).
\]
By \Cref{lem:lb-is-0}, we have $\lambda^\delta_{1, c} = 0$, which means $(\lambda, d, \alpha)$ is $\delta$-unique iff $\lambda \geq \lambda^\delta_{2,c}$.
\end{itemize}

\paragraph{When $\alpha>\frac{1}{d} \cdot \e^{1 + \frac{1}{d}}$:} 
The equation \eqref{eq:equation-w-delta} always has a unique positive solution $w_\delta$ for $\delta\in[0,1)$.
Moreover, since $T_\delta(x) < T_0(x), \forall x \geq 0$, it holds that when $w = w_0$, the equation $T_\delta(x) = 0$ has two positive roots.
By \Cref{lem:T-root-lifecycle}, this means $w_\delta < w_0$ for $\delta > 0$.
\begin{itemize}
\item If $w \leq w_\delta$, then $(\lambda, d, \alpha, w)$ is $\delta$-unique for all $\lambda > 0$. 
\item If $w_\delta < w \leq w_0$, then  for the  $\lambda^\delta_1(w), \lambda^\delta_2(w)$ defined in \Cref{def:x-lambda-delta},
it can be verified that
$\lambda^\delta_1(w) \leq \lambda^\delta_2(w)$ and $(\lambda, d, \alpha, w)$ is $\delta$-unique iff 
\begin{align*}
\lambda \in (0, \lambda^\delta_1(w)] \cup [\lambda^\delta_2(w), +\infty).
\end{align*}
\item 
If $w > w_0$, then $\lambda^0_1(w) \geq \lambda^0_2(w)$ and $(\lambda, d, \alpha, w)$ is $\delta$-unique iff
\begin{align*}
  \lambda \in  & \tp{(0, \lambda^0_2(w)) \cup (\lambda^0_1(w), +\infty)} \cap  \tp{(0, \lambda^\delta_1(w)] \cup [\lambda^\delta_2(w), +\infty)}.
\end{align*}
We note that the region $\tp{(0, \lambda^0_2(w)) \cup (\lambda^0_1(w), +\infty)}$ ensures that the system has exactly one fixed point. This ensures that the characterization of $\delta$-uniqueness given by $T_\delta \geq 0$ can be safely translated from the coordinate system of $(\hat{x}, d, \alpha, w)$ back to the coordinate system of $(\lambda, d, \alpha, w)$.
A typical case is illustrated in the figure below, where we use the fact that
\begin{align*}
  x^\delta_1(w) < x^0_1(w) \leq x^0_2(w) < x^\delta_2(w).
\end{align*}

{\centering
  
\scalebox{1.7}{

\begin{tikzpicture}
\draw[-, thick, draw=blue!50] plot[smooth] coordinates {(0, 1) (1, 0) (2.5, -0.5) (4.5, 0) (5, 0.3)};
\draw[-, thick, draw=orange!50] plot[smooth] coordinates {(0, 2) (1.7, 0) (2.35, -0.3) (3, 0) (5, 2)};
\draw[->] (0, 0) to (5.1, 0);  
\draw[->] (0, 0) to (0, 2.1);  
\node at (5.3, 2) {\tiny \color{orange!90} $T_0(x)$};
\node at (5.3, 0.3) {\tiny \color{blue!70} $T_\delta(x)$};
\draw[->] (0, -3) to (0, -0.9);
\draw[->] (0, -3) to (5.1, -3);
\draw[-, thick] plot[smooth] coordinates {(0, -3) (1.7, -1.7)  (3, -2.5) (5, -1)};
\node at (5.3, -1) {\tiny $\lambda(x)$};
\draw[-, dotted, draw = orange] (1.7, 0) to (1.7, -1.7);
\draw[-, dotted, draw = orange] (3, 0) to (3, -2.5);
\draw[-, draw = orange!90, dashed] (0, -1.69) to (5.1, -1.69);
\draw[-, draw = orange!90, dashed] (0, -2.52) to (5.1, -2.52);
\draw[-, dotted, draw = blue] (1, 0) to (1, -2.15);
\draw[-, dotted, draw = blue] (4.5, 0) to (4.5, -1.45);
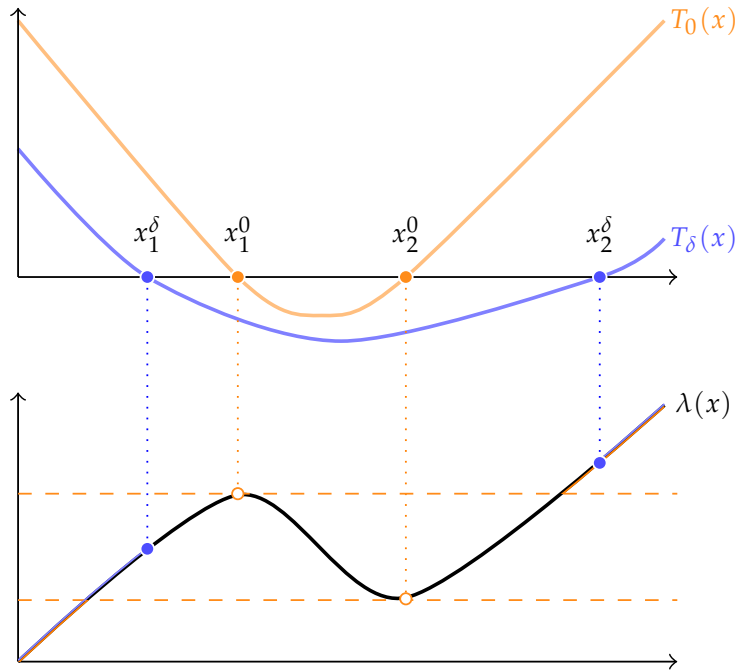
\begin{scope}
  \clip (0, -3) rectangle (5.1, -2.52);
  \draw[-, draw=orange] plot[smooth] coordinates {(0, -3.01) (1.7, -1.71)  (3, -2.5) (5, -1)};
\end{scope}
\begin{scope}
  \clip (0, -1.69) rectangle (5.1, -0.9);
  \draw[-, draw=orange] plot[smooth] coordinates {(0, -3) (1.7, -1.7)  (3, -2.51) (5, -1.01)};
\end{scope}
\begin{scope}
  \clip (0, -3) rectangle (1, -2.12);
  \draw[-, draw=blue!50] plot[smooth] coordinates {(0, -2.99) (1.7, -1.69)  (3, -2.5) (5, -1)};
\end{scope}
\begin{scope}
  \clip (4.5, -1.45) rectangle (5.1, -0.9);
  \draw[-, draw=blue!50] plot[smooth] coordinates {(0, -3) (1.7, -1.7)  (3, -2.49) (5, -0.99)};
\end{scope}

\node[draw = white, fill=blue!70, circle, inner sep=0pt, minimum size = 3pt, label={\tiny $x^\delta_1$}] at (1, 0) {};
\node[draw = white, fill=blue!70, circle, inner sep=0pt, minimum size = 3pt, label={\tiny $x^\delta_2$}] at (4.5, 0) {};
\node[draw = white, fill=orange!90, circle, inner sep=0pt, minimum size = 3pt, label={\tiny $x^0_1$}] at (1.7, 0) {};
\node[draw = white, fill=orange!90, circle, inner sep=0pt, minimum size = 3pt, label={\tiny $x^0_2$}] at (3, 0) {};

\node[draw = orange!90, fill=white, circle, inner sep=0pt, minimum size = 2.5pt] at (1.7, -1.69) {};
\node[draw = orange!90, fill=white, circle, inner sep=0pt, minimum size = 2.5pt] at (3, -2.51) {};
\node[draw = white, fill=blue!70, circle, inner sep=0pt, minimum size = 3pt] at (1, -2.12) {};
\node[draw = white, fill=blue!70, circle, inner sep=0pt, minimum size = 3pt] at (4.5, -1.45) {};

\end{tikzpicture}

} 

\par}

By the continuity of $\lambda^{\delta}_i(w)$ in $w$, taking the intersection of the above uniqueness regimes over all $w$,
the tuple $(\lambda, \alpha, d)$ is $\delta$-unique iff
\begin{align*}
  \lambda \in \tp{(0, \lambda^0_{1,c}] \cup [\lambda^0_{2,c}, +\infty)} \cap \tp{(0, \lambda^\delta_{1,c}] \cup [\lambda^\delta_{2,c}, +\infty)}
\end{align*}
where $\lambda^\delta_{i, c}$ for $\delta\in[0,1)$ and $i\in\{0,1\}$ are defined in \Cref{def:x-lambda-delta}, and we use the following observation
\begin{align*}
  \inf_{w>w_0} \lambda^0_2(w) &= 0 = \lambda^0_{1,c}, \quad \text{by \Cref{lem:lb-is-0}} \\
    \text{and} \quad \sup_{w > w_0} \lambda^0_1(w) &= \sup_{w > w_0} \lambda^0_2(w) = \lambda^0_{2,c}, \quad \text{by \Cref{lem:lambda-c-monotone-0} below and \Cref{lem:T-root-lifecycle}}.
\end{align*}

Here, the boundaries that involve $\lambda^0_{i,c}$ are closed, because $\lambda^0_{i}(w)$ is monotonically decreasing in $w$ (which is formally verified later in \Cref{lem:lambda-c-monotone-0}).

By \Cref{lem:lb-is-0}, it holds that $\lambda^\delta_{1,c} = \lambda^0_{1,c} = 0$.
Together with \Cref{lem:lambda-c-monotone-delta} stated below, we can conclude that $(\lambda, \alpha, d)$ is $\delta$-unique iff $\lambda \geq \lambda^\delta_{2,c}$.
\end{itemize}

\begin{lemma} \label{lem:lambda-c-monotone-0}
  If $\alpha > \frac{1}{d} \cdot \e^{1 + \frac{1}{d}}, w > w_0$, then $\lambda^0_i(w)$ is monotonically decreasing in $w$, for $i\in \{1, 2\}$.
\end{lemma}
\begin{proof}
  For $i \in \{1, 2\}$, 
  taking derivative of $\log \lambda^0_i(w)$ with respect to $w$, we have
  \begin{align*}
    x_i^0(w) & (x_i^0(w)+1) \left(\alpha +(x_i^0(w)+1)^w\right)
               \frac{\partial \log \lambda_i^0(x^0_i(w), w)}{\partial w}  \\
             &= [x_i^0]'(w) \tp{x_i^0(w) \left(\alpha -\alpha  d w+(x_i^0(w)+1)^w\right) +\left(\alpha +(x_i^0(w)+1)^w\right)} \\
             & \quad -\alpha d x_i^0(w)^2 \log (x_i^0(w)+1) - \alpha  d x_i^0(w) \log (x_i^0(w)+1) \\
             & = -\alpha d x_i^0(w)^2 \log (x_i^0(w)+1) - \alpha  d x_i^0(w) \log (x_i^0(w)+1) < 0,
  \end{align*}
  where the last equation holds by the fact that $T_0(x_i^0) = 0$ for $i \in \{1, 2\}$.
  The existence of $[x_i^0]'(w)$ is ensured by the implicit function theorem.
  Hence, both $\lambda_1^0(w)$ and $\lambda_2^0(w)$ are monotonically decreasing in $w$.
\end{proof}

\begin{lemma}\label{lem:lambda-c-monotone-delta}
  If $\alpha>\frac{1}{d}\cdot \e^{1 + \frac{1}{d}}$, then $ \lambda^\delta_{2,c}>\lambda^0_{2,c}$ for $\delta \in (0, 1)$.
\end{lemma}
\begin{proof}
  By \Cref{lem:lambda-c-monotone-0} both $\lambda_1^0(w)$ and $\lambda_2^0(w)$ are monotonically decreasing in $w$ for $w \geq w_0$.
  This implies $\lambda^0_{2,c} = \lambda^0_1(w_0)$.
  Note that $T_\delta(x^0_1(w_0)) < 0$, which means
  \begin{align*}
    x^\delta_1(w_0) < x^0_1(w_0) = x^0_2(w_0) < x^\delta_2(w_0).
  \end{align*}
  Therefore, by the mean value theorem, there exists a $\zeta \in (x^0_1(w_0), x^\delta_2(w_0))$ such that 
  \begin{align*}
    \lambda^\delta_2(w_0) - \lambda^0_1(w_0)
    &= \lambda(x^\delta_2(w_0)) - \lambda(x^0_1(w_0)) \\
    &= \lambda'(\zeta) \cdot \tp{ x^\delta_2(w_0) - x^0_1(w_0) }.
  \end{align*}
  Due to \Cref{fact:T-delta},
  we have $T_0(x^0_1(w_0)) = T'_0(x^0_1(w_0)) = 0$ and $T''_0(x) > 0$ for all $x \geq 0$.
  This means $T_0(x) > 0$ for all $x > x^0_1(w_0)$.
  Note that $\zeta > x^0_1(w_0)$, so we also have $T_0(\zeta) > 0$, which implies that $\lambda'(\zeta) > 0$ since $\mathrm{sign}(\lambda'(x))=\mathrm{sign}(T_0(x))$.
  Hence, we have $\lambda^\delta_{2,c} \geq \lambda^\delta_2(w_0) > \lambda^0_1(w_0)=\lambda^0_{2,c}$.
\end{proof}

In summary, the $\delta$-uniqueness regime for $(\lambda,d,\alpha)$ is described in \Cref{tab:lambda2}.
\begin{table}[h]
\centering
\begin{tabular}{|c|c|}
\hline
condition satisfied by $(\alpha,d)$ 
& $\delta$-uniqueness regime for $\lambda>0$\\
\hline
$\alpha \leq \frac{1-\delta}{d} \cdot \e^{1 + \frac{1-\delta}{d}}$ 
& $[0,+\infty)=[\lambda^\delta_{2,c},+\infty)$ \\
\hline
$\frac{1-\delta}{d} \cdot \e^{1 + \frac{1-\delta}{d}} < \alpha \leq \frac{1}{d} \cdot \e^{1 + \frac{1}{d}}$
& $[\lambda^\delta_{2,c},+\infty)$\\
\hline
$\alpha > \frac{1}{d} \cdot \e^{1 + \frac{1}{d}}$
&  $[\lambda^\delta_{2,c},+\infty)$\\
\hline
\end{tabular}
\caption{uniqueness regime for $\lambda$ under different $\alpha$ and $d$. \label{tab:lambda2}}
\end{table}

This proves \Cref{thm:delta-unique-implicit}, the implicit characterization of the $\delta$-uniqueness.
%


%



\subsection{Explicit $\delta$-uniqueness regime for $(\lambda,d,\alpha)$}\label{sec:delta-unique-general}


In this section, we prove \Cref{thm:resolve-lambda-2-c},
which explicitly resolves the critical threshold $\lambda_c=\lambda_{2,c}^\delta$.
In particular, we show that $\lambda^\delta_2(w)$ achieves its maximum at a unique point $w = w_c$, and hence the $(w_c,x_c)$ claimed in  \Cref{thm:delta-unique} is given by this $w_c$ and $x_c=x^\delta_2(w_c)$.

Fix any $\delta\in[0,1)$, and assume $\alpha>\frac{1-\delta}{d} \cdot \e^{1 + \frac{1-\delta}{d}}$.
Then the equation in \eqref{eq:equation-w-delta} has a unique positive solution $w_\delta$.
From now on until the end of the section, we assume this setting, as it is the only nontrivial case.
Recall the definition of $M_\delta(x)$ in \eqref{eq:M},
\[
        M_\delta(x) := w \log(1 + x) (\alpha d - (1 + x)^{w+1}) + \delta(x + 1)(\alpha + (x + 1)^w).
\]

We note that the function $M_\delta$ arises naturally when we use the first order condition to analysis the monotonicity of $\lambda^\delta_2(w)$, as an implicit function of $w$ (see \Cref{lem:sign-d-lambda} below).

\begin{remark}
	In the rest of this section, unless otherwise stated, we assume $M_\delta, T_\delta,$ are functions of $x$; and $x^\delta_i, \lambda^\delta_i, i\in\{1, 2\}$ are functions of $w$.
	Then, $M_\delta', T_\delta'$ means we are taking derivatives with respect to $x$, and $(x^\delta_i)', (\lambda^\delta_i)', i \in \{1, 2\}$ means we are taking derivatives with respect to $w$.
	We note, however, that in certain cases we may also consider them as multivariate functions. In those cases, we will explicitly use the notation $\partial_z F$ to denote $\partial F / \partial z$ for a function $F$ and a variable $z$.
\end{remark}
When $\delta = 0$, by \Cref{lem:lambda-c-monotone-0}, it holds that $\lambda^0_{2, c} = \lambda^0_1(w_0)$, where $w_0$ is defined as in \Cref{lem:T-root-lifecycle}.
By \Cref{lem:T-root-lifecycle}, the equation $T_0(x) = 0$ has a unique positive solution $x^0_1(w_0) = x^0_2(w_0)$.
By \Cref{fact:T-delta}, $T_0$ is a strictly convex function, these imply that $T_0(x)$ achieves its global minimum at $x^0_1(w_0) = x^0_2(w_0)$.
This implies that $T_0(x^0_1) = T_0'(x^0_1) = 0$ which is equivalent to the system defined in \eqref{thm:resolve-lambda-2-c} assuming $\delta = 0$ (using the same calculation as in \eqref{eq:T'(x) = 0}).
In the rest of this section, without loss of generality, we assume $\delta > 0$.

To prove \Cref{thm:resolve-lambda-2-c}, we need following technical lemmas.

By the strict convexity of $T_\delta(x)$ stated in \Cref{fact:T-delta}, we have the following fact for its derivative.
\begin{fact}\label{fact:T-delta-derivative} 
 Assume $\alpha > \frac{1-\delta}{d} \cdot \e^{1 + \frac{1-\delta}{d}}$.
  \begin{enumerate}
  \item If $w = w_\delta$, then $x^\delta_1(w) = x^\delta_2(w)$ and $T'_\delta(x^\delta_1) = 0$.
  \item If $w > w_\delta$, then $T_\delta'(x^\delta_1) < 0$ and $T_\delta'(x^\delta_2) > 0$.
  \end{enumerate}
\end{fact}

\begin{lemma}\label{lem:sign-d-lambda}
Let $i\in\{1,2\}$. For $w > w_{\delta}$,
\[
\-{sign}\tp{\frac{\partial \log \lambda^\delta_i(w)}{\partial w}} = \-{sign}\tp{T_\delta'(x^\delta_i(w)) \cdot M_\delta(x^\delta_i(w))}.
\]
\end{lemma}

\begin{lemma} \label{lem:sign-d-xi}
  Let $i \in \{1, 2\}$. For $w > w_{\delta}$, if $M_\delta(x^\delta_i(w)) \geq 0$, then 
  \begin{align*}
    \-{sign}\tp{\frac{\partial x^\delta_i(w)}{\partial w}}
    = \-{sign}\tp{T'_\delta(x^\delta_i(w))}.
  \end{align*}
\end{lemma}

\begin{lemma} \label{lem:sign-d-xm}
The equation $M_\delta(x) = 0$ has a unique positive solution $x^\delta_M(w)$, which satisfies the followings: 
  \begin{enumerate}
  \item it holds that
    \begin{align*}
    M_\delta(x)
      \begin{cases}
         > 0 & \text{if } 0 \leq x <  x^\delta_M(w) \\
         = 0 & \text{if } x = x^\delta_M(w) \\
        < 0 & \text{if } x > x^\delta_M(w)
      \end{cases};
    \end{align*}
  \item $\frac{\partial x^\delta_M(w)}{\partial w} < 0$ for all $w \ge w_{\delta}$.
  \end{enumerate}
\end{lemma}

\Cref{lem:sign-d-lambda}, \ref{lem:sign-d-xi}, and \ref{lem:sign-d-xm} are proved in \Cref{sec:sign-d-lambda}, \ref{sec:sign-d-xi}, and \ref{sec:sign-d-xm}, respectively.

The following lemma compares the unique positive root $x^\delta_M(w)$ of $M_\delta(x)=0$ defined in \Cref{lem:sign-d-xm}, with the first positive root $x^\delta_1(w)$ of $T_\delta(x)=0$ defined in \Cref{def:x-lambda-delta}.
\begin{lemma} \label{lem:x1-le-xm}
  For $w \geq w_\delta$, 
  it holds that $x^\delta_1(w) < x^\delta_M(w)$ and $M_\delta(x) > 0$ for $x = x^\delta_1(w)$.  
\end{lemma}
\begin{proof}
We first prove that $M_\delta(x) > 0$ for $x = x^\delta_1(w)$. By \Cref{fact:T-delta-derivative}, for $w \geq w_\delta$, $T'_\delta(x) \leq 0$, and hence
\begin{align}\label{eq:proof-x1-le-xm-1}
  (1 - \delta)(\alpha + (1 + x)^w) &\leq \alpha d w - (1 - \delta) w (1 + x)^w.
\end{align}
Let $x=x^\delta_1(w)$. By  \Cref{def:x-lambda-delta}, $T_\delta(x) = 0$. Therefore,
\begin{align}
  \nonumber & &T_\delta(x) = 0 
  \nonumber  &= (1 - \delta)(1 + x)(\alpha + (1 + x)^w) - \alpha d w x\\
  \nonumber\text{(by \eqref{eq:proof-x1-le-xm-1})}&\implies  &\alpha d w x &\leq (1 + x) (\alpha d w - (1 - \delta) w (1 + x)^w) \\
  \label{eq:x1-fact} &\implies &0 &\leq \alpha d - (1 - \delta) (1 + x)^{w+1}.
\end{align}
Then, for $x=x^\delta_1(w)$,
\begin{align*}
  M_\delta(x)
  &= w \log(1 + x) (\alpha d - (1 + x)^{w+1}) + \delta(x + 1) (\alpha + (1 + x)^w) \\
  (\text{by \eqref{eq:x1-fact}})\qquad &\geq - \frac{\delta}{1-\delta} \alpha d w \log (1 + x) + \delta(x + 1) (\alpha + (1 + x)^w) \\
  (\text{by }T_\delta(x)=0)\qquad&{=} - \frac{\delta}{1-\delta} \alpha d w \log (1 + x) + \frac{\delta}{1-\delta} \alpha d w x \\
  (x > \log(1 + x)\text{ for }x>0)\qquad&> 0.
\end{align*}
This proves $M_\delta(x) > 0$ for $x = x^\delta_1(w)$. And $x^\delta_1(w) < x^\delta_M(w)$ follows from this and \Cref{lem:sign-d-xm}.
\end{proof}



\begin{lemma} \label{lem:x2-xm}
  There is a real number $w_c > w_\delta$ such that
  \begin{enumerate}
  \item if $w \in (w_\delta, w_c)$, then $x^\delta_2(w) < x^\delta_M(w)$;
  \item if $w = w_c$, then $x^\delta_2(w) = x^\delta_M(w)$;
  \item if $w \in (w_c, +\infty)$, then $x^\delta_2(w) > x^\delta_M(w)$.
  \end{enumerate}
\end{lemma}
\begin{proof}
  First, we claim that there is at least a point $w_\star \in (w_\delta, +\infty)$ such that $x^\delta_2(w_\star) = x^\delta_M(w_\star)$.
  Now, suppose $w_\star \in (w_\delta, +\infty)$ is the smallest point such that $x^\delta_2(w_\star) = x^\delta_M(w_\star)$.
  To understand the behavior of $x^\delta_2(w)$ and $x^\delta_M(w)$, we consider three cases: (1) $w \in (w_\delta, w_\star)$; (2) $w = w_\star$; (3) $w \in (w_\star, +\infty)$.

  \paragraph{Case(1):} If $w \in (w_\delta, w_\star)$, note that by \Cref{fact:T-delta-derivative} and \Cref{lem:x1-le-xm}, it holds that $x^\delta_2(w_\delta) = x^\delta_1(w_\delta) < x^\delta_M(w)$.
  By definition, recall that $w_\star \in (w_\delta, +\infty)$ is the smallest point such that $x^\delta_2(w_\star) = x^\delta_M(w_\star)$.
  Hence, by continuity of the function $x^\delta_2(w)$ and $x^\delta_M(w)$, for all $w \in (w_\delta, w_\star)$, we have $x^\delta_2(w) < x^\delta_M(w)$.

  \paragraph{Case(2):} If $w = w_\star$, then $x^\delta_2(w) = x^\delta_M(w)$ holds by the definition of $w_\star$.

  \paragraph{Case(3):} If $w \in (w_\star, +\infty)$, we will show that $x^\delta_2(w) > x^\delta_M(w)$.
  The proof will be done by a contradiction, by supposing the contrary.
  Note that by \Cref{fact:T-delta-derivative}, it holds that $T'_\delta(x^\delta_2(w_\star)) > 0$.
  By the fact that $T'_\delta(x^\delta_2(w_\star)) > 0$, $M_\delta(x^\delta_2(w_\star)) = 0$, and \Cref{lem:sign-d-xi}, it holds that $\-{sign}\tp{\frac{\partial x^\delta_2(w)}{\partial w}\vert_{w = w_\star}}  = \-{sign}\tp{T'_\delta(x^\delta_2(w_\star))} > 0$.
  Moreover, by \Cref{lem:sign-d-xm}, we know that $\frac{\partial x^\delta_M(w)}{\partial w} \vert_{w = w_\star} < 0$.
  By the continuity of $\frac{\partial x^\delta_2(w)}{\partial w}$ and $\frac{\partial x^\delta_M(w)}{\partial w}$, it holds that there is an $\epsilon > 0$ such that for $w \in (w_\star, w_\star + \epsilon)$, we have $x^\delta_2(w) > x^\delta_M(w)$.
  Now, suppose $w_\circ$ is the smallest point such that $w_\circ > w_\star$ and $x^\delta_2(w) = x^\delta_M(w)$.
  By definition of $w_\circ$ and the fact that $w \in (w_\star, w_\star + \epsilon)$, for any $w \in (w_\star, w_\circ)$, it holds that $x^\delta_2(w) > x^\delta_M(w)$.
  By the continuity of $\frac{\partial (x^\delta_2(w) - x^\delta_M(w))}{\partial w}$ and the mean value theorem, this implies that $\frac{\partial (x^\delta_2(w) - x^\delta_M(w))}{\partial w}\vert_{w = w_\circ} \leq 0$.
  However, by a similar argument on $w_\circ$ as for $w_\star$, we know that $\frac{\partial x^\delta_2(w)}{\partial} \vert_{w = w_0} > 0$ and $\frac{\partial x^\delta_M(w)}{\partial w} \vert_{w = w_\circ} < 0$, which implies that $\frac{\partial (x^\delta_2(w) - x^\delta_M(w))}{\partial w}\vert_{w = w_\circ} > 0$.
  This leads to a contradiction.
  
  By combining \textbf{Case (1) (2) (3)}, we finish the proof of \Cref{lem:x2-xm}.
  Now, we are only left to prove the claim that we made first: there is at least a point $w_\star \in [w_\delta, +\infty)$ such that $x^\delta_2(w_\star) = x^\delta_M(w_\star)$.
  We prove this by a proof of contradiction, by supposing the contrary of the claim.
  Note that by \Cref{fact:T-delta-derivative} and \Cref{lem:x1-le-xm}, $x^\delta_2(w_\delta) = x^\delta_1(w_\delta) < x^\delta_M(w_\delta)$.
  Suppose such $w_\star$ does not exists.
  Then, by the continuity of $x^\delta_2(w)$ and $x^\delta_M(w)$, it holds that $x^\delta_2(w) < x^\delta_M(w)$ for all $w > w_\delta$.

  On the one hand, by \Cref{lem:lb-is-0}, we have $\lim_{w \to \infty} x^\delta_2(w) = 0$.

  On the other hand, for any $w > w_\delta$, since $x^\delta_2(w) < x^\delta_M(w)$, by \Cref{lem:sign-d-xm}, it holds that $M_\delta(x^\delta_2(w)) > 0$.
  By \Cref{fact:T-delta-derivative}, it also holds that $T'_\delta(x^\delta_2(w)) > 0$.
  By $M_\delta(x^\delta_2(w)) > 0, T'_\delta(x^\delta_2(w)) > 0$ and \Cref{lem:sign-d-xi}, it holds that $\frac{\partial x^\delta_2(w)}{\partial w} > 0$ for all $w > w_\delta$.
  
  However, we know that $x^\delta_2(w_\delta) > 0$ (since $T_\delta(0) > 0$ for any $w$) and this fact together with the facts that $\lim_{w\to\infty} x^\delta_2(w) = 0$ and $\frac{\partial x^\delta_2(w)}{\partial w} > 0, \forall w > w_\delta$ leads to a contradiction!
\end{proof}

\begin{proof}[Proof of \Cref{thm:delta-unique}]
  Recall that by \Cref{tab:lambda2}, to prove \Cref{thm:delta-unique}, it suffices to resolve
  \begin{align*}
    \lambda^\delta_{2, c} = \sup_{w > w_\delta} \lambda^\delta_2(w).
  \end{align*}
  To do this, all we need to do is to show that the function $\lambda^\delta_2(w)$ achieves its maximum at $w_c$, where $(x_c, w_c)$ is the unique positive solution of the following system
  \begin{align} \label{eq:TM-sys}
  \begin{cases}
    T_\delta = (1 - \delta)(x + 1)(\alpha + (1 + x)^w) - \alpha d w x = 0, \\
    M_\delta = w \log(1 + x) (\alpha d - (1 + x)^{w+1}) + \delta(x + 1)(\alpha + (x + 1)^w) = 0.
  \end{cases}
  \end{align}
  
  First, we finish the easy part, that is showing \eqref{eq:TM-sys} has a unique positive solution.
  For a fixed $w \geq w_\delta$, recall that for $i \in \{1, 2\}$, $x^\delta_i$ is the solution of the equation $T_\delta(x) = 0$ and $x^\delta_M(w)$ is the solution of the equation $M_\delta(x) = 0$.
  Note that a pair $(x, w)$ is the solution of \eqref{eq:TM-sys} iff $x^\delta_i(w) = x^\delta_M(w)$ for some $i\in \{1, 2\}$.
  This means that the pair $(x, w)$ with $w = w_\delta$ can not be a solution, since by \Cref{fact:T-delta-derivative} and \Cref{lem:x1-le-xm}, we have $x^\delta_2 = x^\delta_1 < x^\delta_M$.
  By \Cref{lem:x1-le-xm} and \Cref{lem:x2-xm}, it is straightforward to verify that the only possible situation that $x^\delta_i(w) = x^\delta_M(w)$ happens at $w = w_c$, where $w_c$ is the one defined in \Cref{lem:x2-xm} and $x_c = x^\delta_2(w_c) = x^\delta_M(w_c)$.
  Hence, \eqref{eq:TM-sys} has a unique positive solution $(x_c, w_c)$ for $x_c = x^\delta_2(w_c)$.

  Now, we are only left to show that the function $\lambda^\delta_2(w)$ achieves its maximum at the point $w = w_c$.
  Recall that according to \Cref{fact:T-delta-derivative}, it holds that $T_\delta'(x^\delta_2) > 0$ for all $w > w_\delta$.
  By \Cref{lem:x2-xm} and the sign of $T_\delta'(x^\delta_2)$, we divide all the cases of $w > w_\delta$ into three cases to understand the monotonicity of $\lambda^\delta_2(w)$: (1) $w \in (w_\delta, w_c)$; (2) $w = w_c$; (3) $w \in (w_c, +\infty)$.


  \paragraph{Case (1):}
  If $w \in (w_\delta, w_c)$, then by \Cref{fact:T-delta-derivative}, $T'_\delta(x^\delta_2) > 0$.
  Moreover, by \Cref{lem:x2-xm}, we have $x^\delta_2(w) < x^\delta_M(w)$.
  Combining $x^\delta_2(w) < x^\delta_M(w)$ and \Cref{lem:sign-d-xm}, we have $M_\delta(x^\delta_2(w)) > 0$.
  By $M_\delta(x^\delta_2(w)) > 0$, $T'_\delta(x^\delta_2) > 0$, and \Cref{lem:sign-d-lambda}, we have
  \begin{align*}
    \-{sign}\tp{\frac{\partial \log \lambda^\delta_2(w)}{\partial w}} = \-{sign}\tp{T'_\delta(x^\delta_2) \cdot M_\delta(x^\delta_2)} > 0.
  \end{align*}

  \paragraph{Case (2):} If $w = w_c$, by \Cref{lem:x2-xm}, it holds that $x^\delta_2(w) = x^\delta_M(w)$, which, by \Cref{lem:sign-d-xm}, implies $M_\delta(x^\delta_2) = 0$.
  Combining $M_\delta(x^\delta_2) = 0$ and \Cref{lem:sign-d-lambda} gives us 
  \begin{align*}
    \-{sign}\tp{\frac{\partial \log \lambda^\delta_2(w)}{\partial w}} = \-{sign}\tp{T'_\delta(x^\delta_2) \cdot M_\delta(x^\delta_2)} = 0.
  \end{align*}

  \paragraph{Case (3):} If $w > w_c$, by \Cref{lem:x2-xm}, it holds that $x^\delta_2(w) > x^\delta_M(w)$, which, by \Cref{lem:sign-d-xm}, implies that $M_\delta(x^\delta_2) < 0$.
  Furthermore, by \Cref{fact:T-delta-derivative}, it holds that $T'_\delta(x^\delta_2) > 0$.
  By $M_\delta(x^\delta_2) < 0$, $T'_\delta(x^\delta_2) > 0$, and \Cref{lem:sign-d-lambda}, it holds that
  \begin{align*}
    \-{sign}\tp{\frac{\partial \log \lambda^\delta_2(w)}{\partial w}} = \-{sign}\tp{T'_\delta(x^\delta_2) \cdot M_\delta(x^\delta_2)} < 0.
  \end{align*}

  Combining \textbf{Case (1) (2) (3)}, it holds that $\lambda^\delta_2(w)$ as a function of $w$ is monotonically increasing when $w \in (w_\delta, w_c)$, and is monotonically decreasing when $w \in [w_c, +\infty)$.
  Hence its maximum is achieved at the point $w = w_c$.
  This finishes the proof.
\end{proof}

\subsubsection{Monotonicity of $\lambda^\delta_i(w)$} \label{sec:sign-d-lambda}
In this section, we prove \Cref{lem:sign-d-lambda}.
Recall the definition of $T_\delta$ in \Cref{eq:function-T-delta},
\begin{align*}
  T_\delta := (1 - \delta)(x + 1)(\alpha + (1 + x)^w) - \alpha d w x.
\end{align*}
Fix $i \in \{1, 2\}$ arbitrarily and let $x(w) := x_i^\delta(w)$.
Recall that $x(w)$ is well-defined as long as $w >w_{\delta}$.
Taking derivative on both sides of $T_\delta(x(w), w) = 0$, we have
$ \partial_x T_\delta(x, w) \cdot x'(w) + \partial_w T_\delta(x, w) = 0$, where $x = x(w)$.
This means
\begin{align} \label{eq:d-xi}
  x'(w) &= - \frac{\partial_w T_\delta(x, w)}{\partial_x T_\delta(x, w)} = \frac{1}{\partial_x T_\delta(x, w)} \tp{\alpha  d x(w)-(1-\delta ) (x(w)+1)^{w+1} \log (x(w)+1)}.
\end{align}

Recall that $\lambda(w) := \lambda^\delta_i(w) = x(w)(1 + \alpha(1 + x(w))^{-w})^d$ in \Cref{def:x-lambda-delta}, we have
\begin{align*}
  x(w) (x(w)+1) & \left(\alpha +(x(w)+1)^w\right)
  \frac{\partial \log \lambda (x(w), w)}{\partial w}  \\
  =& x'(w) ((x(w) + 1)(\alpha + (1 + x(w))^w) - \alpha d w x(w))  \\
  & \quad -\alpha d x(w)^2 \log (x(w)+1) - \alpha  d x(w) \log (x(w)+1), \\
  =& x'(w) \cdot \frac{\delta}{1 - \delta} \alpha d w x(w) -\alpha d x(w)^2 \log (x(w)+1) - \alpha  d x(w) \log (x(w)+1),
\end{align*}
where the last equation comes from the fact that $T_\delta(x) = 0$.
Plugging in $x'(w)$ as in \eqref{eq:d-xi}, and note that the denominator of $x'(w)$ is $\partial_x T_\delta(x, w) = T_\delta'(x(w))$,
\begin{align*}
  x&(w)  (x(w)+1) \left(\alpha +(x(w)+1)^w\right) \cdot \frac{(1 - \delta) T_\delta'(x(w))}{\alpha d x(w)} \cdot \frac{\partial \log \lambda}{\partial w}  \\
  = & (1 - \delta)\log(1 + x(w))(1 + x(w))\tp{\alpha  (d w+\delta -1)+(\delta -w-1) (x(w)+1)^w} + \delta \cdot \alpha d w x(w).
\end{align*}
Note that
\begin{align*}
  - &(1 + x(w))\tp{\alpha  (d w+\delta -1)+(\delta -w-1) (x(w)+1)^w} \\
  &= (w + 1 - \delta)(1 + x(w))^{w+1} - \alpha(d w  - (1 - \delta)) (1 + x(w)) \\
  &= w((1 + x(w))^{w+1} - \alpha d) + (1 - \delta) (1 + x(w))(\alpha + (1 + x(w))^w) - \alpha d w x(w) \\
  &= w((1 + x(w))^{w+1} - \alpha d),
\end{align*}
where the last equation comes from the fact that $T_{\delta}(x(w)) = 0$.
Hence we have
\begin{align}
  \nonumber x(w) & (x(w)+1) \left(\alpha +(x(w)+1)^w\right) \cdot \frac{T_\delta'(x(w))}{\alpha d x(w)} \cdot \frac{\partial \log \lambda}{\partial w}  \\
 \nonumber &= -w\log(1 + x(w)) ((1 + x(w))^{w+1} - \alpha d) + \frac{\delta}{1 - \delta} \alpha d w x(w) \\
 \label{eq:is-M} &= w\log(1 + x(w)) (\alpha d - (1 + x(w))^{w+1}) + \delta (x(w) + 1)(\alpha + (x(w) + 1)^w).
\end{align}
where the last equation comes from the fact that $T_\delta(x(w)) = 0$.
We finish the proof by noticing that \eqref{eq:is-M} is exactly $M_\delta(x(w))$, where we recall the definition of $M_\delta$ in \Cref{thm:resolve-lambda-2-c}.

\subsubsection{Monotonicity of $x^\delta_i(w)$} \label{sec:sign-d-xi}
In this section, we prove \Cref{lem:sign-d-xi}.
Fix $i \in \{1, 2\}$ arbitrarily and let $x(w) := x_i^\delta(w)$.
Recall that $x(w)$ is well-defined as long as $w >w_{\delta}$.
Recalling \eqref{eq:d-xi}, we know that the denominator of $x'(w)$ is $T_\delta'(x(w))$.
The numerator is
\begin{align*}
  \alpha  d x(w)-(1-\delta ) (x(w)+1)^{w+1} \log (x(w)+1).
\end{align*}
Recall the definition of $M_\delta$ at \Cref{thm:resolve-lambda-2-c}.
The proof of \Cref{lem:sign-d-xi} can be finished by the following observation.

\begin{observation}
  Fix any $w > w_{\delta}$ such that $M_\delta(x^\delta_i(w)) \geq 0$. If we have
  \begin{align*}
     w\log(1 + x(w)) (\alpha d - (1 + x(w))^{w+1}) + \delta (x(w) + 1)(\alpha + (x(w) + 1)^w) \geq 0,
  \end{align*}
  then it holds that 
  \begin{align*}
    \alpha  d x(w)-(1-\delta ) (x(w)+1)^{w+1} \log (x(w)+1) > 0.
  \end{align*}
\end{observation}
\begin{proof}
  According to the assumption, we have
  \begin{align} \label{eq:M-ge-0}
    w\log(1 + x(w)) & (1 + x(w))^{w+1} \leq \alpha d w\log(1 + x(w)) + \delta (x(w) + 1)(\alpha + (x(w) + 1)^w).
  \end{align}
  Recall the definition of $T_\delta$ at \Cref{eq:function-T-delta}, $T_\delta(x(w)) = 0$ implies that
  \begin{align} \label{eq:T-eq-0}
    \frac{\alpha d w x(w)}{1 - \delta} &= (1 + x(w))(\alpha + (1 + x(w))^w).
  \end{align}
  Hence, by $\eqref{eq:T-eq-0} - \eqref{eq:M-ge-0}$ at both side, we have
  \begin{align*}
    \frac{\alpha d w x(w)}{1 - \delta} & - w\log(1 + x(w)) (1 + x(w))^{w+1} \\
    \geq& (1 + x(w))(\alpha + (1 + x(w))^w)  - \alpha d w\log(1 + x(w)) - \delta (x(w) + 1)(\alpha + (x(w) + 1)^w) \\
    =& (1 - \delta)(1 + x(w))(\alpha + (1 + x(w))^w) - \alpha d w\log(1 + x(w)) \\
    \overset{(\star)}{=}& \alpha d w x(w) - \alpha d w \log(1 + x(w)) > 0,
  \end{align*}
  where $(\star)$ holds by $T_\delta(x(w)) = 0$ and in the last inequality, we use the fact that $x(w) > 0$ for finite $w$ (i.e. $T_\delta(0) > 0$) and $x > \log(1 + x)$ for $x > 0$.
\end{proof}

\subsubsection{Monotonicity of $x^\delta_M(w)$} \label{sec:sign-d-xm}
In this section, we prove \Cref{lem:sign-d-xm}.
Recall that we have
\begin{align*}
  M_\delta(x) &= w \log(1 + x) (\alpha d - (1 + x)^{w+1}) + \delta (x + 1) (\alpha + (x + 1)^w).
\end{align*}
And it holds that
\begin{align*}
  M'_\delta(x) &= \frac{1}{x+1} \cdot \left(w \left(\alpha  d+(\delta -1)(x+1)^{w+1}\right) \right. \\
  & \hspace{2cm} \left. +\delta  (x+1) \left(\alpha +(x+1)^w\right)-w (w+1) (x+1)^{w+1} \log (x+1)\right) \\
  M''_\delta(x) &= -\frac{1}{(x+1)^2} \left(w \left(\alpha  d + (1 -\delta) (x+1)^{w+1} \right. \right.\\
  & \hspace{2cm} \left. \left. +(2 - \delta ) w (x+1)^{w+1}+ w (w+1) (x+1)^{w+1} \log (x+1)\right) \right),
\end{align*}
which means $M''_\delta(x) < 0$ for all $x > 0$.
Moreover, we have
\begin{align} \label{eq:M-delta-0-inf}
  \lim_{x \to 0} M_\delta(x) &= (\alpha +1) \delta > 0 \quad \text{and} \quad \lim_{x \to +\infty} M_\delta(x) = -\infty.
\end{align}
Hence the equation $M_\delta(x) = 0$ has a unique positive solution $x^\delta_m(w)$, with $M_\delta(x) > 0$ for $x < x^\delta_m(w)$ and $M_\delta(x) < 0$ for $x > x^\delta_m(w)$.

Now, we are only left to show $\frac{\partial x^\delta_m(w)}{\partial w} < 0$.
For simplicity, let $x(w) := x^\delta_m(w)$.
In this case, we treat $M_\delta(x(w), w)$ as a function on both $x(w)$ and $w$ as follows:
\begin{align*}
	M_\delta(x, w) := w \log(1 + x) (\alpha d - (1 + x)^{w+1}) + \delta (x + 1) (\alpha + (x + 1)^w).
\end{align*}
Taking derivatives on both sides of $M_\delta(x(w), w) = 0$, we have $\partial_x M_\delta(x, w) \cdot x'(w) + \partial_w M_\delta(x, w) = 0$, where $x = x(w)$.
This means
\begin{align*}
  x'(w) &= -\frac{\partial_w M_\delta(x, w)}{\partial_x M_\delta(x, w)} = \frac{-\partial_wM_\delta(x, w)}{M_\delta'(x(w))}.
\end{align*}

Recall the definition of $M_\delta$ at \Cref{thm:resolve-lambda-2-c}.
The numerator, which is $-\partial_wM_\delta(x, w)$, of $x'(w)$ is
\begin{align*}
  - \partial_wM_\delta(x, w) &= \log (x(w)+1) \left(-\alpha  d-\delta  (x(w)+1)^{w+1}+(x(w)+1)^{w+1} \right.\\
  & \hspace{5cm} \left. +w (x(w)+1)^{w+1} \log (x(w)+1)\right), 
\end{align*}
The sign of $- \partial_w M_\delta(x, w)$ is given by,
\begin{align*}
  - &\alpha  d + (x(w)+1)^{w+1} -\delta  (x(w)+1)^{w+1}+w (x(w)+1)^{w+1} \log (x(w)+1) \\
    &= \frac{\delta (x(w) + 1) (\alpha + (x(w) + 1)^w)}{w \log (x(w) + 1)} -\delta  (x(w)+1)^{w+1} \\
  & \hspace{2cm} +w (x(w)+1)^{w+1} \log (x(w)+1) \\
  &> \frac{\delta (x(w) + 1)^{w+1}}{w \log (x(w) + 1)} -\delta  (x(w)+1)^{w+1}+w (x(w)+1)^{w+1} \log (x(w)+1) \\
  &\overset{(\star)}{\geq} 2 \sqrt{\delta} (x(w) + 1)^{w+1} - \delta (x(w) + 1)^{w+1} \geq 0,
\end{align*}
where the first equation holds by $M_\delta(x(w)) = 0$ (defined in \Cref{thm:delta-unique}), $(\star)$ holds by AM-GM inequality, and in the last inequality, we use the fact that $\delta \in (0, 1)$.

Note that the denominator of $x'(w)$ is actually $M_\delta'(x(w))$.
According to \eqref{eq:M-delta-0-inf}, we have $M_\delta''(x) < 0$ for all $x > 0$ and $\lim_{x\to 0} M_\delta(x) > 0$.
By the mean value theorem, there is $\zeta \in (0, x(w))$ such that
\begin{align*}
  M'_\delta(\zeta) = \frac{M_\delta(x(w)) - M_\delta(0)}{x(w) - 0} < 0.
\end{align*}
Now, since $x(w) > \zeta$ and $M''_\delta(x) < 0, \forall x > 0$, it holds that $M'_\delta(x(w)) < 0$.

\section{Spectral independence and correlation decay analysis}
\label{sec:SI}
In this section, we prove \Cref{thm:SI}.
Let $\delta \in (0, 1)$, $\lambda, \alpha > 0$ be real numbers, and $d\ge 1$ be an integer.
Let $\mu$ be the hardcore distribution on $G= ((L, R), E)$ with maximum degree at most $\Delta = d + 1$ and fugacity $\lambda$ on $L$, and fugacity $\alpha$ on $R$. Let $\nu = \mu_L$ be the projection of $\mu$ on $L$.

Throughout the section, we assume that $(\lambda, d, \alpha)$ is $\delta$-unique,
and we are going to show that $\nu$ is $O(\frac{1}{\delta})$-spectrally independent.

%
%

Let $\*w = (w_1, w_2, \cdots, w_d) \in \^Z^d_{\geq 0}$. Let $F_{\*w}: [0, +\infty]^{\sum_{k} w_k} \to \^R$ be the two-step tree-recursion for marginal ratios in the hardcore model (as deduced in \Cref{sec:uniqueness}), formally defined by
\begin{align*}
  F_{\*w}(\*x) := \lambda \prod_{k=1}^d\tp{1 + \alpha \prod_{h=1}^{w_k}(1 + x_{kh})^{-1}}^{-1}.
\end{align*}
Let $\psi: \^R_{\geq 0} \to \^R_{> 0}$ be integrable.
Let $H_{\*w}^\psi: [0, +\infty]^{\sum_k w_k} \to \^R$ be defined by
\begin{align}
  H_{\*w}^\psi(\*x)
  \nonumber &:= \sum_{k=1}^d \sum_{h=1}^{w_k} \frac{\psi(F_{\*w}(\*x))}{\psi(x_{kh})} \cdot \abs{\frac{\partial (\log \circ F_{\*w} \circ \exp) (\*y)}{\partial y_{kh}}} & \text{(where $y_{kh} = \ln x_{kh}$)}\\
  \label{eq:H-psi} &= \sum_{k = 1}^d \sum_{h = 1}^{w_k} \frac{\psi(F_{\*w}(\*x))}{\psi(x_{kh})} \frac{\alpha \prod_{j=1}^{w_k}(1 + x_{kj})^{-1}}{1 + \alpha \prod_{j=1}^{w_k}(1 + x_{kj})^{-1}} \frac{x_{kh}}{1 + x_{kh}}.
\end{align}

Intuitively, the maximum value of $H_{\*w}^\psi=\norm{\nabla (\Phi \circ \log \circ F_{\*w} \circ \exp \circ \Phi^{-1})}_1$ upper bounds the decay of errors in $\infty$-norm,
after a change of variables $z=\Phi(\log(x))$ applied to the original tree-recursion $F_{\*w}$ for some $\Phi$.
Here, $\Phi:\mathbb{R}\to \mathbb{R}$ is a monotone (invertible) and differentiable potential function satisfying $\Phi'(\log x) = \psi(x)$. 
In the rest of section, we will be working with its derivative, $\psi$, rather than $\Phi$.
Hence, we are going to call $\psi$  the \emph{potential function} from now on.

We have the following abstract lemma for establishing spectral independence.
\begin{lemma} \label{lem:a-c-SI}
  For $\beta \in (0, 1)$ and $c > 0$,
  if there is a function $\psi: \^R_{\geq 0} \to \^R_{> 0}$ such that
  \begin{enumerate}
  \item (contraction) for any $1 \leq d_v \leq d$, $\*w \in \^Z^{d_v}_{\geq 0}$, and $\*x \in [\lambda(1 + \alpha)^{-d}, \lambda]^{\sum_k w_k}$,
  \[
  H^\psi_{\*w}(\*x) \leq 1 - \beta;
  \]
  \item (boundedness) for any $a, b \in [\lambda(1 + \alpha)^{-d}, \lambda]$,  $1 \leq d_r \leq \Delta$,  $\*w \in \^Z^{d_r}_{\geq 0}$, and  $\*x \in [\lambda(1 + \alpha)^{-d}, \lambda]^{\sum_k w_k}$,
    \[
    \frac{\psi(a)}{\psi(b)} \cdot H^\psi_{\*w}(\*x) \leq c \cdot (1 - \beta);
    \]
  \end{enumerate}
  then it holds that $\nu$ is $\frac{c}{\beta}$-spectrally independent.
\end{lemma}
\Cref{lem:a-c-SI} can be proved by 
following the same strategy developed in \cite{chen2020rapid} for proving a similar abstract result. 
We did not find a ``black-box'' application of their result to prove \Cref{lem:a-c-SI}. 
Therefore, we include a formal proof of \Cref{lem:a-c-SI} in \Cref{sec:a-c-SI} for the completeness.

In the rest of the section, we assume the following concrete choice of potential function: 
\begin{align}\label{eq:potential-function-psi}
\psi(x) := \frac{x}{(x +1)\log (x + 1)}.
\end{align}

\begin{lemma}[contraction] \label{lem:contract}
Let $\delta \in (0, 1)$ and $d = \Delta - 1\ge 1$ be an integer.  
If $(\lambda, d, \alpha)$ is $\delta$-unique, then for any $1 \leq d_v \leq \Delta - 1$, $\*w \in \^Z_{\geq 0}^{d_v}$ and $\*x \in [\lambda(1 + \alpha)^{-d}, \lambda]^{\sum_{k} w_k}$, it holds that $H^\psi_{\*w}(\*x) \leq 1 - \delta$.
\end{lemma}

\begin{lemma}[boundedness] \label{lem:boundedness}
  Let $\delta \in (0, 1)$ and $d = \Delta - 1\ge 1$ be an integer.
   If $(\lambda, d, \alpha)$ is $\delta$-unique, then for any $a, b \in [\lambda(1 + \alpha)^{-d}, \lambda]$, $1 \leq d_v \leq \Delta$, $\*w \in \^Z_{\geq 0}^{d_v}$ and $x \in [\lambda(1 + \alpha)^{-d}, \lambda]^{\sum_{k} w_k}$, it holds that
  \begin{align*}
    \frac{\psi(a)}{\psi(b)} \cdot H^\psi_{\*w}(\*x) \leq \frac{\Delta}{d} (1 + \alpha)^\Delta \cdot (1 - \delta).
  \end{align*}
\end{lemma}

\Cref{thm:SI} follows straightforwardly from \Cref{lem:a-c-SI}, \Cref{lem:contract}, and \Cref{lem:boundedness}. 

\Cref{lem:contract} is proved in \Cref{sec:sym} and \Cref{sec:critical-contraction}. 
And \Cref{lem:boundedness} is proved in \Cref{sec:boundedness}.

\subsection{Symmetrization}
\label{sec:sym}
Recall that we assume the form of the potential function $\psi(x)$ in \eqref{eq:potential-function-psi}. And define 
\[
\phi(x) := \psi(x) / x =\frac{1}{(x + 1) \log (x + 1)}.
\]
The function $H_{\*w}(\*x)=H^{\psi}_w(\*x)$ in \eqref{eq:H-psi} becomes:
\begin{align}
  H_{\*w}(\*x)
  \label{eq:multi-var-H} &= F(\*x) \sum_{k = 1}^d \sum_{h = 1}^{w_k} \frac{\phi(F(\*x))}{\phi(x_{kh})} \frac{\alpha \prod_{j=1}^{w_k}(1 + x_{kj})^{-1}}{1 + \alpha \prod_{j=1}^{w_k}(1 + x_{kj})^{-1}} \frac{1}{1 + x_{kh}} \\
  \nonumber &= \sum_{k=1}^d \sum_{h=1}^{w_k} \abs{\frac{\partial F(\*x)}{\partial x_{kh}}} \frac{\phi(F(\*x))}{\phi(x_{kh})}.
\end{align}
We note that our choice of $\phi$ is exactly the derivative of the potential function used in~\cite{liu2015fptas}, 
which is proven to be very useful in handling heterogeneous degrees $\*w = (w_1, \cdots, w_d)$ in the second level of recursion such that the $H_{\*w}(\*x)$ can be symmetrized to a univariate function.

We remark that the potential function plays a different role in our proof.  
In many analysis of correlation decay (including \cite{liu2015fptas}), 
the potential function is used for amortizing the contraction.
  In contrast, we use the potential function to reduce every $\delta$-unique parameters to a set of  ``exact'' $\delta$-unique parameters (see \Cref{cond:critical}).
  Then, we show that the contraction of $H_{\*w}$ is bounded by the contraction at the fixpoint of the univariate tree recursion $F$ encoded by the ``exact'' $\delta$-unique parameters.
  Our strategy allows for a more meaningful analysis, as most calculations are done exactly at the fixpoint of the univariate tree recursion $F$.

\begin{lemma}[\text{\cite[Claim 4.5]{liu2015fptas}}] \label{lem:sym-s}
  For $\*x \in [0, +\infty]^{\sum_{k=1}^d w_k}$, there is a $z \in [1, 1 + \alpha]$, such that 
  \[
  H_{\*w}(\*x)\le U(z),
  \]
  where
  \begin{align*}
  U(z) = U_{\lambda, d, \alpha}(z) := \frac{\lambda z^{-d}}{(1 + \lambda z^{-d}) \log (1 + \lambda z^{-d})} \cdot d \cdot \frac{z - 1}{z} \log \tp{\frac{\alpha}{z - 1}}.
  \end{align*}
\end{lemma}
\Cref{lem:sym-s} was proved in \cite[Claim 4.5]{liu2015fptas} for the special case with $\lambda=\alpha=1$.
For completeness, we formally verify that the same proof works for all fugacity.
\begin{proof}[Proof of \Cref{lem:sym-s}]
The function $H_{\*w}(\*x)$ in \eqref{eq:multi-var-H} can be rewritten as:
\begin{align*}
  \frac{F(\*x)}{(1 + F(\*x)) \log (1 + F(\*x))} \sum_{k=1}^d \frac{\alpha \prod_{j=1}^{w_k}(1 + x_{kj})^{-1}}{1 + \alpha \prod_{j=1}^{w_k}(1 + x_{kj})^{-1}} \log\tp{\prod_{h=1}^{w_k} (1 + x_{kh})}.
\end{align*}
For $1\le k\le d$, let $z_k := 1 + \alpha \prod_{j=1}^{w_k} (1 + x_{kj})^{-1}$.
Note that $z_k \in [1, 1 + \alpha]$ and we have
\begin{align*}
  H_{\*w}(\*x) = \frac{\lambda \prod_{k=1}^d z_k^{-1}}{(1 + \lambda \prod_{k=1}^d z_k^{-1})\log(1 + \lambda \prod_{k=1}^d z_k^{-1})} \sum_{k=1}^d \frac{z_k-1}{z_k} \log \tp{\frac{\alpha}{z_k-1}}.
\end{align*}
Let $z := (\prod_{k=1}^d z_k)^{1/d}$, it is sufficient to show that
\begin{align} \label{eq:jensen}
  \sum_{k=1}^d \frac{z_k - 1}{z_k} \log \tp{\frac{\alpha}{z_k - 1}} \leq d \cdot \frac{z - 1}{z} \log \tp{\frac{\alpha}{z - 1}},
\end{align}
which due to Jensen's inequality,  follows from the concavity of the function on $[0, \ln(1 + \alpha)]$:
\begin{align*}
  R(x) := \frac{\e^x - 1}{\e^x} \log \tp{\frac{\alpha}{\e^x - 1}}.
\end{align*}
The concavity of $R(x)$ on  $[0, \ln(1 + \alpha)]$ is guaranteed by that $R''(x) = \frac{1}{1 - \e^x} - \e^{-x}\log\tp{\frac{\alpha}{\e^x - 1}}<0$ for $x\in[0, \ln(1 + \alpha)]$, since
$R''(\ln(1+\alpha)) = - 1/\alpha < 0$ and $R'''(x) = \e^{-x} \log \left(\frac{\alpha }{\e^x-1}\right)+\frac{2 \e^x-1}{\left(\e^x-1\right)^2} > 0$.
\end{proof}

\begin{remark}
According to \Cref{lem:sym-s},
the supremum of the univariate function $U(z)$ on $[1,1+\alpha]$  upper bounds the contraction. 
However, it is still technically challenging to bound the maximum value of $U(z)$ under the assumption of $\delta$-uniqueness.
Previously, this upper bound is only known for the special case where $\lambda = \alpha = 1$, through numerical experiments for $d\in\{ 1, 2, 3, 4\}$~\cite{liu2015fptas}.
In this work, we are able to go much further by leveraging critical information about fixpoints that we established in \Cref{sec:delta-unique}.
This allows us to give tight upper bounds for $U(z)$ analytically for all $d\ge 1$.
\end{remark}

Recall the univariate tree-recursion $F(x)$ in \eqref{eq:univariate-tree-recursion-F(x)}. Let $w, d \in \^R_{>0}$. 
\[
F(x) := \lambda (1 + \alpha(1 + x)^{-w})^{-d}.
\]
We also define the univariate variant of the function $H_{\*w}(\*x)$ in \eqref{eq:multi-var-H}:
\begin{align}
H(x) = H_{\lambda, d, \alpha, w}(x) &:= \frac{\phi(F(x))}{\phi(x)} F'(x)\notag\\
&= \frac{(1 + x)\log(1 + x)}{(1 + F(x))\log(1 + F(x))} \cdot d w \cdot \frac{\alpha (1 + x)^{-w}}{1 + \alpha (1 + x)^{-w}} \cdot \frac{1}{1 + x} \cdot F(x).\label{eq:univariate-function-H(x)}
\end{align}
In fact, the $H(x)$ and $U(z)$ are equivalent under a change of variables, such that  
  \begin{align}\label{eq:sym-x}
  \forall w > 0, \forall x \ge 0:\quad U_{\lambda, d, \alpha}(z)= H_{\lambda, d, \alpha, w}(x)\quad \text{ for }z=1 + \alpha (1 + x)^{-w}. 
  \end{align}
Note that for any fixed $w > 0$,  the mapping $x \mapsto 1 + \alpha(1 + x)^{-w}$ is a bijection from $[0, +\infty]$ to $[1, 1 + \alpha]$. 
The following is easy to verify:
  \begin{align}\label{eq:sym-x-sup}
  \forall w > 0: \quad \sup_{z\in[1,1 + \alpha]}U_{\lambda, d, \alpha}(z)= \sup_{x\in [0, +\infty]}H_{\lambda, d, \alpha, w}(x). 
  \end{align}
  Intuitively, the freedom of choosing $w$ in $H_{\lambda, d, \alpha, w}(x)$ provides much more flexibility than $U_{\lambda, d, \alpha}(z)$, in that it allows us to use not only the information at a single fixpoint, but information from a whole family of fixpoints.

\subsection{Contraction}
\label{sec:critical-contraction}
By symmetrization,
the contraction stated in \Cref{lem:contract} is implied by the following lemma.
\begin{theorem}[$\delta$-contraction up to $\delta$-uniqueness] \label{thm:contract-sym}
  Let $\delta \in (0, 1)$. If $(\lambda, d, \alpha)$ is $\delta$-unique, then for $w > 0$, 
  \begin{align*}
   \forall  x \geq 0:\quad  H_{\lambda, d, \alpha, w}(x) \leq 1 - \delta.
  \end{align*}
\end{theorem}



\begin{proof}[Proof of \Cref{lem:contract}]
By \Cref{lem:sym-s}, we have $H^\psi_{\*w}(\*x) \le U(z)$ for any $\*w\in\mathbb{Z}_{\ge 0}^{d_v}$. 
  Further note that for any fixed $z \in [1, 1 + \alpha]$, the value of $U(z)$ is increasing in $d$, since
  \begin{align*}
    \frac{\partial \log U}{\partial d} &= \frac{1}{d} + \frac{\log (z) \tp{-z^d + \frac{\lambda}{\log(1 + z^{-d} \lambda)}}}{z^d + \lambda} > 0.
  \end{align*}
  Therefore, we can assume $d_v = d$ in $U(z)$. 
  Since $(\lambda, d, \alpha)$ is $\delta$-unique, by \eqref{eq:sym-x} and \Cref{thm:contract-sym}, we have $U(z)= H(x) \leq 1 - \delta$.
\end{proof}

We now prove \Cref{thm:contract-sym}.
Due to the monotonicity of $U(z)$ in $\lambda$ and $\alpha$, we only need to focus on the following critical case of $\delta$-uniqueness for $\delta \in [0, 1)$. 

\begin{condition}[critical condition] \label{cond:critical}
  $\alpha > \frac{1-\delta}{d} \e^{1 + \frac{1-\delta}{d}}$, $\lambda = \lambda_c$, and $w = w_c$, 
  where $\lambda_c=\lambda_c(\delta,d,\alpha)$ and $w_c=w_c(\delta,d,\alpha)$ are the critical thresholds for $\lambda$ and $w$  defined in \Cref{thm:delta-unique} for $\delta$-uniqueness.
\end{condition}

\begin{lemma} \label{lem:contract-critical}
  Assuming \Cref{cond:critical}, for all $x \geq 0$, it holds for the $H(x)$ defined in \eqref{eq:univariate-function-H(x)} that
  \begin{align*}
    H(x) = \frac{\phi(F(x))}{\phi(x)} F'(x) \leq 1 - \delta.
  \end{align*}
  Moreover, $H(x_c) = 1 - \delta$, for $x_c = x_c(\delta, d, \alpha)$ defined in \Cref{thm:delta-unique}.
\end{lemma}

The prove of \Cref{lem:contract-critical} will be given in \Cref{sec:contract-critical}.

The contraction up to $\delta$-uniqueness claimed in  \Cref{thm:contract-sym}, can be reduced to this contraction for the critical case guaranteed in \Cref{lem:contract-critical}. This is proved as follows.

\begin{proof}[Proof of \Cref{thm:contract-sym}]
Let $(\lambda, d, \alpha)$ be $\delta$-unique  and let $w > 0$.
We claim that there always exists $(\lambda_c, d, \alpha',w_c)$ satisfying \Cref{cond:critical} such that $\lambda \geq \lambda_c$ and $\alpha\leq \alpha'$. 
\begin{itemize}
\item
If $\alpha > \frac{1-\delta}{d} \e^{1 + \frac{1-\delta}{d}}$, then $\alpha'=\alpha$, $\lambda_c=\lambda_c(\delta,d,\alpha')$ and $w_c=w_c(\delta,d,\alpha')$ satisfy our requirement, since $\lambda \geq \lambda_c$ is guaranteed by \Cref{thm:delta-unique} and the $\delta$-uniqueness of $(\lambda, d, \alpha)$.
\item
If $\alpha \le \frac{1-\delta}{d} \e^{1 + \frac{1-\delta}{d}}$, we choose a small enough $\alpha'> \frac{1-\delta}{d} \e^{1 + \frac{1-\delta}{d}}$ such that $\lambda_c(\delta,d,\alpha')\le \lambda$. Such an $\alpha'$ always exists because by \Cref{lem:lb-is-0}, we have $\lambda_c(\delta,d,\alpha')\downarrow 0$  as $\alpha' \downarrow \frac{1-\delta}{d} \e^{1 + \frac{1-\delta}{d}}$. We let $\lambda_c=\lambda_c(\delta,d,\alpha')$ and $w_c=w_c(\delta,d,\alpha')$. Clearly, $(\lambda_c, d, \alpha',w_c)$ satisfies our requirement.
\end{itemize}
By the first order condition, $U_{\lambda, d, \alpha}(z)$ is decreasing in $\lambda>0$, and is increasing in $\alpha>0$, thus
  \begin{align*} 
    \sup_{z\in[1, 1+\alpha]} U_{\lambda, d, \alpha}(z) \leq \sup_{z \in [1, 1 + \alpha]} U_{\lambda_c, d, \alpha}(z) \leq \sup_{z\in [1, 1 + \alpha']} U_{\lambda_c, d, \alpha'}(z),
  \end{align*}
  where in the last inequality, we also use $[1, 1 + \alpha] \subseteq [1, 1 + \alpha']$.
  Combining this with \eqref{eq:sym-x-sup} gives
  \begin{align*}
    \sup_{x\in [0, +\infty]} H_{\lambda, d, \alpha, w}(x)=\sup_{z\in[1, 1+\alpha]} U_{\lambda, d, \alpha}(z)
    &\leq \sup_{z\in [1, 1 + \alpha']} U_{\lambda_c, d, \alpha'} (z)
     = \sup_{x\in [0, +\infty]} H_{\lambda_c, d, \alpha', w_c} (x),
  \end{align*}
  Finally, by \Cref{lem:contract-critical} and $\delta$-uniqueness of $(\lambda_c, d, \alpha')$, we have 
  \[
  \sup_{x\in [0, +\infty]} H_{\lambda_c, d, \alpha', w_c} (x) \leq 1 - \delta,
  \]
  which implies that $H_{\lambda, d, \alpha, w}(x)\leq 1 - \delta$ for all $x\ge 0$, for the original $(\lambda, d, \alpha)$ and any $w>0$.
  %
\end{proof}

\subsubsection{Contraction for the critical instances}
\label{sec:contract-critical}

It remains to prove \Cref{lem:contract-critical}, the contraction for the critical case.
In order to do so, we look into the behavior of $H(x)$ at the critical fixpoint $x_c$. 

Let $x_c=x_c(\delta,d,\alpha)$ be the critical $x$ defined in \Cref{thm:delta-unique}. Under critical condition \Cref{cond:critical}, such $x_c$ is also the unique fixpoint satisfying $x_c=F(x_c)$.

\begin{lemma} \label{lem:H(xc)}
  Under \Cref{cond:critical}, $H(x_c) = F'(x_c) = 1 - \delta$, where $x_c=x_c(\delta,d,\alpha)$ is defined in \Cref{thm:delta-unique}.
\end{lemma}
\begin{proof}
Under \Cref{cond:critical}, we have $x_c=F(x_c)$, which implies that 
\[
H(x_c)=\frac{\phi(F(x_c))}{\phi(x_c)}F'(x_c)=F'(x_c).
\]
Besides, we have $T_\delta(x_c) = 0$ guaranteed in \Cref{thm:delta-unique}, where $T_\delta(x)$ is defined in \eqref{eq:function-T-delta}, which gives
\begin{align*}
\alpha d w F(x_c)=\alpha d w x_c=(1 - \delta)(x_c + 1)(\alpha + (1 + x_c)^w).
\end{align*}
Applying this identity in the calculation of $F'(x)$ in \eqref{eq:F(x)-derivative}, gives
\begin{align*}
  F'(x_c)
&=  \frac{\alpha d w F(x_c)}{(x_c + 1)(\alpha + (1 + x_c)^w)} 
= 1-\delta. \qedhere
\end{align*}
\end{proof}

Next, define
\begin{align*}
  B_\delta(x) := w \log (x+1) \left(\alpha d \cdot \frac{x+1}{F(x) + 1} -(x+1)^{w+1}\right)+\delta  (x+1) \left(\alpha +(x+1)^w\right).
\end{align*}

\begin{lemma} \label{lem:dH-struct}
The derivative $H'(x)$ of $H(x)$ can be calculated by 
  \begin{align*}
    H'(x) =&c_1(x) \cdot ((1 - \delta) - H(x)) + c_2(x) \cdot B_\delta(x),
  \end{align*}
  for some positive-valued  functions $c_1(x) > 0$ and $c_2(x) > 0$ over $x > 0$.
\end{lemma}

\begin{lemma} \label{lem:sign-B}
    Under \Cref{cond:critical}, $B'_\delta(x) < 0$ at $x=x_c$, and
  \begin{enumerate}
  \item $B_\delta(x) > 0$, when $x < x_c$;
  \item $B_\delta(x) = 0$, when $x = x_c$;
  \item $B_\delta(x) < 0$, when $x > x_c$.
  \end{enumerate}
\end{lemma}

\Cref{lem:dH-struct} and \Cref{lem:sign-B} are proved by straightforward calculations, and they are postponed to \Cref{sec:dH}.

\begin{lemma} \label{cor:dH(xc)}
  Under \Cref{cond:critical}, it holds that $H'(x_c) = 0$ and $H''(x_c) < 0$.
\end{lemma}
\begin{proof}
By \Cref{lem:H(xc)}, we have $H(x_c) =1 - \delta$, and  by \Cref{lem:sign-B},  we have $B_\delta(x_c) = 0$. Then $H'(x_c)$ can be calculated according to \Cref{lem:dH-struct} as $H'(x_c)=c_1(x) \cdot ((1 - \delta) - H(x)) + c_2(x) \cdot B_\delta(x)=0$.
  
By \Cref{lem:dH-struct}, we have $c_1(x) > 0$ and $c_2(x) > 0$ such that
\begin{align*}
  H'(x) =&+c_1(x) \cdot ((1 - \delta) - H(x))+c_2(x)  B_\delta(x).
\end{align*}
Therefore
\begin{align*}
  H''(x) =&c_1'(x) \cdot ((1 - \delta) - H(x)) -c_1(x) \cdot H'(x) +c_2'(x) \cdot  B_\delta(x) +c_2(x) \cdot  B_\delta(x)',
\end{align*}
where, at $x_c$, the first three terms equal $0$ due to \Cref{lem:H(xc)}, \Cref{lem:sign-B}, and that $H'(x_c)=0$.
The last term is negative since $B_\delta(x_c)'<0$ by \Cref{lem:sign-B}.
\end{proof}

\begin{proof}[Proof of \Cref{lem:contract-critical}]
  By \Cref{cor:dH(xc)},  $H'(x_c) = 0$ and $H''(x_c) < 0$.
  By continuity, there exists $\epsilon > 0$ such that
  \begin{align*}
    \forall x \in [x_c - \epsilon, x_c):\quad H'(x) > 0 
    \qquad\text{and} \qquad  \forall x \in (x_c, x_c + \epsilon]:\quad H'(x) < 0.
  \end{align*}
  By \Cref{lem:H(xc)}, $H(x_c) = 1 - \delta$, which implies that
  \begin{align*}
  H(x)\begin{cases}
  <1-\delta & x \in [x_c - \epsilon, x_c)\\
  =1-\delta & x=x_c\\
  <1-\delta & x \in (x_c, x_c + \epsilon]
  \end{cases}.
  \end{align*}
  We will prove $H(x)< 1-\delta $ in two separate cases: $x < x_c$ and $x > x_c$.
  
 \bigskip 
  \noindent\textbf{Case.1:} $x > x_c$. For the sake of contradiction, assume that
   there exist $x > x_c$ such that $H(x) \geq 1 - \delta$. Let
  \begin{align*}
    r := \inf \{x > x_c \mid H(x) \geq 1 - \delta\}.
  \end{align*}
  Due to that $H(x) < 1 - \delta$ for $x \in (x_c, x_c + \epsilon]$ and continuity of $H(x)$, it holds that
  $H(r) = 1 - \delta$ and $H(x) < 1 - \delta$ for all $x \in (x_c, r)$, which implies $H'(r) \geq 0$.
  However, by \Cref{lem:dH-struct}, 
  \begin{align*}
    H'(r)
    &= c_1(r) \cdot ((1 - \delta) - H(r)) + c_2(r) \cdot B_\delta(r) 
    = c_2(r) \cdot B_\delta(r) < 0,
  \end{align*}
 where the inequality holds by the fact that $r > x_c$ and \Cref{lem:sign-B}.
 A contradiction!

 \bigskip 
 \noindent\textbf{Case.2:} $0 \leq x < x_c$. For the sake of contradiction, assume   there exist $0 \leq x < x_c$ such that $H(x) \geq 1 - \delta$. Let
  \begin{align*}
    \ell := \sup \{0 \leq x < x_c \mid H(x) \geq 1 - \delta\}.
  \end{align*}
  By the mean value theorem, there exists  $\zeta \in (\ell, x_c - \epsilon)$ such that
  \begin{align*}
    H'(\zeta) = \frac{H(x_c - \epsilon) - H(\ell)}{x_c - \epsilon - \ell} < 0,
  \end{align*}
  where the inequality holds by the fact that $H(x_c - \epsilon) < 1 - \delta$ and $H(\ell) = 1 - \delta$.
  Moreover, since $\zeta \in (\ell, x_c - \epsilon)$, we know that $H(\zeta) < 1 - \delta$.
  By \Cref{lem:dH-struct}, 
  \begin{align*}
    H'(\zeta) = c_1(\zeta) ((1 - \delta) - H(\zeta)) + c_2(\zeta) B_\delta(\zeta) \geq 0.
  \end{align*}
  where the inequality holds by that $H(\zeta) < 1 - \delta$, $\zeta < x_c$, and \Cref{lem:sign-B}.
  Also a contradiction!

 \bigskip 
  Altogether, we show that $H(x)\le 1-\delta$ for all $x\ge 0$.

  \bigskip
  The moreover part of \Cref{lem:contract-critical} comes directly from \Cref{lem:H(xc)}.
\end{proof}

\subsubsection{Derivatives of $H(x)$} \label{sec:dH}
It only remains to formally proves \Cref{lem:dH-struct} and \Cref{lem:sign-B}, which are regarding the first and second derivatives of $H(x)$ respectively.
\begin{proof}[Proof of \Cref{lem:dH-struct}]
For the choice of potential function $\phi(x) =\frac{1}{(x + 1) \log (x + 1)}$, we have 
\begin{align}
  H(x)= \frac{\phi(F(x))}{\phi(x)} F'(x)
  \label{eq:H} &=\frac{\alpha  d w F(x) \log (x+1)}{(F(x)+1) \log (F(x)+1) \left(\alpha +(x+1)^w\right)}.
\end{align}
The derivative $H'(x)$ is given by
\begin{align*}
  H'(x) = &-\frac{\alpha  d w F(x) \log (x+1) F'(x)}{(F(x)+1)^2 \log ^2(F(x)+1) \left(\alpha +(x+1)^w\right)} 
          +\frac{\alpha  d w \log (x+1) F'(x)}{(F(x)+1) \log (F(x)+1) \left(\alpha +(x+1)^w\right)} \\
          &-\frac{\alpha  d w F(x) \log (x+1) F'(x)}{(F(x)+1)^2 \log (F(x)+1) \left(\alpha +(x+1)^w\right)} 
          -\frac{\alpha  d w^2 F(x) (x+1)^{w-1} \log (x+1)}{(F(x)+1) \log (F(x)+1) \left(\alpha +(x+1)^w\right)^2} \\
          &+\frac{\alpha  d w F(x)}{(x+1) (F(x)+1) \log (F(x)+1) \left(\alpha +(x+1)^w\right)}
\end{align*}
Substituting  $F'(x)= \alpha d w (x + 1)^{-w-1} (1 + \alpha (x + 1)^{-w})^{-1} F(x)$, we have,
\begin{align}
  H'(x) =
  \nonumber &-\frac{\alpha ^2 d^2 w^2 F(x)^2 (x+1)^{-w-1} \log (x+1)}{(F(x)+1)^2 \log ^2(F(x)+1) \left(\alpha +(x+1)^w\right) \left(\alpha  (x+1)^{-w}+1\right)} \\
  \nonumber &+\frac{\alpha ^2 d^2 w^2 F(x) (x+1)^{-w-1} \log (x+1)}{(F(x)+1) \log (F(x)+1) \left(\alpha +(x+1)^w\right) \left(\alpha  (x+1)^{-w}+1\right)} \\
  \nonumber &-\frac{\alpha ^2 d^2 w^2 F(x)^2 (x+1)^{-w-1} \log (x+1)}{(F(x)+1)^2 \log (F(x)+1) \left(\alpha +(x+1)^w\right) \left(\alpha  (x+1)^{-w}+1\right)} \\
  \nonumber &-\frac{\alpha  d w^2 F(x) (x+1)^{w-1} \log (x+1)}{(F(x)+1) \log (F(x)+1) \left(\alpha +(x+1)^w\right)^2} \\
  \label{eq:dH} &+\frac{\alpha  d w F(x)}{(x+1) (F(x)+1) \log (F(x)+1) \left(\alpha +(x+1)^w\right)}.
\end{align}
Define $c_0(x) := (x+1) (F(x)+1)^2 \log ^2(F(x)+1) \left(\alpha +(x+1)^w\right)^2 / d w \alpha F(x)$. Clearly, $c_0(x)> 0$ for all $x> 0$.
And
\begin{align*}
  H'(x) \cdot c_0(x) 
  = &-\alpha  d w F(x) \log (x+1)+F(x) (x+1)^w \log (F(x)+1) \\
    &+(x+1)^w \log (F(x)+1) +\alpha  \log (F(x)+1)+\alpha  F(x) \log (F(x)+1) \\[5pt]
    &+\alpha  d w \log (x+1) \log (F(x)+1)-w (x+1)^w \log (x+1) \log (F(x)+1) \\
    &-w F(x) (x+1)^w \log (x+1) \log (F(x)+1) \\[5pt]
  = &+ (1 + F(x))\log(F(x) + 1)(\alpha + (1 + x)^w) \\
  &\hspace{2cm} \times \tp{1 - \frac{\alpha  d w F(x) \log (x+1)}{(F(x)+1) \log (F(x)+1) \left(\alpha +(x+1)^w\right)}} \\
    &+w \log (x+1) \log (F(x)+1) \left(\alpha  d - (F(x) + 1)(x+1)^w\right) \\[5pt]
  \overset{\eqref{eq:H}}{=} &+ (1 + F(x))\log(F(x) + 1)(\alpha + (1 + x)^w) \tp{1 - H(x)} \\
    &+w \log (x+1) \log (F(x)+1) \left(\alpha  d - (F(x) + 1)(x+1)^w\right). 
\end{align*}
We introduce the gap $\delta \in (0, 1)$ here, and have
\begin{align*}
  H'(x) \cdot c_0(x) 
  = &+ (1 + F(x))\log(F(x) + 1)(\alpha + (1 + x)^w) \tp{1 - \delta - H(x)} \\
    &+ \left(w \log (x+1) \left(\alpha  d-(F(x)+1) (x+1)^w\right)+\delta  (F(x)+1) \left(\alpha +(x+1)^w\right)\right) \\
    &\hspace{2cm} \times \log (F(x)+1)  \\
 = &+ (1 + F(x))\log(F(x) + 1)(\alpha + (1 + x)^w) \tp{(1 - \delta) - H(x)} \\
   &+ \log (F(x)+1) \cdot \frac{F(x) + 1}{x + 1} \cdot B_\delta(x)
\end{align*}
We can now finish the proof by defining
\begin{align*}
  c_1(x) &:= (1 + F(x))\log(F(x) + 1)(\alpha + (1 + x)^w) / c_0(x) \\
 c_2(x) &:= \log (F(x)+1) \cdot \frac{F(x) + 1}{x + 1} / c_0(x).
\end{align*}
It is easy to verify that $c_1(x) > 0$ and $c_2(x) > 0$ for $x > 0$.
\end{proof}

\begin{proof}[Proof of \Cref{lem:sign-B}]
Consider the following function,
\begin{align*}
  C(x) &:= w \log (x+1) \left(\alpha d \cdot \frac{x+1}{x_c + 1} -(x+1)^{w+1}\right)+\delta  (x+1) \left(\alpha +(x+1)^w\right).
\end{align*}
Note that when $x < x_c$, we have $\frac{x+1}{x_c + 1} < \frac{x+1}{F(x) + 1}$ and vice versa.
This also implies that $B_\delta'(x_c) \leq C'(x_c)$.
Hence, it is sufficient to show that
\begin{align} \label{eq:sign-C-tar}
  \forall x < x_c, C(x) > 0,  \quad \forall x > x_c, C(x) < 0, \quad \text{and} \quad C'(x_c) < 0.
\end{align}
Under  \Cref{cond:critical}, by \Cref{thm:delta-unique}, for $x = x_c$, we have
\begin{align*}
  M_\delta = w \log(1 + x) (\alpha d - (1 + x)^{w+1}) + \delta(x + 1)(\alpha + (x + 1)^w) = 0,
\end{align*}
which implies $C(x_c) = 0$ by noticing that $F(x_c) = x_c$.
Define $G(x) := \frac{x+1}{x_c + 1}$. 
{\allowdisplaybreaks
\begin{align*}
  C(0) =& (1 + \alpha) \delta > 0;\\[6pt]
  C'(x) =& \alpha  \delta +\alpha  d w \log (x+1) G'(x)+\frac{\alpha  d w G(x)}{x+1} +\delta  w (x+1)^w \\
         &\hspace{0.7cm} +\delta  (x+1)^w -w (x+1)^w-w (w+1) (x+1)^w \log (x+1) \\
  =& \alpha  \delta +\alpha  d w (\log (x+1)+1) G'(x) \\
   & \hspace{0.7cm} -\left((x+1)^w (-\delta  (w+1)+(w+1) w \log (x+1)+w)\right); \\[6pt]
  C'(0) =& \frac{\alpha  d w}{x_c+1}+ (\alpha + 1)\delta - (1 - \delta) w; \\[6pt]
  C''(x) =& \frac{w}{1+x} \cdot \left(\alpha  d (x+1) (\log (x+1)+1) G''(x) +\alpha  d G'(x) \right. \\
         & \hspace{0.7cm} \left. +(x+1)^w (\delta +(\delta -2) w-w (w+1) \log (x+1)-1)\right) \\
  =& \frac{w}{1 + x} \left(\alpha  d G'(x) 
  +(x+1)^w (\delta +(\delta -2) w-w (w+1) \log (x+1)-1) \right); \\[6pt]
  C''(0) =& \frac{\alpha  d w}{x_c+1}- (2 - \delta) w^2 - (1 - \delta) w; \\[6pt]
  \tp{\frac{1+x}{w} C''(x)}' =& -w (x+1)^{w-1} ((2-\delta) +(3 - \delta) w + w (w+1) \log (x+1))
   < 0. 
\end{align*}
}
In above, we use the fact that $G(x) = (1 + x) G'(x)$ and $G''(x) = 0$.

Note that $\frac{1+x}{w} C''(x)$ has the same sign as $C''(x)$.
Hence, $\tp{\frac{1+x}{w} C''(x)}' < 0$ means that
\begin{align} \label{eq:d3-C}
  C''(x_0) < 0 \text{ (or $\leq 0$) for some $x_0 \geq 0$} \Longrightarrow C''(x) < 0 \text{ (or $\leq 0$) for all $x \geq x_0$}.
\end{align}
%
Fortunately, we also have $C''(0) < C'(0)$.
The sign of $C'(0)$ falls into $2$ cases.
\begin{itemize}
%
%
%
%
\item
\textbf{Case.1:} $C'(0) < 0$. Since $C'(0) < 0$, by $C''(0) < C'(0) < 0$ and \eqref{eq:d3-C}, it holds that, $C''(x) < 0$ for all $x \geq 0$, which implies that $C'(0) < 0$ for all $x \geq 0$.
Since $C(0) > 0$ and $C(x_c) = 0$ (because $M_\delta(x_c) = 0$), we have \eqref{eq:sign-C-tar}.
\item
\textbf{Case.2:} $C'(0) \geq 0$. Since $C(x_c) = 0$ and $C(0) > 0$, by the mean value theorem, there is $x_1 \in (0, x_c)$ such that $C'(x_1) < 0$.
Since $C'(0) \geq 0$ and $C'(x_1) < 0$, by the mean value theorem, there is $x_2 \in (0, x_1)$ such that $C''(x_2) < 0$.
This, by \eqref{eq:d3-C}, means $C''(x) < 0$ for all $x \geq x_2$, which implies that $C'(x) < 0$ for all $x \geq x_1$.
In particular, we have $C'(x_c) < 0$.

By the intermediate value theorem there is $x_3 \in [0, x_1)$ such that $C'(x_3) = 0$.
Moreover, if there are many of them, let $x_3$ be the largest one among them.

If $C''(0) < 0$, $C(x)$ is monotonically increasing in $[0, x_3)$ and monotonically decreasing in $[x_3, +\infty)$, which implies \eqref{eq:sign-C-tar}.

If $C''(0) \geq 0$, by the intermediate value theorem, there is $x_4 \in [0, x_2)$ such that $C''(x_4) = 0$. 
Then, by \eqref{eq:d3-C}, $C''(x) \geq 0$ for $x \in [0, x_4)$ and $C''(x) \leq 0$ for $x \in [x_4, +\infty)$.
Since $C'(0) \geq 0$, it holds that $x_3 \geq x_4$.
Moreover, it holds that $C'(x) \geq 0$ when $x \in [0, x_3)$ and $C'(x) \leq 0$ when $x \in [x_3, +\infty)$.
This implies \eqref{eq:sign-C-tar}.
\end{itemize}
Altogether, we have \eqref{eq:sign-C-tar}.
This proves \Cref{lem:sign-B}.
\end{proof}

\subsection{Implications to $\delta$-uniqueness}
\label{sec:delta-unique-eq-7}
We now detour from the analysis of spectral independence, and use the contraction property that we have established to characterize the $\delta$-uniqueness condition. 
Specifically, we 
complete the proof of  \Cref{thm:delta-unique-eq},  showing that $\lambda= \alpha \leq (1 - \delta)\lambda_c(\Delta)$ implies $\Theta(\delta)$-uniqueness for all $\delta\in(0,1)$.

  As explained just before \Cref{thm:solve-sys-eq}, by \Cref{thm:delta-unique-implicit} and \Cref{thm:resolve-lambda-2-c}, all we need to do it to resolve the following system of $(x, \alpha, w)$ as defined in \eqref{eq:delta-sys-eq} for the given parameters $d \geq 2, \delta \in (0, 1)$.
\begin{align} \label{eq:delta-sys-eq-restate}
  \begin{cases}
    T_\delta := (1 - \delta)(x + 1)(\alpha + (1 + x)^w) - \alpha d w x = 0, \\
    M_\delta := w \log(1 + x) (\alpha d - (1 + x)^{w+1}) + \delta(x + 1)(\alpha + (x + 1)^w) = 0, \\
    G       := x(1 + \alpha(1 + x)^{-w})^d - \alpha = 0.
  \end{cases}
\end{align}
In \Cref{sec:delta-unique}, we have already tried to resolve \eqref{eq:delta-sys-eq-restate} exactly.
As stated in \Cref{thm:solve-sys-eq}, we are only able to do so for sufficiently small $\delta$.

We begin by noticing that \Cref{lem:contract-critical} builds a precise correspondence between $\delta$-uniqueness and the contraction.
Thus, we can leverage such correspondence to extend the result in \Cref{thm:solve-sys-eq} for all $\delta \in (0, 1)$.
First, we have the following observation.
\begin{observation} \label{obs:sol-exists-delta-sys-eq}
  For every $d \geq 2$ and $\delta \in (0, 1)$, the system in \eqref{eq:delta-sys-eq-restate} has a solution $(x_c, \alpha_c, w_c)$ with
  \[\textstyle \alpha_c \in \tp{\frac{1-\delta}{d} \e^{1 + \frac{1-\delta}{d}}, \lambda_c(d+1)}. \]
\end{observation}
\begin{proof}
  We consider the system of equations defined by $T_\delta = M_\delta = 0$ as in \Cref{thm:resolve-lambda-2-c} and treat $\alpha$ as a parameter.
  Note that the constraint $G = 0$ in \eqref{eq:delta-sys-eq-restate} actually says that $\lambda^\delta_{2,c} = \alpha$.
  We recall that $\lambda^\delta_{2,c}$ is the critical threshold for $\lambda$, which is defined in \Cref{def:x-lambda-delta}.
  
  By \Cref{thm:resolve-lambda-2-c}, $\lambda^\delta_{2,c}$ is an implicit function of $\alpha$.
  When $\alpha = \frac{1 - \delta}{d}\e^{1 + \frac{1-\delta}{d}}$, it holds that $\lambda^\delta_{2,c}(\alpha) = 0$.
  When $\alpha = \lambda_c(d+1)$, by \Cref{thm:solve-sys-eq}, then $\lambda^0_{2,c}(\alpha) = \alpha$.
  Then, by \Cref{lem:lambda-c-monotone-delta}, $\lambda^\delta_{2,c}(\alpha) > \lambda^0_{2,c}(\alpha) = \alpha$.
  
  By continuity, there is $\alpha_c \in (\frac{1 - \delta}{d}\e^{1 + \frac{1-\delta}{d}}, \lambda_c(d+1))$ such that $\alpha_c = \lambda^\delta_{2,c}(\alpha_c)$.
  Fix $\alpha = \alpha_c$, let $(x_c, w_c)$ be the solution of $T_\delta = M_\delta = 0$ in \Cref{thm:resolve-lambda-2-c}.
  Then $(x_c, \alpha_c, w_c)$ is a solution of \eqref{eq:delta-sys-eq-restate}.
\end{proof}

Recall that given $\Delta = d + 1 \geq 3$, \Cref{thm:solve-sys-eq} only produce solutions of \eqref{eq:delta-sys-eq-restate} in a neighborhood of $\delta = 0$.
%
However, by leveraging \Cref{obs:sol-exists-delta-sys-eq}, \Cref{thm:solve-sys-eq}, and the analysis of the correlation decay at current section,
we are able to extend the result in \Cref{thm:solve-sys-eq} for all $\delta \in (0, 1)$.

\begin{lemma} \label{lem:delta-unique-monotone-eq}
  Given $d \geq 2$ and $\delta \in [0, 1)$. Let $\lambda_\star > 0$ be such that $(\lambda_\star, d)$ is $\delta$-unique. Then it holds that for all $0 < \lambda \leq \lambda_\star$, $(\lambda, d)$ is $\delta$-unique.
\end{lemma}

\begin{lemma} \label{lem:alpha-c-eq-lb}
  Let $d \geq 2$ and $\delta \in [0, 1)$.
  Let $(x_c, \alpha_c, w_c)$ be a solution of \eqref{eq:delta-sys-eq-restate} such that
  \[\textstyle \alpha_c \in \tp{\frac{1-\delta}{d} \e^{1 + \frac{1-\delta}{d}}, \lambda_c(d+1)}.\]
  Then it holds that
  \[ \alpha_c \geq (1 - 10\delta) \lambda_c(d+1). \]
\end{lemma}
\Cref{lem:delta-unique-monotone-eq} indicates that when $\delta$-uniqueness holds for larger $\lambda$, then it will automatically hold for smaller $\lambda$.
Then, \Cref{lem:alpha-c-eq-lb} extends \Cref{thm:solve-sys-eq} for all $\delta \in (0, 1)$.
The proof of \Cref{lem:delta-unique-monotone-eq} and \Cref{lem:alpha-c-eq-lb} relies on the correlation decay analysis in previous subsections, and will be given in \Cref{sec:monotone-delta-unique-eq} and \Cref{sec:alpha-c-eq-lb}, respectively.

\Cref{thm:delta-unique-eq} is then proved by combining \Cref{thm:delta-unique-implicit}, \Cref{thm:resolve-lambda-2-c},  \Cref{lem:delta-unique-monotone-eq}, and \Cref{lem:alpha-c-eq-lb}.
\begin{proof}[Proof of \Cref{thm:delta-unique-eq}]
Due to \Cref{obs:sol-exists-delta-sys-eq} and \Cref{lem:alpha-c-eq-lb}, \eqref{eq:delta-sys-eq-restate} has a solution $(x_c, \alpha_c, w_c)$ with $\alpha_c = \lambda^\delta_{2,c}$.
And we have $\alpha_c \in (\frac{1-\delta}{d}\e^{1+\frac{1-\delta}{d}}, \lambda_c(d+1))$.
Hence, $\lambda_\star := \alpha_c \geq \lambda^\delta_{2,c}$, which, by \Cref{thm:delta-unique-implicit}, implies that $(\lambda_\star, d)$ is $\delta$-unique.
By \Cref{lem:alpha-c-eq-lb}, it holds that $\lambda_\star \geq (1 - 10 \delta)\lambda_c(\Delta)$.
So, for $\lambda \leq (1 - 10\delta)\lambda_c(\Delta) \leq \lambda_\star$, by \Cref{lem:delta-unique-monotone-eq}, $(\lambda, d)$ is also $\delta$-unique.
Finally, we get \Cref{thm:delta-unique-eq} by substituting $\delta$ to $\delta/10$.
%
%
\end{proof}

\subsubsection{Consequence of contraction: monotonicity of $\delta$-uniqueness for $\lambda = \alpha$}
\label{sec:monotone-delta-unique-eq}
We now prove \Cref{lem:delta-unique-monotone-eq}.
%
%
Let $\^T \subseteq \^R^2$ be a triangular region defined by
\begin{align} \label{eq:U-triangle}
  \^T := \{(\lambda, z) \mid \lambda \geq 0 \text{ and } 1 \leq z \leq \lambda\}.
\end{align}
Given $d \geq 2$, consider the function $U: \^T  \to \^R$ defined in  \Cref{lem:sym-s} by letting $\lambda = \alpha$ as
\begin{align}\label{eq:U-lambda-eq-alpha}
  U(\lambda, z) := \frac{\lambda z^{-d}}{(1 + \lambda z^{-d}) \log (1 + \lambda z^{-d})} \cdot d \cdot \frac{z - 1}{z} \log \tp{\frac{\lambda}{z - 1}},
\end{align}
where we consider $U$ as function of $\lambda$ and $z$.
Note that by definition, $U(\lambda, z) = U_{\lambda, d, \lambda}(z)$.
For convenience, we use the notation $\partial_1 U := \partial U/ \partial \lambda$ and $\partial_2 U := \partial U / \partial z$.

\begin{fact} \label{fact:trivial-boundaries}
  For $d \geq 2$, we have
  \begin{align*}
    \lim_{(\lambda, z) \to (0, 0)} U(\lambda, z) = 0.
  \end{align*}
  Moreover, for a fixed $\lambda > 0$, it holds that 
  \begin{align*}
    \lim_{z \to 1} U(\lambda, z) = \lim_{z \to 1 + \lambda} U(\lambda, z) = 0.
  \end{align*}
\end{fact}

\begin{lemma} \label{lem:monotone-at-extreme-point}
  Let $d \geq 2$.
  For every $\lambda > 0$ and every $z \in (1, 1+\lambda)$, we have
  \begin{align*}
    \partial_2 U(\lambda, z) = 0 \quad \text{implies} \quad \partial_1 U(\lambda, z) > 0.
  \end{align*}
\end{lemma}
The proof of \Cref{lem:monotone-at-extreme-point} is postponed to the end of \Cref{sec:monotone-delta-unique-eq}.

\begin{lemma} \label{lem:monotone-U}
  Let $d \geq 2$ and $\lambda_\star > 0$. 
  Define the region $\^T_\star := \^T \cap \{(\lambda, z) \mid \lambda \leq \lambda_\star\}$.
  Then, it holds that, 
  \begin{align*} 
    \sup_{(\lambda, z) \in \^T_\star} U(\lambda, z) = \sup_{z \in [1, 1 + \lambda_\star]} U(\lambda_\star, z).
  \end{align*}
\end{lemma}
\begin{proof}
  For every interior point $(\lambda, z) \in \^T_\star$, by \Cref{lem:monotone-at-extreme-point}, it holds that
  \begin{align*}
    \min\{\partial_1 U(\lambda, z), \partial_2(\lambda, z)\} > 0.
  \end{align*}
  Hence, there is a neighbor of $(\lambda', z')$ of $(\lambda, z)$ such that
  \begin{align*}
    U(\lambda, z) < U(\lambda', z').
  \end{align*}
  This indicates that the supremum of $U$ is achieved at the boundary of $\^T_\star$.
  According to \Cref{fact:trivial-boundaries}, 
  \begin{align*}
    \sup_{(\lambda, z) \in \^T_\star} U(\lambda, z) &= \sup_{z \in [1, 1 + \lambda_\star]} U(\lambda_\star, z). \qedhere
  \end{align*}
\end{proof}

Now, we are ready to prove \Cref{lem:delta-unique-monotone-eq}.

\begin{proof} [Proof of \Cref{lem:delta-unique-monotone-eq}]
  Since $(\lambda_\star, d)$ is $\delta$-unique, pick a $w > 0$, by \Cref{thm:contract-sym}, it holds that 
  \begin{align} \label{eq:U-lambda-star-z}
    \sup_{z \in [1, 1 + \lambda_\star]} U(\lambda_\star, z) = \sup_{z\in[1, 1+\lambda_\star]} U_{\lambda_\star, d, \lambda_\star}(z) \overset{\eqref{eq:sym-x-sup}}{=} \sup_{x \in [0, +\infty]} H_{\lambda_\star, d, \lambda_\star, w}(x) \leq 1 - \delta,
  \end{align}
  where the last inequality holds by \Cref{thm:contract-sym}.
  Then, by \Cref{lem:monotone-U}, it holds that
  \begin{align} \label{eq:U-lambda-z}
    \sup_{z \in [1, 1 + \lambda]} U_{\lambda, d, \lambda}(z) = \sup_{z \in [1, 1 + \lambda]} U(\lambda, z) \leq \sup_{(\lambda, z) \in \^T_\star} U(\lambda, z) \leq \sup_{z \in [1, 1 + \lambda_\star]} U(\lambda_\star, z).
  \end{align}
  Combining \eqref{eq:U-lambda-star-z} and \eqref{eq:U-lambda-z}, it holds that $\sup_{z \in [1, 1+\lambda]} U_{\lambda, d, \lambda}(z) \leq 1 - \delta$.
  
  Now, we prove that $(\lambda, d)$ is $\delta$-unique by contradiction.
  Suppose $(\lambda, d)$ is not $\delta$-unique, then by \Cref{thm:delta-unique}, it holds that $\lambda > \frac{1-\delta}{d} \e^{1 + \frac{1-\delta}{d}}$ and $\lambda < \lambda_c(\delta, d, \alpha = \lambda)$, where $\lambda_c$ is the critical threshold defined in \Cref{thm:delta-unique}.
  Hence, it holds that
  \begin{align} \label{eq:U-lambda-c-less-than}
    \sup_{z \in [1, 1 + \lambda]} U_{\lambda_c, d, \lambda}(z) < \sup_{z \in [1, 1+\lambda]} U_{\lambda, d, \lambda}(z) \leq 1 - \delta,
  \end{align}
  where the explanation of the first inequality is stated as follows.
  Note that for a fixed $z \in (1, 1 + \alpha)$, the function $U_{\lambda, d, \alpha}$ is strictly decreasing in $\lambda$.
  Moreover, since $\lim_{z\to 1} U_{\lambda, d, \alpha}(z) = \lim_{z\to 1+\alpha} U_{\lambda, d, \alpha}(z) = 0$, once $\alpha, \lambda, d > 0$, it holds that the supremum of $U_{\lambda, d, \alpha}(z)$ in the interval $[1, 1 + \alpha]$ is not obtained at $z = 1$ or $z = 1+\alpha$.

  On the other hand, let $w_c = w_c(\delta, d, \alpha = \lambda)$ be the critical threshold in \Cref{thm:delta-unique}.
  According to \Cref{lem:contract-critical}, it holds that
  \begin{align} \label{eq:U-lambda-c-eq}
    1 - \delta = \sup_{x \in [0, +\infty]} H_{\lambda_c, d, \lambda, w_c} \overset{\eqref{eq:sym-x-sup}}{=} \sup_{z \in [1, 1 + \lambda]} U_{\lambda_c, d, \lambda}(z).
  \end{align}
  Combining \eqref{eq:U-lambda-c-less-than} and \eqref{eq:U-lambda-c-eq}, we get $1 - \delta < 1 - \delta$, a contradiction!
  Hence $(\lambda, d)$ is $\delta$-unique.
\end{proof}

\begin{proof}[Proof of \Cref{lem:monotone-at-extreme-point}]
  By $\partial_2 U(\lambda, z) = 0$, it holds that
  \begin{align*}
    \frac{d \lambda  \left(\log \left(\frac{\lambda }{z-1}\right) \left(\left(-d z^{d+1}+(d+1) z^d+\lambda \right) \log \left(\lambda  z^{-d}+1\right)+d \lambda  (z-1)\right)-z \left(z^d+\lambda \right) \log \left(\lambda  z^{-d}+1\right)\right)}{z^2 \left(z^d+\lambda \right)^2 \log ^2\left(\lambda  z^{-d}+1\right)} &= 0,
  \end{align*}
  which implies that
  \begin{align}
    z \left(z^d+\lambda \right)
    \nonumber & \log \left(\lambda  z^{-d}+1\right) \\
     \label{eq:U-z} &= \log \left(\frac{\lambda }{z-1}\right) \left(\left(-d z^{d+1}+(d+1) z^d+\lambda \right) \log \left(\lambda  z^{-d}+1\right)+d \lambda  (z-1)\right)
  \end{align}

  Now, note that
  \begin{align} 
    \partial_1 U(\lambda, z)
    \nonumber &= \frac{d (z-1) \left(\left(z^d+\lambda \right) \log \left(\lambda  z^{-d}+1\right)+\log \left(\frac{\lambda }{z-1}\right) \left(z^d \log \left(\lambda  z^{-d}+1\right)-\lambda \right)\right)}{z \left(z^d+\lambda \right)^2 \log ^2\left(\lambda  z^{-d}+1\right)} \\
   \label{eq:U-lambda} &= c_1(\lambda, z) \cdot \left(\left(z^d+\lambda \right) \log \left(\lambda  z^{-d}+1\right)+\log \left(\frac{\lambda }{z-1}\right) \left(z^d \log \left(\lambda  z^{-d}+1\right)-\lambda \right)\right),
  \end{align}
  where $c_1(\lambda, z) := d(z-1) / (z \left(z^d+\lambda \right)^2 \log ^2\left(\lambda  z^{-d}+1\right))$.
  Plugging \eqref{eq:U-z} into \eqref{eq:U-lambda}, it holds that
  \begin{align*}
    \partial_1 U(\lambda, z)
    &= c_2(\lambda, z) \cdot \tp{\left(-(d-1) z^{d+1}+(d+1) z^d+\lambda \right) \log \left(\lambda  z^{-d}+1\right)+\lambda  (d (z-1)-z)},
  \end{align*}
  where $c_2(\lambda, z) := d(z-1)\log\tp{\frac{\lambda}{z-1}} / (z^2 \left(z^d+\lambda \right)^2 \log ^2\left(\lambda  z^{-d}+1\right))$.
  By the fact that $\frac{x}{1+x} \leq \log\tp{1 + x} \leq x$ for all $x > -1$, we have
  \begin{align*}
    \partial_1 U(\lambda, z) \geq c_2(\lambda, z) \cdot \min\left\{\lambda ^2 z^{-d}+\lambda, \frac{\lambda  \left(z^d+(d-1) \lambda  (z-1)\right)}{z^d+\lambda }\right\} > 0,
  \end{align*}
  where we use the fact that $z \in (1, 1 + \lambda)$ and $d \geq 2$.
\end{proof}

\subsubsection{Consequence of contraction: $\frac{\delta}{10}$-uniqueness for some $\lambda = \alpha \geq (1 - \delta) \lambda_c(\Delta)$}
\label{sec:alpha-c-eq-lb}
We now prove \Cref{lem:alpha-c-eq-lb}.
%
%
We have the following results regarding the $U(\lambda, z)$ defined in~\eqref{eq:U-lambda-eq-alpha}.
\begin{lemma} \label{lem:dU-lambda-lb}
  For any $d \geq 2$, $\lambda > 0$ and $z = 1 + \epsilon \lambda$ such that $\epsilon \in [3/20, 1]$, it holds that 
  \begin{align*}
    \left. \frac{\partial \log U(\e^t, z)}{\partial t} \right\vert_{\e^t = \lambda} \geq 1/10.
  \end{align*}
\end{lemma}

\begin{lemma} \label{lem:max-U-z-lb}
  Fix integer $d \geq 2$ and real $0 < \lambda \leq \lambda_c(d+1)$, the function $z \mapsto U(\lambda, z)$ achieves its maximum in the interval $[1, 1 + \lambda]$ at the point $z_\star = 1 + \epsilon_\star \lambda$ such that $\epsilon_\star \in [3/20, 1]$.
\end{lemma}

\begin{proof}[Proof of \Cref{lem:alpha-c-eq-lb}]
  Since $(x_c, \alpha_c, w_c)$ is a solution of \eqref{eq:delta-sys-eq-restate}.
  It is straightforward to verify that $(\lambda_\star := \lambda_c = \alpha_c, d, \alpha := \alpha_c, w := w_c)$ satisfies \Cref{cond:critical}.
  Hence by \Cref{lem:contract-critical} and \eqref{eq:sym-x-sup}, we know that $\sup_{z \in [1, 1 + \lambda_\star]} U(\lambda_\star, z) = 1 - \delta$.
  This means $(d, \lambda_\star)$ is ``exactly'' $\delta$-unique.
  By \Cref{lem:max-U-z-lb}, we know that the supremum of $U$ is achieved at $z_\star = 1 + \epsilon_\star \lambda_\star$ where $\epsilon_\star \in [3/20, 1]$.
  This also implies $U(\lambda_\star, z_\star) = 1 - \delta$.
  Then, by an integration, it holds that
  \begin{align} 
    \nonumber \log \frac{U(\lambda_c, z_\star)}{U(\lambda_\star, z_\star)}
    &= \int_{\log \lambda_\star}^{\log \lambda_c} \frac{\partial \log U(\e^t, z_\star)}{\partial t} \-d t \\
    \label{eq:lambda-c-star-dist} (\text{by \Cref{lem:dU-lambda-lb}}) &\geq \frac{1}{10} \log \frac{\lambda_c}{\lambda_\star}.
  \end{align}
  On the other hand, we know that $(\lambda_c, d)$ is $0$-unique by \Cref{thm:delta-unique-implicit}, \Cref{thm:resolve-lambda-2-c}, and \Cref{thm:solve-sys-eq}.
  Hence, by \Cref{lem:contract-critical}, we know that $U(\lambda_c, z_\star) \leq 1$.
  Together with \eqref{eq:lambda-c-star-dist}, we have
  \begin{align*}
    \frac{1}{10} \log \frac{\lambda_c}{\lambda_\star} \leq \log \frac{U(\lambda_c, z_\star)}{U(\lambda_\star, z_\star)} \leq \log \frac{1}{1 - \delta}.
  \end{align*}
  Therefore $\frac{1}{10} \log \frac{\lambda_\star}{\lambda_c} \geq \log (1 - \delta)$ and $\alpha_c = \lambda_\star \geq (1 - \delta)^{10} \lambda_c \geq (1 - 10 \cdot \delta) \lambda_c$.
  We note that the last inequality holds by the Bernoulli inequality.
  %
\end{proof}

Now, it remains to prove \Cref{lem:dU-lambda-lb} and \Cref{lem:max-U-z-lb}.
\begin{proof}[Proof of \Cref{lem:dU-lambda-lb}]
  First, note that
  \begin{align*}
    \left. \frac{\partial \log U(\e^t, z)}{\partial t} \right\vert_{\e^t = \lambda}
    &= \frac{z^d-\frac{\lambda }{\log \left(\lambda  z^{-d}+1\right)}}{z^d+\lambda }+\frac{1}{\log \left(\frac{\lambda }{z-1}\right)} \\
    &\geq \frac{1}{\log \left(\frac{\lambda }{z-1}\right)}-\frac{\lambda }{2 \left(z^d+\lambda \right)},
  \end{align*}
  where in the last inequality, we use the fact that $\log(1 + x) \geq \frac{2 x}{2 + x}$ for $x \geq 0$ (here, $x = \lambda z^{-d}$).
  Hence let $z = 1 + \epsilon \lambda$, it holds that
  \begin{align}
    \left. \frac{\partial \log U(\e^t, 1 + \epsilon \lambda)}{\partial t} \right\vert_{\e^t = \lambda}
    \nonumber &\geq \frac{1}{\log (1/\epsilon)}-\frac{\lambda }{2 \left(\left(\epsilon \lambda +1\right)^d+\lambda \right)} \\
    \nonumber &\overset{(\star)}{\geq} \frac{1}{\log (1/\epsilon)} - \frac{\lambda}{2(1 + \epsilon d \lambda + \lambda)} \\
    \label{eq:dU-lb} &\geq \frac{1}{\log(1/\epsilon)} - \frac{1}{2(\epsilon d + 1)},
  \end{align}
  where in $(\star)$, we use the Bernoulli inequality which states that $(1 + x)^d \geq 1 + d x$ for $x \geq -1$, $d \geq 1$.
  By \eqref{eq:dU-lb} and $d \geq 2$, it holds that when $\epsilon \geq 3/20$, we have
  \begin{align*}
    \left. \frac{\partial \log U(\e^t, 1 + \epsilon \lambda)}{\partial t} \right\vert_{\e^t = \lambda} \geq 1/10. & \qedhere
  \end{align*}
\end{proof}

\begin{proof}[Proof of \Cref{lem:max-U-z-lb}]
  To find the $z_\star$ that achieves the maximum of $U$, we calculate $\frac{U(\lambda, z)}{\partial z}$:
  \begin{align*}
    &\frac{\partial U(\lambda, z)}{\partial z} = \\
    &\frac{d \lambda  \left(\log \left(\frac{\lambda }{z-1}\right) \left(\left(-d z^{d+1}+(d+1) z^d+\lambda \right) \log \left(\lambda  z^{-d}+1\right)+d \lambda  (z-1)\right)-z \left(z^d+\lambda \right) \log \left(\lambda  z^{-d}+1\right)\right)}{z^2 \left(z^d+\lambda \right)^2 \log ^2\left(\lambda  z^{-d}+1\right)}.
  \end{align*}
  Note that the sign of $\frac{\partial U(\lambda, z)}{\partial z}$ is determined by 
  \begin{align*}
    S(\lambda, z)
    &:= \textstyle \log \left(\frac{\lambda }{z-1}\right) \left(\left(-d z^{d+1}+(d+1) z^d+\lambda \right) \log \left(\lambda  z^{-d}+1\right)+d \lambda  (z-1)\right) \\
    &\hspace{2cm} \textstyle -z \left(z^d+\lambda \right) \log \left(\lambda  z^{-d}+1\right).
  \end{align*}
  We claim the following fact to be hold when $d \geq 2$, $\lambda \leq \lambda_c(d+1)$ and $z \in [1, 1 + \lambda]$:
  \begin{enumerate}
  \item \label{item:dS(z)-le-0}
    $\frac{\partial S(\lambda, z)}{\partial z} \leq 0$;
  \item \label{item:S-lambda-z-star-ge-0}
    and $S(\lambda, z) \geq 0$ for $z = 1 + \frac{3}{20} \lambda$.
  \end{enumerate}
  Fix $\lambda > 0$, if the equation $S(\lambda, z) = 0$ has no positive solution, by \Cref{item:S-lambda-z-star-ge-0}, it holds that $S(\lambda, z) \geq 0$ for any $z$.
  Hence the maximum of $U(\lambda, z)$ in the interval $z \in [1, 1 + \lambda]$ is achieved at $z_\star = 1 + \lambda$.
  Otherwise, by \Cref{item:dS(z)-le-0}, the equation $S(\lambda, z) = 0$ has a unique positive solution $z_\circ$.
  By \Cref{item:S-lambda-z-star-ge-0}, we know that $z_\circ \geq 1 + \frac{3}{20} \lambda$, which indicates that the maximum of $U(\lambda, z)$ in the interval $z \in [1, 1 + \lambda]$ is achieved at $z_\star = \min\{z_\circ, 1 + \lambda\}$.
  In both case, we have $z_\star = 1 + \epsilon_\star \lambda$ such that $\epsilon_\star \in [\frac{3}{20}, 1]$.

  Now, we are only left to prove \Cref{item:dS(z)-le-0} and \Cref{item:S-lambda-z-star-ge-0}.
  
  \paragraph{Proof of \Cref{item:S-lambda-z-star-ge-0}}
  Note that if
  \begin{align*}
    \textstyle
    \log \left(\frac{\lambda }{z-1}\right) \left(-d z^{d+1}+(d+1) z^d+\lambda \right) \log \left(\lambda  z^{-d}+1\right) -z \left(z^d+\lambda \right) \log \left(\lambda  z^{-d}+1\right) \geq 0,
  \end{align*}
  then $S(\lambda, z) \geq \log\tp{\frac{\lambda}{z-1}} \cdot d \lambda (z - 1) \geq 0$ and we are done.
  
  Otherwise, we use the fact that $\log(1 + x) \leq x$ for $x > -1$ (here, $x = \lambda z^{-d}$), it holds that
  \begin{align*}
    S(\lambda, z) \geq \lambda  z^{-d} \left(z^d+\lambda \right) \left(\log \left(\frac{\lambda }{z-1}\right) - z\right).
  \end{align*}
  Since $d \geq 2$, we have $\lambda \leq \lambda_c(3) = 4$, which implies that $S(\lambda, 1 + \frac{3}{20}\lambda) > 0$.
  \paragraph{Proof of \Cref{item:dS(z)-le-0}}
  By calculation, it holds that
  \begin{align} 
    \nonumber &(z-1) z \left(z^d+\lambda \right) \cdot \frac{\partial S(\lambda, z)}{\partial z} \\
    \nonumber =& -d (z-1)^2 \log \left(\frac{\lambda }{z-1}\right) \left((d+1) z^d \left(z^d+\lambda \right) \log \left(\lambda  z^{-d}+1\right)-\lambda  \left((d+1) z^d+\lambda \right)\right) \\
    \nonumber & \hspace{2cm} -z^2 \left(z^d+\lambda \right)^2 \log \left(\lambda  z^{-d}+1\right) \\ 
    \label{eq:dS} \overset{(\star)}{\leq}& \frac{\lambda  \left(-d \lambda  (z-1)^2 \left((d-1) z^d-\lambda \right) \log \left(\frac{\lambda }{z-1}\right)-2 z^2 \left(z^d+\lambda \right)^2\right)}{2 z^d+\lambda },
  \end{align}
  where $(\star$) holds by the fact that $\log(1 + x) \geq \frac{2x}{2+x}$ for $x \geq 0$ (here, $x = \lambda z^{-d}$).

  When $d \geq 3$, it holds that $\lambda \leq \lambda_c(4) = \frac{27}{16} < 2$.
  Hence $(d - 1)z^d - \lambda \geq 2 - \frac{27}{16} > 0$, which, by \eqref{eq:dS}, implies $\frac{\partial S(\lambda, z)}{\partial z} \leq 0$.

  When $d = 2$, \eqref{eq:dS} simplifies to 
  \begin{align*}
    (z-1) z \left(z^2+\lambda \right) \cdot \frac{\partial S(\lambda, z)}{\partial z} 
    \leq& \frac{-2 z^2 \left(\lambda +z^2\right)^2-2 \lambda  (z-1)^2 \left(z^2-\lambda \right) \log \left(\frac{\lambda }{z-1}\right)}{2 z^2+\lambda }.
  \end{align*}
  If $z^2 - \lambda \geq 0$, we are done.
  Otherwise, we use the fact that $\log(1 + x) \leq x$ for $x > -1$ (here, $x = \frac{\lambda}{z-1} - 1$) and get
  \begin{align*}
    (z-1) z \left(z^2+\lambda \right) \cdot \frac{\partial S(\lambda, z)}{\partial z} 
    \leq& \frac{2 \lambda  (z-1) (-\lambda +z-1) \left(z^2-\lambda \right)-2 z^2 \left(\lambda +z^2\right)^2}{2 z^2+\lambda }. 
  \end{align*}
  Since $z^2 \leq \lambda \leq 4$, in the rest of the proof, we assume that $1 \leq z \leq 2$ and $0 \leq \lambda \leq 4$.
  The numerator is bounded by
  \begin{align*}
    &2 \lambda  (z-1) (-\lambda +z-1) \left(z^2-\lambda \right)-2 z^2 \left(\lambda +z^2\right)^2 \\
   =& -2 \lambda ^3-2 \lambda ^2-2 z^6-2 \lambda  z^4-2 \lambda ^2 z^3-4 \lambda  z^3-2 \lambda ^2 z^2+2 \lambda  z^2+2 \lambda ^3 z+4 \lambda ^2 z \\
   \leq& -2 \lambda ^3-2 \lambda ^2-2 z^6\hphantom{-2 \lambda  z^4\;}-2 \lambda ^2 z^3-4 \lambda  z^3-2 \lambda ^2 z^2\hphantom{+2 \lambda  z^2\;}+2 \lambda ^3 z+4 \lambda ^2 z \\
   \leq& -2 \lambda ^3-2 \lambda ^2-2 z^6\hphantom{-2 \lambda  z^4\;}\hphantom{-2 \lambda ^2 z^3\;}-4 \lambda  z^3\hphantom{-2 \lambda ^2 z^2\;}\hphantom{+2 \lambda  z^2\;}+2 \lambda ^3 z \hphantom{+4 \lambda ^2 z} \\
   \leq& -2 \lambda ^3-2 \lambda ^2\hphantom{-2 z^6\;}\hphantom{-2 \lambda  z^4\;}\hphantom{-2 \lambda ^2 z^3\;}-4 \lambda  z^3\hphantom{-2 \lambda ^2 z^2\;}\hphantom{+2 \lambda  z^2\;}+2 \lambda ^3 z \hphantom{+4 \lambda ^2 z},
  \end{align*}
  where we use the fact that $z \geq 1$.
  Hence, in order to show that $\frac{\partial S(\lambda, z)}{\partial z} \leq 0$, it suffices to show 
  \begin{align} \label{eq:aux-quadratic}
   -2 \lambda^2-2 \lambda -4 z^3+2 \lambda ^2 z \leq 0.
  \end{align}
  We solve \eqref{eq:aux-quadratic} for $\lambda$, and get
  \begin{align*}
    \lambda_1 := \frac{1-\sqrt{8 z^4-8 z^3+1}}{2 (z-1)} \leq \lambda \leq \lambda_2 := \frac{\sqrt{8 z^4-8 z^3+1}+1}{2 (z-1)}.
  \end{align*}
  To prove \eqref{eq:aux-quadratic}, we need to show that $\lambda_1 \leq 0 \leq 4 \leq \lambda_2$.
  Since $z \geq 1$, it holds that $8z^4 - 8z^3 + 1 \geq 0$.
  Hence $\lambda_1 \leq 0$ holds directly.
  It remains to show that $\lambda_2 \geq 4$, 
  \begin{align*}
    &\frac{\sqrt{8 z^4-8 z^3+1}+1}{2 (z-1)} \geq 4 \\
    \Longleftrightarrow \quad  & \sqrt{8 z^4-8 z^3+1} \geq 4(2(z-1)) - 1 \\
    \overset{(+)}{\Longleftrightarrow} \quad  & 8 z^4-8 z^3+1 \geq (4(2(z-1)) - 1)^2 \\
    \Longleftrightarrow \quad  & 8 (-1 + z) (10 - 8 z + z^3) \geq 0,
  \end{align*}
  where $(+)$ holds by the fact that $1 \leq z \leq 2$.
  Hence, we only need to show that for $1 \leq z \leq 2$, $10 - 8 z + z^3 \geq 0$.
  By taking derivative, we know that the function $f(z) := 10 - 8 z + z^3$ achieves its minimum on the interval $[1, 2]$ at $z = \sqrt{8/3}$.
  And note that $f(\sqrt{8/3}) \approx 1.2907 \geq 0$.
  %
  %
\end{proof}

\subsection{Boundedness} \label{sec:boundedness}
We now prove \Cref{lem:boundedness}, the boundedness of the potential function \eqref{eq:potential-function-psi}.
We need the following.

\begin{lemma} \label{lem:phi-dis}
  For $\phi(z) := \frac{1}{(z + 1)\log(z + 1)}$, it holds that 
  \begin{align*}
    \frac{y}{x} \leq \frac{\phi(x)}{\phi(y)} \leq \left(\frac{y}{x}\right)^2, \quad \text{for $0<x < y$.}
  \end{align*}
\end{lemma}
\begin{proof}
  Note that by the mean value theorem, for some $\eta \in (x, y)$,
  \begin{align*}
    \log \left( \frac{\phi(x)}{\phi(y)}\right)
    =&(\log \circ \phi \circ \exp) (\log x) - (\log \circ \phi \circ \exp) (\log y) \\
    = &(\log \circ \phi \circ \exp)' (\log \eta) \cdot (\log x - \log y)\\
    = &\frac{\eta (\log (\eta+1)+1)}{(\eta+1) \log (\eta+1)} \cdot \log \frac{y}{x}.
  \end{align*}
  Moreover, it can be verified that $\forall \eta \in [0, +\infty)$, 
  \begin{align*}
    1 \leq \frac{\eta (\log (\eta+1)+1)}{(\eta+1) \log (\eta+1)} \leq 2.
  \end{align*}
  Therefore, we have
    \begin{align*}
\log \left(\frac{y}{x}\right) \le\log \left(\frac{\phi(x)}{\phi(y)}\right) \le 2\log \left(\frac{y}{x}\right)
   \end{align*}
  Note that $x < y$ so that $\log \left(\frac{y}{x}\right) > 0$.
  Taking $\exp$ at both side, we have
  \begin{align*}
     \frac{y}{x} &\leq \frac{\phi(x)}{\phi(y)} \leq \left(\frac{y}{x}\right)^2. \qedhere
  \end{align*}
\end{proof}

Recall the $\psi(x) := x \phi(x)=\frac{x}{(x+1)\log (x+1)}$ in \eqref{eq:potential-function-psi}, and we have the following corollary.
\begin{corollary} \label{cor:psi-by-psi}
  It holds that
  \begin{align*}
    \max_{a, b \in [\lambda(1 + \alpha)^{-d}, \lambda]} \frac{\psi(a)}{\psi(b)} \leq \max_{a, b \in [\lambda(1 + \alpha)^{-d}, \lambda]} \frac{b}{a} \leq (1 + \alpha)^{d}.
  \end{align*}
\end{corollary}


\begin{proof}[Proof of \Cref{lem:boundedness}]
  By \Cref{cor:psi-by-psi}, \Cref{lem:sym-s}, \Cref{eq:sym-x-sup}, and \Cref{thm:contract-sym}, it holds that when $1 \leq d_v \leq \Delta - 1$, we have
  \begin{align*}
    \frac{\psi(a)}{\psi(b)} \cdot H_{\*w}^\psi(\*x) \leq (1 + \alpha)^{d} \cdot (1 - \delta).
  \end{align*}
  When $d_v = \Delta$, recall that $d = \Delta - 1$, for all $z \in [1, 1 + \alpha]$, it holds that
  \begin{align*}
    U_\Delta(z) / U_d(z)
    &= \frac{(1 + \lambda z^{-d})\log(1 + \lambda z^{-d})}{(1 + \lambda z^{-\Delta})\log(1 + \lambda z^{-\Delta})} \cdot \frac{z^{-\Delta}}{z^{-d}} \cdot \frac{\Delta}{d} \leq \frac{\Delta}{d} z \leq \frac{\Delta}{d} (1 + \alpha),
  \end{align*}
  where we use \Cref{lem:phi-dis}.
  Hence, by \Cref{cor:psi-by-psi}, \Cref{lem:sym-s}, \Cref{eq:sym-x-sup}, and \Cref{thm:contract-sym}, 
  \begin{align*}
    \frac{\psi(a)}{\psi(b)} \cdot H_{\*w}^\psi(\*x) &\leq \frac{\Delta}{d} (1 + \alpha)^{\Delta} \cdot (1 - \delta). \qedhere
  \end{align*}
\end{proof}

\section{Rapid mixing in a subcritical regime}
\label{sec:GD-good}
In this section, we prove \Cref{lem:GD-good-0},
the rapid mixing of the Glauber dynamics on one side in a subcritical regime.
This allows us to efficiently simulate the block updates in the field dynamics.

Let $G = ((L, R), E)$  be a bipartite graph with $n=|L|$ vertices in $L$, and maximum degree at most $\Delta$ on $L$.
Let $\mu$ be the hardcore distribution on $G$ with fugacity $\lambda>0$ on $L$ and $\alpha>0$ on $R$, and $\nu = \mu_L$ the projection of $\mu$ on $L$.

Let $(X_t)_{t\in \^N}$ and $(Y_t)_{t\in \^N}$ be the Glauber dynamics on $\nu$ with arbitrary initial states $X_0$ and $Y_0$.
%
The following shows that the chain is rapidly mixing for the subcritical $\lambda,\alpha$ satisfying~\eqref{eq:good-regime-lambda-alpha}.
\begin{lemma} \label{lem:GD-good}
  Let $k \geq \e^9$ be a real number.
  If
  \begin{align}\label{eq:good-regime-lambda-alpha}
    \lambda (1 + \alpha)^{-\Delta} \geq k \cdot \Delta \log n,
  \end{align}
  then there is a coupling of Markov chains $(X_{t})_{t\in\mathbb{N}}$, $(Y_{t})_{t\in\mathbb{N}}$ such that  for $T \geq 21 \cdot n\log n$,
  \begin{align*}
    \Pr{X_T \neq Y_T} &\leq n^{-\ftp{\frac{T}{21n\log n}}}.
  \end{align*}
\end{lemma}

\Cref{lem:GD-good-0} follows from \Cref{lem:GD-good}:
\begin{proof}[Proof of \Cref{lem:GD-good-0}]
  Let $G = ((L, R), E)$ be a bipartite graph.
  Let $\lambda_\star > 0$ be the fugacity on both $L$ and $R$.
  Let $\mu$ be the hardcore distribution  and $\nu=\mu_{L}$.
  Let $S \subseteq L$ and $\tau \in \Omega(\mu_S)$, and let $\ell = \abs{L \setminus S}$. 
  Let $(Z_t)_{t\in\^N}$ be the Glauber dynamics on $(\theta^{-1} * \nu)^{\tau}_{L\setminus S}$, where $\theta$ satisfies 
  \[
  \theta^{-1} \geq (1 + \lambda_\star)^{\Delta} \cdot \frac{k\Delta}{\lambda_\star}  \log n,  \text{ where }k \geq \e^9.
  \]
  \Cref{lem:GD-good-0} is proved once we show that   for $T \geq 21 \cdot \ell \log \ell$,
  \begin{align}
    \DTV{Z_T}{(\theta^{-1} * \nu)^\tau_{L\setminus S}} \leq 2\ell^{-\ftp{\frac{T}{21 \cdot \ell \log \ell}}}.\label{eq:GD-good-0-proof-goal}
  \end{align}
  
  Construct a new instance. Let $G'=((L'\cup R'),E)$ be the induced subgraph of $G$ on $L'=L\setminus S$ and $R'=R\setminus \bigcup_{v:\tau_v=+1}\Gamma_v$, i.e.~$G'$ is the graph obtained from $G$ by removing $S$ from $L$ and also deleting the vertices in $R$ adjacent to those vertices being occupied in $\tau$.
  Let $\pi$ be the hardcore distribution on $G'$ with fugacity $\theta^{-1}\lambda_\star$ on $L'$ and fugacity $\lambda_\star$ on $R'$.
  Then it can be verified that 
  \[
  (\theta^{-1} * \nu)^\tau_{L\setminus S}=\pi_{L\setminus S}.
  \] 
%
%
Then apply \Cref{lem:GD-good} on such hardcore  instance on $G'$, whose condition is satisfied by 
  \begin{align*}
   \theta^{-1}\lambda_\star(1 + \lambda_\star)^{-\Delta} 
    &\geq  k \Delta \log n\ge k \Delta \log \ell.
 \end{align*}
 By \Cref{lem:GD-good}, there is a coupling $(X_t,Y_t)_{t\in\^N}$ of the chain $(Z_t)_{t\in\^N}$ such that for $T \geq 21 \cdot \ell \log \ell$, the probability of not being coupled is bounded by $2\ell^{-\ftp{{T}/{21 \ell \log \ell}}}$. Hence due to the standard coupling lemma for Markov chains (e.g.~\cite[Lemma 4.10]{levin2017markov}), we have \eqref{eq:GD-good-0-proof-goal}, which proves \Cref{lem:GD-good-0}.
%
\end{proof}
In the rest of this section, we prove \Cref{lem:GD-good}.

Let $m > 0$.
We define the following \emph{good events} on samples in $\Omega(\nu)$:
\begin{align*}
  \forall u \in L, \quad \+G_{u}(m)
  &:=
  \left\{ \sigma \in \Omega(\nu)
  \left\vert
  \begin{array}{l}
    \forall i \in R: 
   \abs{\Gamma_i} > m \implies \exists j \in \Gamma_i \setminus \{u\}, s.t., \sigma_j = +1
  \end{array}
  \right.
  \right\}.
\end{align*}
Recall that for $i\in L\cup R$, we use $\Gamma_i$ to denote $i$'s neighborhood in $G$.
Furthermore, define
\begin{align*}
  \+G(m) := \bigcap_{u \in L} \+G_{u}(m).
\end{align*}

\begin{lemma} \label{lem:good-event}
  For $t \in \^N$, $\theta > 0$, if $t \geq \theta \cdot n\log n$, then it holds for all $m>0$ that
  \begin{align*}
    \Pr{X_t \not\in \+G(m)} &\leq \Delta n^2 \tp{\frac{(1 + \alpha)^\Delta}{\lambda + (1 + \alpha)^\Delta}}^m + n^{2-\theta},\\
    \Pr{Y_t \not\in \+G(m)} &\leq \Delta n^2 \tp{\frac{(1 + \alpha)^\Delta}{\lambda + (1 + \alpha)^\Delta}}^m + n^{2-\theta}.
  \end{align*}
\end{lemma}

\begin{lemma} \label{lem:couple-good}
 There is a coupling of Markov chains $(X_{t})_{t\in\mathbb{N}}$, $(Y_{t})_{t\in\mathbb{N}}$ such that
  for all $t \in \^N$, $m>0$, 
  \begin{align*}
  \forall X_t, Y_t \in \+G(m),\quad
    \E{\abs{X_{t+1} \oplus Y_{t+1}} \mid X_t, Y_t} \leq \tp{1 - \frac{1 - \eta \cdot \Delta  m}{n}} \abs{X_t \oplus Y_t},
  \end{align*}
  where $\eta=\frac{\lambda}{\lambda + 1} - \frac{\lambda}{\lambda + (1 + \alpha)^\Delta}$.
\end{lemma}

\Cref{lem:good-event} and \Cref{lem:couple-good} will be proved in \Cref{sec:good-event} and \Cref{sec:couple-good}, respectively.

We apply the following ``coupling with stationary'' principle for mixing. 
\begin{lemma}[\text{\cite[Theorem 3.1]{hayes2006coupling}}] \label{lem:coupling-with-stationary}
  Let $\delta, \epsilon \in (0, 1)$ be real numbers.
  Let $(X_t)_{t\in \^N}$ and $(Y_t)_{t\in\^N}$ be coupled Markov chains with state space $\{-1, +1\}^n$ such that for $0 \leq t \leq T-1$, the following holds
  \begin{align*}
    \Pr{\E{\abs{X_{t+1} \oplus Y_{t+1}} \mid X_t, Y_t} \leq (1 - \epsilon) \abs{X_t \oplus Y_t}} \leq \delta,
  \end{align*}
  then it holds that
  \begin{align*}
    \Pr{X_T \neq Y_T} \leq n \cdot \left((1 - \epsilon)^T +  \frac{\delta}{\epsilon}\right).
  \end{align*}
\end{lemma}


This principle, combined with the 
\Cref{lem:good-event} and \Cref{lem:couple-good}, is sufficient to prove \Cref{lem:GD-good}.

\begin{proof}[Proof of \Cref{lem:GD-good}]
  Assume $n \geq 2$, otherwise the Glauber dynamics mixes trivially.
  The coupling between $(X_t)_{t\in\^N}$ and $(Y_t)_{t\in\^N}$ is constructed as follows. 
  For $\theta>0$, $T_2 >0$ to be fixed later, let
  \begin{align*}
  T_1 &= \theta \cdot n\log n\quad\text{ and }\quad T =T_1 + T_2.
  \end{align*}
  In the beginning, simulate the two chains $(X_t)_{0\leq t \leq T_1}$ and $(Y_t)_{0 \leq t\leq T_1}$ independently for $T_1$ steps.
  And for $T_1 < t \leq T=T_1+T_2$, at time $t$, for some suitable $m>0$ to be fixed later:
  \begin{enumerate}
  \item if $X_{t-1} \in \+G(m)$ and $Y_{t-1} \in \+G(m)$, then generate $(X_{t}, Y_{t})$ conditioning on $(X_{t-1},Y_{t-1})$ using the coupling in \Cref{lem:couple-good};
  \item otherwise, generate $X_{t}$ and $Y_{t}$ independently according to transition rule of the chain.
  \end{enumerate}
%

  We may treat $(X'_{t})_{0\leq t \leq T_2}:=(X_{t+T_1})_{0\leq t \leq T_2}$, $(Y'_{t})_{0\leq t \leq T_2}:=(Y_{t+T_1})_{0\leq t \leq T_2}$ as coupled Markov chains starting from the initial states $X'_0=X_{T_1}$ and $Y'_0=Y_{T_1}$.

  By \Cref{lem:good-event},  for $0< t \le T_2$,  it holds that
  \begin{align}\label{eq:good-event}
    \Pr{X'_t \not\in \+G(m) \text{ or } Y'_t \not\in \+G(m)} \leq 2\Delta n^2 \tp{\frac{(1 + \alpha)^\Delta}{\lambda + (1 + \alpha)^\Delta}}^m + 2 n^{2-\theta}.
  \end{align}
  By \Cref{lem:couple-good}, for $0< t \le T_2$,  it holds that
   \begin{align}\label{eq:couple-good}
  \forall X'_t, Y'_t \in \+G(m),\quad
    \E{\abs{X'_{t+1} \oplus Y'_{t+1}} \mid X'_t, Y'_t} \leq \tp{1 - \frac{1 - \eta \cdot \Delta  m}{n}} \abs{X'_t \oplus Y'_t},
  \end{align}
  where $\eta=\frac{\lambda}{\lambda + 1} - \frac{\lambda}{\lambda + (1 + \alpha)^\Delta}$.
  Denote $C= (1 + \alpha)^\Delta \geq 1$. Suppose that the followings are true:
  \[
  \epsilon:=\frac{1 - \eta \cdot \Delta m}{n}\in(0,1) \quad \text{ and }\quad \delta:=2 \Delta n^2 \tp{\frac{C}{\lambda + C}}^m + 2n^{2-\theta}\in(0,1),
  \]
  \Cref{eq:good-event} and \eqref{eq:couple-good} together imply that for $0< t \le T_2$, 
  \[
      \Pr{\E{\abs{X'_{t+1} \oplus Y'_{t+1}} \mid X'_t, Y'_t} \leq (1 - \epsilon) \abs{X'_t \oplus Y'_t}} \leq \delta.
  \]
  Then we can apply \Cref{lem:coupling-with-stationary} with the above $\epsilon,\delta\in(0,1)$, which gives
  \begin{align*}
    \Pr{X_T \neq Y_T} 
    &= \Pr{X'_{T_2} \neq Y'_{T_2}}
    \le  n \cdot \left((1 - \epsilon)^{T_2} + \frac{\delta}{\epsilon}\right)\\
    &\leq 
    n \cdot \tp{\tp{1 - \frac{1 - \eta \cdot \Delta m}{n}}^{T_2} + \frac{2\Delta n^3 \tp{\frac{C}{\lambda + C}}^m + 2n^{3-\theta}}{1 - \eta \cdot \Delta m}}\\
    &\leq 
    n \cdot \tp{\tp{1 - \frac{1 - \frac{C}{\lambda} \Delta m}{n}}^{T_2} + \frac{2\Delta n^3 \tp{\frac{C}{\lambda}}^m + 2n^{3-\theta}}{1 - \frac{C}{\lambda} \cdot \Delta m}},
  \end{align*}
  where the last inequality uses the facts $\eta= \frac{\lambda}{\lambda + 1} - \frac{\lambda}{\lambda + C} = \frac{C - 1}{\lambda + \frac{C}{\lambda} + 1 + C} \leq \frac{C}{\lambda}$ and $\frac{C}{\lambda + C}\le \frac{C}{\lambda}$.

  Due to \eqref{lem:GD-good}, we have $\lambda / C = k \cdot \Delta \log n$. Fix $m = \log n$. Then
  \begin{align*}
    \Pr{X_T \neq Y_T}
    &\leq n \cdot \tp{\tp{1 - \frac{1}{2n}}^{T_2} + \frac{2k}{k-1}\tp{\Delta n^3 \cdot n^{-\log\tp{k \cdot \Delta \log n}} + n^{3-\theta}}} \\
    &\leq n^{1 - \frac{T_2}{2 n \log n}} + 4 n^{4-\log k} + 4 n^{4-\theta}. 
  \end{align*}
  We choose $\theta = 9$ and $T_2 = 12 \cdot n\log n$,  hence $T=T_1+T_2=21\cdot n\log n$.
  Since $k \geq \e^{9}$ and $n \geq 2$,
  \begin{align*}
    \Pr{X_T \neq Y_T} \le  9 \cdot n^{-5}  \leq n^{-1}.
  \end{align*}
  By the geometric convergence of Markov chain, for any $T \geq 21 \cdot n\log n$ it holds that
  \begin{align*}
    \Pr{X_T \neq Y_T} &\leq n^{- \ftp{\frac{T}{21 n\log n}}}. \qedhere
  \end{align*}
\end{proof}

\subsection{Occurrence of good events with high probability} \label{sec:good-event}
We now prove \Cref{lem:good-event}.
It is sufficient to show that for $t \geq \theta \cdot n\log n$,
\begin{align} \label{eq:good-event-tar}
  \Pr{Y_t \not\in \+G(m)} \leq \Delta n^2 \tp{\frac{(1 + \alpha)^\Delta}{\lambda + (1 + \alpha)^\Delta}}^m + n^{2-\theta}.
\end{align}
Consider the chain $Y_t$ to be generated by the following equivalent process:
let $((v_s, R_s))_{1 \leq s \leq t}$ be a sequence of ``random seeds'', where each $v_s\in L$ and $R_s\in [0, 1]$ are drawn uniformly and independently at random;
for $1 \leq s \leq t$,  the new state $Y_s$ is constructed from $Y_s$ as that $Y_s(u)=Y_{s-1}(u)$ for all $u\neq v_s$, and
%
%
\begin{align} \label{eq:GD-update}
  Y_s(v_s) =
  \begin{cases}
    +1 & \text{if }R_s \leq \frac{\lambda}{\lambda + (1 + \alpha)^{F(Y_{s-1}, v_s)}} \\
    -1 & \text{otherwise}
  \end{cases},
\end{align}
where $F(Y_{s-1}, v_s)$ counts the number of ``free'' neighbors of $v_s$, whose spin states are not fixed by the current configuration $Y_{s-1}$ on $L$, formally:
\begin{align}
\forall Z \in \{-1, +1\}^L,
\forall w \in L,\quad
  F(Z, w) := \sum_{v\in \Gamma_w} \*1[\forall i \in \Gamma_v \setminus \{w\}: Z(i) = -1].\label{eq:free-neighbors}
\end{align}
It is easy to verify that this process faithfully simulates the Glauber dynamics on $\nu=\mu_L$.

Let $A_t$ be the event that all the vertices in $L$ have been updated by the Glauber dynamics for at least once by time $t$,
that is $A_t := \{ \{v_1, v_2, \cdots v_t\} = L \}$.
Then we have
\begin{align*}
  \Pr{Y_t \not\in \+G_u(m)}
  &= \Pr{Y_t \not\in \+G_u(m) \land A_t} + \Pr{Y_t \not\in \+G_u(m) \land \neg A_t} \\
  &\leq \Pr{Y_t \not\in \+G_u(m) \mid A_t} + \Pr{\neg A_t}.
\end{align*}
By the coupon collector, we have for $t \geq \theta \cdot n\log n$,
\begin{align*}
  \Pr{\neg A_t} \leq n \tp{1 - \frac{1}{n}}^t \leq n \e^{-\frac{t}{n}} = n^{1 - \frac{t}{n\log n}} \leq n^{1-\theta}.
\end{align*}
%
Now, fix any $v_1, v_2, \cdots, v_t$ that make $A_t$ occur.
Let $u \in L$ and let $v \in R$ be a vertex with $\abs{\Gamma_v} > m$.
In above generation of $Y_t$, for each $w \in \Gamma_v \setminus \{u\}$, and $j^*(w):=\max\{1\le j\le t: v_j = w\}$, by~\eqref{eq:GD-update}, 
\begin{align*}
  \Pr{Y_t(w) = -1 \mid (R_s)_{1\leq s\leq t, s\neq j^*(w)}} \leq \frac{(1 + \alpha)^\Delta}{\lambda + (1 + \alpha)^\Delta}.
\end{align*}
For distinct $w \in \Gamma_v \setminus \{u\}$, the $j^*(w)=\max\{1\le j\le t: v_j = w\}$ are obviously distinct, and hence $R_{j^*(w)}$ are mutually independent.
Therefore, by a chain rule, we have for $\abs{\Gamma_v}> m$
\begin{align*}
  \Pr{\forall w \in \Gamma_v \setminus \{u\}, Y_t(w) = -1}
  &\leq \tp{\frac{(1 + \alpha)^\Delta}{\lambda + (1 + \alpha)^\Delta}}^{\abs{\Gamma_v \setminus \{u\}}}
  \leq \tp{\frac{(1 + \alpha)^\Delta}{\lambda + (1 + \alpha)^\Delta}}^m.
\end{align*}
Note that $\abs{R} \leq \Delta \abs{L} = \Delta n$. Thus by  union bound, it holds that
\begin{align*}
  \Pr{Y_t \not\in G_u(m) \mid A_t} \leq \Delta n \cdot \tp{\frac{(1 + \alpha)^\Delta}{\lambda + (1 + \alpha)^\Delta}}^m,
\end{align*}
which implies that
\begin{align*}
  \Pr{Y_t \not\in G_u(m)} \leq \Delta n \cdot \tp{\frac{(1 + \alpha)^\Delta}{\lambda + (1 + \alpha)^\Delta}}^m + n^{1 - \theta}.                                                                                                
\end{align*}
Finally, \eqref{eq:good-event-tar} can be proved by applying another union bound over all $u \in L$.

\subsection{Contraction of path coupling conditioning on good events}
\label{sec:couple-good}
We now prove \Cref{lem:couple-good}.
Let $m>0$ and $ \+G= \+G(m)$.
Fix any $X_t, Y_t \in \+G$. 
Denote by $u_1, u_2, \cdots, u_i$ the vertices in $Y_t^{-1}(+1) \setminus X_t^{-1}(+1)$, and  by $v_{i+1}, v_{i+2},\cdots, v_{i+j}$ the  vertices in $X_t^{-1}(+1) \setminus Y_t^{-1}(+1)$.
Let $P = (P_0, P_1, \cdots, P_{i+j})$ be a path of configurations  from $P_0=X_t$ to $P_{i+j}=Y_t$ defined by:
\begin{align*}
\forall 1 \leq x \leq i + j,\quad   P_x :=
  \begin{cases}
    \text{$P_{x-1}$ except that $P_{x-1}(u_x)$ is set to $+1$} & \text{if }1 \leq x \leq i, \\
    \text{$P_{x-1}$ except that $P_{x-1}(v_x)$ is set to $-1$} & \text{if }i < x \leq i + j.
  \end{cases}
\end{align*}
Note that the length of this path is precisely $\abs{X_t \oplus Y_t}$ and $P_x \in \+G$ for every $0 \leq x \leq i + j$.
Hence, by the standard path coupling argument~\cite{bubley1997path}, it is sufficient to construct a coupling $(X_{t+1}, Y_{t+1})$ for those $X_t, Y_t \in \+G$ with $\abs{X_t \oplus Y_t} = 1$, such that
\begin{align} \label{eq:couple-tar}
  \E{\abs{X_{t+1} \oplus Y_{t+1}} \mid X_t, Y_t} \leq 1 - \frac{1}{n} + \frac{\eta\cdot \Delta m}{n}.
\end{align}

Suppose $X_t \oplus Y_t = \{u\}$, without loss of generality, assume that $X_t(u) = +1$ and $Y_t(u) = -1$.
We generate $X_{t+1}, Y_{t+1}$ by the following coupling procedure:
\begin{enumerate}
\item pick a vertex $i \in L$ uniformly at random;
\item generate $(X_{t+1}(i), Y_{t+1}(i))$ according to the optimal coupling of their marginal distributions;
\item let $X_{t+1}(L\setminus\{i\}) = X_t(L\setminus \{i\})$ and $Y_{t+1}(L\setminus\{i\}) = Y_t(L\setminus \{i\})$.
\end{enumerate}
Due to this construction, it holds that
\begin{align}\label{eq:couple-good-contraction}
  \E{\abs{X_{t+1} \oplus Y_{t+1}} \mid X_t, Y_t} \leq 1 - \frac{1}{n} + \sum_{w\in L\setminus\{ u\}} \frac{\-d(w)}{n},
\end{align}
where the function $\-d: L \to [0, 1]$ is defined as: for $w \in L$,
\begin{align*}
  \-d(w) := \frac{\lambda}{\lambda + (1 + \alpha)^{F(X_t,w)}} - \frac{\lambda}{\lambda + (1 + \alpha)^{F(Y_t, w)}},
\end{align*}
where recall that $F(Y_t, w)$, as defined in \eqref{eq:free-neighbors}, counts the number of free neighbors of $w$ given the configuration $Y_t$.
Note that $\-d(w)\le \eta :=\frac{\lambda}{\lambda + 1} - \frac{\lambda}{\lambda + (1 + \alpha)^\Delta}$ in the worst case.
And by our assumption, we have $X_t(u) = +1$ and $Y_t(u) = -1$, thus
\begin{align*}
  F(Y_t, w) - F(X_t, w) &= \sum_{v \in \Gamma_w \cap \Gamma_u} \*1[\forall i \in \Gamma_v \setminus \{w\}: Y_t(i) = -1].
\end{align*}
Therefore, given any $w \in L$, 
if $\deg_G(v) > m$ for all $v \in \Gamma_w \cap \Gamma_u$, then by the fact that $X_t, Y_t \in \+G_w(m)$, it holds that $F(X_t, w) = F(Y_t, w)$, which means $\-d(w) = 0$ in this case.
Therefore, $\-d(w)$ can be upper bounded as:
\begin{align*}
  \-d(w) \leq
  \begin{cases}
    0 & \text{if }\forall v \in \Gamma_w \cap \Gamma_u: \deg_G(v) > m \\
    \eta & \text{otherwise}
  \end{cases}.
\end{align*}
Applying this upper bound to \eqref{eq:couple-good-contraction},
we have
\begin{align*}
  \E{\abs{X_{t+1} \oplus Y_{t+1}} \mid X_t, Y_t}
  &\leq 1 - \frac{1}{n} + \sum_{w\in L\setminus\{ u\}} \frac{\eta}{n} \cdot \*1\left[\exists v \in \Gamma_w \cap \Gamma_u\text{ s.t.~}\deg_G(v) \leq m \right] \\
  &\leq 1 - \frac{1}{n} + \frac{\eta}{n} \cdot \Delta m.
\end{align*}
The very last inequality is due to the fact that 
any such vertex $w \in L$ satisfying the condition that $\exists v \in \Gamma_w \cap \Gamma_u$ s.t.~$\deg_G(v) \leq m$, 
can be found by enumerating the neighbors of  the low-degree ($\deg\le m$) neighbors of $u$.
And there are at most $\Delta m$ of them.

\section{Mixing of the field dynamics on one side}
\label{sec:FD-ent-decay}
In this section, we prove \Cref{thm:FD-ent-decay-short}, the entropic decay of the field dynamics.
The proof uses the analysis of negative-fields localization developed in~\cite{chen2022localization},
where the negative-fields localization process is instantiated by the field dynamics. 
This stochastic localization process is powerful enough to produce a block factorization of entropy for the field dynamics,
which implies the desired entropy decay by the argument developed in \cite{chen2022optimal, chen2022near}.

The field dynamics have been introduced in \Cref{sec:ent-decay-field-general}, specifically in \Cref{def:field-dynamics-ori}.
Here, for technical convenience we introduce another equivalent definition of the field dynamics.

Let $\theta \in (0, 1)$ be a real number. 
For any distribution  $\nu$ over $\{-1, +1\}^n$,
the field dynamics $\FD{\nu}$ on $\nu$ with parameter $\theta$ is a Markov chain $(X_t)_{t \in \^N}$ on space $\Omega(\nu)$. 
In its $i$-th transition, 
\begin{enumerate}
\item \label{item:FD-9-1} generate a set $R\subseteq X_{i-1}^{-1}(+1)$ by including each $v\in[n]$ with $X_{i-1}(v) = +1$ into $R$ with probability $1 - \theta$;
\item \label{item:FD-9-2} sample $X_i$ according to the joint distribution $\sim (\theta * \nu)^{\*1_R}$, that is, $X_i(v) = 1$ for $v \in R$, and $X_i([n]\setminus R)$ is sampled as in $\theta * \nu$ conditioned on all $v\in R$ being occupied.
\end{enumerate}
Note that this is equivalent to the process defined in \Cref{def:field-dynamics-ori},
where the above transition rule  is defined in terms of the complement set $S=[n]\setminus R$.

It is also helpful to think of the field dynamics as a composition of 
the following  operators.
\begin{definition} \label{def:FD-down-up}
  Let $\theta \in (0, 1)$ be a real number.
  For any distribution $\nu$ over $\{-1, +1\}^n$, we define two operators $\FDdown[\theta]{\nu} \in \^R^{\Omega(\nu) \times 2^{[n]}}$ and $\FDup[\theta]{\nu} \in \^R^{2^{[n]} \times \Omega(\nu)}$ such that for $X \in \Omega(\nu)$ and $R \subseteq [n]$, 
  \begin{align*}
    \FDdown{\nu}(X, R) &:= \*1[R \subseteq X^{-1}(1)] \cdot (1 - \theta)^{\abs{R}} \theta^{\abs{X}_+ - \abs{R}}, \\
    \FDup{\nu}(R, X) &:= (\theta * \nu)^{\*1_R}(X),
  \end{align*}
  where we recall that $\abs{X}_+$ denotes the number of $+1$ in vector $X$.
\end{definition}


By \Cref{def:FD-down-up}, it holds that
\begin{itemize}
\item for a fixed $X \in \Omega(\nu)$, $\FDdown{\nu}(X, \cdot)$ is the distribution of random set $R \subseteq [n]$, such that for each $v \in [n]$ independently,  we have $v\in R$ with probability $1 - \theta$ if $X(v) = +1$, and $v\not\in R$ if $X(v) = -1$.
This correspond to \Cref{item:FD-9-1} above.
\item for a fixed $R \subseteq [n]$, it is straightforward to note that $\FDup{\nu}(R, \cdot) = (\theta * \mu)^{\*1_R}$. This correspond to \Cref{item:FD-9-2} above.
\end{itemize}
So, it is straightforward to verify that 
\begin{align}\label{eq:FD-up-down}
\FD{\nu} = \FDdown{\nu}\FDup{\nu}.
\end{align}

%


The following is a technical restatement of  \Cref{thm:FD-ent-decay-short}.
\begin{theorem}\label{thm:FD-ent-decay}
  Let $\theta \in (0, 1), K, \eta \geq 1$ be real numbers.
  If a distribution $\nu$ over $\{-1, +1\}^n$ satisfies: 
  \begin{enumerate}
  \item $\lambda * \nu$ is $K$-marginally stable for $\lambda \in [\theta, 1]$,
  \item $\lambda * \nu$ is $\eta$-spectrally independent for $\lambda \in [\theta, 1]$,
  \end{enumerate}
  then for  $\kappa := \theta^{2\cdot 10^3 \eta K^4}$, for any distribution $\pi$ that is absolutely continuous respect to $\nu$, 
  we have
  \begin{align*}
    \DKL{\pi \FDdown{\nu}}{\nu \FDdown{\nu}} \leq (1 - \kappa) \DKL{\pi}{\nu}.
  \end{align*}
  \end{theorem}
Note that the conclusion in \Cref{thm:FD-ent-decay-short}:
  \begin{align*}
    \DKL{\pi \FD{\nu}}{\nu \FD{\nu}} \leq (1 - \kappa) \DKL{\pi}{\nu},
  \end{align*}
follows immediately from \Cref{thm:FD-ent-decay}, by \eqref{eq:FD-up-down} and  the data processing inequality.

In the rest of this section, we prove \Cref{thm:FD-ent-decay}.

We adopt the following notion of {``negative-field'' localization process} introduced in \cite{chen2022localization}
\begin{definition}[negative-field localization process]
Let $\nu$ be a distribution over $\{-1, +1\}^n$.
A \emph{negative-field localization process} for $\nu$, denoted by $(\nu_t)_{t \geq 0}$, is a continuous-time stochastic processes defined as follows.
Let $(R_t)_{t \geq 0}, (\nu_t)_{t\geq 0}$ be continuous-time stochastic processes such that for any $t \geq 0$, $R_t \subseteq [n]$ is a subset of $[n]$ and $\nu_t$ is a distribution over $\{-1, +1\}^n$.
Given that $(R_t)_{t\geq 0}$ has been generated, 
the process $(\nu_t)_{t\geq 0}$ can be generated such that $\nu_t := (\e^{-t} * \nu)^{\*1_{R_t}}$ for every $t \geq 0$.
Now, we describe how the process $(R_t)_{t\geq 0}$ is generated.

Let $R_0 = \emptyset$. Suppose that $(R_t)_{0\leq t \leq t_\star}$ has been generated up to some time threshold $t_\star\ge 0$.
Iteratively, the process $(R_t)_{t_\star < t \leq \tau}$ is generated from time $t_\star$ to some stopping time $\tau > t_\star$ as below:
\begin{itemize}
\item If $R_{t_\star} = [n]$, then for every $t > t_\star$, let $R_t = R_{t_\star}$.
\item If otherwise, 
  \begin{enumerate}
  \item for $i \in [n] \setminus R_{t_\star}$, let $T_i$ be mutually independent random variables  such that for $s\ge t$,
    \begin{align*}
      \Pr{T_i > s} := \exp\tp{- \int_t^s (\e^{-r} * \nu)^{R_{t_\star}}_i \; \-d r};
    \end{align*}
  \item define 
    \begin{align*}
      \tau := \min_{i \in [n] \setminus R_t} T_i \quad \text{and} \quad J := \argmin_{i \in [n] \setminus R_t} T_i;
    \end{align*}
  \item extend $R_t$ from time $t_\star$ to $\tau$ and include $J$ into $R_{\tau}$, that is, for every $t \in (t_\star, \tau]$, let
    \begin{align*}
      R_t =
      \begin{cases}
        R_t & t_\star < t < \tau, \\
        R_t \cup \{J\} & t = \tau.
      \end{cases}
    \end{align*}
  \end{enumerate}
\end{itemize}
\end{definition}

It is straightforward  to see that extending $R_0$ to $(R_t)_{t\geq 0}$ requires at most $n+1$ rounds of iteration described above.
Hence the processes $(R_t)_{t\geq 0}$ and $(\nu_t)_{t\geq 0}$ are well defined.

Note that the process $(\nu_t)_{t\geq 0}$ is time-homogeneous in the definitions of the random variables $T_i, i \in [n]$.
Conditioning on $(\nu_r)_{r \in [0, s]}$, the transition rule of the process $(\nu_{s + t})_{t \geq 0}$ is 
identical to that of $(\tilde{\nu}_t)_{t\geq 0}$ with the starting measure $\tilde{\nu}_0 = \nu_s$.

The \emph{negative-fields localization scheme} is a map between a starting distribution $\nu$ to the negative-field localization process $(\nu_t)_{t\geq 0}$ with $\nu_0 = \nu$.
The negative-fields localization scheme can be analyzed through the lens of entropic independence introduced in \cite{anari2022entropic}.
We restate \Cref{def:EI-intro} for the entropic independence here.

\begin{definition}[entropic independence] \label{def:EI}
  Let $\alpha > 0$ be a real number.
  A distribution $\nu$ over $\{-1, +1\}^n$ is said to be $\alpha$-entropically independent if for every distribution $\mu$ which is absolutely continuous with respect to $\nu$, it holds that
  \begin{align*}
    \sum_{i\in[n]}\DKL{\mu_i}{\nu_i}:=\sum_{i \in [n]} \tp{\nu_i(+1) \log \frac{\nu_i(+1)}{\mu_i(+1)} + \nu_i(-1) \log \frac{\nu_i(-1)}{\mu_i(-1)}} \leq \alpha \cdot \DKL{\mu}{\nu}.
  \end{align*}
\end{definition}

Recall the following notion of entropy of a function $f:\Omega\to\mathbb{R}_{\ge 0}$ with respect to a distribution $\mu$, defined in \Cref{sec:prelim-MC}:
\begin{align*}
  \Ent[\mu]{f} := \E[\mu]{f\log \frac{f}{\E[\mu]{f}}},\quad\text{with the convention }0\log 0 = 0.
\end{align*}

\begin{lemma}[\text{\cite[Proposition 41]{chen2022localization}}] \label{lem:dEnt}
  Let $\nu$ be a distribution over $\{-1, +1\}^n$ and $(\nu_t)_{t\geq 0}$ be a negative-field localization process for $\nu$.
  Fix $t \geq 0$ and $\nu_t$. If $\nu_t$ is $\alpha$-entropically independent, then for all $f:\{-1, +1\}^n \to \^R_{\geq 0}$, and for all $h \geq 0$,
  \begin{align*}
    \E{\Ent[\nu_{t+h}]{f} \mid \nu_t} \geq \Ent[\nu_t]{f} (1 - 4 h \alpha) + o(h).
  \end{align*}
\end{lemma}

\begin{lemma}[\text{\cite[Theorem 67]{chen2022localization}}] \label{lem:EI}
  Let $\nu$ be a distribution over $\{-1, +1\}^n$ and let $\eta \geq 1$.
  Suppose that $\nu$ is $\eta$-spectrally independent and $K$-marginally stable, then it holds that $\nu$ is $384\eta K^4$-entropically independent.
\end{lemma}
\begin{remark}
	The notion of $K$-marginal stability we used is slightly different from the one in the statement of  in \cite[Theorem 67]{chen2022localization}.
	In their version of $K$-marginal stability, the distribution $\rho$ as in \Cref{def:K-stable}, is exactly $\nu$ instead of allowing the freedom of $\rho \in \{\nu, \overline{\nu}\}$. %
	We note that their Theorem 67 applies to the relaxed definition of marginal stability, simply because it is  straightforward to verify that
  \begin{itemize}
  \item $\nu$ is $\eta$-spectrally independent $\Leftrightarrow$ $\overline{\nu}$ is $\eta$-spectrally independent;
  \item $\nu$ is $\alpha$-entropically independent $\Leftrightarrow$ $\overline{\nu}$ is $\alpha$-entropically independent.
  \end{itemize}
  Hence, if we know that $\nu$ is spectrally independent and $\overline{\nu}$ satisfies \eqref{eq:K-stable}, we are still able to use Theorem 67 of \cite{chen2022localization} and show that $\nu$ is entropically independent.
\end{remark}

It was pointed out that the negative-fields localization process for a joint distribution $\nu$ generates exactly the field dynamics on $\nu$.
However, we are not aware of any explicit proof of this equivalence.
Here, we build this connection explicitly.

\begin{lemma} \label{lem:NFL-is-FD}
  Let $\nu$ be a distribution over $\{-1, +1\}^n$.
  Let $(R_t)_{t\geq 0}$ be a negative-field localization process for~$\nu$.
  Fix $t \geq 0, h \geq 0$, and let $A \subseteq B \subseteq [n]$. It holds that
  \begin{align*}
    \Pr{R_{t+h} = B \mid R_t = A} = \sum_{X: X_B = \*1} (e^{-t} * \nu)^{\*1_A}(X) \cdot (1 - \e^{-h})^{\abs{B} - \abs{A}}(\e^{-h})^{\abs{X}_+ - \abs{B}},
  \end{align*}
  where $\abs{X}_+$ denotes the number of $+1$ in the vector $X$.
\end{lemma}

\Cref{lem:NFL-is-FD} is proved in \Cref{sec:NFL-is-FD}.

More intuitively, \Cref{lem:NFL-is-FD} means that conditioning on $\nu_t$ and $R_t$, the random variable $\nu_{t+h}$ follows the following law:
\begin{enumerate}
\item sample $X \sim \nu_t = (\e^{-t} * \nu)^{R_t}$;
\item $R_{t+h} = R_t$; \textbf{for} $i\in [n] \setminus R_t$ with $X(i) = 1$: add $i$ to $R_{t+h}$ with probability $1 - \e^{-h}$;
\item let $\nu_{t+h} = (\e^{-(t+h)} * \nu)^{\*1_{R_{t+h}}}$.
\end{enumerate}

Now, we are ready to prove \Cref{thm:FD-ent-decay}.
\begin{proof}[Proof of \Cref{thm:FD-ent-decay}]
  According to the assumptions of \Cref{thm:FD-ent-decay}, for any possible $\nu_t = (\e^{-t} * \nu)^{\*1_R}$ and $R \subseteq [n]$, 
   $\nu_t$ is a $\eta$-spectrally independent distribution over $\{-1, +1\}^n$ and it is $K$-marginally stable.
  Then by \Cref{lem:EI}, we have that $\nu_t$ is $384\eta K^4$-entropically independent.
  By \Cref{lem:dEnt}, let $\zeta := 2 \cdot 10^3 \eta K^4$, for any $t \in [0, - \log \theta], h \geq 0$ and any $\nu_t$, we have that for all $f: \{-1, +1\}^n \to \^R_{\geq 0}$,
  \begin{align*}
    \E{\Ent[\nu_{t+h}]{f} \mid \nu_t} \geq \Ent[\nu_t]{f} (1 - \zeta \cdot h) + o(h).
  \end{align*}
  Without loss of generality, we assume $\Ent[\nu]{f} \neq 0$.
  Taking expectation and logarithm at both side,
  \begin{align*}
    \log \E{\Ent[\nu_{t+h}]{f}} - \log \E{\Ent[\nu_t]{f}} &\geq - \zeta \cdot h + o(h) \\
    \frac{\-d \log \E{\Ent[\nu_t]{f}}}{\-d t} &\geq -\zeta.
  \end{align*}
  By an integrating, we know that for $t = -\log \theta$,
  \begin{align*}
    \frac{\E{\Ent[\nu_t]{f}}}{\Ent[\nu]{f}} \geq \exp(- \zeta \cdot t) = \theta^{\zeta},
  \end{align*}
  which can be rewritten as
  \begin{align} \label{eq:ent-factor}
    \Ent[\nu]{f} \leq \theta^{-\zeta} \cdot \E{\Ent[\nu_t]{f}}.
  \end{align}
  Now, by \Cref{lem:NFL-is-FD}, for $\pi := \e^{-t} * \nu$, the expectation on the right hand side can be calculated as
  \begin{align*}
    \E{\Ent[\nu_t]{f}}
    &= \sum_{R\subseteq [n]} \Pr{R_t = R} \cdot \Ent[\pi^{\*1_R}]{f} \\
    &= \sum_{R\subseteq [n]} \sum_{X:X_R = \*1} \nu(X) \cdot (1 - \theta)^{\abs{R}} \theta^{\abs{X}_+ - \abs{R}} \cdot \Ent[\pi^{\*1_R}]{f}.
  \end{align*}
  Note that by the definition of $\pi$, we have
  \begin{align*}
    \forall \sigma \in \{-1, +1\}^n, \quad \pi(\sigma) = \frac{\nu(\sigma) \theta^{\abs{\sigma}_+}}{Z_\pi},
  \end{align*}
  where $Z_\pi := \sum_{\sigma \in \{-1, +1\}^n} \nu(\sigma) \theta^{\abs{\sigma}_+}$.
  Hence, we have
  \begin{align*}
    \E{\Ent[\nu_t]{f}}
    &= \frac{Z_\pi}{\theta^n} \sum_{R\subseteq [n]} (1 - \theta)^{\abs{R}} \theta^{n - \abs{R}} \sum_{X: X_R = \*1} \pi(X) \cdot \Ent[\pi^{\*1_R}]{f} \\
    &= \frac{Z_\pi}{\theta^n} \sum_{R\subseteq [n]} (1 - \theta)^{\abs{R}} \theta^{n - \abs{R}} \cdot \pi_R(\*1_R) \cdot \Ent[\pi^{\*1_R}]{f}.
  \end{align*}
  Therefore,  \eqref{eq:ent-factor} can be expressed as
  \begin{align} \label{eq:ent-mag-factor}
    \Ent[\nu]{f} \leq \theta^{-\zeta} \cdot \frac{Z_\pi}{\theta^n} \sum_{R\subseteq [n]} (1 - \theta)^{\abs{R}} \theta^{n - \abs{R}} \cdot \pi_R(\*1_R) \cdot \Ent[\pi^{\*1_R}]{f},
  \end{align}
  which is known as the $\theta$-magnetized block factorization of entropy introduced in \cite{chen2022optimal}.
  According to \cite[Equation (20)]{chen2022near}, letting $f = \frac{\mu}{\nu}$, where $\mu$ is a distribution that is absolutely continuous with respect to $\nu$, \Cref{eq:ent-mag-factor} is equivalent to
  \begin{align*}
    \DKL{\mu}{\nu} \leq \theta^{-\zeta} \cdot \tp{\DKL{\mu}{\nu} - \DKL{\mu\FDdown{\nu}}{\nu\FDdown{\nu}}},
  \end{align*}
  which is equivalent to
  \begin{align*}
    \DKL{\mu\FDdown{\nu}}{\nu\FDdown{\nu}} &\leq \tp{1 - \theta^{\zeta}} \DKL{\mu}{\nu}. \qedhere
  \end{align*}
\end{proof}

\subsection{Negative-fields localization and field dynamics}
\label{sec:NFL-is-FD}
To prove \Cref{lem:NFL-is-FD},
we keep track the following quantity:
for $t \geq 0, h \geq 0$ and $A \subseteq B \subseteq [n]$,  define
\begin{align*}
  \mathbb{P}^{\nu}_{A, B}(t \to t + h) &:= \sum_{X: X_B = \*1} (\e^{-t} * \nu)^{\*1_A} (X) \cdot (1 - \e^{-h})^{\abs{B} - \abs{A}} (\e^{-h})^{\abs{X}_+ - \abs{B}}.
\end{align*}
We claim the following chain rule for the quantity defined above.
\begin{lemma}[chain rule] \label{lem:field-chain}
  For $s, t, h \geq 0$ and $A \subseteq B \subseteq [n]$, 
  \begin{align*}
    \mathbb{P}^{\nu}_{A, B}(s \to s + t + h) = \sum_{M: A \subseteq M \subseteq B} \mathbb{P}^{\nu}_{A, M}(s \to s + t) \mathbb{P}^{\nu}_{M, B}(s + t \to s + t + h).
  \end{align*}
\end{lemma}

\begin{lemma} \label{lem:field-approx}
Let $(\nu_t)_{t \geq 0}$, $(R_t)_{t\geq 0}$ be negative-fields localization processes for a joint distribution $\nu$ on $\{-1,+1\}^n$.
  It holds that for $A \subsetneq B \subseteq [n]$, 
  \begin{align*}
    \Pr{R_{t + h} = B \mid R_t = A} = \mathbb{P}^\nu_{A, B}(t \to t + h) + o(h).
  \end{align*}
\end{lemma}

\begin{proof}
  By the definition of the random variables $T_v, v \in [n]$, it holds that
  \begin{align*}
    \Pr{T_v \in [t, t + h]} = h \cdot (\e^{-t} * \nu)^{\*1_A}_v + o(h).
  \end{align*}
  Hence, it holds that
  \begin{align*}
    \Pr{R_{t + h} = B \mid R_t = A}
    &= \begin{cases}
         h \cdot (\e^{-t} * \nu)^{\*1_A}_v + o(h) & B\setminus A = \{v\}, \\
         o(h) & \abs{B\setminus A} > 1.
       \end{cases}
  \end{align*}
  On the other hand, it holds that
  \begin{align*}
    \mathbb{P}^\nu_{A, B}(t \to t + h)
    &= \sum_{X: X_B = 1} (\e^{-t} * \nu)^{\*1_A} (x) \cdot (1 - \e^{-h})^{\abs{B} - \abs{A}} (\e^{-h})^{\abs{X}_+ - \abs{B}} \\
    &= \sum_{X: X_B = 1} (\e^{-t} * \nu)^{\*1_A} (x) \cdot (h + o(h))^{\abs{B} - \abs{A}} (1 - h + o(h))^{\abs{X}_+ - \abs{B}} \\
    &= \begin{cases}
         h \cdot (\e^{-t} * \nu)^{\*1_A}_v + o(h) & B\setminus A = \{v\}, \\
         o(h) & \abs{B\setminus A} > 1.
       \end{cases}
  \end{align*}
  Hence we have $\Pr{R_{t+h} = B \mid R_t = A} - \mathbb{P}_{A, B}(t \to t + h) = o(h)$.
\end{proof}

Now, to prove \Cref{lem:NFL-is-FD}, it is sufficient to strengthen  \Cref{lem:field-approx} to show for any $A \subseteq B \subseteq [n]$,
\begin{align*}
  \Pr{R_{t+h} = B \mid R_t = A} &= \mathbb{P}^\nu_{A, B}(t \to t + h).
\end{align*}
We prove this based on an induction on $\ell := \abs{B \setminus A}$.

For the induction basis: when $\ell = 0$
by definition of the process $(\nu_t)_{t \geq 0}$, it holds that
\begin{align*}
  \Pr{R_{t+h} = A \mid R_t = A}
  &= \Pr{\forall i \in [n], T_i > t + h} \\
  &= \prod_{i \in [n] \setminus A} \exp\tp{- \int_t^{t+h} (\e^{-r} * \nu)^{\*1_A}_i \; \-d r}
\end{align*}
Hence we have
\begin{align*}
  \left. \frac{\-d \log \Pr{R_{t+h} = A \mid R_t = A}}{\-d h} \right\vert_{h = h_0} &= - \sum_{i \in [n] \setminus A} (\e^{-(t + h_0)} * \nu)^{\*1_A}_i.
\end{align*}
On the other hand, we have
\begin{align*}
  &\hspace{-2cm}\left. \frac{\-d \log \mathbb{P}^\nu_{A, A}(t \to t + h)}{\-d h} \right\vert_{h = h_0} \\
  &= \lim_{\Delta \to 0_+} \frac{\log \mathbb{P}^\nu_{A,A}(t \to t + h_0 + \Delta) - \log \mathbb{P}^\nu_{A,A}(t \to t + h_0)}{\Delta} \\
  (\text{by \Cref{lem:field-chain}}) &= \lim_{\Delta \to 0_+} \frac{\log \mathbb{P}^\nu_{A,A}(t + h_0 \to t + h_0 + \Delta)}{\Delta} \\
  &= \lim_{\Delta \to 0_+} \frac{\log \sum_{X: X_A = \*1} (\e^{-(t + h_0)} * \nu)^{\*1_A} (X) \cdot (\e^{-\Delta})^{\abs{X}_+ - \abs{A}}}{\Delta} \\
  &= - \frac{\sum_{X: X_A = \*1} (\e^{-(t + h_0)} * \nu)^{\*1_A}(X) \cdot \tp{\abs{X}_+ - \abs{A}}}{\sum_{X: X_A = \*1} (\e^{-(t + h_0)} * \nu)^{\*1_A} (X)} \\
  &= - \sum_{X} (\e^{-(t + h_0)} * \nu)^{\*1_A} (X) \cdot \tp{\abs{X}_+ - \abs{A}} \\
  &= - \sum_{i \in [n] \setminus A} (\e^{-(t + h_0)} * \nu)^{\*1_A}_i.
\end{align*}
Hence, it holds that $\Pr{R_{t + h} = A \mid R_t = A} = \mathbb{P}^\nu_{A, A}(t \to t + h)$.

Now consider general $\ell \geq 1$. 

As the induction hypothesis, assume that for any $A \subseteq B \subseteq [n]$ such that $\abs{B\setminus A} < \ell$, 
\[
\Pr{R_{t + h} = B \mid R_t = A} = \mathbb{P}^\nu_{A, B}(t \to t + h).
\]

Let $k > 0$ be an integer and divide the interval $[t, t + h]$ into $k$ equal-sized subintervals (recall that $t,h$ are real numbers).
According to the chain rule in probability, we have
\begin{align*}
  &\Pr{R_{t+h} = B \mid R_t = A} \\
  &\hspace{2cm} = \sum_{A = M_0 \subseteq M_1 \subseteq \cdots \subseteq M_k = B} \prod_{i = 1}^k \Pr{R_{t + ih/k} = M_i \mid R_{t + (i-1)h/k} = M_{i-1}}.
\end{align*}
On the other hand, by \Cref{lem:field-chain}, the following  chain rule also holds for the function $\mathbb{P}$: 
\begin{align*}
  \mathbb{P}^\nu_{A, B}(t \to t + h) &= \sum_{A = M_0 \subseteq M_1 \subseteq \cdots \subseteq M_k = B} \prod_{i = 1}^k \mathbb{P}^\nu_{M_{i-1}, M_i}(t + (i-1)h/k \to t + ih/k).
\end{align*}
Fix $1\leq i \leq k$. We have the following cases:
\begin{itemize}
\item when $\abs{M_i\setminus M_{i-1}} < \ell$, by the induction hypothesis, it holds that
\begin{align*}
  \Pr{R_{t + ih/k} = M_i \mid R_{t + (i-1)h/k} = M_{i-1}} &= \mathbb{P}^\nu_{M_{i-1}, M_i}(t + (i-1)h/k \to t + ih/k).
\end{align*}
\item
when $\abs{M_i \setminus M_{i-1}} = \ell$, by \Cref{lem:field-approx}, it holds that
\begin{align*}
  \Pr{R_{t + ih/k} = M_i \mid R_{t + (i-1)h/k} = M_{i-1}} &- \mathbb{P}^\nu_{M_{i-1}, M_i}(t + (i-1)h/k \to t + ih/k)=o\left(\frac{h}{k}\right).
\end{align*}
Also note that when $\abs{M_i \setminus M_{i-1}} = \ell$, then it holds that $M_0 = M_1 = \cdots = M_{i-1} = A$ and $M_i = M_{i+1} = \cdots = M_k = B$.
Overall, this case happens at most $k$ times.
\end{itemize}

Calculating $\Pr{R_{t + h} = B \mid R_t = A} - \mathbb{P}^\nu_{A, B}(t \to t + h)$, we have
\begin{align*}
  &\Pr{R_{t + h} = B \mid R_t = A} - \mathbb{P}^\nu_{A, B}(t \to t + h) \\
  = &\sum_{j=1}^k \Pr{R_{t + (j-1)h/k} = A \mid R_t = A} \cdot o\left(\frac{h}{k}\right) \cdot \Pr{R_{t + h} = B \mid R_{t + jh/k} = B} \\
  = &k \cdot o\left(\frac{1}{k}\right)=o(1).
\end{align*}
We have $\Pr{R_{t + h} = B \mid R_t = A} - \mathbb{P}^\nu_{A, B}(t \to t + h) = 0$ as $k \to \infty$.

\subsection{Verifying the chain rule (Proof of \Cref{lem:field-chain})}
We now prove \Cref{lem:field-chain}, the chain rule for $\mathbb{P}^\nu$.
Note that for $A \subseteq B \subseteq [n]$, and $U \subseteq [n]$ such that $U \cap A = U \cap B = \emptyset$, it holds that 
\begin{align*}
  &\mathbb{P}^\nu_{A \cup U, B \cup U}(s + t \to s + t + h) \\
  = &\sum_{X: X_{B\cup U} = \*1} (\e^{-(t + s)} * \nu)^{\*1_{A \cup U}} (X) \cdot (1 - \e^{-h})^{\abs{B} - \abs{A}} (\e^{-h})^{\abs{X}_+ - \abs{B} - \abs{U}} \\
  = &(\e^{-h})^{-\abs{U}} \cdot \sum_{X: X_{B} = \*1} (\e^{-t} *  (\e^{-s} * \nu)^{\*1_U})^{\*1_A} (X) \cdot (1 - \e^{-h})^{\abs{B} - \abs{A}} (\e^{-h})^{\abs{X}_+ - \abs{B}} \\
  = &(\e^{-h})^{-\abs{U}} \cdot \mathbb{P}^{\mu}_{A,B}(t \to t + h), \quad \text{where $\mu = (\e^{-s} * \nu)^{\*1_U}$}.
\end{align*}
So, in order to prove
\begin{align*}
  \mathbb{P}^{\nu}_{A, B}(s \to s + t + h) = \sum_{M: A \subseteq M \subseteq B} \mathbb{P}^{\nu}_{A, M}(s \to s + t) \mathbb{P}^{\nu}_{M, B}(s + t \to s + t + h),
\end{align*}
it is sufficient to show for $\mu = (\e^{-s} * \nu)^{\*1_A}$
\begin{align*}
  \mathbb{P}^\mu_{\emptyset, B\setminus A} (0 \to t + h)
  &= \sum_{M: A \subseteq M \subseteq B} \mathbb{P}^{\mu}_{\emptyset, M\setminus A} (0 \to t) \mathbb{P}^\mu_{M\setminus A, B\setminus A}(t \to t + h),
\end{align*}
In fact, we have the following result, which implies \Cref{lem:field-chain} immediately.

\begin{lemma}
  For any $t, h \geq 0$, and $B \subseteq [n]$, it holds that
  \begin{align*}
    \mathbb{P}^\nu_{\emptyset, B}(0 \to t + h) &= \sum_{A \subseteq B} \mathbb{P}^\nu_{\emptyset, A}(0 \to t) \mathbb{P}^\nu_{A, B}(t \to t + h).
  \end{align*}
\end{lemma}

\begin{proof}
  We prove this by a brute force calculation.
  By definition,
  \begin{align*}
    \text{RHS}
    &= \sum_{A \subseteq B} \mathbb{P}^\nu_{\emptyset, A}(0 \to t) \cdot \sum_{Y: Y_B = \*1} (\e^{-t} * \nu)^{\*1_A}(Y) \cdot (1 - \e^{-h})^{\abs{B}- \abs{A}} (\e^{-h})^{\abs{Y}_+ - \abs{B}} \\
    &= \sum_{A \subseteq B}  \sum_{X: X_A = \*1} \nu(X) (1 - \e^{-t})^{\abs{A}} (\e^{-t})^{\abs{X}_+ - \abs{A}} \\
    &\quad \times \sum_{Y: Y_B = \*1} (\e^{-t} * \nu)^{\*1_A}(Y) \cdot (1 - \e^{-h})^{\abs{B}-\abs{A}} (\e^{-h})^{\abs{Y}_+ - \abs{B}}.
  \end{align*}
  Recall that it holds that
  \begin{align*}
    (\e^{-t} * \nu)^{\*1_A}(Y) &= \frac{(\e^{-t})^{\abs{Y}_+} \nu(Y)}{\sum_{Z: Z_A = \*1} (\e^{-t})^{\abs{Z}_+} \nu(Z)}, 
  \end{align*}
  which implies that
  \begin{align*}
    \text{RHS}
    &= \sum_{A \subseteq B}  \sum_{X: X_A = \*1} \nu(X) (1 - \e^{-t})^{\abs{A}} (\e^{-t})^{\abs{X}_+ - \abs{A}} \\
    &\quad \times \sum_{Y: Y_B = \*1} \frac{(\e^{-t})^{\abs{Y}_+} \nu(Y)}{\sum_{Z: Z_A = \*1} (\e^{-t})^{\abs{Z}_+} \nu(Z)} \cdot (1 - \e^{-h})^{\abs{B}- \abs{A}} (\e^{-h})^{\abs{Y}_+ - \abs{B}}.
  \end{align*}
  Now, we change the order to delay all the calculation that involved with $A$,
  \begin{align*}
    \text{RHS}
    &= \sum_{Y:Y_B = \*1} \nu(Y) (\e^{-(t+h)})^{\abs{Y}_+} \\
    &\quad \times \sum_{A\subseteq B} (1 - \e^{-h})^{\abs{B} - \abs{A}} (\e^{-h})^{-\abs{B}} \frac{\sum_{X:X_A = \*1} \nu(X) (1 - \e^{-t})^{\abs{A}} (\e^{-t})^{\abs{X}_+ - \abs{A}}}{\sum_{Z:Z_A = \*1} \nu(Z) (\e^{-t})^{\abs{X}_+}}.
  \end{align*}
  Note that the enumerator and the denominator can be canceled out.
  \begin{align*}
    \text{RHS}
    &= \sum_{Y:Y_B = \*1} \nu(Y) (\e^{-(t+h)})^{\abs{Y}_+} \\
    &\quad \times \sum_{A\subseteq B} (1 - \e^{-h})^{\abs{B} - \abs{A}} (\e^{-h})^{-\abs{B}} (1 - \e^{-t})^{\abs{A}} (\e^{-t})^{-\abs{A}} \\
    &= \sum_{Y:Y_B = \*1} \nu(Y) (\e^{-(t+h)})^{\abs{Y}_+ - \abs{B}} \cdot \sum_{A\subseteq B} (1 - \e^{-h})^{\abs{B} - \abs{A}} (1 - \e^{-t})^{\abs{A}} (\e^{-t})^{\abs{B}-\abs{A}} \\
    &= \sum_{Y:Y_B = \*1} \nu(Y) (\e^{-(t+h)})^{\abs{Y}_+ - \abs{B}} \cdot \sum_{A\subseteq B} (\e^{-t} - \e^{-(t+h)})^{\abs{B} - \abs{A}} (1 - \e^{-t})^{\abs{A}} \\
    &= \sum_{Y:Y_B = \*1} \nu(Y) (\e^{-(t+h)})^{\abs{Y}_+ - \abs{B}} \cdot \tp{(\e^{-t} - \e^{-(t + h)}) + (1 - \e^{-t})}^{\abs{B}} \\
    &= \sum_{Y:Y_B = \*1} \nu(Y) (\e^{-(t+h)})^{\abs{Y}_+ - \abs{B}} \tp{1 - \e^{-(t + h)}}^{\abs{B}}.
  \end{align*}
  On the other hand, by definition, it also holds that
  \begin{align*}
    \text{LHS}
    &= \sum_{Y:Y_B = \*1} \nu(Y) (\e^{-(t+h)})^{\abs{Y}_+ - \abs{B}} \tp{1 - \e^{-(t + h)}}^{\abs{B}}.
  \end{align*}
  This finishes the proof.
\end{proof}

\section{Rapid mixing of the Glauber dynamics}
\label{sec:GD-mu-mix}
In this section, we prove \Cref{thm:GD-mu-mix-intro},
the rapid mixing of the single-site Glauber dynamics $\GD[\mu]$ on the hardcore distribution $\mu$.
As explained in \Cref{sec:main-GD}, \Cref{thm:GD-mu-mix-intro}  is proved by a comparison between the field dynamics on one side and the Glauber dynamics on the entire graph.
Specifically,  \Cref{thm:GD-mu-mix-intro}  follows directly from 
\Cref{lem:sgap-nu} and \Cref{lem:comp-nu-mu}.
In the following, we will prove \Cref{lem:sgap-nu}  in \Cref{sec:sgap-nu}, and prove \Cref{lem:comp-nu-mu} in \Cref{sec:comp-nu-mu}.

Throughout this section, we assume the following setting.
Let $G = ((L, R), E)$ be a bipartite graph with $n = \abs{L}$ vertices and degree bound $\Delta = d + 1\ge 2$ on one side $L$.
Let $\mu$ be the hardcore distribution on $G$ with fugacity $\lambda< \lambda_c(\Delta)$, and let $\nu=\mu_L$.

We consider the single-site Glauber dynamics $\GD[\mu]$ on the hardcore distribution $\mu$ 
and the single-site Glauber dynamics $\GD[\nu]$ on the distribution $\nu=\mu_L$ projected from $\mu$ on one side $L$.

\subsection{Spectral gap of the Glauber dynamics on one side} \label{sec:sgap-nu}
First, we prove \Cref{lem:sgap-nu}. 
Specifically, let $n=|L|$, $C = (1 + \lambda)^\Delta$, and $\delta\in(0,1)$.
Assume that $(\lambda, d)$ is $\delta$-unique. We will show that
\[
\sgap{\GD[\nu]} \geq (22 n)^{-1} \cdot  \tp{\frac{C\cdot \e^9 \Delta \log n}{\lambda}}^{-10^5 C^5/\delta},
\]
where $\sgap{\GD[\nu]}:=1-\lambda_2(\GD[\nu])$ denotes the spectral gap of the Glauber dynamics $\GD[\nu]$,  which is formally defined in \Cref{sec:prelim-MC}.

The proof strategy, as explained in \Cref{sec:main-GD},
is to use the ``field dynamics comparison lemma''  stated as \Cref{lem:var-boost}.
Applying this tool requires two elements: (1) a spectral gap of the Glauber dynamics in a subcritical regime, and (2) the variance decay (spectral gap) for the field dynamics in the uniqueness regime.
We establish them separately in the following.



\paragraph{Spectral gap of the Glauber dynamics in a subcritical regime.}

The following is a corollary to the rapid mixing in a subcritical regime proved in \Cref{lem:GD-good-0}.
\begin{corollary} \label{cor:sgap-good}
  Assume  $\theta^{-1} \geq \frac{C \cdot \e^9 \Delta }{\lambda}\log n$.
  Fix any $S\subseteq L$ and $\tau \in \Omega(\nu_S)$. For the Glauber dynamics $\GD[\pi]$ with stationary distribution $\pi$, where $\pi=(\theta^{-1} * \nu)^\tau$, 
 it holds that
  \begin{align*}
    \sgap{\GD[\pi]} \geq \frac{1}{22 n}.
  \end{align*}
\end{corollary}

\begin{proof}
It is sufficient to prove the same bound for the Glauber dynamics $\GD[\pi']$ on a marginal distribution $\pi':=\pi_{L\setminus S}=(\theta^{-1} * \nu)^\tau_{L\setminus S}$, because the configurations in $\pi$ are pinned on $S$.
And such a bound is readily available by combining \Cref{lem:GD-good-0} and \Cref{lem:mix=var} as follows:
\begin{align*}
  \lambda_2(\GD[\pi'])=\lambda_\star(\GD[\pi'])
  &= \lim_{t \to \infty} \max_{X\in \Omega(\pi')}  \DTV{P^t(X, \cdot)}{\pi'}^{1/t} 
  \leq \lim_{t \to \infty} \tp{2\ell^{- \frac{t}{21\cdot \ell\log \ell} + 1}}^{1/t}  \\
  &= \exp\tp{- \frac{1}{ 21 \ell}}
   \le 1 - \frac{1}{22\ell}\;,
\end{align*}
where $\ell:=|L\setminus S|\le n$, and $\lambda_2(\GD[\pi'])=\lambda_\star(\GD[\pi'])$  since the Glauber dynamics $\GD[\pi']$ has nonnegative spectrum due to to \Cref{lem:GD-psd}.

Since the vertices in $S$ are pinned, $\GD[\pi]$ will not move to other states once it picks a vertex in $S$; while it will move exactly the same as $\GD[\pi']$ if it pick a vertex in $L\setminus S$.
So, $\GD[\pi]$ is exactly the $\frac{\ell}{n}$-lazy version of $\GD[\pi']$.
Thus, by a standard comparison between the Dirichlet form of $\GD[\pi]$ and $\GD[\pi']$,
\begin{align*}
  \sgap{\GD[\pi]} &\ge \sgap{\GD[\pi']} \cdot \frac{\ell}{n} \ge \frac{1}{22 n}. \qedhere
\end{align*}
\end{proof}

\paragraph{Variance decay of the field dynamics from entropic independence.}
Recall the distribution $\overline{\nu}$ over $\{-1, +1\}^L$ defined in \eqref{eq:sign-flipping-nu}: $\overline{\nu}$ is obtained by flipping the signs in $\nu$ as:
$\overline{\nu}(\sigma) := \nu(\*{-1} \odot \sigma)$ for $\sigma \in \{-1, +1\}^L$,
where $\odot$ denotes the entry-wise product of two vectors.

\begin{lemma} \label{lem:sgap-field}
  For any $\theta \in (0, 1)$, it holds that $\sgap{\FD{\overline{\nu}}} \geq  \theta^{10^5 C^5/\delta} $.
\end{lemma}

\begin{proof}
  Let $g: \Omega(\overline{\nu}) \to \^R_{\geq 0}$ be a function with $\E[\overline{\nu}]{g} = 1$.
  Let  $\overline{\pi}$ be a distribution over $\{-1, +1\}^L$ constructed as that $\overline{\pi}(X) := g(X) \overline{\nu}(X)$ for all $X \in \Omega(\overline{\nu})$.
  Note that $\overline{\pi}$ is absolutely continuous with respect to $\overline{\nu}$.
  By \Cref{thm:FD-ent-decay}, whose requirements are ensured by \Cref{lem:verify-K-tame} and \Cref{thm:SI},  let $\zeta := 10^5 C^5/\delta$, it holds that
  \begin{align*}
    \DKL{\overline{\pi} \FDdown{\overline{\nu}}}{\overline{\nu} \FDdown{\overline{\nu}}} \leq (1 - \theta^{\zeta}) \DKL{\overline{\pi}}{\overline{\nu}}.
  \end{align*}
  By some calculation (e.g.~\cite[Equation (20)]{chen2022near}), one can verify that this is equivalent to 
  \begin{align} \label{eq:mag-fac}
    \Ent[\overline{\nu}]{g} \leq \theta^{-\zeta} \cdot \frac{Z_{\rho}}{\theta^{\abs{L}}} \sum_{S\subseteq L} (1 - \theta)^{\abs{S}} \theta^{\abs{L} - \abs{S}} \cdot \rho_S(\*1_S) \cdot \Ent[\rho^{\*1_S}]{g},
  \end{align}
  where  $\rho := \theta * \overline{\nu}$ and $Z_\rho := \sum_{\sigma \in \Omega(\overline{\nu})} \overline{\nu}(\sigma) \theta^{\abs{\sigma}_+}$ is a normalizing factor.
  Previously in~\cite{chen2022optimal}, such inequality as~\eqref{eq:mag-fac} was called a \emph{$\theta$-magnetized block factorization of entropy}.
  By definition of  $\Ent{\cdot}$ in~\Cref{sec:prelim-MC}, it holds that $c \cdot \Ent[\overline{\nu}]{g} = \Ent[\overline{\nu}]{c \cdot g}$ for any constant $c > 0$.
  Hence \eqref{eq:mag-fac} actually holds for any $g$ with $\E[\overline{\nu}]{g} > 0$.

  Now, for any function $f: \Omega(\overline{\nu}) \to \^R$, applying \eqref{eq:mag-fac} and \Cref{lem:ent-to-var}, for all sufficiently small $\epsilon > 0$, for $g = 1 + \epsilon f$, we have that for the variance $\Var{\cdot}$ defined in~\Cref{sec:prelim-MC},
  \begin{align*}
    \frac{\epsilon^2}{2} \Var[\overline{\nu}]{f} \leq \theta^{-\zeta} \cdot \frac{Z_{\rho}}{\theta^{\abs{L}}} \sum_{S\subseteq L} (1 - \theta)^{\abs{S}} \theta^{\abs{L} - \abs{S}} \cdot \rho_S(\*1_S) \cdot \frac{\epsilon^2}{2} \Var[\rho^{\*1_S}]{f} + o(\epsilon^2).
  \end{align*}
  Dividing $\epsilon^2/2$ on both sides and letting $\epsilon \to 0$, we have 
  \begin{align*}
    \Var[\overline{\nu}]{f} \leq \theta^{-\zeta} \cdot \frac{Z_{\rho}}{\theta^{\abs{L}}} \sum_{S\subseteq L} (1 - \theta)^{\abs{S}} \theta^{\abs{L} - \abs{S}} \cdot \rho_S(\*1_S) \cdot \Var[\rho^{\*1_S}]{f}.
  \end{align*}
  By \cite[Lemma 4.1]{chen2021rapid}, this is equivalent to
  \begin{align*}
    \Var[\overline{\nu}]{f} \leq \theta^{-\zeta} \cdot \+E_{\FD{\overline{\nu}}}(f).
  \end{align*}
  The lemma then follows by the Poincar\'{e} inequality \eqref{eq:Poincare-inequality}.
\end{proof}

Now we are ready to prove  \Cref{lem:sgap-nu}.
\begin{proof}[Proof of \Cref{lem:sgap-nu}]
Fix any $S \subseteq L$ and $\tau \in \Omega(\nu_S)$. Note that $(\theta^{-1} * \nu)^\tau$ is equivalent to $(\theta * \overline{\nu})^{\overline{\tau}}$ by flipping the roles of $-1$ and $+1$.
Fix $\theta = \tp{\frac{C\cdot \e^9 \Delta}{\lambda} \log n}^{-1}$. By \Cref{cor:sgap-good}, it holds that $\sgap{\GD[(\theta * \overline{\nu})^{\overline{\tau}}]} \geq 1 / (22 n)$, which holds for any $\overline{\tau} \in \Omega(\overline{\nu}_S)$.
Hence, by \Cref{lem:var-boost} and \Cref{lem:sgap-field}, 
\begin{align*}
  \sgap{\GD[\overline{\nu}]} \geq \frac{1}{22 n}\cdot  \tp{\frac{C\cdot \e^9\Delta \log n}{\lambda}}^{-10^5 C^5/\delta} .
\end{align*}
Finally, note that $\nu$ and $\overline{\nu}$ are isomorphic to each other by flipping the roles of $-1$ and $+1$, which implies that $\sgap{\GD[\nu]} = \sgap{\GD[\overline{\nu}]}$.
This finishes the proof of \Cref{lem:sgap-nu}.
\end{proof}

\subsection{Comparison between the one-side and two-side Glauber dynamics} \label{sec:comp-nu-mu}

Next, we prove \Cref{lem:comp-nu-mu}.
Formally, 
for the spectral gaps of the Glauber dynamics $\GD[\mu]$ on the hardcore distribution $\mu$ and the Glauber dynamics $\GD[\nu]$ on the measure $\nu=\mu_L$ projected from $\mu$ on one side $L$, 
we will prove the following under the assumption $\lambda\leq (1 - \delta)\lambda_c(\Delta)$:
\[
\sgap{\GD[\mu]} \geq \sgap{\GD[\nu]} \cdot \zeta \cdot ((\Delta+1) n)^{-1},
\]
where  $n=|L|$, and 
\begin{align}
\zeta
=
\begin{cases}
50^{-400/\delta}
& \text{if }\Delta \geq 3\\
(9 \cdot 4^7(1+\lambda)^8)^{-1},
& \text{if }\Delta = 2.
\end{cases}\label{eq:constant-zeta}
\end{align}

As remarked in \Cref{sec:main-GD}, directly comparing $\GD[\nu]$ are $\GD[\mu]$ is difficult, 
because they run on different state space and have different stationary distributions. 
So we introduce the following block dynamics $\PB$ on state space $\Omega(\mu)$ with stationary distribution $\mu$, as a proxy of comparison.
Specifically, let $\PB$ be a Markov chain $(X_t)_{t\in\^N}$ on space $\Omega(\mu)$. In the $t$-th transition, it does:
\begin{enumerate}
\item pick a vertex $v\in L$ uniformly at random;
\item sample $X_t \sim \mu(\cdot\mid {X_{L\setminus \{v\}}})$.
\end{enumerate}
This chain $\PB$ is just a block version of the Glauber dynamics, where the block for each update is $R\cup\{v\}$ for a uniform random $v\in L$.
Therefore, it is easy to check that $\PB$ is irreducible and aperiodic, and is reversible with respect to the stationary distribution $\mu$.

Furthermore, $\PB$ also has nonnegative spectrum.
\begin{lemma} \label{lem:PB-psd}
  $\PB$ only has non-negative eigenvalues.
\end{lemma}
\begin{proof}
  Let $\+S := \{(v, Y) \mid v \in L, Y \in \Omega(\mu_{L\setminus\{v\}})\}$.
 Let $\PBdown \in \^R^{\Omega(\mu) \times \+S},\PBup \in \^R^{\+S \times \Omega(\mu)}$ be defined by
  \begin{align*}
    \forall X \in \Omega(\mu), (v, Y) \in \+S: & \quad \PBdown(X, (v, Y)) = \frac{1}{n} \cdot \*1[Y = X_{L \setminus \{v\}}],\\
    \forall (v, X) \in \+S, Y \in \Omega(\mu): & \quad \PBup((v, X), Y) = \mu^{X}(Y).
  \end{align*}
  It is easy to see that $\PB = \PBdown\PBup$.
  Let $\mu_0 = \mu \PBdown$. It holds that for all $(v, Y) \in \+S$ and $X \in \Omega(\mu)$, 
  \begin{align} 
    \mu_0((v, Y)) \PBup((v, Y), X)
    \nonumber &= \frac{1}{n} \sum_{Z: Z_{L\setminus \{v\}} = Y} \mu(Z) \cdot \mu^{Y}(X) \\
    \nonumber &= \frac{1}{n} \cdot \mu(X) \cdot \*1[X_{L\setminus \{v\}} = Y] \\
    \label{eq:reversible} &= \mu(X) \PBdown(X, (v, Y)).
  \end{align}
  And for any $f: \Omega(\mu) \to \^R$ and $g: S \to \^R$, it holds that
  \begin{align*}
    \inner{f}{\PBdown\; g}_\mu
    &= \sum_{X \in \Omega(\mu)} \mu(X) f(X) [\PBdown g](X)
      = \sum_{X \in \Omega(\mu), (v,Y)\in S} \mu(X) f(X) \PBdown(X, (v,Y)) g((v, Y)) \\
    &\overset{\eqref{eq:reversible}}{=} \sum_{X \in \Omega(\mu), (v,Y)\in S} f(X) \mu_0((v, Y) \PBup((v, Y), X)  g((v, Y)) \\
    &= \sum_{(v, Y) \in S} \mu_0((v, Y)) [\PBup f]((v, Y)) g((v, Y))
     = \inner{\PBup\; f}{g}_{\mu_0},
  \end{align*}
  which means $\PB$ is a self-adjoint operator with respect to the inner product $\inner{\cdot}{\cdot}_{\mu}$.
  Moreover,  for all $f: \Omega(\mu) \to \^R$, it holds that
  \begin{align*}
    \inner{f}{\PB f}_\mu = \inner{f}{\PBdown\PBup f}_\mu = \inner{\PBup f}{\PBup f}_{\mu_0} \geq 0,
  \end{align*}
  which implies that $\PB$ only has nonnegative eigenvalues.
\end{proof}

In the rest, we will prove:
\begin{itemize}
\item 
$\sgap{\PB}\ge \sgap{\GD[\nu]}$, which follows from a coupling between the chains $\PB$ and $\GD[\nu]$;
\item
$\sgap{\GD[\mu]} \geq \sgap{\PB} \cdot \zeta \cdot ((\Delta+1) n)^{-1}$,
where $\zeta$ is specified in \eqref{eq:constant-zeta},
which is proved by a comparison of the spectral gaps between $\PB$ and $\GD[\mu]$ through a {block factorization} of variance.
\end{itemize}
\Cref{lem:comp-nu-mu} then follows by directly combining these two bounds for spectral gaps.

In the rest of this section, we prove the above two bounds one-by-one.

\paragraph{Compare $\PB$ and $\GD[\nu]$.}
First, we prove the following result.
\begin{lemma} \label{lem:sgap-PB}
  $\sgap{\PB} \geq \sgap{\GD[\nu]}$.
\end{lemma}
\begin{proof}
  Let $(X_t)_{t\in \^N}$ be the chain generated by $\PB$  and $(Y_t)_{t\in \^N}$ be the chain generated by $P=\GD[\nu]$.
  Given that $Y_0 = (X_0)_L$, then by a natural coupling between  $(X_t)_{t\in \^N}$ and $(Y_t)_{t\in \^N}$, we can show that
  \begin{align}
    \forall t \geq 1, \quad \DTV{\PB^t(X_0, \cdot)}{\mu} &\leq \DTV{P^t(Y_0, \cdot)}{\nu}.\label{eq:dtv-PB}
  \end{align}
  To see this, consider the following coupling between $(X_t)_{t\in \^N}$ and $(Y_t)_{t \in \^N}$: 
  in the  $t$-th step,  
  \begin{enumerate}
  \item pick $v\in L$ \emph{u.a.r.}~and let  $X_t(u) = X_{t-1}(u)$ and $Y_t(u) = Y_{t-1}(u)$ for all $u\in L\setminus\{v\}$;
  \item sample $X_t(v) = Y_t(v) \sim \mu^{X_{t-1}(L\setminus \{v\})}_v$ and  $X_t(R) \sim \mu^{Y_{t}}_R$.
  \end{enumerate}
  It is easy to see that this gives a valid coupling of $(X_t)_{t\in \^N}$ and $(Y_t)_{t \in \^N}$, and $Y_t = (X_t)_L$ for all $t \geq 1$.
  Then \eqref{eq:dtv-PB} follows by \Cref{lem:nu-to-mu}.
%
  Furthermore, by \eqref{eq:dtv-PB} and \Cref{lem:mix=var}, it holds that
  \begin{align*}
    \lambda_\star(\PB)
    &= \lim_{t\to\infty} \max_{X_0 \in \Omega(\mu)} \DTV{\PB^t(X_0, \cdot)}{\mu}^{1/t} \\
    &\leq \lim_{t\to\infty} \max_{Y_0 \in \Omega(\nu)} \DTV{(\GD[\nu])^t(Y_0, \cdot)}{\nu}^{1/t} 
    = \lambda_\star(\GD[\nu]).
  \end{align*}
  Finally, by \Cref{lem:GD-psd} and \Cref{lem:PB-psd}, we have $\lambda_\star(\PB) = \lambda(\PB)$ and $\lambda_\star(\GD[\nu]) = \lambda(\GD[\nu])$.
\end{proof}


\paragraph{Compare $\PB$ and $\GD[\mu]$.}
In this part, we prove the following result for comparing $\PB$ to $\GD[\mu]$.
\begin{lemma}\label{lem:sgap-PB-GD}
$\sgap{\GD[\mu]} \geq \sgap{\PB} \cdot \zeta \cdot \abs{L\cup R}^{-1}$,
where $\zeta$ is specified in \eqref{eq:constant-zeta}.
\end{lemma}
Then \Cref{lem:comp-nu-mu} immediately follows from \Cref{lem:sgap-PB} and \Cref{lem:sgap-PB-GD}.

We start by introducing a few abstract notations to be used in the proof.
Let $\pi$ be a distribution supported on a finite set $\Omega(\pi) \subseteq \{-1, +1\}^U$, where $U$ is some ground set.
Any random variable on the sample space $\Omega(\pi)$ can be represented as a function $f: \Omega(\pi) \to \^R$.
Let $\tau \sim \pi$. We use $\Var[\pi]{f} = \Var{f(\tau)}$ to denote the variance of the random variable $f(\tau)$.
For $S \subseteq U$, we use the following notation:
\begin{align*}
  \pi[\Var[S]{f}] := \E{\Var{f(\tau) \mid \tau_{U\setminus S}}} = \E[\tau \sim \pi]{\Var[\pi^{\tau_{U\setminus S}}]{f}}.
\end{align*}
In other worlds, $\pi[\cdot]$ takes expectation with respect to $\pi$ and $\Var[S]{f}: \Omega(\pi) \to \^R$ is a function defined by 
\begin{align*}
  \forall \tau \in \Omega(\pi), \quad \Var[S]{f}\; (\tau) := \Var[\pi^{\tau_{U\setminus S}}]{f}.
\end{align*}
More generally, for $T \subseteq U$, $\sigma \in \Omega(\pi_T)$, and $S\subseteq U \setminus T$, we use the notation
\begin{align*}
  \pi^{\sigma}[\Var[S]{f}] := \E[\tau \sim \pi^\sigma]{\Var[\pi^{\tau_{U\setminus S}}]{f}}.
\end{align*}
We note that $\pi^\sigma[\Var[S]{f}]$ can also be seen as a  function in $\Omega(\pi) \to \^R$, which reads a configuration $\rho \in \Omega(\pi)$ and fixes $\sigma = \rho_{T}$.
That is, for $\rho \in \Omega(\pi)$, we have $\pi^\sigma[\Var[S]{f}] (\rho) = \pi^{\rho_T}[\Var[S]{f}]$.

The following can thus be verified
\begin{align*}
  \pi[\pi^\sigma[\Var[S]{f}]] = \pi[\Var[S]{f}].
\end{align*}

The comparison is done via a special form of the Poincar{\'e} inequality, described using the notation that we have introduced above.
This process is also known as the \emph{approximate block factorization} of variance~\cite{caputo2015approximate, caputo2021block, chen2021optimal}.

\begin{lemma} \label{lem:re-dirichlet}
  Let $\pi$ be a distribution over $\{-1, +1\}^n$ and $\+D$ be a distribution over $2^{[n]}$.
  Let $(X_t)_{t \in \^N}$ be a Markov chain with stationary $\pi$ 
  and in its $t$-th transition, it does the followings:
  \begin{enumerate}
  \item sample $S \sim \+D$;
  \item sample $X_t \sim \pi$ conditioning on that $X_t([n]\setminus S)={X_{t-1}([n]\setminus S)}$.
  \end{enumerate}
  Let $P$ be the transition matrix of this Markov chain, then it holds that
  \begin{align*}
    \forall f: \Omega(\pi) \to \^R, \quad \+E_P(f) &= \E[S\sim \+D]{\pi[\Var[S]{f}]},
  \end{align*}
  where $\+E_P(f)$ is the Dirichlet form defined in \Cref{sec:prelim-MC}.
\end{lemma}
\begin{proof}
  This abstract lemma follows from the following straightforward calculation:
  \begin{align*}
    \+E_P(f)
    &= \frac{1}{2} \sum_{X, Y \in \Omega(\pi)} \pi(X) P(X, Y) (f(X) - f(Y))^2 \\
    &= \frac{1}{2} \E[S\sim \+D]{\sum_{X \in \Omega(\pi)} \sum_{Y \in \Omega(\pi^{X([n]\setminus S)})} \pi(X) \pi^{X([n]\setminus S)}(Y) (f(X) - f(Y))^2} \\
    &= \frac{1}{2} \E[S\sim \+D]{\sum_{Z\in \Omega(\pi_{[n]\setminus S})} \pi_{[n]\setminus S}(Z) \sum_{X \in \Omega(\pi)} \sum_{Y \in \Omega(\pi^{X([n]\setminus S)})} \pi^{Z([n]\setminus S)}(X) \pi^{Z([n]\setminus S)}(Y) (f(X) - f(Y))^2} \\
    &= \frac{1}{2} \E[S\sim \+D]{\E[Z \sim \pi]{\sum_{X, Y \in \Omega(\pi^{Z([n]\setminus S)})} \pi^{Z([n]\setminus S)}(X) \pi^{Z([n]\setminus S)}(Y) (f(X) - f(Y))^2}} \\
    &= \E[S\sim \+D]{\E[Z \sim \pi]{\Var{f(Z) \mid Z_{[n]\setminus S}}}} \\
    &= \E[S\sim \+D]{\pi[\Var[S]{f}]} \qedhere.
  \end{align*}
\end{proof}

Now, consider the example of $\PB$. 
We apply \Cref{lem:re-dirichlet} to the chain $\PB$, with $U = L \cup R$, $\pi=\mu$ being the hardcore distribution,  and $\+D$ being the distribution that generates $v\cup R$ for a uniformly random $v \in L$.
Then the Poincar{\'e} inequality \eqref{eq:Poincare-inequality} for $\PB$ gives
\begin{align} \label{eq:var-PB}
  \Var[\mu]{f} &\le  \sgap{\PB}^{-1} \+E_P(f)= \sgap{\PB}^{-1} \abs{L}^{-1}  \cdot \sum_{v \in L} \mu[\Var[v \cup R]{f}],
\end{align}
where we write $v\cup R=\{v\}\cup R$ for convenience.

Since $G=((L, R), E)$ is a bipartite graph, the induced subgraph $G[v\cup R]$ is almost an empty graph with isolated vertices, 
with the only exceptions being $v$ and its neighbors.
Hence, the factorization of variance on product distribution can apply (see~\cite{cesi2001quasi, caputo2015approximate, chen2021optimal}).

\begin{lemma}[\text{\cite{caputo2015approximate, chen2021optimal}}] \label{lem:var-fac-prod}
  For every subset $S \subseteq U$, every boundary condition $\tau \in \Omega(\mu_{V\setminus S})$, and every function $f: \Omega(\mu) \to \^R$, we have
  \begin{align*}
    \Var[\mu^\tau]{f} \leq \sum_{M \in \+C(S)} \mu^\tau[\Var[M]{f}],
  \end{align*}
  where $\+C(S)$ is the family of all the connected components in $G[S]$.
\end{lemma}

For any $v \in L$ and $\tau \in \Omega(\mu_{L\setminus \{v\}})$, by \Cref{lem:var-fac-prod}, we have
\begin{align} \label{eq:var-ind}
  \Var[\mu^\tau]{f}
  &\leq \mu^\tau[\Var[v\cup \Gamma_v]{f}] + \sum_{u \in R\setminus \Gamma_v} \mu^\tau[\Var[u]{f}].
\end{align}

Due to this factorization, in our instance $G[v\cup R]$,
it is sufficient to bound the $\Var[\mu^\sigma]{f}$ for $v \in L$ and $\sigma \in \Omega(\mu_{U\setminus (v\cup \Gamma_v)})$.
Note that $\mu^\sigma_{v\cup \Gamma_v}$ is a hardcore model on a star graph of at most $\Delta + 1$ vertices with maximum degree $\Delta$,
whose spectral gap was already known, assuming $\lambda\le (1-\delta)\lambda_c(\Delta)$.

\begin{lemma}[\text{\cite{chen2021rapid}}] \label{lem:sgap-star}
Let $v \in L$ and $\sigma \in \Omega(\mu_{U\setminus (v\cup \Gamma_v)})$.
  Let $P$ be the Glauber dynamics on $\mu^\sigma_{v \cup \Gamma_v}$.
  If $\lambda\le (1-\delta)\lambda_c(\Delta)$, then $\sgap{P}\ge\zeta/(\Delta + 1)$, where $\zeta$ is specified in \eqref{eq:constant-zeta}.
\end{lemma}
\begin{proof}
Note that $P$ is the Glauber dynamics for a hardcore model on a star graph with $\Delta+1$ vertices.
  The lower bound of the spectral gap $\gamma(P)$ for the $\Delta \geq 3$ case follows directly from \cite[Theorem 1.3]{chen2021rapid}.
  It remains to calculate the $\Delta = 2$ case.
  When $\lambda \leq 1/(2\Delta) = 1/4$, there is a coupling for $P$ that decays step-wise, which implies that $\sgap{P} \geq 1/10$ (see \cite{chen1998trilogy} and \cite[Chapter 13.1]{levin2017markov}).
  When $\lambda \geq 1/4$, by the Cheeger's inequality (see \cite[Theorem 13.10]{levin2017markov}), it holds that
  \begin{align*}
    \sgap{P} \geq \frac{\Phi^2_\star}{2} \geq \frac{1}{2} \tp{\frac{(\mu^\sigma_{v\cup \Gamma_v})_{\min} \cdot \frac{1}{3} \min\{\frac{1}{1+\lambda}, \frac{\lambda}{1 + \lambda}\}}{1/2}}^{2} \geq \frac{1}{9 \cdot 4^7(1 + \lambda)^8},
  \end{align*}
  where $\Phi_\star$ is the conductance (a.k.a. the bottleneck ratio) and we use the fact that $\lambda \geq 1/4$ and $\abs{v \cup \Gamma_v} \leq 3$.
\end{proof}

By \Cref{lem:sgap-star} and \Cref{lem:re-dirichlet},  the Poincar{\'e} inequality for the Glauber dynamics on $\mu^\sigma_{v \cup \Gamma_v}$ can be expressed as follows. For any $g: \Omega(\mu^\sigma_{v \cup \Gamma_v}) \to \^R$, 
\begin{align} \label{eq:var-star-pre}
  \Var[\mu^\sigma_{v\cup \Gamma_v}]{g} \leq \gamma(P)^{-1}(\Delta+1)^{-1} \cdot \sum_{u \in v\cup \Gamma_v} \mu^\sigma_{v\cup \Gamma_v}[\Var[u]{g}]\le \zeta^{-1}\cdot\sum_{u \in v\cup \Gamma_v} \mu^\sigma_{v\cup \Gamma_v}[\Var[u]{g}],
\end{align}
which implies that for any $f: \Omega(\mu) \to \^R$,
\begin{align} \label{eq:var-star}
  \Var[\mu^\sigma]{f} \leq \zeta^{-1} \cdot \sum_{u \in v\cup \Gamma_v} \mu^\sigma[\Var[u]{f}].
\end{align}
This is because in \eqref{eq:var-star}, only the values of $f(X)$ with $X(U\setminus(v\cup \Gamma_v)) = \sigma$ are used.
Hence, for every $f: \Omega(\mu) \to \^R$ in \eqref{eq:var-star}, we can define $g: \Omega(\mu^\sigma_{v\cup \Gamma_v}) \to \^R$ in \eqref{eq:var-star-pre} as $g(\tau) := f(\tau \uplus \sigma)$ for all $\tau \in \Omega(\mu^\sigma_{v\cup\Gamma_v})$,
through which \eqref{eq:var-star} and \eqref{eq:var-star-pre} become the same.

By \eqref{eq:var-ind} and \eqref{eq:var-star},  for every $v \in L$ and $\tau \in \Omega(\mu_{U\setminus(v \cup \Gamma_v)})$, we have
\begin{align}
  \Var[\mu^\tau]{f}
  \nonumber &\leq \zeta^{-1}  \sum_{w\in v\cup \Gamma_v} \mu^\tau[\Var[w]{f}] + \sum_{u \in R\setminus \Gamma_v} \mu^\tau[\Var[u]{f}] \\
  \label{eq:var-block} &\leq \zeta^{-1} \sum_{u \in v \cup R} \mu^\tau[\Var[u]{f}],
\end{align}
where we use the fact that $\zeta^{-1} \geq 1$.

Finally, combining \eqref{eq:var-PB} and \eqref{eq:var-block}, we have
\begin{align*}
  \Var[\mu]{f}
  &\leq \sgap{\PB}^{-1}\abs{L}^{-1} \cdot \zeta^{-1} \cdot \sum_{v \in L} \sum_{u \in v \cup R} \mu[\Var[u]{f}] \\
  &= \sgap{\PB}^{-1}\abs{L}^{-1} \cdot \zeta^{-1} \cdot \tp{\sum_{v \in L} \mu[\Var[v]{f}] + \sum_{u \in R} \abs{L} \cdot \mu[\Var[u]{f}]} \\
  &\leq  \sgap{\PB}^{-1} \cdot \zeta^{-1}\sum_{v \in L\cup R} \mu[\Var[v]{f}],
\end{align*}
where in the last inequality, we use the fact that $\deg_G(u) \leq \abs{L}$ for all $u\in R$.
Finally, by applying \Cref{lem:re-dirichlet} to the chain $\GD[\mu]$, we have 
\begin{align*}
  \sgap{\GD[\mu]} &\geq \sgap{\PB} \cdot \zeta \cdot \abs{L\cup R}^{-1}.
\end{align*}
This proves \Cref{lem:sgap-PB-GD}.

\bibliographystyle{alpha}
\bibliography{refs}

\appendix

\section{The reduction to the Ising model for $\Delta = 2$ and an alternative sampler}
\label{sec:ising-case}
Let $G = ((L, R), E)$ be a bipartite graph with maximum degree $\Delta = 2$ on $L$, $n$ be the number of vertices in $L$ and $\lambda > 0$ be the fugacity.
We show that such a bipartite hardcore model can be reduced to an Ising model, with consistent local fields.
As a result, one can obtain alternative samplers for $\Delta=2$ from samplers for the Ising model.

Since $\Delta = 2$, we can build a new graph $H = (R, F)$, where 
\[
F := \{(u, v) \mid u, v \in R \text{ and } \abs{\Gamma_u \cap \Gamma_v} > 0\}.
\]
Let $\*\gamma = \*1_F$ and let $\*\beta \in \^R^F$ and $\*\rho \in \^R^R$ be two vectors defined as follows.
For $e = (u, v) \in F$, let $j_e := \abs{\Gamma_u \cap \Gamma_v}$ be the number of common neighbors of $u$ and $v$, and $\beta_e$ is defined as $\beta_e := (1 + \lambda)^{j_e}$.
For $u \in R$, let $k_u := \abs{\{v \in \Gamma_u \mid \Gamma_v = \{u\}\}}$, and $\rho_u$ is defined as $\rho_u := (1 + \lambda)^{k_u} / \lambda$.

Consider the two-spin system defined by $H$ and $\*\beta, \*\gamma, \*\lambda$:
\begin{align} \label{eq:wt-R}
  \forall \sigma \in \{-1, +1\}^R, \quad \-{wt}(\sigma) := \prod_{e \in m_{-}(\sigma)} \beta_e \prod_{e \in m_+(\sigma)} \gamma_e \prod_{v: \sigma_v = -1} \rho_v,
\end{align}
where we use $m_{\pm}(\sigma) := \{e = (u, v) \in F \mid \sigma_u = \sigma_v = \pm 1\}$ to denote the set of $\pm 1$-monochromatic edges.
We define a Gibbs distribution $\nu$ as $\nu(\sigma) = \-{wt}(\sigma) / Z$ where $Z = \sum_\sigma \-{wt}(\sigma)$ is the partition function.
We claim that $\nu$ is exactly the same distribution as $\mu_R$.
This can be verified by noticing that, once we fix a configuration $\sigma$ on $R$ in $\mu$, the spins on $L$ become independent and their contribution can be counted as in \eqref{eq:wt-R}, 
up to a normalizing factor of $\lambda^{\abs{R}}$.

Now, we will use a standard holographic transformation~\cite{goldberg2003computational} to show that the two-spin system encoded by $H$, and $\*\beta, \*\gamma, \*\lambda$ is equivalent to a ferromagnetic Ising model on $H$.
Given an edge $f = (u, v) \in F$, it holds that $\forall \sigma \in \{-1, +1\}^R$,
\begin{align*}
  \-{wt}(\sigma)
  &= \beta_e^{\*1[\sigma_u = \sigma_v = -1]} \gamma_e^{\*1[\sigma_u = \sigma_v = +1]} \rho_u^{\*1[\sigma_u = -1]} \rho_v^{\*1[\sigma_v = -1]} \\
  &\hspace{1cm} \times \prod_{e \in m_{-}(\sigma) \setminus f} \beta_e \prod_{e \in m_+(\sigma) \setminus f} \gamma_e \prod_{\substack{w: \sigma_w = -1 \\ w \not\in f}} \rho_w \\
  &= \sqrt{\frac{\gamma_e}{\beta_e}} \times \tp{\sqrt{\beta_e\gamma_e}}^{\*1[\sigma_u = \sigma_v]} \times \tp{\rho_u \sqrt{\frac{\beta_e}{\gamma_e}}}^{\*1[\sigma_u = -1]} \times \tp{\rho_v \sqrt{\frac{\beta_e}{\gamma_e}}}^{\*1[\sigma_v = -1]} \\
  &\hspace{1cm} \times \prod_{e \in m_{-}(\sigma) \setminus f} \beta_e \prod_{e \in m_+(\sigma) \setminus f} \gamma_e \prod_{\substack{w: \sigma_w = -1 \\ w \not\in f}} \rho_w,
\end{align*}
where we note that the last equation can be verified by a brute force enumeration of the configuration of $\sigma_u, \sigma_v$.
If we keep doing this for all the edges $f \in F$ and normalize the factor $\prod_e \sqrt{\gamma_e / \beta_e}$, we will finally get an Ising model with
\begin{align*}
  \forall \sigma \in \{-1, +1\}^R, \quad \-{wt}(\sigma) = \prod_{e \in m(\sigma)} \beta_e^\star \prod_{v: \sigma_v = -1} \lambda_v^\star,
\end{align*}
where $m(\sigma) := \{e = (u, v) \in F \mid \sigma_u = \sigma_v \}$ denotes the set of monochromatic edges, and we define vectors $\*\beta^\star \in \^R^F$ and $\*\lambda^\star \in \^R^R$ as follows
\begin{align*}
  \forall e \in F,& \quad \beta_e^\star = \sqrt{\beta_e \gamma_e} = (1 + \lambda)^{j_e/2} \\
  \forall v \in R,& \quad \lambda_v^\star = \rho_v \cdot \prod_{e \ni v} \sqrt{\frac{\beta_e}{\gamma_e}} = (1 + \lambda)^{k_v + \sum_{e\ni v} j_e/2} / \lambda.
\end{align*}

To apply any existing sampler for the Ising model, we want consistent local fields, that is, $\forall v \in R, \lambda^\star_v \geq 1$.
However, this condition can fail for $v \in R$ with $k_v = 0$ and $\abs{\Gamma_v(H)} \leq 1$, where we use $\Gamma_v(H)$ to denote the set of neighbors of $v$ in $H$.
We handle these special vertices as follows:
\begin{itemize}
\item 
If $\abs{\Gamma_v(H)} = 0$, then $v$ is an isolated vertex and we can sample $v$ independently from other vertices in the Ising model.

\item 
If $\abs{\Gamma_v(H)} = 1$, suppose $v \in e$, then $\lambda^\star_v < 1$ only happens when $j_e = 1$.
Let $e = (u, v)$, in this case, we can remove $v$ from $H$ and modifies $\lambda^\star_u$ to $\lambda^\star_u \cdot (1 + 2\lambda)$.
In order to sample from the original Ising model, it is sufficient for us to sample from the new Ising model obtained by removing $v$, then sample $v$ according to the marginal distribution on $v$.
\end{itemize}
Finally, the problem is reduced to sampling from an Ising model encoded by parameters $\*\beta^\star$ and $\*\lambda^\star$ with $\min_e \beta^\star_v \geq (1 + \lambda)^{1/2}$ and $\min_v \lambda^\star_v \geq \frac{1 + \lambda}{\lambda}$.
According to the recent developments on Ising samplers~\cite{chen2022near}, there is a fast sampler for such a Ising model that runs in time
\begin{align*}
  n \cdot \tp{(1/\epsilon)^{1/\log n} \cdot \log n}^{O_\lambda(1)}.
\end{align*}

\section{Spectral independence via contraction}
\label{sec:a-c-SI}

We prove \Cref{lem:a-c-SI}, the spectral independence bound implied by contraction and boundedness conditions, that is suitable for the bipartite hardcore model. 
This is proved by following the same route as in \cite{chen2020rapid} for the spin models.
The reason that we have to go through their proof instead to applying their conclusions directly is because 
we need a different treatment of the boundedness at the root to deal with the unbounded degrees on one side.

Fix $\Lambda \subseteq L$ and $\sigma \in \Omega(\nu^\sigma)$.
Note that for all $i, j \in L$, it holds that
\begin{align*}
  \Psi_{\nu^\sigma}(i, j)
  &= \Pr[\nu^\sigma]{j \mid i} - \Pr[\nu^\sigma]{j \mid \overline{i}} 
   = \Pr[\mu^\sigma]{j \mid i} - \Pr[\mu^\sigma]{j \mid \overline{i}}.
\end{align*}
Therefore, in order to prove \Cref{lem:a-c-SI}, it is sufficient to fix an arbitrary $u \in L \setminus \Lambda$ with $\Pr[\nu^\sigma]{u} \in (0, 1)$ and bound
\begin{align} \label{eq:inf-target}
  \sum_{v \in L} \abs{\Psi_{\nu^\sigma}(u, v)} = \sum_{v \in L} \abs{\Psi_{\mu^\sigma}(u, v)} \leq \frac{c}{\beta}.
\end{align}
Note that $\mu^\sigma$ is just the hardcore distribution on a smaller graph $G'$ with the same fugacity.
Then by the self-reducibility of hardcore model, we can focus on the case where $\Lambda = \emptyset$ and $\mu^\sigma = \mu$.

\newcommand{\Tsaw}{T_{\-{SAW}}}
\newcommand{\Vsaw}{V_{\-{SAW}}}
\newcommand{\Esaw}{E_{\-{SAW}}}

  In \cite{chen2020rapid}, there is a general result for two-spin systems showing that bounding the total influence on a general bipartite graph can be reduced to bounding the total influence on a tree.
  We specialize that result to the bipartite graph hardcore model, and get the following result.
  \begin{lemma}[\text{\cite[Lemma 8]{chen2020rapid}}] \label{lem:saw-tree-inf}
	  Let $G = (L\cup R, E)$ be a connected bipartite graph with $\lambda$ being the fugacity on $L$, $\alpha$ being the fugacity on $R$, and a vertex $r \in L$.
    There is a self avoiding walk (SAW) tree $\Tsaw = \Tsaw(G, r) = (\Vsaw, \Esaw)$ with root $r \in \Vsaw$, a pinning $\tau_{\-{SAW}}$ over a subset $\Lambda_{\-{SAW}} \subseteq \Vsaw$ and a map $g: \Vsaw \to L\cup R$ such that the following holds.
    \begin{itemize}
	    \item 
    the fugacity for vertices of $\Tsaw$ on even depth is $\lambda$ and the fugacity for vertices on odd depth is $\alpha$;
    \item for every $v \in \Vsaw$, if $v$ is in even depth of $\Tsaw$, then $g(v) \in L$; otherwise, $g(v) \in R$;
    \item for every $v \in \Vsaw \setminus \Lambda_{\-{SAW}}$, $g(v)$ and $v$ have the same degree in $G$ and $\Tsaw$;
    \item let $\mu$ and $\pi$ be the hardcore distributions encoded by $G$ and $\Tsaw$, respectively, then
      \begin{align*}
        \Psi_\mu(r, u) = \sum_{w \in \Vsaw \setminus \Lambda_{\-{SAW}}: g(w) = u} \Psi_{\pi^{\tau_{\-{SAW}}}} (r, w).
      \end{align*}
    \end{itemize}
  \end{lemma}
  Thanks to \Cref{lem:saw-tree-inf} and the self-reducibility of hardcore model, in order to prove \eqref{eq:inf-target}, it is sufficient to show the following inequality for every tree $T$ rooted at $r$ with fugacity $\lambda$ on its even depths, fugacity $\alpha$ on odd depths, and degree bound $\Delta$ on its even depths:
  \begin{align*}
    \sum_{\text{even $k$}} \sum_{w \in L_r(k)} \abs{\Psi_{\mu} (r, w)} \leq \frac{c}{\beta},
  \end{align*}
  where $\mu$ is the hardcore distribution on $T$, and we use $L_r(k)$ to denote the set of vertices in $T_r$ at depth $k$ such that $T_r$ is the subtree of $T$ rooted at $r$.

We prove this by showing that for any fixed integer $k \geq 0$, 
\begin{align} \label{eq:exp-decay}
  \sum_{w \in L_r(2k)} \abs{\Psi_{\mu} (r, w)} \leq c \cdot (1 - \beta)^k.
\end{align}


	We need the following results towards the hardcore distribution $\mu$ on the tree $T$.

\begin{lemma}[\text{\cite[Lemma B.2]{chen2020rapid}}] \label{lem:inf-chain}
  Let $u, v, w$ be distinct vertices in tree $T$, such that $v$ is on the unique path from $u$ to $w$, it holds that
  \begin{align*}
    \Psi_\mu(u, w) = \Psi_\mu(u, v) \cdot \Psi_\mu(v, w).
  \end{align*}
\end{lemma}

Recall that the tree recursion for the hardcore model on the root $r$ with $d$ children is given by
\begin{align*}
  f(\*R) = \lambda \prod_{i=1}^d (1 + R_{v_i})^{-1},
\end{align*}
where $v_i$ is the $i$-th child of $r$, $\nu$ is the hardcore distribution on $T_{v_i}$, and $R_{v_i} := \frac{\nu_{v_i}(+1)}{\nu_{v_i}(+1)}$ is the marginal ratio of $v_i$ in $T_{v_i}$.
\begin{fact}[\text{\cite[Lemma 16]{chen2020rapid}}] \label{lem:inf-quant}
  Let $u, v$ be vertices in tree $T$ such that $v$ is a child of $u$, then it holds that
  \begin{align*}
    \Psi_\mu(u, v) = - \frac{R_u}{1 + R_u}.
  \end{align*}
\end{fact}

Now, we are ready to prove \Cref{lem:a-c-SI}.

\begin{proof}[Proof of \Cref{lem:a-c-SI}]
As we discussed, it suffices to prove \eqref{eq:exp-decay}.
  First, we prove that for any non-root vertex $u$ in $T$ and integer $k \geq 0$,
  \begin{align} \label{eq:contract}
    \psi(\log R_u) \sum_{v \in L_u(2k)} \abs{\Psi_\mu(u, v)} \leq \max_{v \in L_u(2k)} \{\psi(\log R_v)\} \cdot (1 - \beta)^k.
  \end{align}
  We prove this by induction on $k$.
  The basis is $k = 0$. 
  Then $\psi(\log R_u) \abs{\Psi_\mu(u, u)} \leq \psi(\log R_u)$ holds trivially, since $\Psi_\mu(u, u) = 1$ by definition.
  
  Now assume that \eqref{eq:contract} holds for all smaller $k$'s.
  For the induction step with general $k > 0$, 
  suppose that $u$ have $d_u \leq \Delta - 1$ children $\{v_1,  \cdots, v_{d_u}\}$, and each child $v_i$ has $w_i$ children $\{v_{i1},  \cdots, v_{iw_i}\}$.
  \begin{align*}
    &\psi(\log R_u) \sum_{v \in L_u(2k)} \abs{\Psi_\mu(u, v)} \\
    (\text{\small \Cref{lem:inf-chain}})  
    =&\psi(\log R_u) \sum_{k=1}^{d_u} \sum_{h=1}^{w_k} \abs{\Psi_\mu(u, v_{kh})} \cdot \sum_{v \in L_{v_{kh}}(2(k-1))} \abs{\Psi_\mu(v_{kh}, v)} \\
    =& \psi(\log R_u) \sum_{k=1}^{d_u} \sum_{h=1}^{w_k} \abs{\Psi_\mu(u, v_{kh})} \frac{1}{\psi(\log R_{v_{kh}})} \\
    & \hspace{2cm} \times \psi(\log R_{v_{kh}}) \sum_{v \in L_{v_{kh}}(2(k-1))} \abs{\Psi_\mu(v_{kh}, v)} \\
    \tp{\substack{\text{induction} \\ \text{hypothesis}}} 
   \leq& \max_{v \in L_u(2k)} \{\psi(\log R_v)\} (1 - \beta)^{k-1} \cdot \sum_{k=1}^{d_u} \sum_{h=1}^{w_k} \abs{\Psi_\mu(u, v_{kh})} \frac{\psi(\log R_u)}{\psi(\log R_{v_{kh}})} \\
    \tp{\substack{\text{\Cref{lem:inf-chain}} \\ \text{\Cref{lem:inf-quant}}}} 
    =& \max_{v \in L_u(2k)} \{\psi(\log R_v)\} (1 - \beta)^{k-1} \\
    &\hspace{2cm} \times \sum_{k=1}^{d_u} \sum_{h=1}^{w_k} \frac{R_{v_k}}{1 + R_{v_k}} \frac{R_{v_{kh}}}{1 + R_{v_{kh}}} \frac{\psi(\log R_u)}{\psi(\log R_{v_{kh}})} \\
    =& \max_{v \in L_u(2k)} \{\psi(\log R_v)\} (1 - \beta)^{k-1} \\
    &\hspace{2cm} \times \sum_{k=1}^{d_u} \sum_{h=1}^{w_k} \frac{\beta \prod_{j=1}^{w_k} (1 + R_{v_{kj}})^{-1}}{1 + \beta \prod_{j=1}^{w_k} (1 + R_{v_{kj}})^{-1}} \frac{R_{v_{kh}}}{1 + R_{v_{kh}}} \frac{\psi(\log R_u)}{\psi(\log R_{v_{kh}})} \\
    (\text{\small $\beta$-contract}) 
    \leq& \max_{v \in L_u(2k)} \{\psi(\log R_v)\} (1 - \beta)^k,
  \end{align*}
  where we note that when $w_k > 0$, we have $R_{v_{kj}} \in [\lambda(1 + \alpha)^{-d}, \lambda]$ for any $1 \leq j \leq w_k$, where we recall that $d = \Delta - 1$ is the maximum branching number.
  This meets the regime of $\*x$ in both contraction and boundedness condition.
  This proves that \eqref{eq:contract}, which means  \eqref{eq:exp-decay} holds for non-root vertices.
  
  For the root $r$, \eqref{eq:exp-decay} can be proved similarly.
  Suppose $r$ have $d_r \leq \Delta$ child $\{v_1, \cdots, v_{d_r}\}$ and each $v_i$ has $w_i$ child $\{v_{i1}, \cdots, v_{iw_i}\}$,
  \begin{align*}
    &\sum_{v \in L_r(2k)} \abs{\Psi_\mu(r, v)} \\
    =& \sum_{k=1}^{d_r} \sum_{h=1}^{w_k} \abs{\Psi_\mu(r, v_{kh})} \cdot \frac{1}{\psi(\log R_{v_{kh}})} \cdot \psi(\log R_{v_{kh}}) \sum_{v \in L_{v_{kh}}(2(k-1))} \abs{\Psi_\mu(v_{kh}, v)}\\
    (\text{by \eqref{eq:contract}}) \leq& (1 - \beta)^{k-1} \cdot \max_{v \in L_r(2k)} \left\{ \frac{\psi(\log R_v)}{\psi(\log R_r)} \right\} \sum_{k=1}^{d_r} \sum_{h=1}^{w_k} \abs{\Psi_\mu(r, v_{kh})} \cdot \frac{\psi(\log R_r)}{\psi(\log R_{v_{kh}})}\\
    \leq& c \cdot (1 - \beta)^k,
  \end{align*}
  where the last inequality is guaranteed by the boundedness assumption.
\end{proof}

\end{document}